\documentclass[prx,a4paper,aps,twocolumn,superscriptaddress,nofootinbib,10pt]{revtex4-1}

\usepackage{amsmath,amsthm,amsfonts,amssymb,mathrsfs,mathdots,graphicx,xcolor,times,xfrac,mathtools,enumitem,xr,subfigure,bbm,verbatim,comment,dsfont,soul,array,multirow}
\usepackage{comment}
\usepackage{color}
\definecolor{darkgreen}{RGB}{50,190,50}
\definecolor{darkblue}{RGB}{0,0,190}
\definecolor{darkred}{RGB}{238,0,0}
\definecolor{quantum}{RGB}{83,37,127}
\definecolor{quantumlight}{RGB}{169,146,191}
\definecolor{nice}{RGB}{230,0,230}
\definecolor{nicepink}{rgb}{0.858, 0.188, 0.478}
\usepackage[unicode=true,bookmarks=true,bookmarksnumbered=false,bookmarksopen=false,breaklinks=false,pdfborder={0 0 1}, backref=false,colorlinks=true]{hyperref}
\hypersetup{
	bookmarksnumbered,
	pdfstartview={FitH},
	citecolor={darkgreen},
    linkcolor={darkred},
    linktoc={page},
	urlcolor={darkblue},
	pdfpagemode={UseOutlines}}
\usepackage{orcidlink}

\usepackage{tikz}
\usetikzlibrary{shapes,backgrounds,fit,decorations.pathreplacing,arrows,decorations.markings}
\usepackage{verbatim}
\tikzstyle{vecArrow} = [thick, decoration={markings,mark=at position
   1 with {\arrow[semithick]{open triangle 60}}},
   double distance=1.4pt, shorten >= 5.5pt,
   preaction = {decorate},
   postaction = {draw,line width=1.4pt, white,shorten >= 4.5pt}]
\tikzstyle{innerWhite} = [semithick, white,line width=1.4pt, shorten >= 4.5pt]

\hypersetup{
	bookmarksnumbered,
	pdfstartview={FitH},
	citecolor={darkgreen},
	linkcolor={darkred},
	urlcolor={darkblue},
	pdfpagemode={UseOutlines}}
\definecolor{darkgreen}{RGB}{50,190,50}
\definecolor{darkblue}{RGB}{0,0,190}
\definecolor{darkred}{RGB}{238,0,0}
\definecolor{quantum}{RGB}{83,37,127}
\definecolor{quantumlight}{RGB}{169,146,191}
\definecolor{darkorange}{RGB}{255,100,0}
\usepackage{soul}

\usepackage{pifont}

\usepackage{array}

\newcommand{\nl}{\ensuremath{\hspace*{-0.5pt}}}
\newcommand{\nr}{\ensuremath{\hspace*{0.5pt}}}

\newcommand{\subtiny}[3]{\ensuremath{_{\hspace{#1 pt}\protect\raisebox{#2 pt}{\tiny{$ #3$}}}}}
\newcommand{\suptiny}[3]{\ensuremath{^{\hspace{#1 pt}\protect\raisebox{#2 pt}{\tiny{$ #3$}}}}}

\makeatletter
\theoremstyle{plain}
\newtheorem{thm}{\protect\theoremname}
\theoremstyle{plain}
\newtheorem{lem}{\protect\lemmaname}
\theoremstyle{plain}

\theoremstyle{remark}
\newtheorem*{rem*}{\protect\remarkname}
\theoremstyle{plain}

\theoremstyle{plain}
\newtheorem{cor}{\protect\corollaryname}
\theoremstyle{definition}
\newtheorem{defn}{\protect\definitionname}
\theoremstyle{plain}
\newtheorem*{thm*}{\protect\theoremname}
\theoremstyle{plain}
\newtheorem*{lem*}{\protect\lemmaname}
\theoremstyle{plain}

\providecommand{\propositionname}{Proposition}
\providecommand{\theoremname}{Theorem}
\providecommand{\lemmaname}{Lemma}
\providecommand{\remarkname}{Remark}
\providecommand{\conjecturename}{Conjecture}
\providecommand{\definitionname}{Definition}
\providecommand{\corollaryname}{Corollary}

\allowdisplaybreaks

\def\bra#1{\langle{#1}\vert}
\def\ket#1{\vert{#1}\rangle}
\def\braket#1{\langle{#1}\rangle}

\def\BraVert{e.g.,roup\,\mid\,\bgroup}

\newcommand{\ketbra}[2]{\ensuremath{|{#1}\rangle\!\langle{#2}|}}

\newcommand{\Acal}{\mathcal{A}}
\newcommand{\Bcal}{\mathcal{B}}
\newcommand{\Ccal}{\mathcal{C}}

\newcommand{\Hcal}{\mathcal{H}}

\newcommand{\Mcal}{\mathcal{M}}

\newcommand{\Scal}{\mathcal{S}}

\newcommand{\Xcal}{\mathcal{X}}
\newcommand{\Wcal}{\mathcal{W}}
\newcommand{\Zcal}{\mathcal{Z}}

\def\tr#1{\mbox{tr}\left[{#1}\right]}
\newcommand{\ptr}[2]{\mbox{tr}_{\raisebox{-1pt}{\tiny{$#1$}}}\left[ #2 \right]}

\DeclareMathOperator{\diag}{diag}

\let\oldaddcontentsline\addcontentsline
\newcommand{\stoptocentries}{\renewcommand{\addcontentsline}[3]{}}
\newcommand{\starttocentries}{\let\addcontentsline\oldaddcontentsline}

\begin{document}
\stoptocentries

\title{Landauer vs. Nernst: What is the True Cost of Cooling a Quantum System?}

\author{Philip Taranto\,\orcidlink{0000-0002-4247-3901}}
\email{philipguy.taranto@phys.s.u-tokyo.ac.jp}
\thanks{P. T. and F. B. contributed equally.}
\affiliation{Department of Physics, Graduate School of Science, The University of Tokyo, 7-3-1 Hongo, Bunkyo City, Tokyo 113-0033, Japan}
\affiliation{Atominstitut, Technische Universit{\"a}t Wien, 1020 Vienna, Austria}
\affiliation{Institute for Quantum Optics and Quantum Information - IQOQI Vienna, Austrian Academy of Sciences, Boltzmanngasse 3, 1090 Vienna, Austria}

\author{Faraj Bakhshinezhad\,\orcidlink{0000-0002-0088-0672}}
\thanks{P. T. and F. B. contributed equally.}
\affiliation{Atominstitut, Technische Universit{\"a}t Wien, 1020 Vienna, Austria}
\affiliation{Department of Physics and Nanolund, Lund University, Box 118, 221 00 Lund, Sweden}
\affiliation{Institute for Quantum Optics and Quantum Information - IQOQI Vienna, Austrian Academy of Sciences, Boltzmanngasse 3, 1090 Vienna, Austria}

\author{Andreas Bluhm\,\orcidlink{0000-0003-4796-7633}}
\thanks{A. B. and R. S. contributed equally.}
\affiliation{Univ.\ Grenoble Alpes, CNRS, Grenoble INP, LIG, 38000 Grenoble, France}
\affiliation{QMATH, Department of Mathematical Sciences, University of Copenhagen, Universitetsparken 5, 2100 Copenhagen, Denmark}

\author{Ralph Silva\,\orcidlink{0000-0002-4603-747X}}
\thanks{A. B. and R. S. contributed equally.}
\affiliation{Institute for Theoretical Physics, ETH Z\"urich, Wolfgang-Pauli-Str. 27, Z\"urich, Switzerland}

\author{Nicolai Friis\,\orcidlink{0000-0003-1950-8640}}
\affiliation{Atominstitut, Technische Universit{\"a}t Wien, 1020 Vienna, Austria}
\affiliation{Institute for Quantum Optics and Quantum Information - IQOQI Vienna, Austrian Academy of Sciences, Boltzmanngasse 3, 1090 Vienna, Austria}

\author{Maximilian P.~E. Lock\,\orcidlink{0000-0002-8241-8202}}
\affiliation{Atominstitut, Technische Universit{\"a}t Wien, 1020 Vienna, Austria}
\affiliation{Institute for Quantum Optics and Quantum Information - IQOQI Vienna, Austrian Academy of Sciences, Boltzmanngasse 3, 1090 Vienna, Austria}

\author{Giuseppe Vitagliano\,\orcidlink{0000-0002-5563-3222}}
\affiliation{Atominstitut, Technische Universit{\"a}t Wien, 1020 Vienna, Austria}
\affiliation{Institute for Quantum Optics and Quantum Information - IQOQI Vienna, Austrian Academy of Sciences, Boltzmanngasse 3, 1090 Vienna, Austria}

\author{Felix C. Binder\,\orcidlink{0000-0003-4483-5643}}
\affiliation{School of Physics, Trinity College Dublin, Dublin 2, Ireland}
\affiliation{Institute for Quantum Optics and Quantum Information - IQOQI Vienna, Austrian Academy of Sciences, Boltzmanngasse 3, 1090 Vienna, Austria}
\affiliation{Atominstitut, Technische Universit{\"a}t Wien, 1020 Vienna, Austria}

\author{Tiago Debarba\,\orcidlink{0000-0001-6411-3723}}
\affiliation{Departamento Acad{\^ e}mico de Ci{\^ e}ncias da Natureza, Universidade Tecnol{\'o}gica Federal do Paran{\'a} (UTFPR), Campus Corn{\'e}lio Proc{\'o}pio, Avenida Alberto Carazzai 1640, Corn{\'e}lio Proc{\'o}pio, Paran{\'a} 86300-000, Brazil}

\author{Emanuel Schwarzhans\,\orcidlink{0000-0001-8259-9720}} 
\affiliation{Atominstitut, Technische Universit{\"a}t Wien, 1020 Vienna, Austria}
\affiliation{Institute for Quantum Optics and Quantum Information - IQOQI Vienna, Austrian Academy of Sciences, Boltzmanngasse 3, 1090 Vienna, Austria}

\author{Fabien Clivaz\,\orcidlink{0000-0003-0694-8575}} 
\affiliation{Institut f{\"u}r Theoretische Physik und IQST, Universit{\"a}t Ulm, Albert-Einstein-Allee 11, D-89069 Ulm, Germany}
\affiliation{Institute for Quantum Optics and Quantum Information - IQOQI Vienna, Austrian Academy of Sciences, Boltzmanngasse 3, 1090 Vienna, Austria}

\author{Marcus Huber\,\orcidlink{0000-0003-1985-4623}}
\email{marcus.huber@tuwien.ac.at} 
\affiliation{Atominstitut, Technische Universit{\"a}t Wien, 1020 Vienna, Austria}
\affiliation{Institute for Quantum Optics and Quantum Information - IQOQI Vienna, Austrian Academy of Sciences, Boltzmanngasse 3, 1090 Vienna, Austria}

\date{\today}

\begin{abstract}
Thermodynamics connects our knowledge of the world to our capability to manipulate and thus to control it. This crucial role of control is exemplified by the third law of thermodynamics, Nernst's unattainability principle, which states that infinite resources are required to cool a system to absolute zero temperature. But what are these resources and how should they be utilised? And how does this relate to Landauer's principle that famously connects information and thermodynamics? We answer these questions by providing a framework for identifying the resources that enable the creation of pure quantum states. We show that perfect cooling is possible with Landauer energy cost given infinite time or control complexity. However, such optimal protocols require complex unitaries generated by an external work source. Restricting to unitaries that can be run solely via a heat engine, we derive a novel Carnot-Landauer limit, along with protocols for its saturation. This generalises Landauer's principle to a fully thermodynamic setting, leading to a unification with the third law and emphasises the importance of control in quantum thermodynamics.
\end{abstract}

\maketitle


\pdfbookmark[1]{Introduction}{Introduction}
\label{sec:introduction}
\section{Introduction}

\emph{What is the cost of creating a pure state?} Pure states appear as ubiquitous idealisations in quantum information processing and preparing them with high fidelity is essential for quantum technologies such as reliable quantum communication~\cite{GisinRibordyTittelZbinden2002, PirandolaEtAl2020}, high-precision quantum parameter estimation~\cite{GiovannettiLloydMaccone2011, TothApellaniz2014, DemkowiczDobrzanskiJarzynaKolodynski2015}, and fault-tolerant quantum computation~\cite{Preskill1997, Preskill2018}. Fundamentally, pure states are prerequisites for ideal measurements~\cite{Guryanova2020} and precise timekeeping~\cite{Erker_2017,SchwarzhansLockErkerFriisHuber2021}. To answer the above question, one could turn to Landauer's principle, stating that erasing a bit of information has an \emph{energy} cost of at least $k_B T \log(2)$~\cite{Landauer_1961}. Alternatively, one could consult Nernst's unattainability principle (the third law of thermodynamics)~\cite{Nernst_1906}, stating that cooling a physical system to its ground state requires diverging resources. At the outset, it seems that these statements are at odds with one another. However, Landauer's protocol requires infinite time, thus identifying \emph{time} as a resource according to the third law~\cite{Ticozzi_2014,Masanes_2017,Wilming_2017,Freitas_2018,Scharlau_2018}. \emph{Does this mean either infinite energy or time are needed to prepare a pure state?} 

The perhaps surprising answer we give here is: \emph{no}. We show that finite energy and time suffice to perfectly cool any quantum system and we identify the previously hidden resource\textemdash \emph{control complexity}\textemdash that must diverge (in the spirit of Nernst's principle) to do so. Intuitively, the control complexity of a protocol refers to the structure of machine energy gaps that the cooling unitary must couple the system to; we demonstrate that this energy-level spectrum must approximate a continuum in order to cool with minimal time and energy costs. In short, the ultimate limit on the energetic cost of cooling is still provided by the Landauer limit, but in order to achieve it, either time or control complexity must diverge. 

At the same time, heat fluctuations and short coherence times in quantum technologies~\cite{Acin_2018} demand that both energy and time are not only finite, but minimal. Therefore, in addition to proving the necessity of diverging control complexity for perfect cooling with minimal time and energy, we develop explicit protocols that saturate the ultimate limits. We demonstrate that mitigating overall heat dissipation comes at the practical cost of controlling fine-tuned interactions that require a \emph{coherent} external work source, i.e., a quantum battery~\cite{Aberg2013,Skrzypczyk_2014,LostaglioJenningsRudolph2015,Friis2018,CampaioliPollockVinjanampathy2019}. From a thermodynamic perspective, this may seem somewhat unsatisfactory: nonequilibrium resources imply that the total system is not closed, and the optimal protocol (saturating the Landauer bound) is reminiscent of a Maxwellian demon with perfect control. 

\begin{figure*}[ht!]
  \centering
  \includegraphics[width=\textwidth]{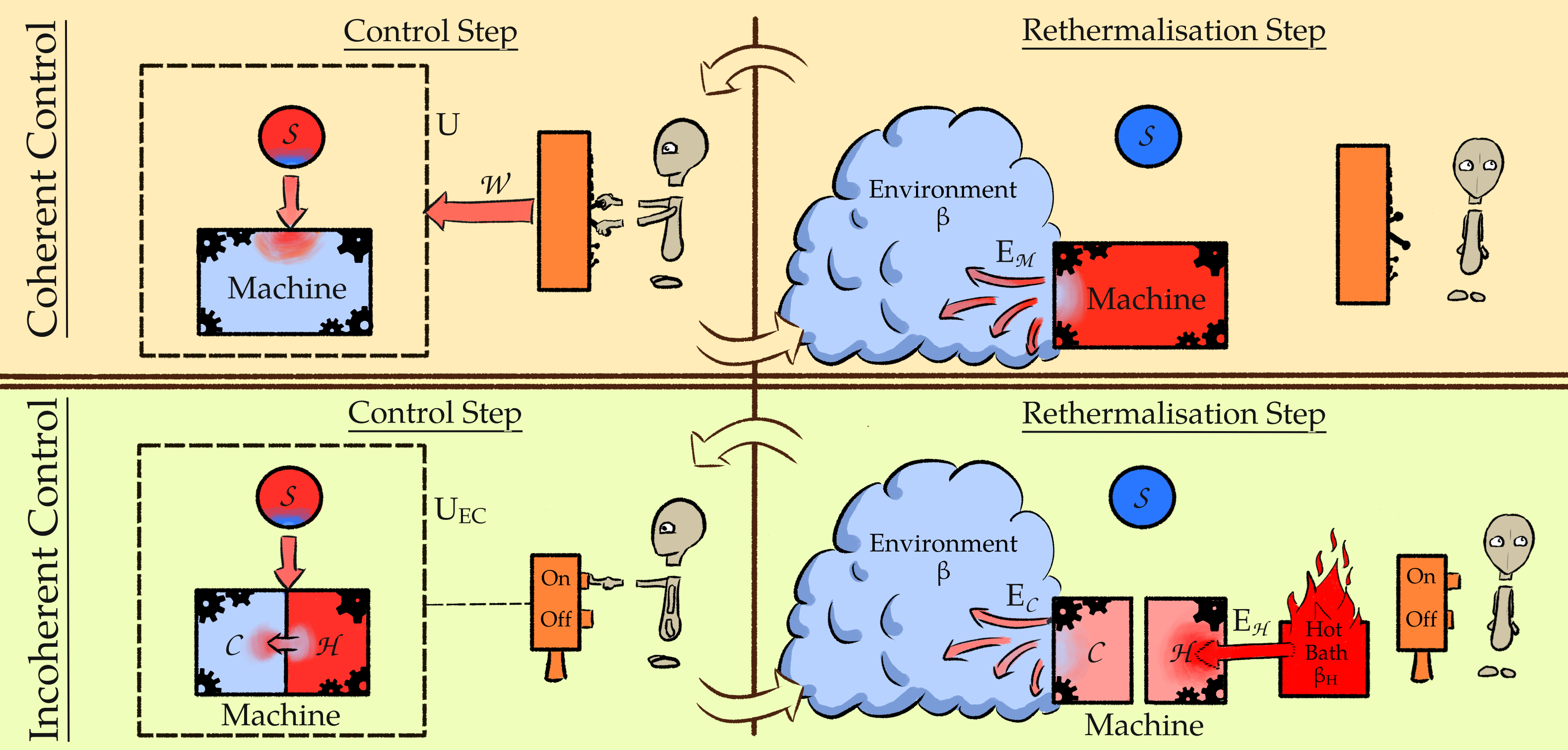}
  \vspace*{-1.5mm}
  \caption{\emph{Framework.} We consider the task of cooling a quantum system in two extremal control scenarios, with each step of both paradigms comprising two primitives. The top panel depicts the coherent-control scenario: in the control step (left), an agent can use a work source $\Wcal$ to implement any global unitary on the system $\Scal$ and machine $\Mcal$, which both begin thermal at inverse temperature $\beta$; in cooling the target, energy and entropy is transferred to the machine. The machine then rethermalises with its environment (right), thereby dissipating the energy it gained in the control step. The bottom panel depicts the incoherent-control scenario: the machine is bipartitioned into a cold part at inverse temperature $\beta$ and a hot part at inverse temperature $\beta_{\raisebox{-1pt}{\tiny{$H$}}} < \beta$. In the control step, the agent switches on an interaction between the three systems, represented by a global energy-conserving unitary $U_{\textup{EC}}$. In the rethermalisation step, the interaction is turned off and both subsystems of the machine rethermalise to their respective initial temperatures; the hot part draws energy from the heat bath while the cold part dissipates heat to its environment. In both paradigms, we quantify the control complexity as the effective dimension accessed by the unitary operation in a given control step (i.e., the dimension of the system-machine Hilbert space upon which the unitary acts nontrivially).} 
  \label{fig:schematic}
\end{figure*}

Accordingly, we also consider an \emph{incoherent} control setting restricted to global energy-conserving unitaries with a heat bath as thermodynamic energy source. This setting corresponds to minimal overall control, where interactions need only be switched on and off to generate transformations, i.e., a heat engine alone drives the dynamics~\cite{Scovil1959,Kosloff2014,Uzdin2015,Mitchison2019,Woods2019maximumefficiencyof}. The incoherent-control setting is therefore fully thermodynamically consistent inasmuch as both the machine state is assumed to be thermal (and to rethermalize between control steps) \emph{and} the permitted control operations are those implementable solely via a heat engine. In this paradigm, we show that the Landauer bound is not attainable, subsequently derive a novel limit---which we dub the \emph{Carnot-Landauer} bound---and construct protocols that saturate it, thereby establishing its significance. The Carnot-Landauer bound follows from an equality phrased in terms of entropic and energetic quantities that must hold for any state transformation in the incoherent control paradigm; in this sense, the Carnot-Landauer equality adapts the equality version of Landauer's principle developed in Ref.~\cite{Reeb_2014} to a fully (quantum) thermodynamic setting.

Our work thus both generalises Landauer’s erasure principle and, at the same time, unifies it with the laws of thermodynamics. By accounting for control complexity, we emphasise a crucial resource that is oftentimes overlooked but, as we show, must be taken into account for any operationally meaningful theory of thermodynamics. Here, we focus on the asymptotic setting that allows us to connect this resource with Nernst's unattainability principle. Beyond the asymptotic case, the gained insights also open the door to a better understanding of the intricate relationship between energy, time, and control complexity when all resources are finite, which will be crucial for practical applications; we additionally provide a preliminary analysis to this end. Lastly, our protocols saturating the Carnot-Landauer bound pave the way for thermodynamically driven (i.e., minimal-control) quantum technologies, which, by mitigating the cost of control at the very outset, could lead to tangible advantages. 


\pdfbookmark[2]{Overview and Summary of Results}{Overview and Summary of Results}
\subsection*{Overview \& Summary of Results}
\label{sec:results}

Loosely speaking, there are two types of thermodynamic laws: those, like the second law, that bound (changes of) characteristic quantities during thermodynamic processes, and those, like the third law, which state the impossibility of certain tasks. Landauer's principle is of the former kind (indeed, it can be rephrased as a version of the second law), associating a minimal heat dissipation to any logically irreversible process, thereby placing a fundamental limit on the energy cost of computation. The paradigmatic logically irreversible process is that of erasing information, i.e., resetting an arbitrary state to a blank register. From a physics perspective, said task can be rephrased as \emph{perfectly cooling} a system to the ground state, or more generally, taking an initially full-rank state to a rank-deficient one.\footnote{Low-temperature thermal states correspond to those with low information content, as they have low entropy or small effective support; viewing cooling more broadly (i.e., not restricting to thermal states and allowing for arbitrary Hamiltonians), we see that cooling indeed encompasses information erasure: States with smaller effective support are ``colder'' than those with greater support according to any meaningful notion of ``cool'' (see Ref.~\cite{Clivaz2020Thesis}).} Note that although there is, in general, a distinction between physical cooling and information erasure, in this paper we focus on erasing quantum information encoded in fundamental degrees of freedom rather than in logical macrostate sectors, and accordingly use the terms somewhat interchangeably. This is justified because in either case, the ultimate limitation (be it cooling to absolute zero or perfectly erasing information) requires a rank-decreasing process, which is what we formally analyse.

Nernst's unattainability principle is of the latter kind of thermodynamic law, stating that perfectly cooling a system requires diverging resources. The resources typically considered are energy and time, whose asymptotic trade-off relation is relatively well established: on the one hand, perfect cooling can be achieved in finite time at the expense of an energy cost that diverges as the ground state is approached; on the other hand, the energy cost can be minimised by implementing a quasistatic process that saturates the Landauer limit but takes infinitely long.\footnote{Note, however, that although the asymptotic trade-off relationship is known, the connection between energy and time in the finite-resource setting remains unresolved: For instance, if one uses twice the amount of energy, it is not clear how much faster a given protocol can be implemented; we provide some preliminary insight to such questions in Sec.~\ref{subsec:imperfectcooling}.} 

These two types of thermodynamic laws are intimately related, but details of their interplay have remained elusive: under which conditions can the Landauer bound be saturated and what are the minimal resources required to do so? Which protocols asymptotically create pure states with given (diverging) resources? What type of control do such protocols require and how difficult are they to implement in practice? We address these questions by considering the task of cooling a quantum system in two extremal control paradigms (see Fig.~\ref{fig:schematic}): One driven by a \emph{coherent} work source and the other by an \emph{incoherent} heat engine.

After laying out the framework, we proceed to analyse the relationship between the aforementioned three resources for cooling. A core idea of this paper originates from the observation that it is possible to perfectly cool a physical system with both finite energy and time. Although said observation is simple in nature inasmuch as it can be obtained by a shift in perspective of Landauer's original protocol, its consequences run deep: indeed, the apparent tension between Landauer cooling and Nernst's unattainability principle that arises when only energy and time are considered as resources is resolved via the inclusion of control complexity as a consideration. Subsequently, we define a meaningful notion of control complexity in terms of the energy-level structure of the machine that the system must be coupled to throughout the cooling protocol and demonstrate its thermodynamic consistency by showing that it indeed must diverge to cool the system to the ground state at minimal energy cost, thereby reconciling the viewpoints of Landauer and Nernst.

Having established the trinity of relevant resources, we present three main results: 
\begin{enumerate}
    \item Perfect cooling is possible with coherent control provided either energy, time, or control complexity diverge. In particular, it is possible in finite time and at Landauer energy cost with diverging control complexity.
    \item Perfect cooling is possible with incoherent control, i.e., with a heat engine, provided either time or control complexity diverge. On the other hand, it is impossible with both finite time and control complexity, regardless of the amount of energy drawn from the heat bath.
    \item No process driven by a finite-temperature heat engine can (perfectly) cool a quantum system at the Landauer limit. Nonetheless, the Carnot-Landauer limit, which we introduce here (as a consequence of a stronger equality), can be saturated for any heat bath, given either diverging time or control complexity.
\end{enumerate}

In the following, we discuss each of these results in turn in more detail and provide a systematic study concerning the asymptotic interplay of energy, time, and control complexity as thermodynamic resources in two extremal control paradigms, as well as develop insight into the finite-resource regime for some special cases. We begin by outlining the framework.

\pdfbookmark[2]{Framework: Cooling a Physical System}{Framework: Cooling a Physical System}
\section{Framework: Cooling a Physical System}
\label{subsec:framework}

Consider a target system $\Scal$ in an initial state $\varrho_{\raisebox{-1pt}{\tiny{$\Scal$}}}$ described by a unit-trace, positive semidefinite operator with associated Hamiltonian $H_{\raisebox{-1pt}{\tiny{$\Scal$}}}$. An auxiliary machine $\Mcal$, initially uncorrelated with $\Scal$ and in equilibrium with a reservoir at inverse temperature $\beta := \tfrac{1}{k_B T}$, is used to cool the target system. The initial state of $\Mcal$ is thus of Gibbs form, 
\begin{align}
    \varrho_{\raisebox{-1pt}{\tiny{$\Mcal$}}} = \tau_{\raisebox{-1pt}{\tiny{$\Mcal$}}}(\beta, H_{\raisebox{-1pt}{\tiny{$\Mcal$}}}) := \frac{e^{-\beta H_{\raisebox{-1pt}{\tiny{$\Mcal$}}}}}{\mathcal{Z}_{\raisebox{-1pt}{\tiny{$\Mcal$}}}(\beta, H_{\raisebox{-1pt}{\tiny{$\Mcal$}}})},
\end{align}
where $H_{\raisebox{-1pt}{\tiny{$\Mcal$}}}$ is the machine Hamiltonian and $\mathcal{Z}_{\raisebox{-1pt}{\tiny{$\Mcal$}}}(\beta, H_{\raisebox{-1pt}{\tiny{$\Mcal$}}}):=\tr{e^{-\beta H_{\raisebox{-1pt}{\tiny{$\Mcal$}}}}}$ its partition function. Throughout this paper we consider only Hamiltonians with discrete spectra, i.e., with an associated separable Hilbert space that has a countable energy eigenbasis. Moreover, for the most part we consider finite-dimensional systems (or sequences thereof) and deal with infinite-dimensional systems separately. 

As shown in Fig.~\ref{fig:schematic}, a single step of a cooling process comprises two subprocedures: first, a joint unitary is implemented during the \emph{control} step; second, the machine \emph{rethermalises} to the ambient temperature. A cooling \emph{protocol} is determined by the initial conditions and any concatenation of such primitives\footnote{One could refer to both $\Mcal$ \emph{and} the transformations applied as the \emph{machine} and call the system $\Mcal$ itself the working \emph{medium} inasmuch as the latter passively facilitates the process, in line with conventional parlance; however, we use the terminology established in the pertinent literature.}. We consider two extremal control paradigms corresponding to two classes of allowed global transformations. The \emph{coherent control} paradigm permits arbitrary unitaries on $\Scal\Mcal$; in general, these change the total energy but leave the global entropy invariant and thus require an external work source $\Wcal$. At the other extreme is the \emph{incoherent control} paradigm, where the energy source is a heat bath. Here, the machine $\Mcal$ is bipartitioned: one part, $\Ccal$, is connected to a \emph{cold} bath at inverse temperature $\beta$, which serves as a sink for all energy and entropy flows; the other, $\Hcal$, is connected to a \emph{hot} bath at inverse temperature $\beta_{\raisebox{-1pt}{\tiny{$H$}}} \leq \beta$, which provides energy. The composite system $\Scal\Ccal\Hcal$ is closed and thus global unitary transformations are restricted to be energy conserving. The temperature gradient causes a natural heat flow away from the hot bath, which carries maximal entropic change with it. Cooling protocols in this setting can be run with minimal external control, i.e., they require only the switching on and off of interactions.  

\pdfbookmark[1]{Coherent Control}{Coherent Control}
\section{Coherent Control }
\label{sec: coherentcontrol}

We begin by considering cooling with coherently controlled resources (see Fig.~\ref{fig:schematic}, top panel). We first analyse energy, time, and control complexity as resources that can be traded off against one another in order to optimise cooling performance, before focusing more specifically on the nature and role of control complexity. 


\pdfbookmark[2]{Energy, Time, and Control Complexity as Resources}{Energy, Time, and Control Complexity as Resources}
\subsection{Energy, Time, and Control Complexity as Resources}
\label{subsec:energytimecontrolcomplexity}

In the coherent-control setting, a transformation $\varrho_{\raisebox{-1pt}{\tiny{$\Scal$}}} \to \varrho^\prime_{\raisebox{-1pt}{\tiny{$\Scal$}}}$ is enacted via a unitary $U$ on $\Scal\Mcal$ involving a thermal machine $\varrho_{\raisebox{-1pt}{\tiny{$\Mcal$}}} = \tau_{\raisebox{-1pt}{\tiny{$\Mcal$}}}(\beta, H_{\raisebox{-1pt}{\tiny{$\Mcal$}}})$, i.e., 
\begin{align}
    \varrho_{\raisebox{-1pt}{\tiny{$\Scal$}}}^\prime := \ptr{\Mcal}{U (\varrho_{\raisebox{-1pt}{\tiny{$\Scal$}}} \otimes \varrho_{\raisebox{-1pt}{\tiny{$\Mcal$}}}) U^\dagger}.
\end{align}
For such a transformation, there are two energy costs contributing to the total energy change, which must be drawn from a work source $\Wcal$. The first is the energy change of the target $\Delta E_{\raisebox{-1pt}{\tiny{$\Scal$}}} := \tr{H_{\raisebox{-1pt}{\tiny{$\Scal$}}} (\varrho_{\raisebox{-1pt}{\tiny{$\Scal$}}}^\prime - \varrho_{\raisebox{-1pt}{\tiny{$\Scal$}}})}$; the second is that of the machine $\Delta E_{\raisebox{-1pt}{\tiny{$\Mcal$}}} := \tr{H_{\raisebox{-1pt}{\tiny{$\Mcal$}}} (\varrho_{\raisebox{-1pt}{\tiny{$\Mcal$}}}^\prime - \varrho_{\raisebox{-1pt}{\tiny{$\Mcal$}}})}$, where $\varrho^\prime_{\raisebox{-1pt}{\tiny{$\Mcal$}}} := \ptr{\Scal}{U (\varrho_{\raisebox{-1pt}{\tiny{$\Scal$}}} \otimes \varrho_{\raisebox{-1pt}{\tiny{$\Mcal$}}}) U^\dagger}$. The latter is associated with the heat dissipated into the environment and is given by~\cite{Reeb_2014} 
\begin{align}\label{eq:landauerequality}
    \beta \Delta E_{\raisebox{-1pt}{\tiny{$\Mcal$}}} = \widetilde{\Delta} S_{\raisebox{-1pt}{\tiny{$\Scal$}}} + I(\Scal: \Mcal)_{\varrho_{\raisebox{-1pt}{\tiny{$\Scal \Mcal$}}}^\prime} + D(\varrho_{\raisebox{-1pt}{\tiny{$\Mcal$}}}^\prime \| \varrho_{\raisebox{-1pt}{\tiny{$\Mcal$}}}),
\end{align}
where $S(\varrho) := - \tr{\varrho \log (\varrho)}$ is the von Neumann entropy, $\widetilde{\Delta} S_{\raisebox{-1pt}{\tiny{$\Acal$}}} := S(\varrho_{\raisebox{-1pt}{\tiny{$\Acal$}}}) - S(\varrho_{\raisebox{-1pt}{\tiny{$\Acal$}}}^\prime)$\footnote{Note the differing sign conventions (denoted by the tilde) that we use for changes in energies, $\Delta E_{\raisebox{-1pt}{\tiny{$\Xcal$}}} := E_{\raisebox{-1pt}{\tiny{$\Xcal$}}}^\prime - E_{\raisebox{-1pt}{\tiny{$\Xcal$}}}$, and in entropies, $\widetilde{\Delta} S_{\raisebox{-1pt}{\tiny{$\Xcal$}}} := S_{\raisebox{-1pt}{\tiny{$\Xcal$}}} - S_{\raisebox{-1pt}{\tiny{$\Xcal$}}}^\prime$, such that energy \emph{increases} and entropy \emph{decreases} are positive.}, $I(\Acal:\Bcal)_{\varrho_{\Acal \Bcal}} := S(\varrho_{\raisebox{-1pt}{\tiny{$\Acal$}}})+S(\varrho_{\raisebox{-1pt}{\tiny{$\Bcal$}}}) - S(\varrho_{\raisebox{-1pt}{\tiny{$\Acal \Bcal$}}})$ (with marginals $\varrho_{\raisebox{-1pt}{\tiny{$\Acal$/$\Bcal$}}} := \ptr{\Bcal/\Acal}{\varrho_{\raisebox{-1pt}{\tiny{$\Acal \Bcal$}}}}$) is the mutual information between $\Acal$ and $\Bcal$, and $D(\varrho\|\sigma) := \tr{\varrho \log(\varrho)} - \tr{\varrho \log(\sigma)}$ is the relative entropy of $\varrho$ with respect to $\sigma$, with $D(\varrho\|\sigma) := \infty$ if $\textup{supp}[\varrho] \nsubseteq \textup{supp}[\sigma]$. We derive Eq.~\eqref{eq:landauerequality} and its generalisation to the incoherent-control setting in Appendix~\ref{app:equalityformsofthecarnot-landauerlimit}. The mutual information is non-negative and vanishes iff $\varrho_{\raisebox{-1pt}{\tiny{$\Acal \Bcal$}}} = \varrho_{\raisebox{-1pt}{\tiny{$\Acal$}}} \otimes \varrho_{\raisebox{-1pt}{\tiny{$\Bcal$}}}$; similarly, the relative entropy is non-negative and vanishes iff $\varrho = \sigma$. Dropping these terms leads to the Landauer bound~\cite{Landauer_1961}
\begin{align}\label{eq:landauerlimit}
    \beta \Delta E_{\raisebox{-1pt}{\tiny{$\Mcal$}}}  \geq \widetilde{\Delta} S_{\raisebox{-1pt}{\tiny{$\Scal$}}}.
\end{align}

\setlength{\tabcolsep}{8pt}

\begin{table}[t!]
\begin{tabular}{llll}
& \textbf{Energy} & \textbf{Time} & \textbf{Complexity} \\\colrule
\parbox[t]{2mm}{\multirow{3}{*}{\rotatebox[origin=c]{90}{Qudit}}} & $\to\infty$ & $1$ & $\tfrac{1}{2} d(d-1)$ \\
& Landauer & $\to\infty$ & $\tfrac{1}{2} d(d-1)$ \\
& Landauer & $1$ & $\to\infty$\\\colrule
\parbox[t]{2mm}{\multirow{4}{*}{\rotatebox[origin=c]{90}{H. O.}}} &$\to\infty$ & $1$ & $ \to \infty$ (Gaussian) \\
&Landauer & $\to\infty$ & $ \to \infty$ (Gaussian) \\
&Finite ($>$ Landauer) & $\to\infty$ & 1 (Non-Gaussian) \\
&Landauer & $1$ & $\to\infty$ (Gaussian)\\\botrule
\end{tabular}
\caption{\emph{Coherent-control cooling protocols for finite-dimensional (qudit) and harmonic oscillator systems.} Landauer energy cost refers to saturation of Eq.~\eqref{eq:landauerlimit} and complexity refers to the proxy measure effective dimension (see Def.~\ref{def:effectivedimension}); time is measured as the number of unitary operations with a fixed complexity. In the qudit case, the system and machine dimensions are equal: $d_{\raisebox{-1pt}{\tiny{$\Scal$}}} = d_{\raisebox{-1pt}{\tiny{$\Mcal$}}} =: d$.}
\label{tab:coherentcontrol}
\end{table}

The Landauer limit holds \emph{independently} of the protocol implemented, i.e., it assumes only that \emph{some} unitary was applied to the target and thermal machine. For large machines, the dissipated heat is typically much greater than the energy change of the target; nonetheless, the contributions can be comparable at the microscopic scale. We assume that the target begins in equilibrium with the reservoir at inverse temperature $\beta$, i.e., in the initial thermal state $\varrho_{\raisebox{-1pt}{\tiny{$\Scal$}}} = \tau_{\raisebox{-1pt}{\tiny{$\Scal$}}}(\beta , H_{\raisebox{-1pt}{\tiny{$\Scal$}}})$, with no loss of generality since such a relaxation can be achieved for free (by swapping the target with a suitable part of the environment; however, see Ref.~\cite{Riechers2021} for a discussion of initial state dependency of the bound). We track all energetic and entropic quantities and refer to the asymptotic saturation of Eq.~\eqref{eq:landauerlimit} with $\varrho_{\raisebox{-1pt}{\tiny{$\Scal$}}}^\prime$ pure as \emph{perfect cooling at the Landauer limit}. 

Although Landauer's limit sets the minimum heat that must be dissipated---and thereby the minimum energy cost---for cooling any physical system, the third law makes no specification that energy must be the resource minimised (or that time must diverge). One might instead consider using a source of unbounded energy to perfectly cool a system as quickly as possible. Additionally, control complexity plays an important role as a resource, inasmuch as its divergence permits perfect cooling at the Landauer limit in finite time (see below). As summarised in Table~\ref{tab:coherentcontrol}, we now present coherently controlled protocols that perfectly cool an arbitrary finite-dimensional target system using thermal machines when any one of the three considered resources\textemdash energy, time or control complexity\textemdash diverges; moreover, the resources that are kept finite saturate protocol-independent ultimate bounds. The following thus provides a comprehensive analysis of cooling with respect to the trinity of resources that can be traded off amongst each other.

\subsection{Perfect Cooling at the Ultimate Limits with Infinite Resources}
\label{subsec:methods-coherentcontrol}

\emph{1. Diverging Energy.---}We first consider the situation in which time and control complexity are fixed to be finite, while the energy cost is allowed to diverge. Here, we present the following:
\begin{thm}\label{thm:divergingenergycoherent}
With diverging energy, any finite-dimensional quantum system can be perfectly cooled using a single interaction of finite complexity.
\end{thm}
\noindent 
The cooling protocol using diverging energy is the simplest. Here, one exchanges all populations of the target system with those of a thermal machine with suitably large energy gaps to sufficiently concentrate the initial machine population in the ground state subspace of the target system. This exchange requires a single system-machine unitary and is of finite complexity (in a sense discussed below). Nonetheless, the energy drawn from the work source in this protocol diverges. Moreover, in addition to being sufficient for perfect cooling with both finite time and control complexity, any protocol that cools perfectly with both finite time and control complexity requires diverging energy. See Appendix~\ref{app:divergingenergy} for details.

We now move to consider the situations in which the energy cost is minimised at the expense of either diverging time or control complexity. Equation~\eqref{eq:landauerequality} provides insight for understanding the conditions required for saturating the Landauer bound. Although for finite-dimensional machines only trivial processes of the form $U_{\raisebox{-1pt}{\tiny{$\Scal \Mcal$}}} = U_{\raisebox{-1pt}{\tiny{$\Scal$}}} \otimes \mathbbm{1}_{\raisebox{-1pt}{\tiny{$\Mcal$}}}$ saturate the Landauer limit~\cite{Reeb_2014}, we show how it can be asymptotically saturated with nontrivial processes by considering diverging machine and interaction properties, as we elaborate on shortly. Any such process must asymptotically exhibit no correlations such that $I(\Scal:\Mcal)_{\varrho_{\raisebox{-1pt}{\tiny{$\Scal \Mcal$}}}^\prime} \to 0$ and effectively not disturb the machine, i.e., yield $\varrho_{\raisebox{-1pt}{\tiny{$\Mcal$}}}^\prime \to \varrho_{\raisebox{-1pt}{\tiny{$\Mcal$}}}$ such that $D(\varrho_{\raisebox{-1pt}{\tiny{$\Mcal$}}}^\prime \| \varrho_{\raisebox{-1pt}{\tiny{$\Mcal$}}}) \to 0$. Indeed, any correlations created between initially thermal systems would come at the expense of an additional energetic cost~\cite{HuberPerarnauHovhannisyanSkrzypczykKloecklBrunnerAcin2015, BruschiPerarnauLlobetFriisHovhannisyanHuber2015,VitaglianoKloecklHuberFriis2019} whose minimisation is a problem that has so far only been partially resolved~\cite{BakhshinezhadEtAl2019}. However, it has been shown that for any (strictly) rank nondecreasing process, there exists a thermal machine and joint unitary such that for any $\epsilon > 0$, the heat dissipated satisfies $\beta \Delta E_{\raisebox{-1pt}{\tiny{$\Mcal$}}} \leq \widetilde{\Delta} S_{\raisebox{-1pt}{\tiny{$\Scal$}}} + \epsilon$~\cite{Reeb_2014}, thereby saturating the Landauer limit. Here, we present protocols that asymptotically achieve both this and perfect cooling (in particular, effectively decrease the rank), and provide necessary conditions on the underlying resources required to do so.

\emph{2. Diverging Time.---}We now present a protocol that uses a diverging number of operations of finite complexity to asymptotically attain perfect cooling at the Landauer limit~\cite{Anders_2013,Reeb_2014,Skrzypczyk_2014}.

\begin{thm}\label{thm:inftimeFinTepFinDim}
With diverging time, any finite-dimensional quantum system can be perfectly cooled at the Landauer limit via interactions of finite complexity.
\end{thm}

\emph{Sketch of proof.\textemdash }We first show that any system can be cooled from $\varrho_{\raisebox{-1pt}{\tiny{$\Scal$}}} = \tau_{\raisebox{-1pt}{\tiny{$\Scal$}}}(\beta, H_{\raisebox{-1pt}{\tiny{$\Scal$}}})$ to $\tau_{\raisebox{-1pt}{\tiny{$\Scal$}}}(\beta^{*}, H_{\raisebox{-1pt}{\tiny{$\Scal$}}})$, with $\beta^* \geq \beta$, using only $\beta^{-1}\, \widetilde{\Delta} S_{\raisebox{-1pt}{\tiny{$\Scal$}}}$ units of energy. Our proof is constructive in the sense that we provide a protocol that achieves the Landauer energy cost as the number of operations diverges. The individual interactions in this protocol are of finite control complexity as they simply swap the target system with one of a sequence of thermal machines with increasing energy gaps. In this way, the final state $\tau_{\raisebox{-1pt}{\tiny{$\Scal$}}}(\beta^*, H_{\raisebox{-1pt}{\tiny{$\Scal$}}})$ can be made to be arbitrarily close to $\ket{0}\!\bra{0}_{\raisebox{-1pt}{\tiny{$\Scal$}}}$ for any initial temperature. \qed

The proof is presented in Appendix~\ref{app:divergingtimecoolingprotocolfinitedimensionalsystems}, along with a more detailed dimension-dependent energy cost function for the special case of equally spaced Hamiltonians.

Through the protocol described above, we see that given a diverging amount of time, the target system can be sequentially coupled with a machine of finite complexity that rethermalizes between control steps in such a way that the final target system state is arbitrarily close to the ground state for any initial temperature. This trade-off between energy and time is well known, and we discuss it only briefly in order to help build intuition and highlight the versatility of our framework. Alternatively, one can compress all the operations applied in the diverging-time protocol into one global unitary that achieves the same final states, thereby achieving perfect cooling at the Landauer limit in a single unit of time but with an infinitely complex interaction. That is, the diverging temporal resource of repeated interactions with a single, finite-size machine is replaced by a single interaction with a larger machine of diverging control complexity.

\emph{3. Diverging Control Complexity.---}By reconsidering the diverging-time protocol above, a trade-off can be made between time and control complexity. As illustrated in Fig.~\ref{fig:complexity}, one can consider all of the operations $\{ U_k=e^{-iH_k t_k} \}_{k=1,\dots,N}$ required in said protocol to make up one single joint interaction $U_{\textup{tot}}:=\lim_{N\to\infty}\prod_{k=1}^{N}U_k=e^{-iH_{\textup{tot}} t_{\textup{tot}}}$ acting on a larger machine, thus setting the time required to be unity (in terms of the number of control operations before the machine rethermalises). In other words, for any finite number $N$ of unitary transformations $U_k$, there exists a total Hamiltonian $H_{\textup{tot}}\suptiny{0}{0}{(N)}$ and a finite time $t\subtiny{0}{0}{N}$ that generates the overall transformation $U_{\textup{tot}}\suptiny{0}{0}{(N)} := \prod_{k=1}^{N}U_k$; since $t\subtiny{0}{0}{N}$ is finite, we can set it equal to one without loss of generality by rescaling the Hamiltonian as $\widetilde{H}_{\textup{tot}}\suptiny{0}{0}{(N)}= t\subtiny{0}{0}{N} H_{\textup{tot}}\suptiny{0}{0}{(N)}$. Here, we refer to the limit $N \to\infty$ as diverging control complexity. Compressing a diverging number of finite-complexity operations thus yields a protocol of diverging control complexity. The fact that there exists such an operation that minimises both the time and energy requirements follows from our constructive proof of Theorem~\ref{thm:inftimeFinTepFinDim}. We therefore have the following:

\begin{cor}\label{cor:infcomplexity}
With diverging control complexity, any finite-dimensional quantum system can be perfectly cooled at the Landauer limit in finite time.
\end{cor}

However, this particular way of constructing complex control protocols is not necessarily unique. It is thus natural to wonder if diverging control complexity is a generic feature necessary to achieve perfect cooling at the Landauer limit in unit time and indeed, how to quantify control complexity that is operationally meaningful between the extreme cases of being either very small or divergent, as we now turn to discuss. Indeed, the inclusion of an explicit quantifier of control complexity regarding thermodynamic tasks---which, although crucial for practical purposes, is oftentimes overlooked---is one of the main novelties of our present work.

\pdfbookmark[2]{Control Complexity in Quantum Thermodynamics}{Control Complexity in Quantum Thermodynamics}
\section{Control Complexity in Quantum Thermodynamics}

Although the protocol described above has diverging control complexity by construction, one need not construct complex protocols in this way, and so the natural concern becomes understanding the generic features that enable perfect cooling at the Landauer limit in unit time. To address this issue, we first provide protocol-independent structural conditions that must be fulfilled by the machine to enable \emph{(1) perfect cooling} and \emph{(2) cooling at Landauer cost}; combined, these independent conditions provide a necessary requirement, namely that the machine must have an unbounded spectrum (from above) and be infinite-dimensional (respectively) for the \emph{possibility} of \emph{(3) perfect cooling at the Landauer limit}. Such properties of the machine Hamiltonian define the \emph{structural complexity}, which sets the potential for how cool the target system can be made and at what energy cost. As the name suggests, this is entailed by the structure of the machine, e.g., the number of energy gaps and their arrangement, and as such provides a static notion of complexity. However, given a machine with particular structural complexity, one may not be able to utilise said potential due to constraints on the dynamics that can be implemented. For instance, one may be restricted to only two-body interactions, or operations involving only a few energy levels at a time. Assuming a sufficient structural complexity at hand, such constraints limit one from optimally manipulating the systems. Thus, the extent to which a machine's potential is utilised depends on properties of the dynamics of a given protocol, i.e., the \emph{control complexity}. We provide a detailed study of structural and control complexity in Appendix~\ref{app:conditionsstructuralcontrolcomplexity}, and here summarise the key methods.

\subsection{Structural \& Dynamical Notions of Complexity}

We split the consideration of complexity into two parts: first, the protocol-independent \emph{structural} conditions that must be fulfilled by the machine and, second, the dynamic \emph{control complexity} properties of the interaction that implements a given protocol (see Fig.~\ref{fig:complexity}). 


\begin{figure}[t!]
  \centering
  \includegraphics[width=\linewidth]{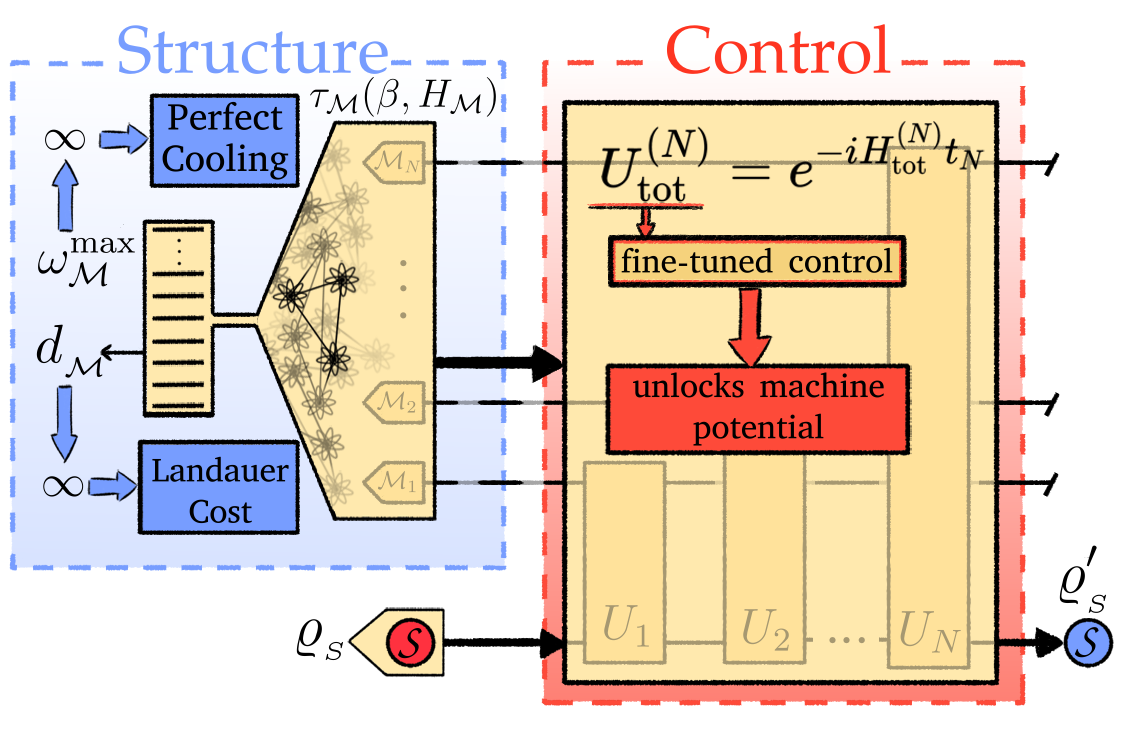}
   \caption{\emph{Complexity.} We consider structural (left) and control complexity (right). Structural complexity concerns properties of the machine Hamiltonian. For perfect cooling it is necessary that the largest energy gap diverges [see Eq.~\eqref{eq:maingeneralpuritybound}]. Moreover, an infinite-dimensional machine with particular energy-level structure is required for saturation of the Landauer bound. Control complexity refers to properties of the unitary that represents a protocol. The yellow box in the foreground represents a unitary $U$ involving the entire machine, whereas the smaller yellow columns in the background represent a potential decomposition (e.g., of the diverging-time protocol) into unitaries $U_{i}$ involving certain subspaces of the overall machine. Not only must the target system interact with all levels of an infinite-dimensional machine for Landauer-cost cooling, it must do so in a fine-tuned way.}
  \label{fig:complexity}
\end{figure}


\subsubsection{Structural Complexity}

Regarding the former, first note that one can lower bound the smallest eigenvalue $\lambda_{\textup{min}}$ of the final state $\varrho_{\raisebox{-1pt}{\tiny{$\Scal$}}}'$ (and hence how cold the system can become) after \emph{any} unitary interaction with a thermal machine by~\cite{Reeb_2014}
\begin{align}\label{eq:maingeneralpuritybound}
    \lambda_{\textup{min}}(\varrho_{\raisebox{-1pt}{\tiny{$\Scal$}}}^\prime) \geq e^{-\beta\, \omega_{\raisebox{0pt}{\tiny{$\Mcal$}}}^{\textup{max}}} \lambda_{\textup{min}}(\varrho_{\raisebox{-1pt}{\tiny{$\Scal$}}}),
\end{align}
where $\omega_{\raisebox{0pt}{\tiny{$\Mcal$}}}^{\textup{max}}:=\max_{i,j}|\omega_{j}-\omega_{i}|$ denotes the largest energy gap of the machine Hamiltonian $H_{\raisebox{-1pt}{\tiny{$\Mcal$}}}$ with eigenvalues $\omega_{i}$. It follows that perfect cooling is only possible under two conditions: either the machine begins in a pure state ($\beta\to\infty$), or $H_{\raisebox{-1pt}{\tiny{$\Mcal$}}}$ is unbounded, i.e., $\omega_{\raisebox{0pt}{\tiny{$\Mcal$}}}^{\textup{max}}\to\infty$. Requiring $\beta<\infty$, a diverging energy gap in the machine Hamiltonian is thus a necessary structural condition for perfect cooling. Independently, another condition required to saturate the Landauer limit can be derived for any amount of cooling: in Ref.~\cite{Reeb_2014}, it was shown that for any finite-dimensional machine, there are correction terms to the Landauer bound which imply that it cannot be saturated; these terms only vanish in the limit where the machine dimension diverges. 

We thus have two independent necessary conditions on the structure of the machine that must be asymptotically fulfilled to achieve relevant goals for cooling: the former is required for perfect cooling; the latter for cooling at the Landauer limit. Together, these conditions imply the following: 
\begin{cor}\label{cor:structuralcondition}
To perfectly cool a target system with energy cost at the Landauer limit using a thermal machine $\tau_{\raisebox{-1pt}{\tiny{$\Mcal$}}}(\beta, H_{\raisebox{-1pt}{\tiny{$\Mcal$}}})$, the machine must be infinite dimensional and $\omega_{\raisebox{0pt}{\tiny{$\Mcal$}}}^{\textup{max}}$, the maximal energy gap of $H_{\raisebox{-1pt}{\tiny{$\Mcal$}}}\,$, must diverge.
\end{cor}

The unbounded structural properties of the machine support the \emph{possibility} for perfect cooling at the Landauer limit; we now move to focus on the control properties of the interaction that \emph{realise} said potential (see Fig.~\ref{fig:complexity}). This leads to the distinct notion of control complexity, which differentiates between protocols that access the machine in a more or less complex manner. The structural complexity properties are protocol independent and related to the energy spectrum and dimensionality of the machine, whereas the control complexity concerns dynamical properties of the unitary that represents a particular protocol. 

\subsubsection{Control Complexity}

Although it is intuitive that a unitary coupling the system to many degrees of freedom of the machine should be considered complex, it is \emph{a priori} unclear how to quantify control complexity in a manner that both
\begin{enumerate}
    \item corresponds to our intuitive understanding of the word ``complex'', meaning ``difficult to implement''; and
    \item is consistent with Nernst's third law in the sense that its divergence is necessary to reach a pure state (when all other considered resources are restricted to be finite).
\end{enumerate}
Many notions of complexity put forth throughout the literature to capture the first point above do not necessarily satisfy the second, as we discuss later. Here, we take the opposite approach and seek a \emph{minimal} notion of complexity that is first and foremost consistent with the third law of thermodynamics, which we hope to develop further to incorporate the idea of quantifying how difficult a protocol is to implement. 

In the following sections, we begin by demonstrating that any cooling protocol that achieves perfect cooling with minimal time and energy resources requires coupling the target system to an infinite-dimensional machine, thereby capturing a notion of control complexity that satisfies the second point above. However, by subsequently analysing the sufficient conditions for such optimal cooling, we see that such a condition is in general insufficient to achieve said goal; furthermore, coupling to an infinite-dimensional machine is not necessarily difficult to implement in practice in certain experimental platforms. The insights gained here finally motivate our more refined notion of control complexity, namely that the system must be coupled to a spectrum of machine energy gaps that approximate a continuum. This condition is indeed difficult to achieve in all experimental settings and therefore provides a reasonable definition of control complexity inasmuch as it satisfies both desiderata outlined above.

\subsection{Effective Dimension as a Notion of Control Complexity}

As a first step in this direction, a good proxy measure of control complexity is the effective dimension of a unitary operation, i.e., the dimension of the subspace of the global Hilbert space upon which the unitary acts nontrivially.\\

\begin{defn}\label{def:effectivedimension}
The \emph{effective dimension} is the minimum dimension of a subspace $\mathcal{A}$ of the joint Hilbert space $\mathscr{H}_{\raisebox{-1pt}{\tiny{$\Scal\Mcal$}}}$ in terms of which the unitary can be decomposed as $U_{\raisebox{-1pt}{\tiny{$\Scal\Mcal$}}} = U_{\raisebox{-1pt}{\tiny{$\Acal$}}} \oplus \mathbbm{1}_{\raisebox{-1pt}{\tiny{$\Acal^\perp$}}}$:
\begin{align}\label{eq:maineffectivedimension}
    d^{\,\textup{eff}} := \min \mathrm{dim}(\mathcal{A}) : U_{\raisebox{-1pt}{\tiny{$\Scal\Mcal$}}} = U_{\raisebox{-1pt}{\tiny{$\Acal$}}} \oplus \mathbbm{1}_{\raisebox{-1pt}{\tiny{$\Acal^\perp$}}}.
\end{align}
\vspace{-2em}
\end{defn}
Intuitively, given any (sufficiently large) machine dimension, the effective dimension captures how much of the machine takes part in the controlled interaction. While any dynamics that requires a high amount of control must accordingly have large effective dimension, the converse does not necessarily follow: there exist dynamics with corresponding large (even infinite) effective dimensions (e.g., Gaussian operations on two harmonic oscillators, such as those enacted by a beam splitter) that are easily implementable and do not require high levels of control, as we discuss further below. Nevertheless, using the definition above, we show that any protocol achieving perfect cooling at the Landauer limit necessarily involves interactions between the target and infinitely many energy levels of the machine. In other words, no interaction restricted to a finite-dimensional subspace suffices. 

We begin by demonstrating that the effective dimension (nontrivially) accessed by a unitary (see Def.~\ref{def:effectivedimension}) must diverge to achieve perfect cooling at the Landauer limit, thereby providing a good proxy for control complexity in the sense that it aligns with Nernst's third law and provides a necessary condition. Intuitively, the effective dimension of a unitary operation is the dimension of the subspace of the global Hilbert space upon which the unitary acts nontrivially, in other words the part of the joint space that is actually accessed by the control protocol. This quantity can be computed by considering a given cooling protocol and finite unit of time $T$ (which we can set equal to unity without loss of generality) with respect to which the target and total machine transform unitarily by decomposing the Hamiltonian in $U_{\raisebox{-1pt}{\tiny{$\Scal\Mcal$}}} = e^{-i H_{\raisebox{-1pt}{\tiny{$\Scal\Mcal$}}} T}$ in terms of local and interaction terms, i.e., $H_{\raisebox{-1pt}{\tiny{$\Scal\Mcal$}}} = H_{\raisebox{-1pt}{\tiny{$\Scal$}}}\otimes \mathbbm{1}_{\raisebox{-1pt}{\tiny{$\Mcal$}}} + \mathbbm{1}_{\raisebox{-1pt}{\tiny{$\Scal$}}}\otimes H_{\raisebox{-1pt}{\tiny{$\Mcal$}}} + H_{\raisebox{-1pt}{\scriptsize{\textup{int}}}}$. The effective dimension then corresponds to $\mathrm{rank}(H_{\raisebox{-1pt}{\scriptsize{\textup{int}}}})$.
With this definition at hand, we have the following:

\begin{thm}\label{thm:variety}
The unitary representing a cooling protocol that saturates the Landauer limit must act nontrivially on an infinite-dimensional subspace of $\operatorname{supp}(H_{\raisebox{-1pt}{\tiny{$\Mcal$}}})$. This implies $d^{\,\textup{eff}} \to \infty$. 
\end{thm}
\noindent Intuitively, we show that if a protocol accesses only a finite-dimensional subspace of the machine, then the machine is effectively finite-dimensional inasmuch as a suitable replacement can be made while keeping all quantities relevant for cooling invariant. Invoking the main result of Ref.~\cite{Reeb_2014} then implies that there are finite-dimensional correction terms such that the Landauer limit cannot be saturated. 

The effective dimension therefore provides a minimal quantifier for control complexity: it is the quantity that must diverge in order to (perfectly) cool at minimal energy cost---thus, it satisfies the above point 2. Moreover, it requires no assumption on the underlying structure of the machine, with the results holding for either collections of finite-dimensional systems or harmonic oscillators without any restrictions on the types of individual operations allowed. This highlights a certain level of generality regarding the definition put forth, inasmuch as it is not tied to any presupposed structure of the systems at hand or the ability of the agent to control them. Additionally, as we discuss below, in many situations of interest, such as a machine comprising a collection of qubits and/or natural gate-set limitations, said definition also corresponds to protocols that are difficult to implement in practice, therefore also satisfying the above point~1. However, such additional restrictions are by no means generic. Moreover, it is \emph{a priori} unclear if having a diverging effective dimension is enough to permit perfect cooling with minimal time and energy cost. We now move on to discuss the connection to practical difficulty in general before analysing sufficient conditions regarding control complexity.

\subsubsection{Correspondence to Practical Difficulty}

Importantly, if one supposes that the system and machines are finite dimensional, then diverging effective dimension implies diverging circuit complexity, where the latter is defined in terms of the minimum number of gates (from a predetermined set of possibilities) required to implement the overall circuit representing a particular protocol. For instance, considering a qubit system and machines, and the ability to perform arbitrary two-qubit gates, the effective dimension is simply the logarithm of the number of distinct machine qubits that the system interacts with throughout the protocol. For any cooling protocol that achieves Landauer energy cost, it is clear that every one of a diverging number of qubit machines must take part in the overall transformation. Moreover, the particular interactions applied can be taken to be \texttt{SWAP} gates, which require the ability for the agent to be able to perform a \texttt{CNOT} gate, which in turn permits universal quantum computation with two-qubit interactions. Thus, given a universal two-qubit gate set, the circuit required to perform perfect cooling at minimal energy cost has a complexity that scales with the number of machine qubits. For higher-dimensional architectures or further restrictions on the gate set, any meaningful notion of control complexity will increase accordingly. This means that the task of cooling a finite-dimensional system with finite-dimensional machines at the Landauer limit is---even with a perfect quantum computer---an impossibly difficult task.

However, although our proposed definition of effective dimension as a notion of control complexity is flexible inasmuch as it applies to arbitrary system-machine structures, the price of such generality comes with the drawback that it tends to overestimate the difficulty of implementing a particular protocol in practice. In other words, without imposing any additional assumptions regarding the situation at hand, the effective dimension does not necessarily satisfy the above point~1. For example, whilst the effective dimension and the circuit complexity coincide for qubits, in higher-dimensional settings, the former overestimates the latter because not all system-machine subspaces are necessarily required to implement a particular protocol (i.e., although using all such subspaces provides one way to achieve it, this is not unique). Thus, the extent to which the circuit complexity is overestimated depends on the allowed gate set that is considered ``simple'' in general. At the extreme end, i.e., for harmonic-oscillator systems and machines, this can be seen from the fact that a single beam-splitter operation (which is a two-mode Gaussian operation, corresponding to a simple circuit complexity in the usual sense considered for infinite-dimensional quantum circuit architectures) already has infinite effective dimension, but is far from sufficient to achieve perfect cooling at Landauer cost.

As a representative for infinite-dimensional systems, we treat harmonic oscillator target systems separately in Appendix~\ref{app:harmonicoscillators}. In the infinite-dimensional setting, the difficulty of implementing an operation is often related to the polynomial degree of its generators. Here, we see some friction with respect to Eq.~\eqref{eq:maineffectivedimension}: as mentioned above, a generic Gaussian unitary operation (i.e., one generated by a Hamiltonian at most quadratic in the mode operators) between a harmonic oscillator target and machine already implies infinite effective dimensionality. In light of this, we first construct a protocol that achieves perfect cooling at the Landauer limit with diverging time using only sequences of Gaussian operations [i.e., those typically considered to be practically easily implementable (cf. Refs.~\cite{BrownFriisHuber2016,Friis2018}), but nonetheless with infinite effective dimensionality according to Def.~\ref{def:effectivedimension}]. This result highlights that the polynomial degree of the generators of a particular protocol would---somewhat counterintuitively, since operations corresponding to high polynomial degree are difficult to achieve in practice---\emph{not} provide a suitable measure of control complexity inasmuch as its divergence is not necessary for Landauer-cost cooling. In contrast, we then present a protocol that demonstrates that perfect cooling is possible given diverging time and operations acting on only a finite effective dimensionality (i.e., using non-Gaussian operations), with a finite energy cost that is greater than the Landauer limit; whether or not a similar protocol that saturates the Landauer limit exists in this setting remains an open question. 

\subsubsection{Sufficiency for Optimal Cooling}

Thus, in general, accessing an infinite-dimensional machine subspace is not sufficient for reaching the Landauer limit. Indeed, in all of the protocols that we present, the degrees of freedom of the machine must be individually addressed in a fine-tuned manner to permute populations optimally, which intuitively corresponds to complicated multipartite gates and demonstrates that an operationally meaningful notion of control complexity must take into account factors beyond the effective dimensionality accessed by an operation. In particular, the interactions couple the target system to a diverging number of subspaces of the machine corresponding to distinct energy gaps. Moreover, there are a diverging number of energy levels of the machine both above and below the first excited level of the target. These observations highlight that fine-tuned control plays an important role. Indeed, both the final temperature of the target as well as the energy cost required to achieve this depends upon how the global eigenvalues are permuted via the cooling process. First, how cool the target becomes depends on the sum of the eigenvalues that are placed into the subspace spanned by the ground state. Second, for any fixed amount of cooling, the energy cost depends on the constrained distribution of eigenvalues within the machine. Thus, in general, the optimal permutation of eigenvalues depends upon properties of both the target and machine. To highlight this, in Appendix~\ref{app:conditionsstructuralcontrolcomplexity}, we consider the task of cooling a maximally mixed target system with the additional constraint that the operation implemented lowers the temperature as much as possible. This allows us to derive a closed-form expression for the distribution of machine eigenvalues alone that must be asymptotically satisfied as the machine dimension diverges. Drawing from these insights, in the coming section we propose a stronger notion of control complexity (in the sense that it bounds the effective dimension from below and that it corresponds to practical difficulty in virtual every setting imaginable) in terms of the energy-gap structure of the machine and demonstrate that this measure too must diverge to cool perfectly with minimal time and energy costs. This concept is even more important in the case where all resources are finite, as particular structures of machines and the types of interactions permitted play a crucial role in both how much time or energy is spent cooling a system and how cold the system can ultimately become (see, e.g., Refs.~\cite{Clivaz_2019E,Taranto_2020,zhen2021}).

\subsection{Energy-Gap Variety as a Notion of Control Complexity}

This analysis motivates searching for a more detailed notion of control complexity that takes the energy-level structure of the machine into account, which should hold across all platforms and dimension scales. The discussion above illustrates some key challenges in defining a measure of control complexity that satisfies natural desiderata: such a measure should correspond to the difficulty of implementing operations in practice and simultaneously cover all possible physical platforms, including finite-dimensional systems such as, e.g., specific optical transitions of electrons in the shell of trapped ions, and infinite-dimensional systems such as the state-space-specific modes of the electromagnetic field. The effective dimension that we introduce above as a proxy manages to cover all such systems and provides a rigorous mathematical criterion that every physical protocol will necessarily have to fulfil in order to cool at minimal energy cost. As we have seen, however, infinite effective dimension is insufficient for cooling at the Landauer limit and it may not be all that difficult to achieve in continuous-variable setups. This begs the question of how this minimal definition of control complexity can be extended in order to more faithfully represent what permits saturation of the ultimate limitations and is difficult to achieve in practice. 

Looking at all of our cooling protocols, a common property that seems to be important in minimising the energy cost of cooling is that the system is coupled to a set of machine energy gaps that are distributed in such a way that they (approximately) densely cover the interval $[\omega_1, \omega^*]$, where $\omega_1$ is the first energy gap of the target system and $\omega^*$ is the maximal energy gap, which sets the final achievable temperature of the system (for perfect cooling to the ground state, note that one requires $\omega^* \to \infty$). Let us denote the number of distinct energy gaps in a (fixed) interval as the \emph{energy-gap variety}. More formally, we have the following:

\begin{defn}\label{def:energygapvariety}
Consider an interval $[\omega_a, \omega_b) \subseteq \mathbbm{R}$. We define the energy-gap variety in terms of the set of machine energy gaps that lie in said interval, i.e., first construct the set
\begin{align}
    \mathcal{E}_{[\omega_a, \omega_b)} := \{ \omega_\gamma := \omega_i-\omega_j \, | \, \omega_i-\omega_j \in [\omega_a, \omega_b) \}_{\gamma}.
\end{align}
The number of distinct elements in such a set is the \emph{energy-gap variety}. 
\end{defn}

On the one hand, it is clear that coupling a system to a large number of distinct and/or closely spaced energy gaps requires fine-tuned control that is difficult in any experimental setting. On the other, the energy-gap variety lower bounds the effective dimension, and thus it is not clear that it needs to diverge in order to cool at Landauer energy cost. In Appendix~\ref{app:conditionsstructuralcontrolcomplexity}, we demonstrate that the energy-gap variety must indeed diverge and, additionally, that the set of energy gaps must densely cover a relevant interval (whose endpoints set the amount of cooling possible) in order to perfectly cool at the Landauer limit by proving the following:
\begin{thm}\label{thm:energygapvariety-main}
    In order to cool $\varrho_{\raisebox{-1pt}{\tiny{$\Scal$}}} \mapsto \ketbra{0}{0}$ with a thermal machine $\tau_{\raisebox{0pt}{\tiny{$\Mcal$}}}(\beta,H_{\raisebox{0pt}{\tiny{$\Mcal$}}})$ at Landauer energy cost with a single control operation, the global unitary $U$ must couple the system to a diverging number of distinct energy gaps that densely cover the interval $[\omega_0,\infty)$, where $\omega_0$ is the smallest energy gap of the target system. 
\end{thm}
Taken in combination with its sufficiency to achieve said task, this result posits the energy-gap variety as a better quantifier of control complexity than the effective dimension, constituting the best thermodynamically meaningful notion of control complexity that we have put forth so far.

The above theorem establishes the relevance of the energy-gap variety regarding the ultimate limitations of perfect cooling. In reality, of course, experimental imperfections abound, and so naturally the question arises: \emph{how robust is the energy-gap variety and to what extent can it incorporate errors?} Regarding the former: note that the above theorem posits the impossibility of cooling at Landauer energy cost unless one has control over an (infinitely) fine-grained energy-gap structure. Any perturbation away from said structure will result in some additional energy requirement for cooling; however, intuitively, small perturbations will correspond to small increases in energy costs. Properly accounting for such impacts, e.g., by bounding the additional energy cost in terms of a difference from the optimal energy-gap structure, is an important next step to understand the practical limitations of cooling. Regarding the latter point, in reality one never has perfect control over microscopic degrees of freedom. For instance, an immediate experimental imperfection that should be accounted for is the fact that two energy gaps which are very close together will be practically indistinguishable. Although a full-fledged error analysis here would constitute a major follow-up work, note that such cases can be formally dealt with within our framework by suitably modifying the definition, i.e., by discretising energy bands to suitably capture the indistinguishability of energy gaps and/or error margins.

Aside from introducing and highlighting the important role of control complexity, we now take a step back to consider the notion of overall control at a higher level. It is clear that the protocols that saturate the Landauer limit for the energy cost of cooling require highly controlled microstate interactions between the system and machine; in turn, such transformations necessitate that the agent has access to a versatile \emph{work source}, i.e., either a quantum battery~\cite{Aberg2013,Skrzypczyk_2014,LostaglioJenningsRudolph2015,Friis2018,CampaioliPollockVinjanampathy2019} or a classical work source with a precise clock~\cite{Erker_2017,SchwarzhansLockErkerFriisHuber2021}. Such control is reminiscent of Maxwell's demon, who can indeed address all microscopic configurations at hand. This level of control is, however, in some sense at odds with the true spirit of thermodynamics. Indeed, the very reason that the machine is taken to begin as a thermal (Gibbs) state in thermodynamics is precisely because it provides the microscopic description that is \emph{both} consistent with macroscopic observations (in particular, average energy) and makes minimal assumptions regarding the information that the agent has about the initial state; thermodynamics as a whole is largely concerned with what can be done with minimal information requirements. Beginning with this, and then going on to permit dynamical interactions that address the full complex microstructure is somewhat contradictory, at least in essence; indeed, it has been argued that ``Maxwell's demon cannot operate''~\cite{Brillouin_1951} as an autonomous thermal being. Thus, a more thermodynamically sound setting would also restrict the transformations themselves to be ones that can be driven with minimal overall control. We now move to analyse the task of cooling within such a context.

\pdfbookmark[1]{Incoherent Control}{Incoherent Control}
\section{Incoherent Control (Heat Engine)}
\label{subsec:methods-incoherentcontrol}

The results presented so far pertain to cooling with the only restriction being that the machines are initially thermal. In particular, there are no restrictions on the allowed unitaries. In general, the operations required for cooling are not energy conserving and require an external work source. With respect to standard considerations of thermodynamics, this may seem somewhat unsatisfactory, as the joint system is, in the coherent setting, open to the universe. When quantifying thermodynamic resources, one typically restricts the permitted transformations to be energy conserving, thereby closing the joint system and yielding a self-contained theory. 

We therefore analyse protocols using energy-conserving unitaries. With this restriction, it is in general not possible to cool a target system with machines that are initially thermal at a single temperature, as was considered in the coherent-control paradigm~\cite{Clivaz_2019L}. Instead, cooling can be achieved by partitioning the machine into one cold subsystem $\Ccal$ that begins in equilibrium at inverse temperature $\beta$ and another hot subsystem $\Hcal$ coupled to a heat bath at inverse temperature $\beta_{\raisebox{-1pt}{\tiny{$H$}}} < \beta$~\cite{Clivaz_2019L,Clivaz_2019E} (see Fig.~\ref{fig:schematic}, bottom panel). In other words, one uses a hot and a cold bath to construct a heat engine that cools the target. As we demonstrate, perfect cooling can be achieved in this setting as pertinent resources diverge. However, the structure of the hot bath plays a crucial role regarding the resource requirements. In particular, we present a no-go theorem that states that perfect cooling with a heat engine using a single unitary of finite control complexity is impossible, even given diverging energy drawn from the hot bath. This result is in stark contrast to its counterpart in the coherent-control setting, where diverging energy is sufficient for perfect cooling and serves to highlight the fact that the incoherent-control setting is a fundamentally distinct paradigm that must be considered independently. Here, we focus on finite-dimensional systems and leave the analysis of infinite-dimensional ones to future work.

\subsection{Ultimate Limits in the Incoherent Control Paradigm}

In the incoherent-control setting, an adaptation of the (equality-form) Landauer bound on the minimum heat dissipated (or, as we phrase it here, the minimum amount of energy drawn from the hot bath) can be derived, which we dub the \emph{Carnot-Landauer limit}: 
\begin{thm}\label{thm:main-landauer-incoherent}
Let $F_{\raisebox{-1pt}{\tiny{$\beta$}}}(\varrho_{\raisebox{-1pt}{\tiny{$\Xcal$}}}) := \mathrm{tr}[H_{\raisebox{-1pt}{\tiny{$\Xcal$}}} \varrho_{\raisebox{-1pt}{\tiny{$\Xcal$}}}] - \beta^{-1} S(\varrho_{\raisebox{-1pt}{\tiny{$\Xcal$}}})$ be the free energy of a state $\varrho_{\raisebox{-1pt}{\tiny{$\Xcal$}}}$ with respect to a heat bath at inverse temperature $\beta$, $\Delta F_{\raisebox{-1pt}{\tiny{$\Scal$}}}^{(\beta)} :=  F_{\raisebox{-1pt}{\tiny{$\beta$}}}(\varrho_{\raisebox{-1pt}{\tiny{$\Scal$}}}^\prime) - F_{\raisebox{-1pt}{\tiny{$\beta$}}}(\varrho_{\raisebox{-1pt}{\tiny{$\Scal$}}})$, and let $\eta := 1 - \beta_{\raisebox{-1pt}{\tiny{$H$}}}/\beta \in (0,1)$ be the Carnot efficiency with respect to the hot and cold baths. In the incoherent-control setting, the quantity \begin{align}\label{eq:landauerincoherent1}
&\Delta F_{\raisebox{-1pt}{\tiny{$\Scal$}}}^{(\beta)} + \eta\, \Delta E_{\raisebox{-1pt}{\tiny{$\Hcal$}}}  \\
    &= -\frac{1}{\beta}[\Delta S_{\raisebox{-1pt}{\tiny{$\Scal$}}} + \Delta  S_{\raisebox{-1pt}{\tiny{$\Ccal$}}} + \Delta  S_{\raisebox{-1pt}{\tiny{$\Hcal$}}} + D(\varrho_{\raisebox{-1pt}{\tiny{$\Ccal$}}}^\prime||\varrho_{\raisebox{-1pt}{\tiny{$\Ccal$}}}) +  D(\varrho_{\raisebox{-1pt}{\tiny{$\Hcal$}}}^\prime||\varrho_{\raisebox{-1pt}{\tiny{$\Hcal$}}})] \notag
\end{align} 
satisfies the inequality
\begin{align}\label{eq:landauerincoherent2}
\Delta F_{\raisebox{-1pt}{\tiny{$\Scal$}}}^{(\beta)} + \eta \Delta E_{\raisebox{-1pt}{\tiny{$\Hcal$}}}  &\leq 0. 
\end{align}
\end{thm}
\noindent Equation~\eqref{eq:landauerincoherent2} holds due to the non-negativity of the sum of local entropy changes and the relative-entropy terms. The derivation is provided in Appendix~\ref{app:equalityformsofthecarnot-landauerlimit}, where we also show that the usual Landauer bound is recovered in the limit of an infinite-temperature heat bath. 

The incoherent-control setting is fundamentally distinct from the coherent-control setting in terms of what can (or cannot) be achieved with given resources. For instance, consider the case where one wishes to achieve perfect cooling in unit time and with finite control complexity with diverging energy cost. In the coherent-control setting, this task is possible in principle (see Theorem~\ref{thm:divergingenergycoherent}). On the other hand, in the incoherent-control setting, we have the following no-go theorem (see Appendix~\ref{app:coolingprotocolsincoherentcontrolparadigm} for a proof):
\begin{thm}\label{thm:infenergyauto}
In the incoherent control scenario, it is not possible to perfectly cool any quantum system of finite dimension in unit time and with finite control complexity, even given diverging energy drawn from the hot bath, for any non-negative inverse temperature heat bath $\beta_{\raisebox{-1pt}{\tiny{$H$}}} \in [0, \beta < \infty)$. 
\end{thm}
\noindent This result follows from the fact that in the incoherent-control setting, the target system can only interact with subspaces of the joint hot-and-cold machine with respect to which it is energy degenerate. For any operation of fixed control complexity, there is always a finite amount of population remaining outside of the accessible subspace, implying that perfect cooling cannot be achieved, independent of the amount of energy drawn from the hot bath.

\subsection{Saturating the Carnot-Landauer Limit}

The above result emphasises the difference between coherent and incoherent controlling, which means that it is \emph{a priori} unclear if the Carnot-Landauer bound is attainable and, if so, how to attain it. Indeed, the restriction to energy-conserving unitaries generally makes it difficult to tell if the ultimate bounds can be saturated in the incoherent-control setting, and which resources would be required to do so. We present a detailed study of cooling in the incoherent-control setting in Appendix~\ref{app:coolingprotocolsincoherentcontrolparadigm}, where we prove the following results. We begin by demonstrating incoherent cooling protocols that saturate the Landauer bound in the regime where the heat-bath temperature goes to infinity. We do so by fine tuning the machine structure such that the desired cooling transitions between the target system and the cold and hot parts of the machine are rendered energy conserving. In particular, we prove the following:
\begin{thm}\label{thm:autoinftempinftime}
In the incoherent control scenario, for an infinite-temperature hot bath $\beta_{\raisebox{-1pt}{\tiny{$H$}}} = 0$, any finite-dimensional system can be perfectly cooled at the Landauer limit with diverging time via interactions of finite control complexity. Similarly, the goal can be achieved in unit time with diverging control complexity. 
\end{thm}

Following our analysis of infinite-temperature heat baths, we study the more general case of finite-temperature heat baths. In Appendix~\ref{app:incoherentcoolingfinitetemperature}, we detail cooling protocols that saturate the Carnot-Landauer limit for any finite-temperature heat bath. More precisely, we prove:
\begin{thm}\label{thm:autofintempinftime}
In the incoherent control scenario, for any finite-temperature hot bath $0 < \beta_{\raisebox{-1pt}{\tiny{$H$}}} < \beta$, any finite-dimensional quantum system can be perfectly cooled at the Carnot-Landauer limit given diverging time via finite control complexity interactions. Similarly, the goal can be achieved in unit time with diverging control complexity. 
\end{thm}
\noindent As in the coherent-control setting, these protocols use either diverging time or control complexity to asymptotically saturate the Carnot-Landauer bound. The results presented in this section therefore provide a comprehensive understanding of the resources required to perfectly cool at minimum energy cost in a setting that aligns with the resource theories of thermodynamics. 

\pdfbookmark[1]{Imperfect Cooling with Finite Resources}{Imperfect Cooling with Finite Resources}
\section{Imperfect Cooling with Finite Resources}
\label{subsec:imperfectcooling}

The above results set the ultimate limitations for cooling inasmuch as the protocols saturate optimal bounds by using diverging resources. In reality, however, any practical implementation is limited to having only finite resources at its disposal. According to the third law, a perfectly pure state cannot be achieved in this scenario. Nonetheless, one can prepare a state of finite temperature by investing said resources appropriately. In this finite-resource setting, the interplay between energy, time, and control complexity is rather complicated. First, the cooling performance is stringent upon the chosen figure of merit for the notion of cool---the ground-state population, purity, average energy,
or temperature of the nearest thermal state are all reasonable candidates, but they differ in general~\cite{Clivaz_2019L}. Second, the total amount of resources available bounds the reachable temperature in any given protocol. Third, the details of the protocol itself influence the energy cost of achieving a desired temperature. In other words, determining the optimal distribution of resources is an extremely difficult task in general and remains an open question.

We therefore focus here on the paradigmatic special case of cooling a qubit target system by increasing its ground-state population in order to highlight some salient points regarding cooling to finite temperatures. First, we compare the finite performance of two distinct coherent control protocols that both asymptotically saturate the Landauer limit; nonetheless, at any finite time, their performance varies. The first protocol simply swaps the target qubit with one of a sequence of machine qubits whose energy gaps are distributed linearly; the second involves interacting the target with a high-dimensional machine with a particular degeneracy structure. Although the latter cannot be decomposed easily into a qubit circuit (thereby making it more difficult to implement in practice), one can compare the two protocols fairly by fixing the total (and effective) dimension to be equal, i.e., comparing the performance of the linear sequential qubit machine protocol after $N+1$ qubits have been accessed with that of the latter protocol with machine dimension $2^{N+1}$. In doing so, we see that the simpler former protocol outperforms the more difficult latter one in terms of the energy cost at finite times, emphasising the fact that difficulty in practice does not necessarily correspond to complexity as a thermodynamic resource. Additionally, we analyse the cooling rates at which energy and time can be traded off amongst each other in the linear qubit sequence protocol by deriving an analytic expression. Lastly, we compare the performance of a coherent and an incoherent control protocol that use a similar machine structure to achieve a desired final temperature. We see that the price one must pay for running the protocol via a heat engine is that either more steps or more complex operations are required to match the performance of the coherent control setting. This example serves to elucidate the connection between the two extremal control scenarios relevant for thermodynamics.

Although throughout most of the paper we focus on the asymptotic achievability of optimal cooling strategies, the protocols that we construct provide insight into how said asymptotic limits are approached. This facilitates a better understanding of the more practically relevant questions that are constrained when all resources are restricted to be finite: i) \emph{how cold can the target system be made?} and ii) \emph{at what energy cost?} In line with Nernst's third law, the answer to the former question cannot be perfectly cold (i.e., zero temperature). The answer depends upon how said resources are configured and utilized. For instance, given a single unitary interaction of finite complexity in the coherent-control setting, the ground-state population of the output state can be upper bounded in terms of the largest energy gap of the machine, $\omega_{\textup{max}}$ [see Eq.~\eqref{eq:maingeneralpuritybound}]. On the other hand, supposing that one can reuse a single machine system multiple times, then as the number of operation steps increases, the ground-state population of the output state approaches $(1+e^{-\beta \omega_{\textup{max}}})^{-1}$ from below~\cite{Clivaz_2019L}. There is clearly a trade-off relation here between time and complexity, and a systematic analysis of the rate at which these quantities can be traded off against one another warrants further investigation. Similarly, the energy cost to reach a desired final temperature also depends upon the distribution of resources, as we now examine. 

Given access to a machine of a certain size (as measured by its dimension), one could ask: \emph{what is the optimal configuration of machine energy spectrum and global unitary to cool a system as efficiently as possible?} Here, we compare two contrasting constructions for the cooling unitary in the coherent-control setting for a qubit target system (with energy gap $\omega_{\raisebox{-1pt}{\tiny{$\Scal$}}}$)---both of which asymptotically achieve Landauer cost cooling, but whose finite behaviour differs. The first protocol considers a machine of $N$ qubits whose energy gaps increase linearly from the first excited state energy level of the system $\omega_1 = \omega_{\raisebox{-1pt}{\tiny{$\Scal$}}}$ to some maximum energy level $\omega_{\raisebox{-1pt}{\tiny{$N$}}} = \omega_{\textup{max}}$, which dictates the final achievable temperature. In this protocol, the target system is swapped sequentially with each of the $N$ qubits in order of increasing energy gaps; we hence refer to it as the \emph{linear qubit machine sequence}. The second protocol we consider is presented in full in Appendix~\ref{app:finetunedcontrolconditions} and inspired by one presented in Ref.~\cite{Reeb_2014} (see Appendix D therein); we hence refer to it as the Reeb \& Wolf \textbf{(RW)} protocol. Here, the global unitary acts on the system and a high-dimensional machine with an equally spaced Hamiltonian whose degeneracy doubles with each increasing energy level, i.e., it has a singular ground state, a twofold degenerate first excited state, a fourfold degenerate second excited state, and so on; the final energy level has an extra state so that the total dimension is $2^{N+1}$ (where $N$ is the number of energy levels). In particular, the unitary performs the permutation that places the maximal amount of population in the ground state of the target system. Due to the structure of both protocols, one can make a fair comparison between them, contrasting the single unitary on a $2^{N}$-dimensional machine in the RW protocol versus the composition of $N$ two-qubit \texttt{SWAP} unitaries in the linear machine sequence, i.e., such that both protocols access a machine of the same size overall. 

\begin{figure}[t!]
    \centering
    \includegraphics[width=\linewidth]{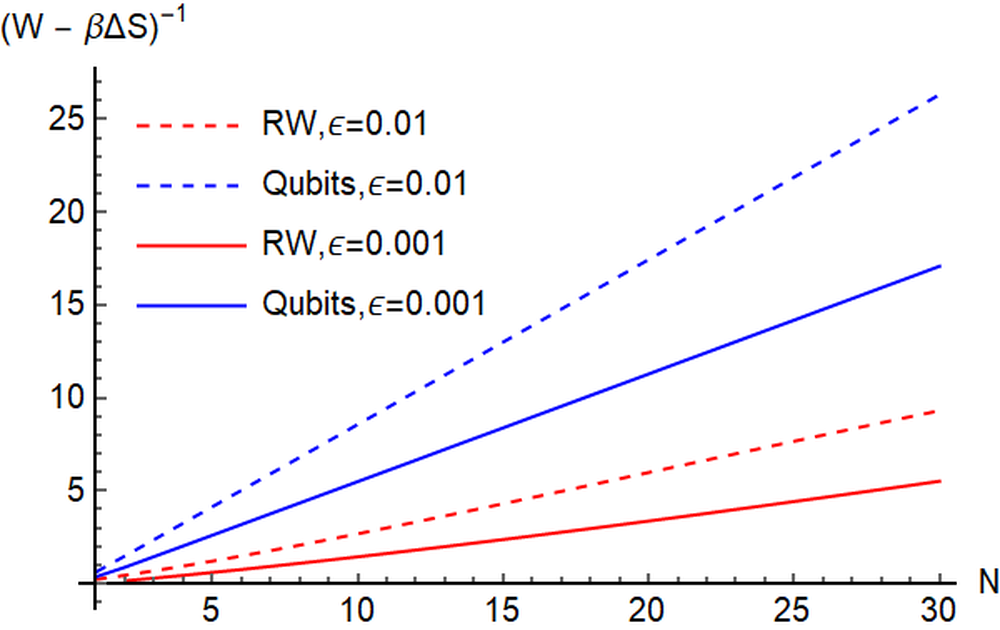}
    \caption{\emph{Imperfect Cooling.} We compare the cooling performance of a degenerate qubit target system using either $N$ machine qubits of linearly increasing energy accessed sequentially or a single unitary on a $2^N$-dimensional machine, the latter being a finite adaptation of a protocol presented in Ref.~\cite{Reeb_2014}. We set $\beta = 1$, choose units such that $\hbar=k_B = 1$, and fix $1-\epsilon$ to be the desired final ground-state population of the target. We plot the inverse of the excess work cost above the Landauer limit, $W-\beta\widetilde{\Delta} S_{\raisebox{-1pt}{\tiny{$\Scal$}}}$ (in units of the smallest machine energy gap, $\omega_{\raisebox{-1pt}{\tiny{$\Mcal$}}}^{\textup{min}}$), confirming that the surplus work cost in both cases scales with $N^{-1}$. Interestingly, we see that the protocol in which the target is sequentially swapped with machine qubits outperforms that which uses a high-dimensional unitary (at equal overall control complexity) in terms of energy cost required to reach a desired temperature.}
    \label{fig:RWqubits}
\end{figure}

As shown in Fig.~\ref{fig:RWqubits}, although both protocols asymptotically tend to the Landauer limit, their finite behaviour differs. Indeed, the work cost of the linear qubit machine sequence protocol outperforms that of the RW protocol. This is somewhat surprising, as the latter is a complex high-dimensional unitary whereas the former a composition of qubit swaps; although both protocols have the same effective dimension in this comparison overall, this highlights that difficulty in the lab setting need not correspond to resourcefulness in a thermodynamic sense. Indeed, developing optimal finite cooling strategies for arbitrary systems and machines is difficult in general and remains an important open question. Nonetheless, in Appendix~\ref{app:imperfectcooling}, we derive the rate of resource divergence of the sequential qubit protocol to further clarify the trade-off between time and energy for this protocol.

Finally, we contrast the two extremal thermodynamic paradigms considered by comparing the energy cost of a coherently controlled cooling protocol to an incoherently controlled one that achieves the same final ground-state population. Intuitively, the latter setting requires more resources to achieve the same performance as the former due to the fact that only energy-resonant subspaces can be accessed by the unitary, and hence only a subspace of the full machine is usable. This implies that a greater number of operations (of fixed control complexity) are required to achieve similar results as the coherent setting, as demonstrated in Appendix~\ref{app:imperfectcooling} explicitly. Indeed, determining the optimal cooling protocols for a range of realistic assumptions remains a major open avenue.

\pdfbookmark[2]{Perfect Cooling with Incoherent Control (Heat Engine)}{Perfect Cooling with Incoherent Control (Heat Engine)}

\pdfbookmark[1]{Discussion}{Discussion}
\section{Discussion}
\label{sec:discussion}

\pdfbookmark[2]{Relation to Previous Works}{Relation to Previous Works}
\subsection*{Relation to Previous Works}
\label{subsec:relationpreviousworks}

A vast amount of the literature concerning quantum thermodynamics considers resource theories (see Refs.~\cite{Ng_2018,Lostaglio_2019} and references therein), whose central question is: \emph{what transformations are possible given particular resources, and how can one quantify the value of a resource?} While this perspective sheds light on what is possible in principle, it does not per se concern itself with the potential implementation of said transformations. Yet, the unitary operations considered in a resource theory will themselves require certain resources to implement in practice. Focusing only on a resource-theoretic perspective would thus overlook the question: \emph{how does one optimally use said resources?} Our results focus on this latter question and highlight the role of control complexity in optimising resource use. 

Concurrently, by considering arbitrary unitary operations (akin to our coherent-control paradigm without limitations on machine size) Refs.~\cite{Anders_2013,Skrzypczyk_2014} and~\cite{Reeb_2014}, studied the potential saturation of the second law of thermodynamics and Landauer's limit, respectively. References~\cite{Skrzypczyk_2014} and~\cite{Anders_2013} develop a similar protocol to our diverging time protocol in the context of work extraction and demonstrate its optimality for saturating the second law. However, these works do not discuss the practical viewpoint that the goal can be achieved in a smaller number of operations by allowing the latter to be more complex, as we emphasise. On the other hand, Ref.~\cite{Reeb_2014} considers the resources required for saturation of the Landauer limit and show an important result regarding structural complexity, namely that the machine must be infinite dimensional to cool at the Landauer limit. Our analysis regarding complexity begins here and continues to elucidate the key complexity properties that enhance the efficiency of a cooling protocol. In particular, we show that an infinite-dimensional machine is not sufficient unless the controlled unitary indeed accesses the entire machine. This first leads to the notion of ``effective dimension'', which provides a good proxy for control complexity that is consistent with Nernst's third law for all types of quantum machines---from finite-dimensional systems to harmonic oscillators. Moreover, we highlight that the optimal interactions must be fine tuned, i.e., they must couple the system to particular energy gaps of the machine in a specific configuration, paving the way for a more nuanced definition of control complexity that takes into account the complicated and precise level of control required, as we present in terms of the ``energy-gap variety''. Lastly, we emphasise that the latter discussion concerns the coherent-control scenario, which is only one of the extremal control paradigms that we consider. In addition, we consider the task of cooling in a more thermodynamically consistent setting, namely the incoherent-control paradigm. There we derive the Carnot-Landauer equality and consequent inequality, which are adaptations of the Landauer equality~\cite{Reeb_2014} and inequality~\cite{Landauer_1961}, respectively, where the protocol can only be run via a heat engine.

On the more practical side, note that our work here concerns erasing quantum information encoded in fundamental rather than logical degrees of freedom. Our reasoning here is twofold: firstly, the ultimate limitations that we aim to understand are the same whether one wishes to cool a physical system or erase information; in other words, although it may be possible to save some finite trade-off costs for imperfect erasure in the coarse-grained setting, the resources required to perform a rank-reducing process asymptotically diverge in both cases. Secondly, it is much more difficult to create coherent superpositions in the case where information is redundantly encoded in macrostates, as this would require all microstates to be in phase (indeed, this is a major reason why quantum computers aim to encode information in fundamental degrees of freedom). For erasing quantum information using bulk (classical) cooling (i.e., coupling to a suitably engineered cold bath), the relevant condition is nondegeneracy of the ground state; additionally, many original Landauer thought experiments consider degenerate Hamiltonians for the computational states. In contrast, our protocols are based upon directly controlled cooling, which works independently of the target system Hamiltonian and as such bridges the gap between various perspectives. Moving forward, it will be interesting to explore how information can be erased cheaper if it is encoded in a coarse-grained fashion, in order to better square our fundamental results presented here with experimental demonstrations. Doing so would require finite versions of all of the systems and resources that we analyse here, which we leave for future exploration.

\pdfbookmark[2]{Conclusions and Outlook}{Conclusions and Outlook}
\subsection*{Conclusions \& Outlook}
\label{subsec:conclusionsoutlook}

The results of this work have wide-ranging implications. We have both generalised and unified Landauer's bound with respect to the laws of thermodynamics. In particular, we have posed the ultimate limitations for cooling quantum systems or erasing quantum information in terms of resource costs and presented protocols that asymptotically saturate these limits. Indeed, while it is well known that heat and time requirements must be minimised to combat the detrimental effects of fluctuation-induced errors and short decoherence times on quantum technologies~\cite{Acin_2018}, we have shown that this comes at a practical cost of greater control. In particular, we have demonstrated the necessity of implementing fine-tuned interactions involving a diverging number of energy levels to minimise energy and time costs, which serves to deliver a cautionary message: control complexity must be accounted for to build operationally meaningful resource theories of quantum thermodynamics. This result posits the energy-gap variety accessed by a unitary protocol as a meaningful quantifier of control complexity that is both fully consistent with the third law of thermodynamics and chimes well with what is difficult to achieve in practice. Our analysis of the incoherent-control setting further provides pragmatic ultimate limitations for the scenario where minimal control is required, in the sense that all transformations are driven by thermodynamic energy and entropy flows between two heat baths, which could be viewed as a thermodynamically driven quantum computer~\cite{Bennett_1982}. Nevertheless, the intricate relationship between various resources here will need to be further explored. 

Looking forward, we believe it will be crucial to go beyond asymptotic limits. While Landauer erasure and the third law of thermodynamics conventionally deal with the creation of pure states, practical results would need to consider cooling to a finite temperature (i.e., creating approximately pure states) with a finite amount of invested resources~\cite{Clivaz_2019E,Taranto_2020,zhen2021}. In this context, the trade-off between time and control complexity will gain more practical relevance, as realistic quantum technologies have limited coherence times and interaction Hamiltonians are limited to few-body terms. Here, operational measures of control complexity that fit the envisioned experimental setup present an important challenge that must be overcome to apply our results across various platforms. 

Our results strengthen the view that, in contrast to classical thermodynamics, the role of control is one of the most crucial issues to address before a true understanding of the limitations and potential of quantum machines is revealed. On the one hand, in classical systems, control is only ever achieved over few bulk degrees of freedom, whereas addressing and designing particular microstate control is within reach of current quantum technological platforms, offering additional routes towards operations enhanced by fine-tuned control. On the other hand, the cost of such control itself can quickly exceed the energy scale of the system, potentially rendering any perceived advantages a mirage. This is exacerbated by the fact that it is not possible to observe (measure) a quantum machine without incurring significant additional thermodynamic costs~\cite{Guryanova2020,DebarbaEtAl2019} and non-negligible backaction on the operation of the machine itself~\cite{Manzano_2019}. A fully developed theory of quantum thermodynamics would need to take these into account and we hope that our study sheds light on the role of control complexity in this endeavour.

\pdfbookmark[1]{Acknowledgments}{Acknowledgments}
\begin{acknowledgments}

The authors thank Elizabeth Agudelo and Paul Erker for very insightful discussions at the early stages of the project. 
P.T. acknowledges support from the Austrian Science Fund (FWF) project: Y879-N27 (START), the European Research Council (Consolidator grant `Cocoquest' 101043705), and the Japan Society for the Promotion of Science (JSPS) by KAKENHI Grant No. 21H03394.
F.B. is supported by FQXi Grant No. FQXi-IAF19-07 from the Foundational Questions Institute Fund, a donor advised fund of Silicon Valley Community Foundation. 
A.B. acknowledges support from the VILLUM FONDEN via the QMATH Centre of Excellence (Grant no.  10059) and from the QuantERA ERA-NET Cofund in Quantum Technologies implemented within the European Union's Horizon 2020 Programme (QuantAlgo project) via the Innovation Fund Denmark. 
R.S. acknowledges funding from the Swiss National Science Foundation via an Ambizione grant PZ00P2\_185986. 
N.F. is supported by the Austrian Science Fund (FWF) projects: P 36478-N and P 31339-N27. 
M.P.E.L. acknowledges financial support by the ESQ (Erwin Schr{\"o}dinger Center for Quantum Science \& Technology) Discovery programme, hosted by the Austrian Academy of Sciences ({\"O}AW) and TU Wien. 
G.V. is supported by the Austrian Science Fund (FWF) projects ZK 3 (Zukunftskolleg) and M 2462-N27 (Lise-Meitner). 
F.C.B. acknowledges support from the European Union's Horizon 2020 research and innovation programme under the Marie Sk{\l}odowska-Curie Grant Agreement No. 801110 and the Austrian Federal Ministry of Education, Science and Research (BMBWF). This project reflects only the authors' view, the EU Agency is not responsible for any use that may be made of the information it contains.
T.D. acknowledges support from the Brazilian agency CNPq INCT-IQ through the project (465469/2014-0). 
E.S. is supported by the Austrian Science Fund (FWF) project: Y879-N27 (START). 
F.C. is supported by the ERC Synergy grant HyperQ (Grant No. 856432).
M.H. is supported by the European Research Council (Consolidator grant `Cocoquest' 101043705), the Austrian Science Fund (FWF) project: Y879-N27 (START), and acknowledges financial support by the ESQ (Erwin Schr{\"o}dinger Center for Quantum Science \& Technology) Discovery programme, hosted by the Austrian Academy of Sciences ({\"O}AW). 
\end{acknowledgments}

\pdfbookmark[1]{References}{References}
\bibliographystyle{apsrev4-1fixed_with_article_titles_full_names_new}
%

\newpage
\onecolumngrid
\appendix

\starttocentries
\pdfbookmark[0]{Supplemental Material}{Supplemental Material}
\vspace*{5mm}
\label{supplementalmaterial}
\phantomsection
\begin{center}
\begin{LARGE}
Supplemental Material
\end{LARGE}
\end{center}

\setcounter{tocdepth}{2}
\vspace*{-5mm}
\tableofcontents

\renewcommand{\thesection}{\Alph{section}}
\renewcommand{\thesubsection}{\Alph{section}\arabic{subsection}}
\renewcommand{\thesubsubsection}{\Alph{section}\arabic{subsection}\alph{subsubsection}}
\makeatletter
\renewcommand{\p@subsection}{}
\renewcommand{\p@subsubsection}{}
\makeatother

\section{Equality Forms of the (Carnot-)Landauer Limit}
\label{app:equalityformsofthecarnot-landauerlimit}

In this section, we present lower bounds on the energy change of the machine (or heat dissipated into its environment) in terms of the entropy change of the target system, both in the coherent and incoherent-control settings outlined in the main text. In the coherent setting, this amounts to the well-known Landauer principle~\cite{Landauer_1961}, whereas the incoherent setting requires an extension of this derivation. These lower bounds are important, because they put limits on the optimal energetic performance of the machines for cooling. Note, finally, that the initial state of the machine is diagonal in its energy eigenbasis and must remain so for any process saturating the (Carnot-)Landauer limit; moreover, the target begins similarly and ends up in the pure state $\ket{0}\!\bra{0}$ when perfect cooling is achieved. As a result, all quantities relevant to perfect cooling at the (Carnot-)Landauer limit can be computed in terms of their ``classical'' counterparts, i.e., $\varrho_{\raisebox{-1pt}{\tiny{$\Xcal$}}} \to p_{\raisebox{-1pt}{\tiny{$\Xcal$}}} := (p_0, \hdots, p_d)$ with $p_n = e^{-\beta E_n}$, $\tr{H \varrho_{\raisebox{-1pt}{\tiny{$\Xcal$}}}} \to \langle E \rangle_{p_{\raisebox{-1pt}{\tiny{$\Xcal$}}}} := \sum_n p_n E_n$, $S(\varrho_{\raisebox{-1pt}{\tiny{$\Xcal$}}}) \to S(p_{\raisebox{-1pt}{\tiny{$\Xcal$}}}) := -\sum_n p_n \log{(p_n)}$, $\Zcal(\beta, H_{\raisebox{-1pt}{\tiny{$\Xcal$}}}) = \sum_n e^{-\beta E_n}$, and so on. Nonetheless, all of the results presented hold for the more general ``quantum'' properties. 

\subsection{Coherent-Control Paradigm: The Landauer Limit}
\label{app:coherentcontrolparadigmlandauerlimit}

The coherent setting was already studied in detail in Ref.~\cite{Reeb_2014}, where the authors derived an equality version of Landauer's principle. We restate the results here for convenience, since we will also use them in the incoherent paradigm. Recall that the setting we consider consists of two parts, the target system $\Scal$ and the machine $\Mcal$. In the beginning, the joint state is $\varrho_{\raisebox{-1pt}{\tiny{$\Scal \Mcal$}}} = \varrho_{\raisebox{-1pt}{\tiny{$\Scal$}}} \otimes \tau_{\raisebox{-1pt}{\tiny{$\Mcal$}}}(\beta, H_{\raisebox{-1pt}{\tiny{$\Mcal$}}})$ for some arbitrary (but fixed) Hamiltonian $H_{\raisebox{-1pt}{\tiny{$\Mcal$}}}$ and $\beta \in \mathbb R$. Note that any full-rank state $\varrho$ can be associated to some chosen temperature $\beta$, which sets the energy scale, and a Hamiltonian $H = -\frac{1}{\beta} \log{(\varrho)}$; as we consider arbitrary Hamiltonians, we only write the state dependence on these parameters when necessary. If the state is not full rank, the rank can be used to redefine the dimension. We assume that both systems are finite dimensional. Let $U$ be a global unitary on $\Scal \Mcal$. We write $\varrho^\prime_{\raisebox{-1pt}{\tiny{$\Scal\Mcal$}}} := U [\varrho_{\raisebox{-1pt}{\tiny{$\Scal$}}} \otimes \tau_{\raisebox{-1pt}{\tiny{$\Mcal$}}}(\beta, H_{\raisebox{-1pt}{\tiny{$\Mcal$}}})]U^\dagger$ and denote by $\varrho_{\raisebox{-1pt}{\tiny{$\Scal$}}}^\prime$ and $\varrho_{\raisebox{-1pt}{\tiny{$\Mcal$}}}^\prime$ the respective reduced states. The quantity $I(\Scal: \Mcal)_{\varrho^\prime_{\raisebox{-1pt}{\tiny{$\Scal \Mcal$}}}} = S(\varrho_{\raisebox{-1pt}{\tiny{$\Scal$}}}^\prime) + S(\varrho_{\raisebox{-1pt}{\tiny{$\Mcal$}}}^\prime) - S(\varrho^\prime_{\raisebox{-1pt}{\tiny{$\Scal \Mcal$}}})$ is the final mutual information between $\Scal$ and $\Mcal$ and $D(\varrho^\prime_{\raisebox{-1pt}{\tiny{$\Mcal$}}}||\varrho_{\raisebox{-1pt}{\tiny{$\Mcal$}}}) = \tr{\varrho^\prime_{\raisebox{-1pt}{\tiny{$\Mcal$}}} \log(\varrho^\prime_{\raisebox{-1pt}{\tiny{$\Mcal$}}})}-\tr{\varrho^\prime_{\raisebox{-1pt}{\tiny{$\Mcal$}}} \log(\varrho_{\raisebox{-1pt}{\tiny{$\Mcal$}}})}$ is the relative entropy of the final machine state with respect to its initial state.

\begin{lem}[{\cite[Lemma 2]{Reeb_2014}}] \label{lem:second-law-lemma}
Let the setting be as above. Then
\begin{equation}
    [S(\varrho_{\raisebox{-1pt}{\tiny{$\Scal$}}}^\prime) - S(\varrho_{\raisebox{-1pt}{\tiny{$\Scal$}}})] + [S(\varrho_{\raisebox{-1pt}{\tiny{$\Mcal$}}}^\prime) - S(\varrho_{\raisebox{-1pt}{\tiny{$\Mcal$}}})] = I(\Scal: \Mcal)_{\varrho^\prime_{\raisebox{-1pt}{\tiny{$\Scal \Mcal$}}}} \geq 0.
\end{equation}
\end{lem}
\begin{proof}
We note that 
\begin{equation}
    [S(\varrho_{\raisebox{-1pt}{\tiny{$\Scal$}}}^\prime) - S(\varrho_{\raisebox{-1pt}{\tiny{$\Scal$}}})] + [S(\varrho_{\raisebox{-1pt}{\tiny{$\Mcal$}}}^\prime) - S(\varrho_{\raisebox{-1pt}{\tiny{$\Mcal$}}})] = S(\varrho_{\raisebox{-1pt}{\tiny{$\Scal$}}}^\prime) + S(\varrho_{\raisebox{-1pt}{\tiny{$\Mcal$}}}^\prime) - S(\varrho_{\raisebox{-1pt}{\tiny{$\Scal \Mcal$}}}^\prime),
\end{equation}
since the von Neumann entropy is additive for product states and invariant under unitary evolution. The assertion follows from the definition of the mutual information and the fact that it is non-negative.
\end{proof}

\begin{thm}[Equality form of Landauer's principle, {\cite[Theorem 3]{Reeb_2014}}]
Let the setting be as above. Then
\begin{equation}
     \beta \, \mathrm{tr}[H_{\raisebox{-1pt}{\tiny{$\Mcal$}}}(\varrho_{\raisebox{-1pt}{\tiny{$\Mcal$}}}^\prime-\varrho_{\raisebox{-1pt}{\tiny{$\Mcal$}}})] - [S(\varrho_{\raisebox{-1pt}{\tiny{$\Scal$}}}) - S(\varrho_{\raisebox{-1pt}{\tiny{$\Scal$}}}^\prime)] =  I(\Scal : \Mcal)_{\varrho^\prime_{\raisebox{-1pt}{\tiny{$\Scal \Mcal$}}}} + D(\varrho_{\raisebox{-1pt}{\tiny{$\Mcal$}}}^\prime||\varrho_{\raisebox{-1pt}{\tiny{$\Mcal$}}}) \geq 0.
\end{equation}
\end{thm}

\begin{proof}
From Lemma \ref{lem:second-law-lemma}, it follows that 
\begin{equation}
    [S(\varrho_{\raisebox{-1pt}{\tiny{$\Scal$}}}) - S(\varrho_{\raisebox{-1pt}{\tiny{$\Scal$}}}^\prime)] + I(\Scal : \Mcal)_{\varrho^\prime_{\raisebox{-1pt}{\tiny{$\Scal \Mcal$}}}} = S(\varrho_{\raisebox{-1pt}{\tiny{$\Mcal$}}}^\prime) - S(\varrho_{\raisebox{-1pt}{\tiny{$\Mcal$}}}).
    \label{eq:th8}
\end{equation}
Using the fact that $\varrho_{\raisebox{-1pt}{\tiny{$\Mcal$}}} = \tau_{\raisebox{-1pt}{\tiny{$\Mcal$}}}(\beta, H_{\raisebox{-1pt}{\tiny{$\Mcal$}}})$, we infer that
$D(\varrho_{\raisebox{-1pt}{\tiny{$\Mcal$}}}^\prime||\varrho_{\raisebox{-1pt}{\tiny{$\Mcal$}}})=-S(\varrho_{\raisebox{-1pt}{\tiny{$\Mcal$}}}^\prime)+\beta \mathrm{tr}[H_{\raisebox{-1pt}{\tiny{$\Mcal$}}} \varrho_{\raisebox{-1pt}{\tiny{$\Mcal$}}}^\prime] + \log{[\mathrm{tr}(e^{-\beta H_{\raisebox{-1pt}{\tiny{$\Mcal$}}}})]}$ and $S(\varrho_{\raisebox{-1pt}{\tiny{$\Mcal$}}}) = \beta \mathrm{tr}[H_{\raisebox{-1pt}{\tiny{$\Mcal$}}} \varrho_{\raisebox{-1pt}{\tiny{$\Mcal$}}}] + \log{[\mathrm{tr}(e^{-\beta H_{\raisebox{-1pt}{\tiny{$\Mcal$}}}})]}$. Re-expressing the first of these for $S(\varrho_{\raisebox{-1pt}{\tiny{$\Mcal$}}}^\prime)$ and inserting both into Eq.~\eqref{eq:th8} yields the claimed equality. The inequality results from non-negativity of relative entropy and mutual information. This completes the proof.
\end{proof}

\subsection{Incoherent-Control Paradigm: The Carnot-Landauer Limit}
\label{app:incoherentcontrolparadigmcarnot-landauerlimit}

Landauer's principle provides a relationship between how much heat must necessarily be dissipated into the thermal background environment upon manipulating the entropy of a given quantum system. Until now, we have assumed that the system of interest can interact arbitrarily with its environment (i.e., the machine); in other words, we have considered general joint unitary interactions between system and machine, without restriction. In doing so, we have tacitly assumed the ability to draw energy from some external resource (i.e., a work source) in order to implement said unitaries, which are in general not energy preserving. The particularities of such a resource are left as an abstraction. However, from a thermodynamicists' perspective, this setting may seem somewhat unsatisfactory, as the joint target-machine system is not energetically closed. In order to provide a more self-contained picture of the cooling procedure, one can explicitly include the energy resource, modelled as a quantum system itself, into the setting. 

To this end, note first that said resource must be out of thermal equilibrium with respect to the target and machine in order to perform any meaningful thermodynamic transformation. Furthermore, it is sensible to assume that the energy resource system is in thermal equilibrium with its own environment to begin with. The joint target-machine-resource system is then considered to be energetically closed; as such, global unitaries in this setting are restricted to be energy conserving. In order to act as a resource for cooling the target in this picture, the energy source here must begin in equilibrium with a heat bath that is hotter than the initial temperature of the machine (assuming that the machine and resource both begin in thermal states), such that a natural heat flow is induced that leads the environment of the machine to act as a final heat sink. This setting is what we call the incoherent-control scenario. In this context, Landauer's principle translates to studying the relationship between the heat that is necessarily dissipated into the machine's environment upon manipulating the entropy of the target system. Finally, note that the relationship between the coherent and the incoherent-control paradigms is interesting in itself: while on the one hand the incoherent setting includes an additional system and therefore increases the dimensionality of the overall joint system, on the other hand by restricting the transformations on this larger space to be energy conserving, one limits the orbit of attainable states.

Now let us consider the incoherent-control setting. Here, we have the target system $\Scal$ and the machine comprises of one part $\Ccal$ coupled to the cold bath and another $\Hcal$ coupled to the hot bath. We assume that all systems are finite-dimensional. Every subsystem $\Acal$ is associated to a Hamiltonian $H_{\raisebox{-1pt}{\tiny{$\Acal$}}}$ and $\Ccal$, $\Hcal$ are initially in a thermal state; the cold bath has inverse temperature $\beta$ and the hot bath has inverse temperature $\beta_{\raisebox{-1pt}{\tiny{$H$}}} < \beta$. We assume $\beta$, $\beta_{\raisebox{-1pt}{\tiny{$H$}}}$. Thus, the initial joint state is $\varrho_{\raisebox{-1pt}{\tiny{$\Scal \Ccal \Hcal$}}} = \varrho_{\raisebox{-1pt}{\tiny{$\Scal$}}} \otimes \tau_{\raisebox{-1pt}{\tiny{$\Ccal$}}}(\beta, H_{\raisebox{-1pt}{\tiny{$\Ccal$}}}) \otimes \tau_{\raisebox{-1pt}{\tiny{$\Hcal$}}}(\beta_{\raisebox{-1pt}{\tiny{$H$}}}, H_{\raisebox{-1pt}{\tiny{$\Hcal$}}})$. The global evolution on $\Scal\Ccal\Hcal$ is implemented via a unitary $U$, leading to $\varrho_{\raisebox{-1pt}{\tiny{$\Scal \Ccal \Hcal$}}}^\prime = U (\varrho_{\raisebox{-1pt}{\tiny{$\Scal \Ccal \Hcal$}}}) U^\dagger $. We further assume that the unitary evolution on the joint system is energy conserving, i.e.,\ $[U, H_{\raisebox{-1pt}{\tiny{$\Scal$}}} + H_{\raisebox{-1pt}{\tiny{$\Ccal$}}} + H_{\raisebox{-1pt}{\tiny{$\Hcal$}}}] = 0$. We write $\Delta S_{\raisebox{-1pt}{\tiny{$\Acal$}}} := S(\varrho^\prime_{\raisebox{-1pt}{\tiny{$\Acal$}}}) - S(\varrho_{\raisebox{-1pt}{\tiny{$\Acal$}}})$ for the entropy change on subsystem $\Acal$ and $\Delta E_{\raisebox{-1pt}{\tiny{$\Acal$}}} := \mathrm{tr}[H_{\raisebox{-1pt}{\tiny{$\Acal$}}}(\varrho^\prime_{\raisebox{-1pt}{\tiny{$\Acal$}}} - \varrho_{\raisebox{-1pt}{\tiny{$\Acal$}}})]$ for the average energy change. Moreover, the free energy of a state $\varrho_{\raisebox{-1pt}{\tiny{$\Acal$}}}$ with respect to the inverse temperature $\beta$ is $F_{\raisebox{-1pt}{\tiny{$\beta$}}}(\varrho_{\raisebox{-1pt}{\tiny{$\Acal$}}}) = \mathrm{tr}[H_{\raisebox{-1pt}{\tiny{$\Acal$}}} \varrho_{\raisebox{-1pt}{\tiny{$\Acal$}}}] - \beta^{-1} S(\varrho_{\raisebox{-1pt}{\tiny{$\Acal$}}})$.

In the incoherent setting, it makes sense to look at the energy decrease in the hot bath $\mathcal{H}$, since the hot bath can be seen as the energetic resource one must to expend in order to cool the system $\Scal$ (alternatively, as we present after the following theorem, one can consider the energy dissipated into the cold bath $\Ccal$, which serves as the heat sink).

\begin{thm} \label{thm:landauer-incoherent}
In the above setting, it holds that
\begin{equation}
    \Delta F_{\raisebox{-1pt}{\tiny{$\Scal$}}}^{(\beta)} + \eta \Delta E_{\raisebox{-1pt}{\tiny{$\Hcal$}}} = - \frac{1}{\beta}[\Delta S_{\raisebox{-1pt}{\tiny{$\Scal$}}} + \Delta S_{\raisebox{-1pt}{\tiny{$\Ccal$}}} + \Delta S_{\raisebox{-1pt}{\tiny{$\Hcal$}}} + D(\varrho_{\raisebox{-1pt}{\tiny{$\Ccal$}}}^\prime||\varrho_{\raisebox{-1pt}{\tiny{$\Ccal$}}}) +  D(\varrho_{\raisebox{-1pt}{\tiny{$\Hcal$}}}^\prime||\varrho_{\raisebox{-1pt}{\tiny{$\Hcal$}}})] \leq 0,
\end{equation}
where $(0,1) \ni \eta := 1 - \beta_{\raisebox{-1pt}{\tiny{$H$}}}/\beta $ is the Carnot efficiency and $\Delta F_{\raisebox{-1pt}{\tiny{$\Scal$}}}^{(\beta)} =  F_{\raisebox{-1pt}{\tiny{$\beta$}}}(\varrho_{\raisebox{-1pt}{\tiny{$\Scal$}}}^\prime) - F_{\raisebox{-1pt}{\tiny{$\beta$}}}(\varrho_{\raisebox{-1pt}{\tiny{$\Scal$}}})$. 
\end{thm}

\begin{proof}
Let us consider
\begin{equation}
    I(\Scal:\Ccal:\Hcal)_{\varrho_{\raisebox{-1pt}{\tiny{$\Scal \Ccal \Hcal$}}}^\prime} := S(\varrho^\prime_{\raisebox{-1pt}{\tiny{$\Scal$}}}) + S(\varrho^\prime_{\raisebox{-1pt}{\tiny{$\Ccal$}}}) + S(\varrho^\prime_{\raisebox{-1pt}{\tiny{$\Hcal$}}}) - S(\varrho_{\raisebox{-1pt}{\tiny{$\Scal \Ccal \Hcal$}}}^\prime) \geq 0.
\end{equation}
Note that the quantity $I(\Scal:\Ccal:\Hcal)_{\varrho_{\raisebox{-1pt}{\tiny{$\Scal \Ccal \Hcal$}}}^\prime}$, which quantifies the tripartite mutual information of the state $\varrho_{\raisebox{-1pt}{\tiny{$\Scal \Ccal \Hcal$}}}^\prime$, is non-negative via subadditivity $S(\varrho_{\raisebox{-1pt}{\tiny{$\Acal$}}}) + S(\varrho_{\raisebox{-1pt}{\tiny{$\Bcal$}}}) \geq S(\varrho_{\raisebox{-1pt}{\tiny{$\Acal\Bcal$}}})$ for any state $\varrho_{\raisebox{-1pt}{\tiny{$\Acal\Bcal$}}}$. Furthermore, since the von Neumann entropy is invariant under unitary transformations and additive for tensor product states, we have 
\begin{equation}
    I(\Scal:\Ccal:\Hcal)_{\varrho_{\raisebox{-1pt}{\tiny{$\Scal \Ccal \Hcal$}}}^\prime} = \Delta S_{\raisebox{-1pt}{\tiny{$\Scal$}}} + \Delta S_{\raisebox{-1pt}{\tiny{$\Ccal$}}} + \Delta S_{\raisebox{-1pt}{\tiny{$\Hcal$}}}.
\end{equation}
We also have that
\begin{equation}
    \Delta S_{\raisebox{-1pt}{\tiny{$\Ccal$}}} = \beta \Delta E_{\raisebox{-1pt}{\tiny{$\Ccal$}}} - D(\varrho_{\raisebox{-1pt}{\tiny{$\Ccal$}}}^\prime||\varrho_{\raisebox{-1pt}{\tiny{$\Ccal$}}})
\end{equation}
and 
\begin{equation}
    \Delta S_{\raisebox{-1pt}{\tiny{$\Hcal$}}} = \beta_{\raisebox{-1pt}{\tiny{$H$}}} \Delta E_{\raisebox{-1pt}{\tiny{$\Hcal$}}} - D(\varrho_{\raisebox{-1pt}{\tiny{$\Hcal$}}}^\prime||\varrho_{\raisebox{-1pt}{\tiny{$\Hcal$}}}).
\end{equation}
Thus,
\begin{equation}\label{eq:incoherentlandauerequalityallterms}
    I(\Scal:\Ccal:\Hcal)_{\varrho_{\raisebox{-1pt}{\tiny{$\Scal \Ccal \Hcal$}}}^\prime} = \Delta S_{\raisebox{-1pt}{\tiny{$\Scal$}}} +  \beta \Delta E_{\raisebox{-1pt}{\tiny{$\Ccal$}}} - D(\varrho_{\raisebox{-1pt}{\tiny{$\Ccal$}}}^\prime||\varrho_{\raisebox{-1pt}{\tiny{$\Ccal$}}}) + \beta_{\raisebox{-1pt}{\tiny{$H$}}} \Delta E_{\raisebox{-1pt}{\tiny{$\Hcal$}}} - D(\varrho_{\raisebox{-1pt}{\tiny{$\Hcal$}}}^\prime||\varrho_{\raisebox{-1pt}{\tiny{$\Hcal$}}}).
\end{equation}
Since the unitary is energy conserving, we infer that $\Delta E_{\raisebox{-1pt}{\tiny{$\Scal$}}} + \Delta E_{\raisebox{-1pt}{\tiny{$\Ccal$}}} + \Delta E_{\raisebox{-1pt}{\tiny{$\Hcal$}}} = 0$. Hence, we have
\begin{equation}
    \Delta S_{\raisebox{-1pt}{\tiny{$\Scal$}}} - \beta \Delta E_{\raisebox{-1pt}{\tiny{$\Scal$}}} + (\beta_{\raisebox{-1pt}{\tiny{$H$}}} - \beta) \Delta E_{\raisebox{-1pt}{\tiny{$\Hcal$}}} = I(\Scal:\Ccal:\Hcal)_{\varrho_{\raisebox{-1pt}{\tiny{$\Scal \Ccal \Hcal$}}}^\prime} + D(\varrho_{\raisebox{-1pt}{\tiny{$\Ccal$}}}^\prime||\varrho_{\raisebox{-1pt}{\tiny{$\Ccal$}}}) +  D(\varrho_{\raisebox{-1pt}{\tiny{$\Hcal$}}}^\prime||\varrho_{\raisebox{-1pt}{\tiny{$\Hcal$}}}).
\end{equation}
Using the free energy, we can rewrite this as
\begin{equation}
    - \beta[F_{\raisebox{-1pt}{\tiny{$\beta$}}}(\varrho_{\raisebox{-1pt}{\tiny{$\Scal$}}}^\prime) - F_{\raisebox{-1pt}{\tiny{$\beta$}}}(\varrho_{\raisebox{-1pt}{\tiny{$\Scal$}}})] - (\beta - \beta_{\raisebox{-1pt}{\tiny{$H$}}}) \Delta E_{\raisebox{-1pt}{\tiny{$\Hcal$}}} = I(\Scal:\Ccal:\Hcal)_{\varrho_{\raisebox{-1pt}{\tiny{$\Scal \Ccal \Hcal$}}}^\prime} + D(\varrho_{\raisebox{-1pt}{\tiny{$\Ccal$}}}^\prime||\varrho_{\raisebox{-1pt}{\tiny{$\Ccal$}}}) +  D(\varrho_{\raisebox{-1pt}{\tiny{$\Hcal$}}}^\prime||\varrho_{\raisebox{-1pt}{\tiny{$\Hcal$}}}).
\end{equation}
Dividing by $-\beta$, we obtain the assertion, since, in particular, $ I(\Scal:\Ccal:\Hcal)_{\varrho_{\raisebox{-1pt}{\tiny{$\Scal \Ccal \Hcal$}}}^\prime} + D(\varrho_{\raisebox{-1pt}{\tiny{$\Ccal$}}}^\prime||\varrho_{\raisebox{-1pt}{\tiny{$\Ccal$}}}) +  D(\varrho_{\raisebox{-1pt}{\tiny{$\Hcal$}}}^\prime||\varrho_{\raisebox{-1pt}{\tiny{$\Hcal$}}}) \geq 0$ by the non-negativity of each term.
\end{proof}

In particular, we have shown that the energy extracted from the hot bath is lower-bounded by the increase in free energy, weighted by the inverse Carnot efficiency:
\begin{equation}
    \mathrm{tr}[H_{\raisebox{-1pt}{\tiny{$\Hcal$}}}(\varrho_{\raisebox{-1pt}{\tiny{$\Hcal$}}} - \varrho_{\raisebox{-1pt}{\tiny{$\Hcal$}}}^\prime)] \geq \frac{1}{\eta}[F_{\raisebox{-1pt}{\tiny{$\beta$}}}(\varrho_{\raisebox{-1pt}{\tiny{$\Scal$}}}^\prime) - F_{\raisebox{-1pt}{\tiny{$\beta$}}}(\varrho_{\raisebox{-1pt}{\tiny{$\Scal$}}})].
\end{equation}
Note that if $\varrho_{\raisebox{-1pt}{\tiny{$\Scal$}}}=\tau_{\raisebox{-1pt}{\tiny{$\Scal$}}}(\beta, H_{\raisebox{-1pt}{\tiny{$\Scal$}}})$, the r.h.s. is non-negative for any nontrivial thermodynamic process, i.e., any for which the target system is heated or\textemdash of particular relevance for us\textemdash cooled. This follows by the Gibbs variational principle, which states that the free energy of $\varrho$ is minimal iff $\varrho$ is the corresponding Gibbs state. 

Finally, in order to make a more concrete connection to the spirit of Landauer's original derivation, note that one can consider bounding the heat dissipated into the cold bath, rather than that drawn from the hot bath. Substituting $\Delta E_{\raisebox{-1pt}{\tiny{$\Hcal$}}} = - (\Delta E_{\raisebox{-1pt}{\tiny{$\Scal$}}} + \Delta E_{\raisebox{-1pt}{\tiny{$\Ccal$}}})$ into Eq.~\eqref{eq:incoherentlandauerequalityallterms} leads to
\begin{align}\label{eq:carnotlandauercold}
    -\widetilde{\Delta} S_{\raisebox{-1pt}{\tiny{$\Scal$}}} - \beta_{\raisebox{-1pt}{\tiny{$H$}}} \Delta E_{\raisebox{-1pt}{\tiny{$\Scal$}}} + (\beta - \beta_{\raisebox{-1pt}{\tiny{$H$}}}) \Delta E_{\raisebox{-1pt}{\tiny{$\Ccal$}}} \geq 0,
\end{align}
which recovers the standard Landauer bound for the dissipated heat in the limit of an infinitely hot heat bath, i.e., $\beta_{\raisebox{-1pt}{\tiny{$H$}}} \to 0$.


\section{Diverging Energy}
\label{app:divergingenergy}

\subsection{Sufficiency: Diverging Energy Cooling Protocol}
\label{app:sufficiencydivergingenergycoolingprotocol}

This cooling protocol is arguably the simplest of those presented. The thermal populations of any target system can be exchanged with a machine system of the same dimension, in the thermal state of $H_{\raisebox{-1pt}{\tiny{$\Mcal$}}}=\omega_{\raisebox{0pt}{\tiny{$\Mcal$}}} \sum_{n=0}^{d-1} n |n\rangle\!\langle n|$. As $\omega_{\raisebox{0pt}{\tiny{$\Mcal$}}}\rightarrow\infty$, the machine state $\tau_{\raisebox{-1pt}{\tiny{$\Mcal$}}}(\beta, H_{\raisebox{-1pt}{\tiny{$\Mcal$}}})$ approaches $ |0\rangle\!\langle0|_{\raisebox{-1pt}{\tiny{$\Mcal$}}}$ independently of $\beta$ (as long as $\beta \neq 0$). Such a population-exchange operation is a single interaction (i.e., the protocol occurs in unit time), which is of finite complexity (in a sense that we discuss below). However, the energy drawn from the resource $\Wcal$ upon performing said \texttt{SWAP} operation is at least $E = (p_{\raisebox{-1pt}{\tiny{$\Scal$}}}^{(1)}-p_{\raisebox{-1pt}{\tiny{$\Mcal$}}}^{(1)}) (\omega_{\raisebox{0pt}{\tiny{$\Mcal$}}}-\omega_{\raisebox{0pt}{\tiny{$\Scal$}}}^{(1)})$, where $p_{\raisebox{-1pt}{\tiny{$\Xcal$}}}^{(1)}$ is the initial population of the first excited level of system $\Xcal$ and $\omega_{\raisebox{0pt}{\tiny{$\Scal$}}}^{(1)}$ is the first energy eigenvalue of the target system. Denoting by $\omega_{\raisebox{0pt}{\tiny{$\Scal$}}}^{(k)}$ the energy eigenvalue of the $k^{\text{th}}$ excited level of the target system, we have above assumed that $\omega_{\raisebox{0pt}{\tiny{$\Scal$}}}^{(0)}=0$ (which we do for all Hamiltonians without loss of generality) and $\omega_{\raisebox{0pt}{\tiny{$\Mcal$}}} > \omega_{\raisebox{0pt}{\tiny{$\Scal$}}}^{(d-1)}$. As such, perfect cooling will incur diverging energy cost. 

\subsection{Necessity of Diverging Energy for Protocols with Finite Time and Control Complexity}\label{app:necessitydivergingenergyprotocolsfinitetimecontrolcomplexity}

Consider the following Hamiltonians for the target system and machine with finite but otherwise arbitrary energy levels, $H_{\raisebox{-1pt}{\tiny{$\Scal$}}}= \sum_{n=0}^{d_{\raisebox{-1pt}{\tiny{$\Scal$}}}-1} \omega_{\raisebox{0pt}{\tiny{$\Scal$}}}^{(n)} |n\rangle\!\langle n|_{\raisebox{-1pt}{\tiny{$\Scal$}}}$ and $H_{\raisebox{-1pt}{\tiny{$\Mcal$}}}= \sum_{n=0}^{d_{\raisebox{-1pt}{\tiny{$\Mcal$}}}-1} \omega_{\raisebox{0pt}{\tiny{$\Mcal$}}}^{(n)} |n\rangle\!\langle n|_{\raisebox{-1pt}{\tiny{$\Mcal$}}}$, respectively. For any finite inverse temperature $\beta$, the initial thermal states $\tau_{\raisebox{-1pt}{\tiny{$\Scal$}}}(\beta, H_{\raisebox{-1pt}{\tiny{$\Scal$}}})$ and $\tau_{\raisebox{-1pt}{\tiny{$\Mcal$}}}(\beta,H_{\raisebox{-1pt}{\tiny{$\Mcal$}}})$ are of full rank. Suppose now that one can implement a single unitary transformation (i.e., a unit time protocol) of finite control complexity on the joint target and machine, yielding the joint output state $\varrho'_{\raisebox{-1pt}{\tiny{$\Scal \Mcal$}}} = \ptr{\Mcal}{U (\tau_{\raisebox{-1pt}{\tiny{$\Scal$}}}(\beta,H_{\raisebox{-1pt}{\tiny{$\Scal$}}}) \otimes \tau_{\raisebox{-1pt}{\tiny{$\Scal$}}}(\beta,H_{\raisebox{-1pt}{\tiny{$\Mcal$}}})) U^\dagger}$, and wishes to attain perfect cooling of the target in doing so. By invariance of the rank under unitary transformations and the fact that the system and machine begin uncorrelated, we have
\begin{align}
    \mathrm{rank}[\tau_{\raisebox{-1pt}{\tiny{$\Scal$}}}(\beta,H_{\raisebox{-1pt}{\tiny{$\Scal$}}})] \, \mathrm{rank}[\tau_{\raisebox{-1pt}{\tiny{$\Mcal$}}}(\beta,H_{\raisebox{-1pt}{\tiny{$\Mcal$}}})] = \mathrm{rank}[\varrho'_{\raisebox{-1pt}{\tiny{$\Scal \Mcal$}}}] \leq \mathrm{rank}[\varrho'_{\raisebox{-1pt}{\tiny{$\Scal$}}}] \, \mathrm{rank}[\varrho'_{\raisebox{-1pt}{\tiny{$\Mcal$}}}], 
\end{align}
where the inequality follows from the subadditivity of the R{\' e}nyi-zero entropy~\cite{vanDam2002}, which is the logarithm of the rank. To achieve perfect cooling of the target, one must (at least asymptotically) attain $\mathrm{rank}[\varrho'_{\raisebox{-1pt}{\tiny{$\Scal$}}}] < \mathrm{rank}[\tau_{\raisebox{-1pt}{\tiny{$\Scal$}}}(\beta,H_{\raisebox{-1pt}{\tiny{$\Scal$}}})]$, which implies that $\mathrm{rank}[\varrho'_{\raisebox{-1pt}{\tiny{$\Mcal$}}}] > \mathrm{rank}[\tau_{\raisebox{-1pt}{\tiny{$\Mcal$}}}(\beta,H_{\raisebox{-1pt}{\tiny{$\Mcal$}}})]$. However, if this condition is achieved, then $D[\varrho'_{\raisebox{-1pt}{\tiny{$\Mcal$}}} \| \tau_{\raisebox{-1pt}{\tiny{$\Mcal$}}}(\beta,H_{\raisebox{-1pt}{\tiny{$\Mcal$}}})]$ diverges, implying a diverging energy cost by Eq.~\eqref{eq:landauerequality}. The above argument already appears in Ref.~\cite{Reeb_2014}.

The other situation that one must consider is the case where one attains a $\varrho'_{\raisebox{-1pt}{\tiny{$\Scal$}}}$ such that $\mathrm{rank}[\varrho'_{\raisebox{-1pt}{\tiny{$\Scal$}}}] = \mathrm{rank}[\tau_{\raisebox{-1pt}{\tiny{$\Scal$}}}(\beta,H_{\raisebox{-1pt}{\tiny{$\Scal$}}})]$ but nonetheless $\varrho'_{\raisebox{-1pt}{\tiny{$\Scal$}}}$ is arbitrarily close to a pure state, as is the case, for instance, in the protocols that we present. Consider a sequence of machines $\varrho_{\raisebox{-1pt}{\tiny{$\Mcal$}}}^{(i)}$ and unitaries $U^{(i)}$ such that $\varrho_{\raisebox{-1pt}{\tiny{$\Mcal$}}}^{(i)} \to \varrho_{\raisebox{-1pt}{\tiny{$\Mcal$}}}$ and $U^{(i)} \to U$. Note that since we fixed the dimensions of $\mathcal S$ and $\mathcal M$, any sequence of machines has a converging subsequence by the Bolzano-Weierstrass theorem and the fact that the set of quantum states is compact. Here, $\varrho_{\raisebox{-1pt}{\tiny{$\Mcal$}}}$ and $U$ achieve perfect cooling. If we fix $\varrho_{\raisebox{-1pt}{\tiny{$\Scal$}}}$, we obtain a corresponding sequence $(\varrho_{\raisebox{-1pt}{\tiny{$\Mcal$}}}^\prime)^{(i)}$ such that $(\varrho_{\raisebox{-1pt}{\tiny{$\Mcal$}}}^\prime)^{(i)} \to \varrho_{\raisebox{-1pt}{\tiny{$\Mcal$}}}^\prime$. Crucially, here, since we restrict the unitary transformation to be of finite control complexity, the states $\varrho_{\raisebox{-1pt}{\tiny{$\Mcal$}}}$ and $\varrho_{\raisebox{-1pt}{\tiny{$\Mcal$}}}^\prime$ are effectively finite dimensional, in the sense that whatever their true dimension, they can be replaced by finite-dimensional versions without changing any of the relevant quantities (see Appendix~\ref{app:conditionsstructuralcontrolcomplexity}). Since the relative entropy $(\varrho, \sigma) \mapsto D(\varrho||\sigma)$ is lower semicontinuous~\cite{Ohya1993, Ohya2010} and since $D(\varrho_{\raisebox{-1pt}{\tiny{$\Mcal$}}}^\prime || \varrho_{\raisebox{-1pt}{\tiny{$\Mcal$}}}) \to\infty$ by the arguments above, we infer that $D[(\varrho_{\raisebox{-1pt}{\tiny{$\Mcal$}}}^\prime)^{(i)} || \varrho_{\raisebox{-1pt}{\tiny{$\Mcal$}}}^{(i)}] \to \infty$ as $i \to \infty$. This argument holds independently of $\mathrm{rank}[\varrho_{\raisebox{-1pt}{\tiny{$\Scal$}}}^\prime]$; in particular, for the special case $\mathrm{rank}[\varrho'_{\raisebox{-1pt}{\tiny{$\Scal$}}}] = \mathrm{rank}[\tau_{\raisebox{-1pt}{\tiny{$\Scal$}}}(\beta,H_{\raisebox{-1pt}{\tiny{$\Scal$}}})]$ that we are considering here. Thus, to approach perfect cooling in finite time and with finite control complexity, one would need a diverging energy cost. Thus, we see that within the resource trinity of energy, time, and control complexity, if the latter two are finite, then energy must diverge to asymptotically achieve a pure state. Whether or not there exist other (unaccounted for) resources that allow one to achieve this with all three of the aforementioned resources being finite remains an open question.

Importantly, the above argument no longer holds if the time or control complexity is allowed to diverge. In such cases, both $\varrho_{\raisebox{-1pt}{\tiny{$\Mcal$}}}$ and $\varrho_{\raisebox{-1pt}{\tiny{$\Mcal$}}}^\prime$ can be infinite dimensional, and because of this the rank argument no longer applies and the relative entropy does not necessarily diverge in the limit of perfect cooling. In contrast, as we show, it is even possible to saturate the Landauer bound.


\section{Diverging Time Cooling Protocol for Finite-Dimensional Systems}
\label{app:divergingtimecoolingprotocolfinitedimensionalsystems}

\subsection{Proof of Theorem~\ref{thm:inftimeFinTepFinDim}}
\label{app:prooftheoreminftime}

\begin{proof} Consider a target system $\Scal$ of dimension $d$ with associated Hamiltonian
\begin{equation}\label{eq:sysHam}
   H_{\raisebox{-1pt}{\tiny{$\Scal$}}}=\sum_{k=0}^{d-1} \, \omega_{k} \ket{k}\!\bra{k}_{\raisebox{-1pt}{\tiny{$\Scal$}}},
\end{equation} 
where we also set $\omega_0 = 0$ without loss of generality. Consider also the machine $\Mcal$ to be composed of $N$ subsystems, $\{\Mcal_n\}_{n=1,\hdots,N}$, each of the same dimension $d$ as the target, whose local Hamiltonians are
\begin{equation}\label{eq:machHam}
    H_{\raisebox{-1pt}{\tiny{$\Mcal$}}}^{(n)}=(1+n\epsilon)H_{\raisebox{-1pt}{\tiny{$\Scal$}}} ,
\end{equation} where $\epsilon=(\beta_{\textup{max}}-\beta)/(N\beta)$. We first cool the system initially at nonzero $\beta$ to some fixed, finite $\beta_{\textup{max}}$, which we eventually take $\beta_{\textup{max}} \to \infty$ in order to asymptotically achieve perfect cooling. We treat the case $\beta = 0$ as a limiting case of $\beta \to 0$: here, as $\beta \to 0$, we let $N \rightarrow \infty$ such that $N \beta \to \infty$, e.g., we specify a suitable function $N(\beta)$ such that $N(\beta) \to \infty$ ``faster'' than $\beta \to 0$. 

We now show that, given the ability to perform a diverging number of operations on such a configuration, one can reach the target state $\tau_{\raisebox{-1pt}{\tiny{$\Scal$}}}(\beta_{\textup{max}}, H_{\raisebox{-1pt}{\tiny{$\Scal$}}})$. In particular, we show that the protocol presented uses the minimal amount of energy to do so, and explicitly calculate this to be $\beta^{-1} \widetilde{\Delta} S$ units of energy, where $\widetilde{\Delta} S:=S[\tau_{\raisebox{-1pt}{\tiny{$\Scal$}}}(\beta, H_{\raisebox{-1pt}{\tiny{$\Scal$}}})]-S[\tau_{\raisebox{-1pt}{\tiny{$\Scal$}}}(\beta_{\textup{max}}, H_{\raisebox{-1pt}{\tiny{$\Scal$}}})]$. In other words, as the number of operations in the protocol diverges, we approach perfect cooling at the Landauer limit, thereby saturating the ultimate bound. 

The diverging time cooling protocol is as follows. At each step, the target system interacts with a single machine labelled by $n$ via the \texttt{SWAP} operator
\begin{align}
    \mathbbm{S}^d_{\Scal\Mcal_n} := \sum_{i,j=0}^{d-1}\ket{i,j}\!\bra{j,i}_{\Scal\Mcal_n}.
\end{align}
As the target and machine subsystems considered here are of the same dimension, we drop the subscript on the states associated to each subsystem, for ease of notation. Such a transformation is, in general, not energy conserving, but one can calculate the energy change for both the target system and the machine due to the $n^{\textup{th}}$ interaction as
\begin{align}
    \Delta E_{\raisebox{-1pt}{\tiny{$\Scal$}}}^{(n)} = \tr{{H}_{\raisebox{-1pt}{\tiny{$\Scal$}}}\,\tau (\beta, H_{\raisebox{-1pt}{\tiny{$\Mcal$}}}^{(n)})}
  -\tr{{H}_{\raisebox{-1pt}{\tiny{$\Scal$}}}\, \tau (\beta, H_{\raisebox{-1pt}{\tiny{$\Mcal$}}}^{(n-1)})},
  \end{align}
and so the total energy change of the system over the entire $N$-step protocol is given by
    \begin{align}
  \Delta E_{\raisebox{-1pt}{\tiny{$\Scal$}}}=\sum_{n=1}^N\Delta E_{\raisebox{-1pt}{\tiny{$\Scal$}}}^{(n)}=\tr{{H}_{\raisebox{-1pt}{\tiny{$\Scal$}}}\,\tau (\beta, H_{\raisebox{-1pt}{\tiny{$\Mcal$}}}^{(N)})}
  -\tr{{H}_{\raisebox{-1pt}{\tiny{$\Scal$}}}\, \tau (\beta, H_{\raisebox{-1pt}{\tiny{$\Mcal$}}}^{(0)})}.
\end{align}
The energy change of the machine subsystem that is swapped with the target system at each step is given by
\begin{align}
\Delta E_{\raisebox{-1pt}{\tiny{$\Mcal$}}}^{(n)}=&\tr{H_{\raisebox{-1pt}{\tiny{$\Mcal$}}}^{(n)}\tau (\beta, H_{\raisebox{-1pt}{\tiny{$\Mcal$}}}^{(n-1)})}
  -\tr{H_{\raisebox{-1pt}{\tiny{$\Mcal$}}}^{(n)}\tau (\beta, H_{\raisebox{-1pt}{\tiny{$\Mcal$}}}^{(n)})}
  = \sum_{k=0}^{d-1} \,(1+n\epsilon) \omega_{k} \left[p_k (\beta, H_{\raisebox{-1pt}{\tiny{$\Mcal$}}}^{(n-1)})-p_k(\beta, H_{\raisebox{-1pt}{\tiny{$\Mcal$}}}^{(n)})\right],
  \label{eq:energy exch machine ar}
\end{align}
where $p_k(\beta, H_{\raisebox{-1pt}{\tiny{$\Mcal$}}}^{(n)}) = e^{-\beta (1+n\epsilon)\omega_k}/\mathcal{Z}_{\Mcal_n}(\beta, H_{\raisebox{-1pt}{\tiny{$\Mcal$}}}^{(n)})$ is the population in the $k^{\textup{th}}$ energy level of the thermal state of the $n^{\textup{th}}$ machine subsystem at inverse temperature $\beta$, with $\mathcal{Z}_{\Mcal_n}(\beta, H_{\raisebox{-1pt}{\tiny{$\Mcal$}}}^{(n)})=\tr{e^{-\beta H_{\raisebox{-1pt}{\tiny{$\Mcal$}}}^{(n)} }}$ being the partition function.

By summing the contributions of the energy changes in each step, one can obtain the total energy change for the overall machine throughout the entire process:
\begin{align}
 \Delta E_{\raisebox{-1pt}{\tiny{$\Mcal$}}}^{(N)}=& \sum_{n=1}^N\Delta E_{\raisebox{-1pt}{\tiny{$\Mcal$}}}^{(n)}
  = \sum_{n=1}^N \sum_{k=0}^{d-1} \, (1+n\epsilon)\omega_{k} \left[p_k (\beta, H_{\raisebox{-1pt}{\tiny{$\Mcal$}}}^{(n-1)})-p_k(\beta, H_{\raisebox{-1pt}{\tiny{$\Mcal$}}}^{(n)})\right],
  \label{eq:total energy exch machine ar}
\end{align}
In general, it is complicated to calculate the energy cost for the protocol up until a finite time step $N$, since this depends on the full energy structure of the target system and machine subsystems involved (we return to resolve this problem for the special case of equally spaced system and machine Hamiltonians in the coming section). Here, we focus on a special case in which $N \to \infty$, i.e., there is a diverging number of machine subsystems that the target system interacts with throughout the protocol. This limit physically corresponds to that of requiring a diverging amount of time (in terms of the number of steps). Furthermore, we take the limit $\epsilon \to 0$ for any fixed $\beta, \beta_{\textup{max}}$. Considering the differentials 
\begin{equation}
    \Delta p_k^{(n)} := p_k(\beta, H_{\raisebox{-1pt}{\tiny{$\Mcal$}}}^{(n)}) - p_k (\beta, H_{\raisebox{-1pt}{\tiny{$\Mcal$}}}^{(n-1)}) ,
\end{equation}
and
\begin{equation}
\Delta x_n := x_{n} - x_{n-1} \qquad \text{with} \quad x_n := 1+n \epsilon .
\end{equation} 
In order for $x_n$ to become infinitesimal, and noting the explicit form of the machine subsystem Hamiltonians $H_{\raisebox{-1pt}{\tiny{$\Mcal$}}}^{(n)} = (1+n\epsilon)H_{\raisebox{-1pt}{\tiny{$\Scal$}}}$, we can make the replacement
\begin{align}
    -\frac{\Delta p_k^{(n)}}{\Delta x_n} \, \Delta x_n \to -\frac{\partial{p_k (\beta, xH_{\raisebox{-1pt}{\tiny{$\Scal$}}})}}{\partial x} \,\textup{d}x
\end{align}
where $x:=1+n \epsilon$ has become a continuous parameter. This way we can express the limit $N\to\infty$ of Eq.~\eqref{eq:total energy exch machine ar} as a Riemann integral in the following form:
  \begin{align}
 \lim_{N\to\infty}\Delta E^{(N)}_{\raisebox{-1pt}{\tiny{$\Mcal$}}}
  =& -\int_{1}^{x_{\textup{max}}} \sum_{k=0}^{d-1} \, x \omega_{k} \,\frac{\partial p_k (\beta, xH_{\raisebox{-1pt}{\tiny{$\Scal$}}})}{\partial x}\,  \textup{d}x,
\end{align}
where $x_{\textup{max}}:=\beta_{\textup{max}}/\beta$.
Both the summation and the integral converge, so one can swap the order of their evaluation. Integrating by parts then gives
\begin{align}
   \lim_{N\to\infty}\Delta E^{(N)}_{\raisebox{-1pt}{\tiny{$\Mcal$}}}
  =& \sum_{k=0}^{d-1} \left[ - x \omega_{k} \, p_k (\beta, xH_{\raisebox{-1pt}{\tiny{$\Scal$}}})\big|_{1}^{x_{\textup{max}}}+\int_{1}^{x_\textup{{max}}}  \,  \omega_{k} \,p_k (\beta, xH_{\raisebox{-1pt}{\tiny{$\Scal$}}})\,  \textup{d}x\right] \nonumber\\
    =&  \sum_{k=0}^{d-1} \left[- x \omega_{k} \, p_k (\beta, xH_{\raisebox{-1pt}{\tiny{$\Scal$}}})\big|_{1}^{x_{\textup{max}}} \right] -\int_{1}^{x_\textup{{max}}} \frac{1}{\beta}\frac{\partial }{\partial x}\big[\log \mathcal{Z}(\beta, x H_{\raisebox{-1pt}{\tiny{$\Scal$}}})\big]\,  \textup{d}x  \nonumber\\
  =& {E}[\tau(\beta, H_{\raisebox{-1pt}{\tiny{$\Scal$}}})]-{E}[\tau(\beta, x_{\textup{max}}H_{\raisebox{-1pt}{\tiny{$\Scal$}}})]-\frac{1}{\beta}\log \mathcal{Z}(\beta, x_{\textup{max}}\, H_{\raisebox{-1pt}{\tiny{$\Scal$}}})+\frac{1}{\beta}\log \mathcal{Z}(\beta, \, H_{\raisebox{-1pt}{\tiny{$\Scal$}}}),
  \label{eq:total energy exch integral}
  \end{align}
where in the second line we again swap the order of the integral and the sum to write $\sum_{k=0}^{d-1} \omega_k p_k(\beta, x H_{\raisebox{-1pt}{\tiny{$\Scal$}}}) = -\tfrac{1}{\beta} \tfrac{\partial}{\partial x} [\log {\mathcal{Z}(\beta, x H_{\raisebox{-1pt}{\tiny{$\Scal$}}})} ]$ and in the last line we invoke ${E}[\tau(\beta,x H)]=\tr{x H\, \tau(\beta,x H)}$. Finally, writing the partition function in terms of the average energy and entropy, i.e., $ \log[ \mathcal{Z}(\beta,x H)]=-\beta \, E [\tau(\beta,x H)]+S[\tau(\beta,x H)]$, the total energy change of the machine is given by
\begin{align}
 \lim_{N\to\infty}\Delta E^{(N)}_{\raisebox{-1pt}{\tiny{$\Mcal$}}}
  =& {E}[\tau(\beta, H_{\raisebox{-1pt}{\tiny{$\Scal$}}})]-{E}[\tau(\beta, x_{\textup{max}}H_{\raisebox{-1pt}{\tiny{$\Scal$}}})]+ \, E [\tau(\beta,x_{\textup{max}} H_{\raisebox{-1pt}{\tiny{$\Scal$}}})]-\frac{1}{\beta} S[\tau(\beta,x_{\textup{max}} H)]- \, E [\tau(\beta, H_{\raisebox{-1pt}{\tiny{$\Scal$}}})]+\frac{1}{\beta}S[\tau(\beta, H_{\raisebox{-1pt}{\tiny{$\Scal$}}})] \nonumber\\
  =&\frac{1}{\beta}\big\{S[\tau(\beta, H_{\raisebox{-1pt}{\tiny{$\Scal$}}})] - S[\tau(\beta_{\textup{max}}, H_{\raisebox{-1pt}{\tiny{$\Scal$}}})]\big\}=\frac{1}{\beta}\,\widetilde{\Delta} S_{\raisebox{-1pt}{\tiny{$\Scal$}}},  \label{eq:total energy exch integral final}
\end{align}
where we make use of the property $\tau_{\raisebox{-1pt}{\tiny{$\Scal$}}}(\beta, x_{\textup{max}} H_{\raisebox{-1pt}{\tiny{$\Scal$}}})=\tau_{\raisebox{-1pt}{\tiny{$\Scal$}}}(\beta_{\textup{max}}, H_{\raisebox{-1pt}{\tiny{$\Scal$}}})$ and the entropy decrease of the target system corresponds to that associated with the transformation $\tau(\beta, H_{\raisebox{-1pt}{\tiny{$\Scal$}}})\to \tau(\beta_{\textup{max}}, H_{\raisebox{-1pt}{\tiny{$\Scal$}}})$. Thus, as the number of timesteps diverges, this cooling process saturates the Landauer limit for the heat dissipated by the machine. In order to achieve perfect cooling at the Landauer limit, i.e., the final target state to approach $\ketbra{0}{0}$ and thus prove Theorem~\ref{thm:inftimeFinTepFinDim}, we can now take the limit $\beta_{\textup{max}} \to \infty$. \end{proof}

The above proof holds for systems and machines of arbitrary (but equal) dimension, either finite or infinite, with arbitrary Hamiltonians. We now present some more detailed analysis regarding the special case where the Hamiltonians of the target system and all machine subsystems are equally spaced; this provides an opportunity both to derive a more detailed formula for the energy costs involved and to build intuition regarding some of the important differences between the finite- and infinite-dimensional settings.

\subsection{Special Case: Equally Spaced Hamiltonians}
\label{app:specialcaseequallyspacedhamiltonians}

Consider a finite $d$-dimensional target system beginning at inverse temperature $\beta$ with an equally spaced Hamiltonian $H_{\raisebox{-1pt}{\tiny{$\Scal$}}}(\omega_{\raisebox{0pt}{\tiny{$\Scal$}}}) = \omega_{\raisebox{0pt}{\tiny{$\Scal$}}} \sum_{n=0}^{d-1} n \ket{n}\!\bra{n}_{\raisebox{-1pt}{\tiny{$\Scal$}}}$. In this case, we can derive a more precise dimension-dependant function for the energy cost dissipated by the machines throughout the optimal cooling protocol presented above. 

Consider an initial target system $\tau_{\raisebox{-1pt}{\tiny{$\Scal$}}}(\beta, H_{\raisebox{-1pt}{\tiny{$\Scal$}}})$ and a diverging number $N$ of machines $\{ \mathcal{M}_{\alpha} \}_{\alpha = 0, \hdots, N}$ of the same dimension $d$ as the target, which all begin in a thermal state at inverse temperature $\beta$ with respect to an equally spaced Hamiltonian whose gaps between neighbouring energy levels $\omega_{\raisebox{0pt}{\tiny{$\Mcal_\alpha$}}}$ are ordered non-decreasingly. Each machine is used once and then discarded; the particular interaction is the aforementioned \texttt{SWAP} between the target system and the $n^\text{th}$ qudit machine, i.e., that represented by the unitary $\mathbbm{S}^d_{\raisebox{-1pt}{\tiny{$\Scal \Mcal_\alpha$}}} := \sum_{i,j=0}^{d-1}\ket{i,j}\!\bra{j,i}_{\raisebox{-1pt}{\tiny{$\Scal\Mcal_\alpha$}}}.$ After applying such an operation, the state of the target system is given by 
\begin{align}\label{eq: opt asy state max}
    \tau_{\raisebox{-1pt}{\tiny{$\Scal$}}}(\beta,\omega_{\alpha}) := \frac{e^{-\beta {H}_{\raisebox{-1pt}{\tiny{$\Scal$}}}(\omega_\alpha) }}{{\mathcal{Z}}_{\raisebox{-1pt}{\tiny{$\Scal$}}}(\beta,\omega_\alpha)},
\end{align}
where ${H}_{\raisebox{-1pt}{\tiny{$\Scal$}}}(\omega_\alpha) := \omega_{\alpha} \sum_{n=0}^{d-1} n  \ket{n}\!\bra{n}_{\raisebox{-1pt}{\tiny{$\Scal$}}}$ and ${\mathcal{Z}}_{\raisebox{-1pt}{\tiny{$\Scal$}}}(\beta,\omega_\alpha) :=\tr{e^{-\beta {H}_{\raisebox{-1pt}{\tiny{$\Scal$}}}(\omega_\alpha) }}$. 

We now calculate the energy cost explicitly for the diverging time cooling protocol, which saturates the Landauer bound in the asymptotic limit. In order to minimise the energy cost of cooling, the target system must be cooled by the qudit system in the machines with the smallest gap between neighbouring energy levels (that permits cooling) as much as possible at each stage. In order to optimally use the given machine structure at hand, we thus order the set of energy gaps $\omega_\alpha $ in non-decreasing order. In addition, the protocol to reach the Landauer erasure bound, i.e., minimal energy cost, dictates that one must infinitesimally increase $\omega_\alpha$ of the machines in order to dissipate as little heat as possible throughout the interactions. Since we are here considering a diverging time limit, we have access to a diverging number of qudit machine with distinct energy gap $\omega_{\alpha}$  at our disposal; the task is then to use these in an energy-optimal manner.

It is straightforward to see that to minimise the total energy cost, one must apply the sequence of unitaries $\mathbbm{S}^d_{\raisebox{-1pt}{\tiny{$\Scal\Mcal_\alpha$}}}$ such that $\mathbbm{S}^d_{\raisebox{-1pt}{\tiny{$\Scal\Mcal_0$}}}$ is first applied to reach the optimally cool $\tau_{\raisebox{-1pt}{\tiny{$\Scal$}}}(\beta,\omega_0)$, then $\mathbbm{S}^d_{\raisebox{-1pt}{\tiny{$\Scal\Mcal_1$}}}$ to reach $\tau_{\raisebox{-1pt}{\tiny{$\Scal$}}}(\beta,\omega_1)$, and so on. The heat dissipated by the reset machines in each stage of such a cooling protocol (i.e., for each value of $\alpha$) can thus be calculated as
\begin{align}
 \Delta E_{\raisebox{-1pt}{\tiny{$\Mcal_\alpha$}}}(\omega_{\alpha})&=-\left\{\tr{{H}_{\raisebox{-1pt}{\tiny{$\Mcal_\alpha$}}}(\omega_{\alpha})\tau_{\raisebox{-1pt}{\tiny{$\Mcal_\alpha$}}} (\beta, \omega_{\alpha})}
  +\tr{{H}_{\raisebox{-1pt}{\tiny{$\Mcal_\alpha$}}}(\omega_{\alpha})\, \tau_{\raisebox{-1pt}{\tiny{$\Mcal_\alpha$}}} (\beta, \omega_{\alpha-1})}\right\}\nonumber\\
  &=-\tr{{H}_{\raisebox{-1pt}{\tiny{$\Scal$}}}(\omega_{\alpha})\left[\tau_{\raisebox{-1pt}{\tiny{$\Scal$}}} (\beta, \omega_{\alpha})
  -\, \tau_{\raisebox{-1pt}{\tiny{$\Scal$}}} (\beta, \omega_{\alpha-1})\right]}.
  \label{eq:energy exch R}
\end{align}
In the second line, we have made use of the fact that the Hamiltonians of both the target system and each of machine are $d$-dimensional and equally spaced. So far, we have obtained the energy dissipated by the reset machines. To investigate the total energy cost of cooling in such a process, we also must consider the contribution of energy transferred to the target system $\Scal$, which is characterised via its local Hamiltonian $H_{\raisebox{-1pt}{\tiny{$\Scal$}}}$ and calculated via
\begin{equation}
  \Delta {E}_{\raisebox{-1pt}{\tiny{$\Scal$}}}(\omega_{\alpha})=
  \tr{{H}_{\raisebox{-1pt}{\tiny{$\Scal$}}}(\omega_{\raisebox{0pt}{\tiny{$\Scal$}}})\,\tau_{\raisebox{-1pt}{\tiny{$\Scal$}}} (\beta, \omega_{\alpha})}
  -\tr{{H}_{\raisebox{-1pt}{\tiny{$\Scal$}}}(\omega_{\raisebox{0pt}{\tiny{$\Scal$}}})\, \tau_{\raisebox{-1pt}{\tiny{$\Scal$}}} (\beta, \omega_{\alpha-1})},
  \label{eq:energy exch SL}
\end{equation}
in which we set $\omega_0=\omega_\Scal$. Using Eqs.~(\ref{eq:energy exch R},~\ref{eq:energy exch SL}), the total energy cost for each stage of cooling is given by
\begin{equation}
    \Delta {E}_{\raisebox{-1pt}{\tiny{$\Scal \Mcal$}}}(\omega_\alpha)=\Delta {E}_{\raisebox{-1pt}{\tiny{$\Scal$}}}(\omega_{\alpha})+\Delta {E}_{\raisebox{-1pt}{\tiny{$\Mcal$}}}(\omega_{\alpha})= \mathrm{tr}\left\{\big[{H}_{\raisebox{-1pt}{\tiny{$\Scal$}}}(\omega_{\Scal})-{H}_{\raisebox{-1pt}{\tiny{$\Scal$}}}(\omega_{\alpha})\big]\big[\tau_{\raisebox{-1pt}{\tiny{$\Scal$}}} (\beta, \omega_{\alpha})-\tau_{\raisebox{-1pt}{\tiny{$\Scal$}}}(\beta,\omega_{\alpha-1})\big]\right\},
    \label{eq:tot energy memory}
\end{equation}
which leads to the overall energy cost after $N$ stages, where $N$ is the number of non-zero distinct energy gaps of the reset machines, as
\begin{align}
    \Delta {E}_{\raisebox{-1pt}{\tiny{$\Scal \Mcal$}}}^{(N)}&= \sum_{\alpha = 1}^N \Delta {E}_{\raisebox{-1pt}{\tiny{$\Scal \Mcal$}}}(\omega_\alpha) 
  =\sum_{\alpha=1}^N \mathrm{tr}\left\{\big[{H}_{\raisebox{-1pt}{\tiny{$\Scal$}}}(\omega_{\Scal})-{H}_{\raisebox{-1pt}{\tiny{$\Scal$}}}(\omega_{\alpha})\big]\big[\tau_{\raisebox{-1pt}{\tiny{$\Scal$}}} (\beta, \omega_{\alpha})-\tau_{\raisebox{-1pt}{\tiny{$\Scal$}}}(\beta,\omega_{\alpha-1})\big]\right\}
  \label{eq:tot energy n}.
\end{align}
Now, we can obtain the total energy cost for each stage of the protocol (i.e., each value of $\alpha$ considered) in terms of the transformation of the target system alone. Note that in this protocol, each stage corresponding to each of the $N$ distinct energy gaps $\{ \omega_\alpha\}$ in itself requires only one operation to perfectly reach $\tau_{\raisebox{-1pt}{\tiny{$\Scal$}}}(\beta, \omega_\alpha)$. The end result of this protocol is that the target system is cooled from the initial thermal state $\tau_{\raisebox{-1pt}{\tiny{$\Scal$}}}(\beta,\omega_{\raisebox{0pt}{\tiny{$\Scal$}}})$, where $\omega_{\raisebox{0pt}{\tiny{$\Scal$}}}$ is the energy gap between each pair of adjacent energy levels in the system, to $\tau_{\raisebox{-1pt}{\tiny{$\Scal$}}}(\beta,\omega_{\textup{max}})$ in the energy-optimal manner. 

Starting from Eq.~\eqref{eq:tot energy n}, we have
\begin{align}
    \Delta {E}_{\raisebox{-1pt}{\tiny{$\Scal \Mcal$}}}^{(N)}&=\sum_{\alpha=1}^{N} \mathrm{tr}\left\{\big[{H}_{\raisebox{-1pt}{\tiny{$\Scal$}}}(\omega_{\raisebox{0pt}{\tiny{$\Scal$}}})-{H}_{\raisebox{-1pt}{\tiny{$\Scal$}}}(\omega_{\alpha})\big]\big[\tau_{\raisebox{-1pt}{\tiny{$\Scal$}}} (\beta, \omega_{\alpha})-\tau_{\raisebox{-1pt}{\tiny{$\Scal$}}}(\beta,\omega_{\alpha-1})\big]\right\}\nonumber\\
    &=\sum_{\alpha=1}^{N} (\omega_{\raisebox{0pt}{\tiny{$\Scal$}}}-\omega_\alpha)\left[\left(\frac{e^{-\beta \omega_{\alpha}}}{1-e^{-\beta \omega_{\alpha}}}-\frac{e^{-\beta \omega_{\alpha-1}}}{1-e^{-\beta \omega_{\alpha-1}}}\right)-\left(\frac{d\,e^{-\beta d \omega_{\alpha}}}{1-e^{-\beta d \omega_{\alpha}}}-\frac{d e^{-\beta d \omega_{\alpha-1}}}{1-e^{-\beta d \omega_{\alpha-1}}}\right)\right]\nonumber\\
    &=\lim_{K\to\infty} \sum_{\alpha=1}^{N} (\omega_{\raisebox{0pt}{\tiny{$\Scal$}}}-\omega_\alpha)\sum_{k=0}^{K}\big[\big(e^{-\beta (k+1) \omega_{\alpha}}-e^{-\beta (k+1) \omega_{\alpha-1}}\big)-d\, \big(e^{-\beta (k+1)d \omega_{\alpha}}-e^{-\beta (k+1)d \omega_{\alpha-1}}\big)\big]\nonumber\\
     &=\lim_{K\to\infty}\sum_{\alpha=1}^{N} (\omega_{\raisebox{0pt}{\tiny{$\Scal$}}}-\omega_\alpha)\sum_{k=0}^{K}\big[e^{-\beta (k+1) \omega_{\alpha}}\big(1-e^{-\beta (k+1) (\omega_{\alpha-1}-\omega_{\alpha})}\big)-d\, e^{-\beta d(k+1) \omega_{\alpha}}\big(1-e^{-\beta d(k+1) (\omega_{\alpha-1}-\omega_{\alpha})}\big)\big].
    \label{eq:tot energy d}
\end{align}
Here, since both ${H}_{\raisebox{-1pt}{\tiny{$\Mcal_\alpha$}}}$ and $H_{\raisebox{-1pt}{\tiny{$\Scal$}}}$ are equally spaced Hamiltonians, the average energy can be written as 
\begin{align}
    E(\omega_x, \omega_y) = \tr{{H}_{\raisebox{-1pt}{\tiny{$\Scal$}}}(\omega_x)\,\tau_{\raisebox{-1pt}{\tiny{$\Scal$}}} (\beta, \omega_y)}= \frac{\sum_{n=0}^{d-1}\, n \omega_{x}e^{-n\beta \omega_y}}{\sum_{n=0}^{d-1}\, e^{-n\beta \omega_y}}=\omega_{x}\left(\frac{e^{-\beta \omega_{y}}}{1-e^{-\beta \omega_{y}}}-\frac{d\,e^{-\beta d\, \omega_{y}}}{1-e^{-\beta d\,\omega_{y}}}\right)
\end{align}
by evaluating the geometric series 
\begin{align}
    \mathcal{Z}(\beta, \omega_y) = \sum_{n=0}^{d-1} e^{- \beta n \omega_y} = \tfrac{1-e^{-\beta d \omega_y}}{1-e^{-\beta \omega_y}}
\end{align}
and writing 
\begin{align}
    E(\omega_x, \omega_y) = \sum_{n=0}^{d-1} n \omega_x \tfrac{e^{- \beta n \omega_y}}{\mathcal{Z}(\beta,\omega_y)} = \tfrac{\omega_x}{\omega_y} \left\{ - \tfrac{\partial}{\partial \beta} \log{\left[ \mathcal{Z}(\beta, \omega_y)\right]} \right\} = - \tfrac{\omega_x}{\omega_y} \tfrac{\partial}{\partial \beta} \left[ \log{\left( 1-e^{-\beta d \omega_y}\right) - \log{\left( 1-e^{-\beta \omega_y}\right)}} \right]
\end{align}
as we do in the second line of Eq.~\eqref{eq:tot energy d} and then using the infinite series expression $(1-x)^{-1}=\lim_{K\to\infty}\sum_{k=0}^K x^k$ for any $|x|<1$ as per the third line.

As we will see in Appendix~\ref{app:hodivergingtimegaussian}, the energy cost for cooling an infinite-dimensional system when both target and machines have equally spaced Hamiltonians (i.e., harmonic oscillators) is similar to the form of Eq.~\eqref{eq:tot energy d}. Importantly, the second term in square parenthesis vanishes as $d\to\infty$, simplifying the expression even further.

We now assume that the energy gaps of the machine are given by $\omega_\alpha =\omega_{\raisebox{0pt}{\tiny{$\Scal$}}}+\epsilon \alpha$ and so the total energy cost can be written as follows:
\begin{align}
    \Delta {E}_{\raisebox{-1pt}{\tiny{$\Scal \Mcal$}}}^{(N)}
     &=- \lim_{K\to\infty} \sum_{\alpha=1}^{N} \alpha\epsilon \sum_{k=0}^{K} e^{-\beta k (\omega_{\raisebox{0pt}{\tiny{$\Scal$}}}+\alpha \epsilon)}\big(1-e^{\beta k  \epsilon}\big)+\lim_{K\to\infty} \sum_{\alpha=1}^{N} \alpha d \epsilon  \sum_{k=0}^{K}\, e^{-\beta k d (\omega_{\raisebox{0pt}{\tiny{$\Scal$}}}+\alpha \epsilon)}\big(1-e^{\beta k d  \epsilon}\big)\nonumber\\
     &= \lim_{K\to\infty} \sum_{k=0}^{K} \big[e^{-\beta k \omega_{\raisebox{0pt}{\tiny{$\Scal$}}}}\big(e^{\beta k  \epsilon}-1\big)\big(\sum_{\alpha=1}^{N} \alpha\epsilon e^{-\beta k \alpha \epsilon}\big)\big]-\lim_{K\to\infty} \sum_{k=0}^{K}e^{-\beta k d \omega_{\raisebox{0pt}{\tiny{$\Scal$}}}}\big[\big(e^{\beta k d  \epsilon}-1\big)\big(\sum_{\alpha=1}^{N} d \alpha\epsilon\, e^{-\beta kd \alpha \epsilon}\big)\big],
    \label{eq:tot energy alphaepsilon}
\end{align}
where we can swap the order of summation since both sums converge and the summands are non-positive. This can be seen from the first line above, using the fact that $e^{-\alpha x}(1-e^x) \in [-1,0]$ for all $\alpha\geq 1 $ and $ x \geq 0$. We now calculate the sum over $\alpha$.
\begin{align}
\sum_{\alpha=1}^{N} \, \alpha\epsilon\, e^{-\beta  \alpha \epsilon}&= -\frac{\partial }{\partial \beta}\sum_{\alpha=0}^{N}   e^{-\beta \alpha \epsilon}=   -\frac{\partial }{\partial \beta}\left(\frac{1-e^{-\beta (N+1)\epsilon}}{1-e^{-\beta \epsilon}}\right)\nonumber\\
  &= -\left(\frac{(N+1)\epsilon e^{-\beta (N+1)\epsilon} -(N+1)\epsilon e^{-\beta (N+2)\epsilon}-\epsilon e^{-\beta \epsilon}+\epsilon e^{-\beta (N+2)\epsilon}}{(1-e^{-\beta \epsilon})^2}\right)\nonumber\\
 & = \frac{\epsilon e^{-\beta \epsilon }}{(1-e^{-\beta \epsilon})^2}\big(1-(N+1) e^{-\beta N\epsilon} +N e^{-\beta (N+1)\epsilon}\big)\nonumber\\
 &=\frac{\epsilon e^{-\beta \epsilon }}{(1-e^{-\beta \epsilon})^2}\left(1- e^{-\beta N\epsilon} -N e^{-\beta N\epsilon}(1-e^{-\beta\epsilon})\right).
\label{eq:sum 1}
\end{align}
Combining Eqs.~(\ref{eq:tot energy alphaepsilon}) and~(\ref{eq:sum 1}), we arrive at 
\begin{align}
    \Delta {E}_{\raisebox{-1pt}{\tiny{$\Scal \Mcal$}}}^{(N)}
     &= \lim_{K\to\infty}\sum_{k=0}^{K} \left[\frac{e^{-\beta k \omega_{\raisebox{0pt}{\tiny{$\Scal$}}}}}{k}\frac{k\epsilon (1- e^{-\beta Nk\epsilon})}{(1-e^{-\beta k \epsilon})}- N \epsilon \, e^{-\beta k ( \omega_{\raisebox{0pt}{\tiny{$\Scal$}}}+N \epsilon)} \right]\nonumber\\
     &
     - \lim_{K\to\infty} \sum_{k=0}^{K} \left[\frac{e^{-\beta k d \omega_{\raisebox{0pt}{\tiny{$\Scal$}}}}}{k}\frac{kd\epsilon (1- e^{-\beta Nkd\epsilon})}{(1-e^{-\beta kd \epsilon})}- N d \epsilon \, e^{-\beta kd ( \omega_{\raisebox{0pt}{\tiny{$\Scal$}}}+N \epsilon)} \right].
    \label{eq:tot energy alphaepsilon 2}
\end{align}
In order to optimise the energy cost, we now assume that the energy gaps of the machines can be chosen to be smoothly increasing in such way that $\epsilon=\Delta \omega/N :=(\omega_{\textup{max}}-\omega_{\raisebox{0pt}{\tiny{$\Scal$}}})/N$. Substituting this expression for $\epsilon$ into the above equation yields
\begin{align}
    \Delta {E}_{\raisebox{-1pt}{\tiny{$\Scal \Mcal$}}}^{(N)} &= \lim_{K\to\infty}\sum_{k=0}^{K} \left[\frac{e^{-\beta k \omega_{\raisebox{0pt}{\tiny{$\Scal$}}}}}{k}\frac{k\Delta \omega (1- e^{-\beta k\Delta \omega})}{N(1-e^{-\beta k \frac{\Delta \omega}{N}})}- \Delta \omega \, e^{-\beta k ( \omega_{\raisebox{0pt}{\tiny{$\Scal$}}}+\Delta \omega)} \right]\nonumber\\ 
    &- \lim_{K\to\infty} \sum_{k=0}^{K} \left[\frac{e^{-\beta k d \omega_{\raisebox{0pt}{\tiny{$\Scal$}}}}}{k}\frac{kd\Delta \omega (1- e^{-\beta kd\Delta \omega})}{N(1-e^{-\beta kd \frac{\Delta \omega}{N}})}- d \Delta \omega \, e^{-\beta kd ( \omega_{\raisebox{0pt}{\tiny{$\Scal$}}}+ \Delta \omega)} \right].
\end{align}
We now wish to take the limit of $N \gg K \to \infty$. This assumption means that energy change of the system is approximately equal to its free energy change; in other words, the process occurs quasi-adiabatically. The ability to switch the order of taking the limits of $K$ and $N$ going to $\infty$ follows from the monotonic convergence of the sum over $k$. In particular, note that the term inside square parentheses in each summand converges and the first term in each summation (which is the only part that depends on $N$) is positive and bounded.

Under this assumption, we can use the approximation $\lim_{\beta x\to 0}\, \frac{x}{1-e^{-\beta x}}=\frac{1}{\beta}$; since $0 < e^{-\beta x} < 1$ for any positive $x$, the sum over $k$ converges to a finite value. In general, this approximation introduces a correction term for the energy change, however under said assumption the error incurred becomes negligible. Then, the total energy change $\Delta {E}_{\raisebox{-1pt}{\tiny{$\Scal \Mcal$}}}^{\textup{tot}}$ for the transformation $ \tau_{\raisebox{-1pt}{\tiny{$\Scal$}}}(\beta, \omega_{\raisebox{0pt}{\tiny{$\Scal$}}})\to \tau_{\raisebox{-1pt}{\tiny{$\Scal$}}}(\beta, \omega_{\textup{max}})$ throughout the overall process is
\begin{align}
\Delta {E}_{\raisebox{-1pt}{\tiny{$\Scal \Mcal$}}}^{\textup{tot}}
     &= \lim_{K\to\infty} \sum_{k=0}^{K} \left[\frac{e^{-\beta k \omega_{\raisebox{0pt}{\tiny{$\Scal$}}}}}{\beta k}-\frac{e^{-\beta k \omega_{\textup{max}}}}{\beta k}
    - (\omega_{\textup{max}}-\omega_{\raisebox{0pt}{\tiny{$\Scal$}}}) \, e^{-\beta k  \omega_{\textup{max}}} \right]\nonumber\\
     &
     -\lim_{K\to\infty}\sum_{k=0}^{K} \left[\frac{e^{-\beta kd \omega_{\raisebox{0pt}{\tiny{$\Scal$}}}}}{\beta k}-\frac{e^{-\beta k d\omega_{\textup{max}}}}{\beta k}
    - d (\omega_{\textup{max}}-\omega_{\raisebox{0pt}{\tiny{$\Scal$}}}) \, e^{-\beta k d \omega_{\textup{max}}} \right].
    \label{eq:tot energy alphaepsilon N infty}
\end{align}
As a side remark, note that here one can see that in the special case of equally spaced Hamiltonians, one indeed requires a diverging number of machine subsystems to attain perfect cooling at the Landauer limit, as this is the only way to fulfil the condition of Theorem~\ref{thm:variety}. This follows from the fact that the approximation $\frac{x}{1-e^{-\beta x}}\approx\frac{1}{\beta}$ only holds for small $\beta x$ and in general one would need to include higher-order terms that lead to an increase in energy cost.

We then have, using the expression for $E(\omega_x, \omega_y)$ derived earlier: 
\begin{align}
\Delta {E}_{\raisebox{-1pt}{\tiny{$\Scal \Mcal$}}}^{\textup{tot}} 
    &= -\frac{1}{\beta}\log (1-e^{-\beta \omega_{\raisebox{0pt}{\tiny{$\Scal$}}}})+\frac{1}{\beta}\log (1-e^{-\beta \omega_{\textup{max}}})-\frac{(\omega_{\textup{max}}-\omega_{\raisebox{0pt}{\tiny{$\Scal$}}}) \, e^{-\beta \omega_{\textup{max}}}}{1- \, e^{-\beta  \omega_{\textup{max}}}}\nonumber\\
    &+\frac{1}{\beta}\log (1-e^{-\beta d \omega_{\raisebox{0pt}{\tiny{$\Scal$}}}})-\frac{1}{\beta}\log (1-e^{-\beta d \omega_{\textup{max}}})+\frac{d (\omega_{\textup{max}}-\omega_{\raisebox{0pt}{\tiny{$\Scal$}}}) \, e^{-\beta d \omega_{\textup{max}}}}{1- \, e^{-\beta  d \omega_{\textup{max}}}}\nonumber\\
  &  = \frac{1}{\beta}\log \left(\frac{1-e^{-\beta d \omega_{\raisebox{0pt}{\tiny{$\Scal$}}}}}{1-e^{-\beta \omega_{\raisebox{0pt}{\tiny{$\Scal$}}}}}\right)- \frac{1}{\beta}\log \left(\frac{1-e^{-\beta d \omega_{\textup{max}}}}{1-e^{-\beta \omega_{\textup{max}}}}\right)
 -(\omega_{\textup{max}}-\omega_{\raisebox{0pt}{\tiny{$\Scal$}}})\left(\frac{ \, e^{-\beta \omega_{\textup{max}}}}{1- \, e^{-\beta  \omega_{\textup{max}}}}-\frac{d \, e^{-\beta d \omega_{\textup{max}}}}{1- \, e^{-\beta  d \omega_{\textup{max}}}}\right)\nonumber\\
 &=\frac{1}{\beta} \log [\mathcal{Z}_{\raisebox{-1pt}{\tiny{$\Scal$}}}(\beta, \omega_{\raisebox{0pt}{\tiny{$\Scal$}}})]-\frac{1}{\beta} \log [\mathcal{Z}_{\raisebox{-1pt}{\tiny{$\Scal$}}}(\beta,\omega_{\textup{max}})]- \tr{{H}_{\raisebox{-1pt}{\tiny{$\Scal$}}}(\omega_{\textup{max}})\,\tau_{\raisebox{-1pt}{\tiny{$\Scal$}}}(\beta, \omega_{\textup{max}})}+\tr{{H}_{\raisebox{-1pt}{\tiny{$\Scal$}}}(\omega_{\raisebox{0pt}{\tiny{$\Scal$}}})\,\tau_{\raisebox{-1pt}{\tiny{$\Scal$}}}(\beta, \omega_{\textup{max}})}\notag\\
  &=\frac{1}{\beta} \log [\mathcal{Z}_{\raisebox{-1pt}{\tiny{$\Scal$}}}(\beta , \omega_{\raisebox{0pt}{\tiny{$\Scal$}}})]-\frac{1}{\beta} \log [\mathcal{Z}_{\raisebox{-1pt}{\tiny{$\Scal$}}}(\beta , \omega_{\textup{max}})]\nonumber\\
  &- \tr{{H}_{\raisebox{-1pt}{\tiny{$\Scal$}}}(\omega_{\textup{max}})\,\tau_{\raisebox{-1pt}{\tiny{$\Scal$}}}(\beta, \omega_{\textup{max}})}+\tr{{H}_{\raisebox{-1pt}{\tiny{$\Scal$}}}(\omega_{\raisebox{0pt}{\tiny{$\Scal$}}})\,\tau_{\raisebox{-1pt}{\tiny{$\Scal$}}}(\beta, \omega_{\raisebox{0pt}{\tiny{$\Scal$}}})}-\tr{{H}_{\raisebox{-1pt}{\tiny{$\Scal$}}}(\omega_{\raisebox{0pt}{\tiny{$\Scal$}}})\,\tau_{\raisebox{-1pt}{\tiny{$\Scal$}}}(\beta, \omega_{\raisebox{0pt}{\tiny{$\Scal$}}})}+\tr{{H}_{\raisebox{-1pt}{\tiny{$\Scal$}}}(\omega_{\raisebox{0pt}{\tiny{$\Scal$}}})\,\tau_{\raisebox{-1pt}{\tiny{$\Scal$}}}(\beta, \omega_{\textup{max}})}\nonumber\\
 & =\frac{1}{\beta}{\Delta} S_{\raisebox{-1pt}{\tiny{$\Scal$}}}
 +\Delta E_{\raisebox{-1pt}{\tiny{$\Scal$}}},
    \label{eq:tot energy alphaepsilon N infty 1}
\end{align}
where we have explicitly written the von Neumann entropy $S(\varrho) = -\tr{\varrho \log (\varrho)}$ of a thermal state at inverse temperature $\beta$ as $S[\tau_{\raisebox{-1pt}{\tiny{$\Scal$}}}(\beta,\omega)]= \log [\mathcal{Z}_{\raisebox{-1pt}{\tiny{$\Scal$}}}(\beta ,\omega)] +\beta \, E [\tau_{\raisebox{-1pt}{\tiny{$\Scal$}}}(\beta ,\omega)]$. Since the energy change of the target system only concerns its local Hamiltonian, we immediately see that the heat dissipated by the resetting of machines in such a cooling process, i.e., $\Delta E_{\raisebox{-1pt}{\tiny{$\Mcal$}}}$, saturates the Landauer bound as it is equal to $\beta^{-1} {\Delta}S_{\raisebox{-1pt}{\tiny{$\Scal$}}}$. The process described is thus energy-optimal.


\section{Conditions for Structural and Control Complexity}
\label{app:conditionsstructuralcontrolcomplexity}

Here we begin by considering the protocol-independent structural conditions that must be fulfilled by the machine Hamiltonian to enable \emph{(1) perfect cooling} and \emph{(2) cooling at Landauer cost}; combined, these independent conditions provide a necessary requirement, namely that the machine must be infinite-dimensional with a spectrum that is unbounded (from above) for the \emph{possibility} of \emph{(3) perfect cooling at the Landauer limit}. We then turn to analyse the control complexity, which concerns the properties of the interaction that implements a given protocol. The properties of the machine Hamiltonian define the \emph{structural complexity}, which set the potential for how cool the target system can be made and at what energy cost; the extent to which a machine's potential is utilised in a particular protocol then depends on the properties of the joint unitary, i.e., the \emph{control complexity}. Here, we show that it is necessary that any protocol achieving perfect cooling at the Landauer limit involves interactions between the target and infinitely-many levels of the machine to realise the full cooling potential. We then analyse some sufficient conditions that arise as observations from our diverging control complexity protocols. This then leads us to demonstrate that individual degrees of freedom of the machine must be addressed in a fine-tuned manner to permute populations, highlighting that an operationally meaningful notion of control complexity must take into account factors beyond the effective dimensionality. 

\subsection{Necessary Complexity Conditions}
\label{app:necessaryconditions}

\subsubsection{Necessary Structural Conditions}
\label{app:necessarystructuralconditions}

\emph{1. Perfect Cooling.\textemdash }Let us consider the task of perfect cooling, independently from protocol-specific constraints, in the envisaged setting. One can lower bound the smallest eigenvalue $\lambda_{\textup{min}}$ of the final state $\varrho_{\raisebox{-1pt}{\tiny{$\Scal$}}}'$ (and hence how cold the system can become) after \emph{any} unitary interaction with a thermal machine by~\cite{Reeb_2014}
\begin{align}\label{eq:generalpuritybound}
    \lambda_{\textup{min}}(\varrho_{\raisebox{-1pt}{\tiny{$\Scal$}}}^\prime) \geq e^{-\beta \omega_{\raisebox{0pt}{\tiny{$\Mcal$}}}^{\textup{max}}} \lambda_{\textup{min}}(\varrho_{\raisebox{-1pt}{\tiny{$\Scal$}}}),
\end{align}
where $\omega_{\raisebox{0pt}{\tiny{$\Mcal$}}}^{\textup{max}}:=\max_{i,j}|\omega_{j}-\omega_{i}|$ denotes the largest energy gap of the machine Hamiltonian $H_{\raisebox{-1pt}{\tiny{$\Mcal$}}}$ with eigenvalues $\omega_{i}$. Without loss of generality, throughout this paper we set the ground-state energy of any system to be zero, i.e., $\omega_0 = 0$, such that the largest energy gap coincides with the largest energy eigenvalue. As we make no restrictions on the size or structure of the target or machine, the above inequality pertains to cooling protocols that could, for instance, be realised via sequences of unitaries on the target and parts of the machine. It follows that perfect cooling is only possible under two conditions: either the machine begins in a pure state ($\beta\to\infty$), or $H_{\raisebox{-1pt}{\tiny{$\Mcal$}}}$ is unbounded, i.e., $\omega_{\raisebox{0pt}{\tiny{$\Mcal$}}}^{\textup{max}}\to\infty$. Requiring $\beta<\infty$, a diverging energy gap in the machine Hamiltonian is thus a necessary structural condition for perfect cooling. Indeed, the largest energy gap of the machine plays a crucial role in limiting how cool the target system can be made (see also, e.g., Refs.~\cite{Allahverdyan_2011,Clivaz_2019L}). We now detail an independent property that is required for cooling with minimal energetic cost. 

\emph{2. Cooling at the Landauer Limit.\textemdash }Suppose now that one wishes to cool an initial target state $\tau_{\raisebox{-1pt}{\tiny{$\Scal$}}}(\beta, H_{\raisebox{-1pt}{\tiny{$\Scal$}}})$ to any thermal state $\tau^\prime_{\raisebox{-1pt}{\tiny{$\Scal$}}}(\beta^*, H_{\raisebox{-1pt}{\tiny{$\Scal$}}})$ with $\beta^*>\beta$ (not necessarily close to a pure state), at an energy cost saturating the Landauer limit. In Ref.~\cite{Reeb_2014}, it was shown that for any finite-dimensional machine, there are correction terms to the Landauer bound, which imply that it cannot be saturated; these terms vanish only in the limit where the machine dimension diverges. Thus, a necessary condition for achieving cooling with energy cost at the Landauer limit is provided by the following:

\begin{thm}\label{thm:landauerstructural}
To cool a target system $\tau_{\raisebox{-1pt}{\tiny{$\Scal$}}}(\beta, H_{\raisebox{-1pt}{\tiny{$\Scal$}}})$ to $\tau_{\raisebox{-1pt}{\tiny{$\Scal$}}}(\beta^*,H_{\raisebox{-1pt}{\tiny{$\Scal$}}})$, with $\beta^* > \beta$, using a machine in the initial state $\tau_{\raisebox{-1pt}{\tiny{$\Mcal$}}}(\beta, H_{\raisebox{-1pt}{\tiny{$\Mcal$}}})$ with energy cost at the Landauer limit, the machine must be infinite dimensional.
\end{thm}
As we will discuss below, this minimal requirement for the notion of complexity is far from sufficient to achieve cooling at Landauer cost. 

\emph{3. Perfect Cooling at the Landauer Limit.\textemdash }We have two independent necessary conditions on the structure of the machine that must be asymptotically achieved to enable relevant goals for cooling: the former is required to achieve perfect cooling; the latter for cooling at the Landauer limit. Together, these conditions imply that in order to achieve perfect cooling at the Landauer limit, one must have an infinite-dimensional machine with a spectrum that is unbounded (from above), as stated in Corollary~\ref{cor:structuralcondition}.

Henceforth, we assume that these conditions are satisfied by the machine. The question then becomes: \emph{how does one engineer an interaction between the target system and machine to achieve perfect cooling at Landauer cost?} 

\subsubsection{Necessary Control Complexity Conditions}
\label{app:necessarycontrolcomplexityconditions}

The unbounded structural properties of the machine support the \emph{possibility} for perfect cooling at the Landauer limit; however, we now focus on the control properties of the interaction that \emph{realise} said potential (see Fig.~\ref{fig:complexity}). This leads to the distinct notion of \emph{control complexity}, which aims to differentiate between protocols that access the machine in a more or less complex manner. The structural complexity properties are protocol independent and related to the energy spectrum and dimensionality of the machine, whereas the control complexity concerns properties of the unitary that represents a particular protocol. For instance, the diverging-time protocol previously outlined comprises a sequence of interactions, each of which is individually not very complex; at the same time, the unconstrained control complexity protocol accesses the total (overall infinite-dimensional) machine ``at once'', and thus the number of (nontrivial) terms in the interaction Hamiltonian, or the effective dimensionality of the machine accessed by the unitary, becomes unbounded. Nonetheless, the net energy cost of this protocol with unconstrained control complexity remains in accordance with the Landauer limit, as the initial and final states of both the system and machine are identical to those in the diverging-time protocol. 

\emph{Effective Dimensionality.---}We begin by considering the effective dimensionality accessed (nontrivially) by a unitary, whose divergence is necessary but insufficient for achieving perfect cooling at the Landauer limit, as we show in the next section. This in turn motivates the desire for a more detailed notion of control complexity that takes into account the energy-level structure of the machine. 

We define the effective dimension as the dimension of the subspace of the global Hilbert space upon which the unitary acts nontrivially, which can be quantified via the minimum dimension of a subspace $\mathcal{A}$ of the joint Hilbert space $\mathscr{H}_{\raisebox{-1pt}{\tiny{$\Scal\Mcal$}}}$ in terms of which the unitary can be decomposed as $U_{\raisebox{-1pt}{\tiny{$\Scal\Mcal$}}} = U_{\raisebox{-1pt}{\tiny{$\Acal$}}} \oplus \mathbbm{1}_{\raisebox{-1pt}{\tiny{$\Acal^\perp$}}}$, i.e.,
\begin{align}\label{eq:appeffectivedimension}
    d^{\,\textup{eff}} := \min \mathrm{dim}(\mathcal{A}) : U_{\raisebox{-1pt}{\tiny{$\Scal\Mcal$}}} = U_{\raisebox{-1pt}{\tiny{$\Acal$}}} \oplus \mathbbm{1}_{\raisebox{-1pt}{\tiny{$\Acal^\perp$}}}.
\end{align}
One can relate this quantity to properties of the Hamiltonian that generates the evolution in a finite unit of time $T$ (which we can set equal to unity without loss of generality) by considering the interaction picture. In general, any global unitary $U_{\raisebox{-1pt}{\tiny{$\Scal\Mcal$}}} = e^{-i H_{\raisebox{-1pt}{\tiny{$\Scal\Mcal$}}} T}$ is generated by a Hamiltonian of the form $H_{\raisebox{-1pt}{\tiny{$\Scal\Mcal$}}} = H_{\raisebox{-1pt}{\tiny{$\Scal$}}}\otimes \mathbbm{1}_{\raisebox{-1pt}{\tiny{$\Mcal$}}} + \mathbbm{1}_{\raisebox{-1pt}{\tiny{$\Scal$}}}\otimes H_{\raisebox{-1pt}{\tiny{$\Mcal$}}} + H_{\raisebox{-1pt}{\scriptsize{\textup{int}}}}$. However, all protocols considered in this work have vanishing local terms, i.e., $H_{\raisebox{-1pt}{\tiny{$\Scal$}}}= H_{\raisebox{-1pt}{\tiny{$\Mcal$}}}=0$. More generally, one can argue that the local terms play no role in how the machine is used to cool the target. As such, one can consider unitaries generated by only the nontrivial term $H_{\raisebox{-1pt}{\scriptsize{\textup{int}}}}$ to be those representing a particular protocol of interest. That is, we can restrict our attention to $U_{\raisebox{-1pt}{\tiny{$\Scal\Mcal$}}} = e^{-i H_{\raisebox{-1pt}{\scriptsize{\textup{int}}}} T}$, where $H_{\raisebox{-1pt}{\scriptsize{\textup{int}}}}$ is a Hermitian operator on $\mathscr{H}_{\raisebox{-1pt}{\tiny{$\Scal\Mcal$}}}$ of the form $\sum_i A^i_{\raisebox{-1pt}{\tiny{$\Scal$}}}\otimes B^i_{\raisebox{-1pt}{\tiny{$\Mcal$}}}$ such that none of the $A^i_{\raisebox{-1pt}{\tiny{$\Scal$}}}, B^i_{\raisebox{-1pt}{\tiny{$\Mcal$}}}$ are proportional to the identity operator. In doing so, it follows that the effective dimension corresponds to $\mathrm{rank}(H_{\raisebox{-1pt}{\scriptsize{\textup{int}}}})$. Lastly, note that the above definition in terms of a direct sum decomposition provides an upper bound on any similar quantification of effective dimensionality based on other tensor factorisations of the joint Hilbert space considered and makes no assumption about the underlying structure. On the other hand, knowledge of said structure would permit a more meaningful notion of complexity to be defined. For instance, the effective dimensionality of a unitary acting on a many qubit system is better captured by considering its decomposition into a tensor product factorisation rather than the direct sum. We leave the exploration of such considerations to future work.

The effective dimensionality provides a minimal quantifier for a notion of control complexity, insofar as its divergence is necessary for saturating the Landauer bound, as we prove in the next section. In fact, we prove a slightly stronger statement, namely that the dimension of the machine Hilbert space to which the unitary (nontrivially) couples the target system to must diverge. However, as we discuss below, $d^{\,\textup{eff}} \to \infty$ is generally insufficient to achieve said goal, and fine-tuned control is required. Nonetheless, the manifestation of such control seems to be system dependent, precluding our ability (so far) to present a universal quantifier of control complexity. Thus, even though further conditions need to be met to achieve perfect cooling at minimal energy cost in unit time (see Theorem~\ref{thm:maineigenvaluecondition}), whenever we talk of an operation with finite control complexity, we mean those represented by a unitary that acts (nontrivially) only on a finite-dimensional subspace of the target system and machine. In contrast, by diverging control complexity, we mean a unitary that couples the target (nontrivially) to a full basis of the machine's Hilbert space, whose dimension diverges. With this notion at hand, we have Theorem~\ref{thm:variety}, which is proven below. Intuitively, we show that if a protocol accesses only a finite-dimensional subspace of the machine, then the machine is effectively finite dimensional inasmuch as a suitable replacement can be made while keeping all quantities relevant for cooling invariant. Invoking then the main result of Ref.~\cite{Reeb_2014}, there are finite-dimensional correction terms that then imply that the Landauer limit cannot be saturated.

Note finally that in Theorem~\ref{thm:variety} no particular structure of the systems is presupposed and the effective dimensionality relates to various notions of complexity put forth throughout the literature (see, e.g., Refs.~\cite{Ladyman2013,Holovatch_2017}). For instance, for a finite-dimensional target system with equally spaced energy levels $\omega_{\raisebox{0pt}{\tiny{$\Scal$}}}$, suppose that the machine structure is decomposed as $N$ qubits with energy gaps $\omega_{\raisebox{0pt}{\tiny{$\Mcal_n$}}} \in \{ \omega_{\raisebox{0pt}{\tiny{$\Scal$}}}+n \epsilon\}_{n=1, \dots, N}$, with arbitrarily small $\epsilon > 0$ and $N\to\infty$. Then the overall unitary that approaches perfect cooling at the Landauer limit has circuit complexity equal to the diverging $N$. 

\subsection{Proof of Theorem~\ref{thm:variety}, Corollary~\ref{cor:structuralcondition}, and Theorem~\ref{thm:landauerstructural}}
\label{app:proofcomplexitytheorems}

Here we prove Theorem~\ref{thm:variety}, which implies Theorem~\ref{thm:landauerstructural} and leads to Corollary~\ref{cor:structuralcondition}.

\begin{proof}
Let $\mathscr{H}_{\raisebox{-1pt}{\tiny{$\Xcal$}}}$ be a separable Hilbert space associated with the system $\Xcal$. Consider
\begin{align}
    H_{\raisebox{-1pt}{\tiny{$\Mcal$}}} = \sum_{n=0}^{\infty} \omega_n \ket{n}\!\bra{n} \quad \quad \textup{and} \quad \quad
    \mathscr H_{\raisebox{-1pt}{\tiny{$\Mcal^\prime$}}} = \textup{span}_{n \leq m}\{\ket{n}\},
\end{align}
for some finite $m$. In other words, $\mathscr H_{\raisebox{-1pt}{\tiny{$\Mcal^\prime$}}}$ is a finite-dimensional restriction of $\mathscr H_{\raisebox{-1pt}{\tiny{$\Mcal$}}}$. We show that any unitary that (nontrivially) interacts the target system with only a subspace spanned by finitely many eigenstates of $H_{\raisebox{-1pt}{\tiny{$\Mcal$}}}$ cannot attain Landauer's bound. Consider a general unitary $U$. Suppose that $U$ couples only $\mathscr{H}_{\raisebox{-1pt}{\tiny{$\Scal$}}}$ with $\mathscr{H}_{\raisebox{-1pt}{\tiny{$\Mcal^\prime$}}}$; whenever we talk of an operation with finite effective dimension in this paper, we mean specifically such a $U$, and by diverging effective dimension we mean a unitary that couples the target to any subspace of $\mathscr H_{\raisebox{-1pt}{\tiny{$\Mcal$}}}$ whose dimension diverges. Since
\begin{align}
    \mathscr{H}_{\raisebox{-1pt}{\tiny{$\Scal$}}} \otimes \mathscr{H}_{\raisebox{-1pt}{\tiny{$\Mcal$}}} &= \mathscr{H}_{\raisebox{-1pt}{\tiny{$\Scal$}}} \otimes (\mathscr{H}_{\raisebox{-1pt}{\tiny{$\Mcal^\prime$}}} \oplus \mathscr{H}_{\raisebox{-1pt}{\tiny{$\Mcal^\prime$}}}^\perp) \simeq (\mathscr{H}_{\raisebox{-1pt}{\tiny{$\Scal$}}} \otimes \mathscr{H}_{\raisebox{-1pt}{\tiny{$\Mcal^\prime$}}} )\oplus(\mathscr{H}_{\raisebox{-1pt}{\tiny{$\Scal$}}}  \otimes \mathscr{H}_{\raisebox{-1pt}{\tiny{$\Mcal^\prime$}}}^\perp),
\end{align}
we can associate the subspace $\mathscr{H}_{\raisebox{-1pt}{\tiny{$\Scal$}}} \otimes \mathscr{H}_{\raisebox{-1pt}{\tiny{$\Mcal^\prime$}}}$ with the label $\Acal$ and $\mathscr{H}_{\raisebox{-1pt}{\tiny{$\Scal$}}}  \otimes \mathscr{H}_{\raisebox{-1pt}{\tiny{$\Mcal^\prime$}}}^\perp$ with $\Bcal$ and write $U = U_{\raisebox{-1pt}{\tiny{$\Acal$}}} \oplus \mathbbm{1}_{\raisebox{-1pt}{\tiny{$\Bcal$}}}$. Then the initial configuration can be expressed as
\begin{align}
    \varrho_{\raisebox{-1pt}{\tiny{$\Scal$}}} \otimes \tau_{\raisebox{-1pt}{\tiny{$\Mcal$}}}(\beta, H_{\raisebox{-1pt}{\tiny{$\Mcal$}}}) = \left[ \begin{array}{cc}
    \varrho_{\raisebox{-1pt}{\tiny{$\Scal$}}} \otimes \varrho_{\raisebox{-1pt}{\tiny{$\Mcal^\prime$}}} & 0 \\
    0 & \varrho_{\raisebox{-1pt}{\tiny{$\Scal$}}} \otimes \varrho_{\raisebox{-1pt}{\tiny{$\Mcal^\prime$}}}^\perp
    \end{array} \right],
\end{align}
where
\begin{align}
    \varrho_{\raisebox{-1pt}{\tiny{$\Mcal^\prime$}}} := \frac{1}{\Zcal_{\raisebox{-1pt}{\tiny{$\Mcal$}}}(\beta, H_{\raisebox{-1pt}{\tiny{$\Mcal$}}})} \sum_{n \leq m} e^{-\beta \omega_n } \ket{n}\!\bra{n} \quad \quad \textup{and} \quad \quad \varrho_{\raisebox{-1pt}{\tiny{$\Mcal^\prime$}}}^\perp := \frac{1}{\Zcal_{\raisebox{-1pt}{\tiny{$\Mcal$}}}(\beta, H_{\raisebox{-1pt}{\tiny{$\Mcal$}}})} \sum_{n > m} e^{-\beta \omega_n } \ket{n}\!\bra{n}
\end{align}
add up to a (normalised) thermal state. 
Now consider the state
\begin{align}
    \widetilde{\varrho}_{\raisebox{-1pt}{\tiny{$\Mcal$}}} = \left[ \begin{array}{cc}
     \varrho_{\raisebox{-1pt}{\tiny{$\Mcal^\prime$}}} & 0 \\
    0 & \tr{\varrho_{\raisebox{-1pt}{\tiny{$\Mcal^\prime$}}}^\perp}
    \end{array} \right].
\end{align}
It is straightforward to check that is indeed a quantum state; moreover, it is the Gibbs state (at inverse temperature $\beta$) associated with the Hamiltonian
\begin{align}
    \widetilde{H}_{\raisebox{-1pt}{\tiny{$\Mcal$}}} = \sum_{n \leq m} \omega_n \ket{n}\!\bra{n} - \frac{1}{\beta} \log\left( \sum_{n>m} e^{-\beta \omega_n}\right) \ket{m+1}\!\bra{m+1}.
\end{align}
To see this, note that $\mathcal{Z}_{\raisebox{-1pt}{\tiny{$\Mcal$}}}(\beta, H_{\raisebox{-1pt}{\tiny{$\Mcal$}}}) = \mathcal{Z}_{\raisebox{-1pt}{\tiny{$\Mcal$}}}(\beta, \widetilde{H}_{\raisebox{-1pt}{\tiny{$\Mcal$}}})$ and that
\begin{align}
    \textup{exp}\left\{-\beta \left[-\frac{1}{\beta} \log\left( \sum_{n>m} e^{-\beta \omega_n}\right)\right]\right\} = \sum_{n>m} e^{-\beta \omega_n}.
\end{align}
Thus $\widetilde{\varrho}_{\raisebox{-1pt}{\tiny{$\Mcal$}}} = \tau_{\raisebox{-1pt}{\tiny{$\Mcal$}}}(\beta, \widetilde{H}_{\raisebox{-1pt}{\tiny{$\Mcal$}}})$. To ease notation in what follows, we write $\widetilde{\omega}_{m+1} := - \frac{1}{\beta} \log\left( \sum_{n>m} e^{-\beta \omega_n}\right)$. In the rest of the proof, we show that the unitary $U$ and the Hamiltonian $H_{\raisebox{-1pt}{\tiny{$\Mcal$}}}$ can be replaced by finite-dimensional versions without changing the quantities relevant for Landauer's principle.

Let $\widetilde{U} = U_{\raisebox{-1pt}{\tiny{$\Acal$}}} \oplus (\mathbbm{1}_{\raisebox{-1pt}{\tiny{$\Scal$}}}\otimes \ket{m+1}\!\bra{m+1})$. We then have
\begin{align}
    \widetilde{U} \left( \varrho_{\raisebox{-1pt}{\tiny{$\Scal$}}} \otimes \widetilde{\varrho}_{\raisebox{-1pt}{\tiny{$\Mcal$}}} \right) \widetilde{U}^\dagger = \left[ \begin{array}{cc}
        U_{\raisebox{-1pt}{\tiny{$\Acal$}}} (\varrho_{\raisebox{-1pt}{\tiny{$\Scal$}}} \otimes \varrho_{\raisebox{-1pt}{\tiny{$\Mcal^\prime$}}}) U_{\raisebox{-1pt}{\tiny{$\Acal$}}}^\dagger & 0 \\
        0 & \frac{e^{-\beta \widetilde{\omega}_{m+1}}}{\mathcal{Z}_{\raisebox{-1pt}{\tiny{$\Mcal$}}}(\beta, H_{\raisebox{-1pt}{\tiny{$\Mcal$}}})} \varrho_{\raisebox{-1pt}{\tiny{$\Scal$}}}
    \end{array} \right]
\end{align}
and
\begin{align}
    \ptr{\Mcal}{\widetilde{U} \left( \varrho_{\raisebox{-1pt}{\tiny{$\Scal$}}} \otimes \widetilde{\varrho}_{\raisebox{-1pt}{\tiny{$\Mcal$}}} \right) \widetilde{U}^\dagger} = \ptr{\Mcal^\prime}{U_{\raisebox{-1pt}{\tiny{$\Acal$}}} \left( \varrho_{\raisebox{-1pt}{\tiny{$\Scal$}}} \otimes \varrho_{\raisebox{-1pt}{\tiny{$\Mcal^\prime$}}} \right) U_{\raisebox{-1pt}{\tiny{$\Acal$}}}^\dagger} + \frac{e^{-\beta \widetilde{\omega}_{m+1}}}{\mathcal{Z}_{\raisebox{-1pt}{\tiny{$\Mcal$}}}(\beta, H_{\raisebox{-1pt}{\tiny{$\Mcal$}}})} \varrho_{\raisebox{-1pt}{\tiny{$\Scal$}}}
\end{align}
Compare this to the expression
\begin{align}
    \ptr{\Mcal}{U (\varrho_{\raisebox{-1pt}{\tiny{$\Scal$}}} \otimes \varrho_{\raisebox{-1pt}{\tiny{$\Mcal$}}}) U^\dagger} &= \mathrm{tr}_{\raisebox{-1pt}{\tiny{$\Mcal^\prime$}}}\left[ \begin{array}{cc}
       U_{\raisebox{-1pt}{\tiny{$\Acal$}}} (\varrho_{\raisebox{-1pt}{\tiny{$\Scal$}}} \otimes \varrho_{\raisebox{-1pt}{\tiny{$\Mcal^\prime$}}}) U^\dagger_{\raisebox{-1pt}{\tiny{$\Acal$}}}  & 0 \\
        0 & \varrho_{\raisebox{-1pt}{\tiny{$\Scal$}}} \otimes \varrho_{\raisebox{-1pt}{\tiny{$\Mcal^\prime$}}}^\perp
    \end{array}\right] \notag \\ &= \ptr{\Mcal^\prime}{U_{\raisebox{-1pt}{\tiny{$\Acal$}}} (\varrho_{\raisebox{-1pt}{\tiny{$\Scal$}}} \otimes \varrho_{\raisebox{-1pt}{\tiny{$\Mcal^\prime$}}}) U^\dagger_{\raisebox{-1pt}{\tiny{$\Acal$}}}} + \tr{\varrho_{\raisebox{-1pt}{\tiny{$\Mcal^\prime$}}}^\perp} \varrho_{\raisebox{-1pt}{\tiny{$\Scal$}}} \notag \\
    &= \ptr{\Mcal^\prime}{U_{\raisebox{-1pt}{\tiny{$\Acal$}}} (\varrho_{\raisebox{-1pt}{\tiny{$\Scal$}}} \otimes \varrho_{\raisebox{-1pt}{\tiny{$\Mcal^\prime$}}}) U^\dagger_{\raisebox{-1pt}{\tiny{$\Acal$}}}} + \frac{e^{-\beta \widetilde{\omega}_{m+1}}}{\mathcal{Z}_{\raisebox{-1pt}{\tiny{$\Mcal$}}}(\beta, H_{\raisebox{-1pt}{\tiny{$\Mcal$}}})} \varrho_{\raisebox{-1pt}{\tiny{$\Scal$}}},
\end{align}
since $\tr{\varrho_{\raisebox{-1pt}{\tiny{$\Mcal^\prime$}}}^\perp} = \frac{1}{\mathcal{Z}_{\raisebox{-1pt}{\tiny{$\Mcal$}}}(\beta, H_{\raisebox{-1pt}{\tiny{$\Mcal$}}})} \sum_{n>m} e^{-\beta \omega_n}$. Thus, the final system state is the same as it would be if we replaced the full initial machine state with $\widetilde{\varrho}_{\raisebox{-1pt}{\tiny{$\Mcal$}}}$; in particular, the entropy decrease of the system for any unitary that cools it is also unchanged.

The last thing we need to check is that the energy change of the machine similarly remains invariant. To that end, we have that
\begin{align}
    \widetilde{\varrho}_{\raisebox{-1pt}{\tiny{$\Mcal$}}}^\prime &= \ptr{\Scal}{\widetilde{U} (\varrho_{\raisebox{-1pt}{\tiny{$\Scal$}}} \otimes \widetilde{\varrho}_{\raisebox{-1pt}{\tiny{$\Mcal$}}}) \widetilde{U}^\dagger} = \ptr{\Scal}{U_{\raisebox{-1pt}{\tiny{$\Acal$}}} (\varrho_{\raisebox{-1pt}{\tiny{$\Scal$}}} \otimes \varrho_{\raisebox{-1pt}{\tiny{$\Mcal^\prime$}}}) U_{\raisebox{-1pt}{\tiny{$\Acal$}}}^\dagger} + \frac{e^{-\beta \widetilde{\omega}_{m+1}}}{\mathcal{Z}_{\raisebox{-1pt}{\tiny{$\Mcal$}}}(\beta, H_{\raisebox{-1pt}{\tiny{$\Mcal$}}})} \ket{m+1}\!\bra{m+1} \notag \\
    \widetilde{\varrho}_{\raisebox{-1pt}{\tiny{$\Mcal$}}} &= \varrho_{\raisebox{-1pt}{\tiny{$\Mcal^\prime$}}} + \frac{e^{-\beta \widetilde{\omega}_{m+1}}}{\mathcal{Z}_{\raisebox{-1pt}{\tiny{$\Mcal$}}}(\beta, H_{\raisebox{-1pt}{\tiny{$\Mcal$}}})} \ket{m+1}\!\bra{m+1}.
\end{align}
Thus, we have
\begin{align}
    \tr{\widetilde{H}_{\raisebox{-1pt}{\tiny{$\Mcal$}}}(\widetilde{\varrho}_{\raisebox{-1pt}{\tiny{$\Mcal$}}}^\prime - \widetilde{\varrho}_{\raisebox{-1pt}{\tiny{$\Mcal$}}})} = \mathrm{tr}\left\{H_{\raisebox{-1pt}{\tiny{$\Mcal$}}} \left[\ptr{\Scal}{U_{\raisebox{-1pt}{\tiny{$\Acal$}}} (\varrho_{\raisebox{-1pt}{\tiny{$\Scal$}}} \otimes \varrho_{\raisebox{-1pt}{\tiny{$\Mcal^\prime$}}}) U_{\raisebox{-1pt}{\tiny{$\Acal$}}}^\dagger} - \varrho_{\raisebox{-1pt}{\tiny{$\Mcal^\prime$}}} \right] \right\},
\end{align}
since $U_{\raisebox{-1pt}{\tiny{$\Acal$}}}$ only acts on $\mathscr{H}_{\raisebox{-1pt}{\tiny{$\Scal$}}} \otimes \mathscr{H}_{\raisebox{-1pt}{\tiny{$\Mcal^\prime$}}}$ and $\widetilde{H}_{\raisebox{-1pt}{\tiny{$\Mcal|\Mcal^\prime$}}} = H_{\raisebox{-1pt}{\tiny{$\Mcal|\Mcal^\prime$}}}$. In the same way, we have
\begin{align}
    \ptr{\Scal}{U (\varrho_{\raisebox{-1pt}{\tiny{$\Scal$}}} \otimes \varrho_{\raisebox{-1pt}{\tiny{$\Mcal$}}}) U^\dagger} &= \ptr{\Scal}{U_{\raisebox{-1pt}{\tiny{$\Acal$}}} (\varrho_{\raisebox{-1pt}{\tiny{$\Scal$}}}\otimes \varrho_{\raisebox{-1pt}{\tiny{$\Mcal^\prime$}}}) U_{\raisebox{-1pt}{\tiny{$\Acal$}}}^\dagger} + \varrho_{\raisebox{-1pt}{\tiny{$\Mcal^\prime$}}}^\perp \notag \\
    \varrho_{\raisebox{-1pt}{\tiny{$\Mcal$}}} &= \varrho_{\raisebox{-1pt}{\tiny{$\Mcal^\prime$}}} + \varrho_{\raisebox{-1pt}{\tiny{$\Mcal^\prime$}}}^\perp.
\end{align}
Thus, the energy difference is also
\begin{align}
    \mathrm{tr}\left\{ H_{\raisebox{-1pt}{\tiny{$\Mcal$}}} \left[ \ptr{\Scal}{U_{\raisebox{-1pt}{\tiny{$\Acal$}}} (\varrho_{\raisebox{-1pt}{\tiny{$\Scal$}}}\otimes \varrho_{\raisebox{-1pt}{\tiny{$\Mcal^\prime$}}}) U_{\raisebox{-1pt}{\tiny{$\Acal$}}}^\dagger} - \varrho_{\raisebox{-1pt}{\tiny{$\Mcal^\prime$}}} \right] \right\}.
\end{align}
Hence, we show that one can replace (a potentially infinite-dimensional) $\Mcal$ by some (finite) $m+1$-dimensional machine $\widetilde \Mcal$ if the joint unitary $U$ acts only on $m$ levels of $H_{\raisebox{-1pt}{\tiny{$\Mcal$}}}$. 
By Theorem 6 of Ref.~\cite{Reeb_2014}, there are finite-dimensional corrections to the Landauer bound, which then imply that it cannot be reached for finite $m$. Thus, the effective machine dimension, i.e., that which is actually (nontrivially) accessed throughout the interaction, must diverge in order for cooling to be possible at the Landauer limit. This proves Theorem~\ref{thm:variety}, which implies Theorem~\ref{thm:landauerstructural}.
\end{proof}

\subsection{Sufficient Complexity Conditions}
\label{app:sufficientconditionscomplexity}

Having shown the necessary requirements for cooling at Landauer cost, namely a control interaction that acts nontrivially on an infinite-dimensional (sub)space of the machine's Hilbert space, let us now return to emphasise the properties of the machine and cooling protocol that are sufficient to achieve perfect cooling at Landauer cost. For simplicity, we consider the case of a qubit, which exemplifies the discussion of finite-dimensional systems. The case of infinite-dimensional systems is treated independently in the next Appendix. 

We first consider the structural properties of the machine. The diverging-time protocol discussed in Appendix~\ref{app:divergingtimecoolingprotocolfinitedimensionalsystems} makes use of a diverging number $N$ of machines. Thus, the machine begins in the thermal state $\tau (\beta, H_{\raisebox{-1pt}{\tiny{$\Mcal$}}}^{\rm tot})$ of a $(2^N)$-dimensional system (with $N$ eventually diverging), with energy-level structure given by the sum of the Hamiltonians in Eq.~\eqref{eq:machHam}, i.e.,
\begin{equation}
    H_{\raisebox{-1pt}{\tiny{$\Mcal$}}}^{\rm tot} = \sum_{n=1}^N H_{\raisebox{-1pt}{\tiny{$\Mcal$}$_n$}}^{(n)} =\sum_n (1+n\epsilon)H_{\raisebox{-1pt}{\tiny{$\Scal$}}}^{(n)} ,
\end{equation}
that acts on the full Hilbert space (we use the usual convention that it acts as identity on unlabelled subspaces, e.g., $H^{(1)}_{\raisebox{-1pt}{\tiny{$\Mcal$}}}\equiv H^{(1)}_{\raisebox{-1pt}{\tiny{$\Mcal$}}} \otimes \openone^{(2)}\otimes \dots \otimes \openone^{(N)}$). 
Let us analyse in detail the properties of this Hamiltonian. The ground state is $\ket{0}^{\otimes N}$, which is set at zero energy. More generally, the energy eigenvalue corresponding to an eigenstate $\ket{i_0,i_1,\dots,i_{\raisebox{-1pt}{\tiny{$N$}}}}$ is given by $\omega_1$ multiplied by the number of indices $i_k$ that are equal to $1$, plus a sum of terms $k \epsilon$ where $k$ is the label of each index equal to $1$. Thus, the energy eigenvalue of the eigenstate $\ket{1,\dots,1}$ diverges as the number of subsystems diverges. At the same time, letting the factor $\epsilon$ go to zero renders all eigenstates with the same (constant) number of indices such that $i_k=1$ approach the same energy. Thus, in the limit $\epsilon \rightarrow 0$, one obtains subspaces of energy $E_{\raisebox{-1pt}{\tiny{$\Mcal$}}}^{(k)}=k \omega_1$ with degeneracy given by $D_k = \binom{N}{k}$, which also diverges for each constant $k$ and diverging $N$. Therefore, in addition to satisfying the structural conditions that are necessary for perfect cooling, as stated in Theorem~\ref{thm:landauerstructural}, the machine used here features additional properties, which are crucially important for this particular protocol, in particular because they are sufficient for perfect cooling at Landauer cost. As a remark, we also emphasise that for fixed (large) $N$ and (small) $\epsilon$, the machine is finite dimensional and has a nondegenerate Hamiltonian without any energy levels formally at infinity. 

Concerning the control complexity properties of the unitary that achieves perfect cooling in unit time, note that it is a cyclic shift operator, which can be written as
\begin{align}
    U_{\raisebox{-1pt}{\tiny{$\Scal \Mcal$}}} &= \Pi_{n=1}^N \mathbbm{S}^2_{\raisebox{-1pt}{\tiny{$\Scal \Mcal$}$_n$}} =  \Pi_n \left( \sum_{i,j_n=0}^{1} \ket{i,j_1,\dots,j_n,\dots,j_{\raisebox{-1pt}{\tiny{$N$}}}}\!\bra{j_n,j_1,\dots ,i,\dots,j_{\raisebox{-1pt}{\tiny{$N$}}}}_{\raisebox{-1pt}{\tiny{$\Scal \Mcal$}}}\right) \notag \\
    &= \sum_{i,j_1\dots j_{\raisebox{-1pt}{\tiny{$N$}}}=0}^{1}\ket{i,j_1,\dots ,j_{\raisebox{-1pt}{\tiny{$N$}}}}\!\bra{j_{\raisebox{-1pt}{\tiny{$N$}}},i,j_1,\dots,j_{N-1}}_{\raisebox{-1pt}{\tiny{$\Scal \Mcal$}}} .
\end{align}
As it is evident from its form, this unitary acts nontrivially on all of the (divergingly many) energy levels of the machine. The only basis vectors of the system-plus-machine Hilbert space that are left invariant are $\ket{i=0,j_1=0,\dots,j_{\raisebox{-1pt}{\tiny{$N$}}}=0}$ and 
$\ket{i=1,j_1=1,\dots,j_{\raisebox{-1pt}{\tiny{$N$}}}=1}$.

\subsection{Fine-Tuned Control Conditions}
\label{app:finetunedcontrolconditions}

Theorem~\ref{thm:variety} captures a notion of control complexity as a resource in a thermodynamically consistent manner, i.e., in line with Nernst's unattainability principle. However, following the discussion around Theorem~\ref{thm:landauerstructural} and that above, the protocols that we present that achieve perfect cooling at Landauer cost make use of machines and interactions with a far more complicated structure than suggested by the necessary condition of diverging effective dimensionality. In particular, we note that the interactions couple the target system to a diverging number of subspaces of the machine corresponding to distinct energy gaps in a fine-tuned manner. Moreover, there are a diverging number of energy levels of the machine both above and below the first excited level of the target. In this section, we begin by outlining the general conditions that perfect cooling at the Landauer limit entails, before presenting a more nuanced notion of control complexity in terms of the variety of distinct energy gaps in the machine in Appendix~\ref{app:energygapvariety}.

This suggests that an operationally meaningful quantifier of control complexity must take into account the energy-level structure of the machine that is accessed throughout any given protocol; additionally that of the target system plays a role. Indeed, both the final temperature of the target as well as the energy cost required to achieve this depends upon how the global eigenvalues are permuted via the cooling process. First, how cool the target becomes depends on the sum of the eigenvalues that are placed into the subspace spanned by the ground state. Second, for any fixed cooling amount, the energy cost depends on the constrained distribution of eigenvalues within the machine. Thus, in general, the optimal permutation of eigenvalues depends upon properties of both the target and machine. 

For instance, consider an arbitrary initially thermal target qubit, whose state is given by $\mathrm{diag}(p, 1-p)$ and a thermal machine of dimension $d_{\raisebox{-1pt}{\tiny{$\Mcal$}}}$ with spectrum $\{\lambda_{\raisebox{-1pt}{\tiny{$\Mcal$}}}^{i}\}_{i = 0, \hdots , d_{\raisebox{-1pt}{\tiny{$\Mcal$}}}-1}$. Now consider the decomposition of the joint Hilbert space into two orthogonal subspaces, $\Bcal_0$ and $\Bcal_1$, corresponding to the ground and excited eigenspaces of the target. The initial joint state is $p \, \mathrm{diag}(\lambda_{\raisebox{-1pt}{\tiny{$\Bcal_0$}}}^{i}) \oplus (1-p) \, \mathrm{diag}(\lambda_{\raisebox{-1pt}{\tiny{$\Bcal_1$}}}^{i})$, where we write $\lambda_{\raisebox{-1pt}{\tiny{$\Bcal_j$}}}^{i}$ to denote the $i^\textup{th}$ machine eigenvalue in the subspace $\Bcal_j$. The total population in the subspaces $\Bcal_0$ and $\Bcal_1$ are $p$ and $(1-p)$ respectively. To achieve perfect cooling one must permute the eigenvalues such that approximately a net transfer of population $(1-p)$ is moved from $\Bcal_1$ to $\Bcal_0$. To do this, one can take any subset $K$ of $\{ \lambda_{\raisebox{-1pt}{\tiny{$\Bcal_1$}}}^{i} \}$ such that as $d_{\raisebox{-1pt}{\tiny{$\Mcal$}}} \to \infty$, $\sum_{i\in K} \lambda_{\raisebox{-1pt}{\tiny{$\Bcal_1$}}}^{i} \to (1-p)$ and a subset $K^\prime$ (with $|K| = |K^\prime|$) from $\{ \lambda_{\raisebox{-1pt}{\tiny{$\Bcal_0$}}}^{i} \}$ such that $\sum_{i \in K^\prime} \{ \lambda_{\raisebox{-1pt}{\tiny{$\Bcal_0$}}}^{i} \} \to 0$ and exchange them. Although the choice of eigenvalues permuted is nonunique, the requirement must be fulfilled for some sets to perfectly cool the target. For any pair of eigenvalues exchanged between the subspaces, demanding that the exchange costs minimal energy amounts to a fine-tuning condition of the form $\lambda_{\raisebox{-1pt}{\tiny{$\Mcal$}}}^{i} \to p \lambda_{\raisebox{-1pt}{\tiny{$\Bcal_0$}}}^{i} + (1-p) \lambda_{\raisebox{-1pt}{\tiny{$\Bcal_1$}}}^{i}$ that must be satisfied. In general, the fine-tuned eigenvalue conditions that must be asymptotically attained depend upon target and machine eigenvalues, making it difficult to derive a closed-form expression. However, in the restricted scenario in which the target qubit begins maximally mixed (i.e., at infinite temperature), the machine begins thermal at some $\beta > 0$ and of dimension $d_{\raisebox{-1pt}{\tiny{$\Mcal$}}}$, and that the unitary implemented is such that the target is cooled as much as possible, one can derive precise conditions in terms of the machine structure alone, as we demonstrate below. The case for higher-dimensional target systems is similar. 

This discussion highlights the importance of capturing properties beyond the effective dimensionality, e.g., those regarding the distribution of machine (and, more generally, target system) eigenvalues, in order to meaningfully quantify control complexity in thermodynamics. Our protocols display similar behaviour to that discussed above asymptotically. Moreover, the machines exhibit an energy-level structure such that every possible energy gap is present, i.e., the set of machine energy gaps $\{ \omega_{ij} = \omega_i - \omega_j \}$ densely covers the interval $[\omega_{\raisebox{0pt}{\tiny{$\Scal$}}}, \infty)$, where $\omega_{\raisebox{0pt}{\tiny{$\Scal$}}}$ is the energy of the first excited level of the target. In Appendix~\ref{app:energygapvariety}, we demonstrate that indeed this condition is necessary for minimal-energy cost cooling.

Before doing so, we here first derive the fine-tuned control conditions that are asymptotically required for cooling at the Landauer limit. We begin with some general considerations before focusing on a special case for which an analytic expression can be derived. Furthermore, we demand that the unitary implemented is such that the target is cooled as much as possible: this does not preclude the possibility for cooling the target system less (albeit still close to a pure state) at a cost closer to the Landauer bound without satisfying all of the fine-tuning conditions. Nonetheless, in general there are a number of such conditions to be satisfied, and the special case serves as a pertinent example that demonstrates how the particular set of fine-tuning conditions for any considered scenario can be similarly derived. 

Consider an arbitrary thermal target system and machine of finite dimensions, with respective spectra $\boldsymbol{\lambda}_{\raisebox{-1pt}{\tiny{$\Scal$}}} := \{ \lambda_{\raisebox{-1pt}{\tiny{$\Scal$}}}^{0}, \hdots, \lambda_{\raisebox{-1pt}{\tiny{$\Scal$}}}^{d_{\raisebox{-1pt}{\tiny{$\Scal$}}}-1}\}$ and $\boldsymbol{\lambda}_{\raisebox{-1pt}{\tiny{$\Mcal$}}} := \{ \lambda_{\raisebox{-1pt}{\tiny{$\Mcal$}}}^{0}, \hdots, \lambda_{\raisebox{-1pt}{\tiny{$\Mcal$}}}^{d_{\raisebox{-1pt}{\tiny{$\Mcal$}}}-1}\}$. The states begin uncorrelated, so the global spectrum of the initial joint state is $\boldsymbol{\lambda}_{\raisebox{-1pt}{\tiny{$\Scal \Mcal$}}} := \{ \lambda_{\raisebox{-1pt}{\tiny{$\Scal \Mcal$}}}^0, \hdots, \lambda_{\raisebox{-1pt}{\tiny{$\Scal \Mcal$}}}^{d_{\raisebox{-1pt}{\tiny{$\Scal$}}} d_{\raisebox{-1pt}{\tiny{$\Mcal$}}} - 1}\} = \{ \lambda_{\raisebox{-1pt}{\tiny{$\Scal$}}}^{0}\lambda_{\raisebox{-1pt}{\tiny{$\Mcal$}}}^{0}, \lambda_{\raisebox{-1pt}{\tiny{$\Scal$}}}^{0}\lambda_{\raisebox{-1pt}{\tiny{$\Mcal$}}}^{1}, \hdots , \lambda_{\raisebox{-1pt}{\tiny{$\Scal$}}}^{d_{\raisebox{-1pt}{\tiny{$\Scal$}}}-1}\lambda_{\raisebox{-1pt}{\tiny{$\Mcal$}}}^{d_{\raisebox{-1pt}{\tiny{$\Mcal$}}}-1}\}$. Consider now a global unitary transformation; such a transformation cannot change the values of the spectrum, but merely permute them. In other words, the spectrum of the final global state after any such unitary is invariant and we have equivalence of the (unordered) sets $\boldsymbol{\lambda}_{\raisebox{-1pt}{\tiny{$\Scal \Mcal$}}}^\prime$ and $\boldsymbol{\lambda}_{\raisebox{-1pt}{\tiny{$\Scal \Mcal$}}}$.

The transformation that cools the target system as much as possible\footnote{We take majorisation among passive states to be the measure of cooling; this implies the highest possible ground state population and purity, and lowest possible entropy and average energy via Schur convexity.} is the one that places the $d_{\raisebox{-1pt}{\tiny{$\Mcal$}}}$ largest of the global eigenvalues into the subspace spanned by the ground state of the target, the second $d_{\raisebox{-1pt}{\tiny{$\Mcal$}}}$ largest into that spanned by the first excited state of the target, and so forth, with the smallest $d_{\raisebox{-1pt}{\tiny{$\Mcal$}}}$ global eigenvalues placed in the subspace corresponding to the highest energy eigenstate of the target system (we prove this statement shortly). More precisely, we denote by $\boldsymbol{\lambda}^{\downarrow}$ the nonincreasing ordering of the set $\boldsymbol{\lambda}$. Since the target and machine begin thermal, the local spectra $\boldsymbol{\lambda}_{\raisebox{-1pt}{\tiny{$\Scal$}}}$ and $\boldsymbol{\lambda}_{\raisebox{-1pt}{\tiny{$\Mcal$}}}$ are already ordered in this way with respect to their energy eigenbases, which we consider to be labelled in nondecreasing order. Cooling the target system as much as possible amounts to achieving the final reduced state of the target
\begin{align}\label{eq:systemreducedoptimallycool}
    \varrho_{\raisebox{-1pt}{\tiny{$\Scal$}}}^\prime = \sum_{i=0}^{d_{\raisebox{-1pt}{\tiny{$\Scal$}}}-1} \left(\sum_{j=0}^{d_{\raisebox{-1pt}{\tiny{$\Mcal$}}}-1} \lambda_{\raisebox{-1pt}{\tiny{$\Scal \Mcal$}}}^{\downarrow id_{\raisebox{-1pt}{\tiny{$\Mcal$}}} + j}\right) \ket{i}\!\bra{i}_{\raisebox{-1pt}{\tiny{$\Scal$}}}.
\end{align}
As a side remark, note that since each of the global eigenvalues are a product of the initial local eigenvalues (due to the initial tensor product structure), which are in turn related to the energy-level structure of the target system and machine (as they begin as thermal states), one can already see here that in order to approach perfect cooling, the machine must have some diverging energy gaps, such that the (finite) sum of the global eigenvalues contributing to the ground-state population of the target approaches 1. 

Of course, there is an equivalence class of unitaries that can achieve the same amount of cooling; in particular, any permutation of the set of the $d_{\raisebox{-1pt}{\tiny{$\Mcal$}}}$ global eigenvalues \emph{within} each energy eigenspace of the target system achieves the same amount of cooling, since it is the sum of these values that contribute to the total population in each subspace. Importantly, although such unitaries cool the target system to the same extent, their effect on the machine differs, and therefore so too does the energy cost of the protocol. However, demanding that such cooling is achieved at minimal energy cost amounts to a unique constraint on the global post-transformation state, namely that it must render the machine energetically passive, leading to the form:
\begin{align}\label{eq:globaloptimallycool}
    \varrho_{\raisebox{-1pt}{\tiny{$\Scal \Mcal$}}}^\prime = \sum_{i=0}^{d_{\raisebox{-1pt}{\tiny{$\Scal$}}}-1} \sum_{j=0}^{d_{\raisebox{-1pt}{\tiny{$\Mcal$}}}-1} \lambda_{\raisebox{-1pt}{\tiny{$\Scal \Mcal$}}}^{\downarrow id_{\raisebox{-1pt}{\tiny{$\Mcal$}}}+j} \ket{ij}\!\bra{ij}_{\raisebox{-1pt}{\tiny{$\Scal \Mcal$}}}.
\end{align}

We can derive the above form of the final joint state as follows. Consider the following ordering for the energy eigenbasis of $\Scal\Mcal$ chosen to match the above form
\begin{align}\label{eq:basisordering}
\{ \ket{00}_{\raisebox{-1pt}{\tiny{$\Scal \Mcal$}}}, \ket{01}_{\raisebox{-1pt}{\tiny{$\Scal \Mcal$}}}, ...,  \ket{0,d_{\raisebox{-1pt}{\tiny{$ \Mcal$}}}-1}_{\raisebox{-1pt}{\tiny{$\Scal \Mcal$}}}, \ket{10}_{\raisebox{-1pt}{\tiny{$\Scal \Mcal$}}}, ...,  \ket{1,d_{\raisebox{-1pt}{\tiny{$\Mcal $}}}-1}_{\raisebox{-1pt}{\tiny{$\Scal \Mcal$}}}, ..., \ket{d_{\raisebox{-1pt}{\tiny{$\Scal $}}}-1,0}_{\raisebox{-1pt}{\tiny{$\Scal \Mcal$}}}, ..., \ket{d_{\raisebox{-1pt}{\tiny{$\Scal $}}}-1,d_{\raisebox{-1pt}{\tiny{$ \Mcal$}}}-1}_{\raisebox{-1pt}{\tiny{$\Scal \Mcal$}}}\}.
\end{align}
This ordering is monotonically nondecreasing primarily with respect to the energy of $\Scal$, and secondarily w.r.t. $\Mcal$. We take the final state $\rho_{\raisebox{-1pt}{\tiny{$\Scal \Mcal$}}}^\prime$ to be expressed in this basis. To maximise the cooling in a single unitary operation, we maximise the sum of the first $k \cdot d_{\raisebox{-1pt}{\tiny{$\Mcal$}}}$ diagonal elements, for each $k\in \{1,2,..., d_{\raisebox{-1pt}{\tiny{$\Scal$}}}\}$, as each sum corresponds to the total population in the $k^\textup{th}$ lowest energy eigenstate of $\Scal$. The initial state $\varrho_{\raisebox{-1pt}{\tiny{$\Scal \Mcal$}}}$ is diagonal in this basis, so the vector of initial diagonal elements, which we label $\boldsymbol{\theta} := \mathrm{diag}(\varrho_{\raisebox{-1pt}{\tiny{$\Scal \Mcal$}}})$, is also the vector of eigenvalues, $\boldsymbol{\lambda}_{\raisebox{-1pt}{\tiny{$\Scal \Mcal$}}}$, i.e., $\boldsymbol{\theta} = \boldsymbol{\lambda}_{\raisebox{-1pt}{\tiny{$\Scal \Mcal$}}}$. Furthermore, since the unitary operation leaves the set of eigenvalues invariant, we have via the Schur-Horn lemma~\cite{2011Marshall} that the vector of final diagonal elements, which we label $\boldsymbol{\theta}^\prime := \mathrm{diag}(\varrho_{\raisebox{-1pt}{\tiny{$\Scal \Mcal$}}}^\prime)$, is majorised by the vector of initial ones, i.e., $\boldsymbol{\theta}^\prime \prec \boldsymbol{\theta}$. It follows that the partial sums we wish to maximise are upper bounded by the corresponding partial sums of the $k \cdot d_{\raisebox{-1pt}{\tiny{$ \Mcal$}}}$ largest diagonal elements of the initial state. We claim that the unitary that cools this maximal cooling amount at minimum energy cost is the one that permutes the diagonal elements to be ordered w.r.t. the basis ordering in Eq.~\eqref{eq:basisordering}.

More precisely, via the Schur-Horn lemma, one can always write $\boldsymbol{\theta}^\prime = D \boldsymbol{\theta}$, with $D$ a doubly stochastic matrix. The partial sums of the $k \cdot d_{\raisebox{-1pt}{\tiny{$ \Mcal$}}}$ first elements are linear functions of the elements of $\boldsymbol{\theta}$. Thus the maximum values are obtained at the extremal points of the convex set of doubly stochastic matrices, which are the permutation matrices, via the Birkhoff-von Neumann theorem~\cite{2011Marshall}. One can see by inspection that the optimal permutation matrices are the ones that place the largest $d_{\raisebox{-1pt}{\tiny{$ \Mcal$}}}$ diagonal elements in the first block (i.e., the ground-state eigenspace of $\Scal$), the next largest $d_{\raisebox{-1pt}{\tiny{$ \Mcal$}}}$ elements in the second block (i.e., the first excited-state eigenspace of $\Scal$), and so on. Within each block, the ordering does not affect the cooling of the target, so there is an equivalence class of permutations that satisfy the maximal cooling criterion. However, adding the optimisation over the energy cost eliminates this freedom. We may consider the reduced set of stochastic matrices that satisfy maximal cooling, generated by the permutations described above. Since the average energy of the final state is again a linear function of the diagonal elements, here too the minimum corresponds to a permutation matrix. Clearly the permutation that minimises the average energy is the one that orders the elements within each block to be decreasing w.r.t. the energies of $\Mcal$. Thus, the unique\footnote{Note that degeneracies in energy eigenvalues would lead to sets of equal diagonal elements, and prevent one from choosing a unique permutation. However, as the state in such degenerate subspaces is proportional to the identity matrix, we may take any unitary that is block diagonal w.r.t. the degeneracies without affecting the state, and hence the final cooling or average energy change.} stochastic matrix $D$ that leads to maximal cooling at the least energy cost possible is the one that permutes the energy eigenvalues to be ordered decreasing primarily w.r.t. the system energies, and secondarily w.r.t. the machine energies. The action of the stochastic matrix on diagonal elements of the state is related to the unitary operation on the entire quantum state by $|U_{ij}|^2 = D_{ij}$, so that the unitary operation is also a permutation (up to an energy-dependent phase, which is irrelevant since the initial and final states are diagonal).

We may understand this optimal operation through the notion of passivity, by noting that it cools at minimal energy cost by rendering the machine into the most energetically passive reduced state in the joint unitary orbit with respect to the cooling constraint on the target. Intuitively, one has cooled the target system maximally at the expense of heating the machine as little as possible. The final reduced state of the machine corresponding to this energetically optimal cooling transformation is
\begin{align}
    \varrho_{\raisebox{-1pt}{\tiny{$\Mcal$}}}^\prime = \sum_{j=0}^{d_{\raisebox{-1pt}{\tiny{$\Mcal$}}}-1} \left(\sum_{i=0}^{d_{\raisebox{-1pt}{\tiny{$\Scal$}}}-1} \lambda_{\raisebox{-1pt}{\tiny{$\Scal \Mcal$}}}^{\downarrow i d_{\raisebox{-1pt}{\tiny{$\Mcal$}}} + j}\right) \ket{j}\!\bra{j}_{\raisebox{-1pt}{\tiny{$\Mcal$}}}.
\end{align}

In general, any unitary that achieves these desired conditions simultaneously depends upon the energy-level structure of both the target system and machine, precluding a closed-form set of conditions that can be expressed only in terms of the machine. However, for the special case of a maximally mixed initial target state (i.e., cooling a thermal state at infinite temperature or erasing quantum information from its most entropic state), one can deduce this ordering precisely and moreover relate it directly to properties of the machine Hamiltonian, as we now demonstrate. In the following, we assume that $d_{\raisebox{-1pt}{\tiny{$\Mcal$}}}$ is even; the case for odd $d_{\raisebox{-1pt}{\tiny{$\Mcal$}}}$ can be derived similarly.

\begin{thm}\label{thm:maineigenvaluecondition}
Consider the target system to begin in the maximally mixed state and a thermal machine at temperature $\beta > 0$, whose eigenvalues are labelled in nonincreasing order, $\{\lambda_{\raisebox{-1pt}{\tiny{$\Mcal$}}}^{\downarrow i}\}_{i = 0, \hdots , d_{\raisebox{-1pt}{\tiny{$\Mcal$}}}-1}$. In order to cool the target perfectly, with the restriction that the target must be cooled as much as possible, at an energy cost that saturates the Landauer limit, the machine eigenvalues must satisfy
\begin{align}\label{eq:maineigenvaluesum} 
    \sum_{i=0}^{\frac{d_{\raisebox{-1pt}{\tiny{$\Mcal$}}}}{2}-1} \lambda_{\raisebox{-1pt}{\tiny{$\Mcal$}}}^{\downarrow i} \to 1, \quad
    \sum_{i=\frac{d_{\raisebox{-1pt}{\tiny{$\Mcal$}}}}{2}}^{d_{\raisebox{-1pt}{\tiny{$\Mcal$}}}-1} \lambda_{\raisebox{-1pt}{\tiny{$\Mcal$}}}^{\downarrow i} \to 0,
\end{align} 
    and
\begin{align}\label{eq:maingenerictermeigenvalue} 
    \frac{\frac{1}{2} \left(\lambda_{\raisebox{-1pt}{\tiny{$\Mcal$}}}^{\downarrow \lfloor \frac{i}{2} \rfloor} + \lambda_{\raisebox{-1pt}{\tiny{$\Mcal$}}}^{\downarrow  \frac{d_{\raisebox{-1pt}{\tiny{$\Mcal$}}}}{2}  + \lfloor \frac{i}{2} \rfloor}\right)}{\lambda_{\raisebox{-1pt}{\tiny{$\Mcal$}}}^{\downarrow i}} \to 1
\end{align}
for all $i \in \{ 0, \hdots, d_{\raisebox{-1pt}{\tiny{$\Mcal$}}}-1\}$, where $\lfloor \cdot \rfloor$ denotes the floor function and $\rightarrow$ denotes that the condition is satisfied asymptotically, \textup{i.e.}, as $d_{\raisebox{-1pt}{\tiny{$\Mcal$}}} \to \infty$\footnote{Strictly speaking, in the limit $d_{\raisebox{-1pt}{\tiny{$\Mcal$}}} \to \infty$ the conditions in Eq.~\eqref{eq:maingenerictermeigenvalue} must only be satisfied for almost all $i$, i.e., for all but a small subset that contributes negligibly to the relative entropy, as we discuss below.}.
\end{thm}

\begin{proof}
We consider a qubit for simplicity, but the generalisation to cooling an arbitrary-dimensional maximally mixed state is straightforward. The initial joint spectrum of the system and machine is
\begin{align}\label{eq:maximallymixedglobalspectrum}
    \boldsymbol{\lambda}_{\raisebox{-1pt}{\tiny{$\Scal \Mcal$}}} = \tfrac{1}{2} \{ \boldsymbol{\lambda}_{\raisebox{-1pt}{\tiny{$\Mcal$}}}^{\downarrow}, \boldsymbol{\lambda}_{\raisebox{-1pt}{\tiny{$\Mcal$}}}^{\downarrow}\} = \tfrac{1}{2} \{ \lambda_{\raisebox{-1pt}{\tiny{$\Mcal$}}}^{\downarrow 0}, \lambda_{\raisebox{-1pt}{\tiny{$\Mcal$}}}^{\downarrow 1} , \hdots, \lambda_{\raisebox{-1pt}{\tiny{$\Mcal$}}}^{\downarrow d_{\raisebox{-1pt}{\tiny{$\Mcal$}}}-1},  \lambda_{\raisebox{-1pt}{\tiny{$\Mcal$}}}^{\downarrow 0}, \lambda_{\raisebox{-1pt}{\tiny{$\Mcal$}}}^{\downarrow 1} , \hdots, \lambda_{\raisebox{-1pt}{\tiny{$\Mcal$}}}^{\downarrow d_{\raisebox{-1pt}{\tiny{$\Mcal$}}}-1}\}.
\end{align}
As each $\lambda_{\raisebox{-1pt}{\tiny{$\Mcal$}}}^{\downarrow i} = \tfrac{1}{\mathcal{Z}_{\raisebox{-1pt}{\tiny{$\Mcal$}}}(\beta, H_{\raisebox{-1pt}{\tiny{$\Mcal$}}})} e^{-\beta \omega_i}$ for any thermal state with Hamiltonian $H_{\raisebox{-1pt}{\tiny{$\Mcal$}}} = \sum_{i} \omega_i \ket{i}\!\bra{i}_{\raisebox{-1pt}{\tiny{$\Mcal$}}}$ written with respect to nondecreasing energy eigenvalues, it follows that the globally ordered spectrum is
\begin{align}\label{eq:appdgloballyorderedfinalspectrum}
     \boldsymbol{\lambda}_{\raisebox{-1pt}{\tiny{$\Scal \Mcal$}}}^{\downarrow} = \tfrac{1}{2} \{ \lambda_{\raisebox{-1pt}{\tiny{$\Mcal$}}}^{\downarrow 0}, \lambda_{\raisebox{-1pt}{\tiny{$\Mcal$}}}^{\downarrow 0}, \lambda_{\raisebox{-1pt}{\tiny{$\Mcal$}}}^{\downarrow 1} , \lambda_{\raisebox{-1pt}{\tiny{$\Mcal$}}}^{\downarrow 1} , \hdots, \lambda_{\raisebox{-1pt}{\tiny{$\Mcal$}}}^{\downarrow d_{\raisebox{-1pt}{\tiny{$\Mcal$}}}-1}, \lambda_{\raisebox{-1pt}{\tiny{$\Mcal$}}}^{\downarrow d_{\raisebox{-1pt}{\tiny{$\Mcal$}}}-1}\}.
\end{align}
Expressing the global states with respect to the product of local energy eigenbases, we have that the initial joint state is $\tfrac{\mathbbm{1}_{\raisebox{-1pt}{\tiny{$\Scal$}}}}{2} \otimes \tau_{\raisebox{-1pt}{\tiny{$\Mcal$}}}(\beta,H_{\raisebox{-1pt}{\tiny{$\Mcal$}}}) = \mathrm{diag}(\boldsymbol{\lambda_{\raisebox{-1pt}{\tiny{$\Scal \Mcal$}}}})$ [see Eq.~\eqref{eq:maximallymixedglobalspectrum}] and the unitary that cools the target as much as possible at minimum energy cost is the one achieving the globally passive final joint state $\varrho_{\raisebox{-1pt}{\tiny{$\Scal \Mcal$}}}^\prime = \mathrm{diag}(\boldsymbol{\lambda}_{\raisebox{-1pt}{\tiny{$\Scal \Mcal$}}}^{\downarrow})$. This leads to the following reduced states
\begin{align}
    \varrho_{\raisebox{-1pt}{\tiny{$\Scal$}}}^\prime &= \left(\sum_{i=0}^{\frac{d_{\raisebox{-1pt}{\tiny{$\Mcal$}}}}{2} - 1} \lambda_{\raisebox{-1pt}{\tiny{$\Mcal$}}}^{\downarrow i}\right) \ket{0}\!\bra{0}_{\raisebox{-1pt}{\tiny{$\Scal$}}} + \left(\sum_{i=\frac{d_{\raisebox{-1pt}{\tiny{$\Mcal$}}}}{2}}^{d_{\raisebox{-1pt}{\tiny{$\Mcal$}}}-1} \lambda_{\raisebox{-1pt}{\tiny{$\Mcal$}}}^{\downarrow i}\right) \ket{1}\!\bra{1}_{\raisebox{-1pt}{\tiny{$\Scal$}}}, \\
    \varrho_{\raisebox{-1pt}{\tiny{$\Mcal$}}}^\prime &= \frac{1}{2} \left( \lambda_{\raisebox{-1pt}{\tiny{$\Mcal$}}}^{\downarrow 0} + \lambda_{\raisebox{-1pt}{\tiny{$\Mcal$}}}^{\downarrow \frac{d_{\raisebox{-1pt}{\tiny{$\Mcal$}}}}{2}}\right) \ket{0}\!\bra{0}_{\raisebox{-1pt}{\tiny{$\Mcal$}}} + \frac{1}{2} \left( \lambda_{\raisebox{-1pt}{\tiny{$\Mcal$}}}^{\downarrow 0} + \lambda_{\raisebox{-1pt}{\tiny{$\Mcal$}}}^{\downarrow \frac{d_{\raisebox{-1pt}{\tiny{$\Mcal$}}}}{2}}\right) \ket{1}\!\bra{1}_{\raisebox{-1pt}{\tiny{$\Mcal$}}} + \notag \\
    &+ \frac{1}{2} \left( \lambda_{\raisebox{-1pt}{\tiny{$\Mcal$}}}^{\downarrow 1} + \lambda_{\raisebox{-1pt}{\tiny{$\Mcal$}}}^{\downarrow \frac{d_{\raisebox{-1pt}{\tiny{$\Mcal$}}}}{2}  + 1}\right) \ket{2}\!\bra{2}_{\raisebox{-1pt}{\tiny{$\Mcal$}}} + \frac{1}{2} \left( \lambda_{\raisebox{-1pt}{\tiny{$\Mcal$}}}^{\downarrow 1} + \lambda_{\raisebox{-1pt}{\tiny{$\Mcal$}}}^{\downarrow  \frac{d_{\raisebox{-1pt}{\tiny{$\Mcal$}}}}{2} + 1}\right) \ket{3}\!\bra{3}_{\raisebox{-1pt}{\tiny{$\Mcal$}}} + \hdots  \label{eq:finalmachinestateoptimalenergy}
\end{align}
Intuitively, the reduced target state has the larger half of the initial machine eigenvalues in the ground state and the smaller half in the excited state; the reduced machine state has the sum of the largest elements from each of these halves in its ground state, the next largest element from each half (which, in this case, is equal to the first) in its first excited state, and so forth. Let us denote the spectrum of the final state of the machine by $\boldsymbol{\lambda}_{\raisebox{-1pt}{\tiny{$\Mcal$}}}^{\prime \downarrow} := \{\lambda_{\raisebox{-1pt}{\tiny{$\Mcal$}}}^{\prime \downarrow 0} , \lambda_{\raisebox{-1pt}{\tiny{$\Mcal$}}}^{\prime \downarrow 1}, \hdots, \lambda_{\raisebox{-1pt}{\tiny{$\Mcal$}}}^{\prime \downarrow d_{\raisebox{-1pt}{\tiny{$\Mcal$}}}-1} \} = \tfrac{1}{2}\{ \lambda_{\raisebox{-1pt}{\tiny{$\Mcal$}}}^{\downarrow 0} + \lambda_{\raisebox{-1pt}{\tiny{$\Mcal$}}}^{\downarrow \frac{d_{\raisebox{-1pt}{\tiny{$\Mcal$}}}}{2} }, \lambda_{\raisebox{-1pt}{\tiny{$\Mcal$}}}^{\downarrow 0} + \lambda_{\raisebox{-1pt}{\tiny{$\Mcal$}}}^{\downarrow \frac{d_{\raisebox{-1pt}{\tiny{$\Mcal$}}}}{2} }, \hdots, \lambda_{\raisebox{-1pt}{\tiny{$\Mcal$}}}^{\downarrow \frac{d_{\raisebox{-1pt}{\tiny{$\Mcal$}}}}{2} + 1} + \lambda_{\raisebox{-1pt}{\tiny{$\Mcal$}}}^{\downarrow d_{\raisebox{-1pt}{\tiny{$\Mcal$}}}-1} \}$. Importantly, by construction, the reduced state of the final machine has its local eigenvalues in nonincreasing order, i.e., it is energetically passive. 

We therefore have the final reduced states of the protocol that cools the initially maximally mixed target as much as possible at minimal energy cost, in particular with minimal heat dissipation by the machine, given the structural resources at hand. We can now analyse the properties that are required to saturate the Landauer limit by considering the terms on the r.h.s. of Eq.~\eqref{eq:landauerequality} for any fixed initial inverse temperature of the machine $\beta \geq 0$. 

First note that cooling the target system by any amount fixes the change in entropy of the target system, so the first term is irrelevant. The second term concerns the mutual information built up between the target system and machine. In general, this is nonvanishing, although one can achieve any desired amount of cooling without generating such correlations (as per our constructions). Furthermore, in the case where one wants to consider attaining a perfectly cool final state, as we do here, the final reduced state of the target is approximately pure and so $I(S:M)_{\varrho_{\raisebox{-1pt}{\tiny{$\Scal \Mcal$}}}^\prime} \to 0$. In terms of the reduced states above, this means that $\sum_{i=0}^{\frac{d_{\raisebox{-1pt}{\tiny{$\Mcal$}}}}{2}-1} \lambda_{\raisebox{-1pt}{\tiny{$\Mcal$}}}^{\downarrow i} \to 1$ and $\sum_{i=\frac{d_{\raisebox{-1pt}{\tiny{$\Mcal$}}}}{2}}^{d_{\raisebox{-1pt}{\tiny{$\Mcal$}}}-1} \lambda_{\raisebox{-1pt}{\tiny{$\Mcal$}}}^{\downarrow i} \to 0$, which can occur only if the largest half of energy eigenvalues of the machine, i.e., $\omega_{i}$ for all $i \geq \frac{d_{\raisebox{-1pt}{\tiny{$\Mcal$}}}}{2}$, diverge (since the summation contains only non-negative summands).

The final term that must be minimised to saturate the Landauer limit is the relative entropy of the final with respect to the initial machine state, $D(\varrho_{\raisebox{-1pt}{\tiny{$\Mcal$}}}^\prime \| \varrho_{\raisebox{-1pt}{\tiny{$\Mcal$}}})$. Here one can already see that an infinite-dimensional machine is required to saturate the Landauer bound: from Ref.~\cite{Reeb_2014}, $D(\varrho_{\raisebox{-1pt}{\tiny{$\Mcal$}}}^\prime \| \varrho_{\raisebox{-1pt}{\tiny{$\Mcal$}}}) \geq f(\Delta S_{\raisebox{-1pt}{\tiny{$\Mcal$}}}, d_{\raisebox{-1pt}{\tiny{$\Mcal$}}})$, where $f$ is a dimension-dependant function of the entropy difference of the machine that exhibits non-negative correction terms that vanish only in the limit $d_{\raisebox{-1pt}{\tiny{$\Mcal$}}} \to \infty$. The relative entropy vanishes iff $\varrho_{\raisebox{-1pt}{\tiny{$\Mcal$}}} = \varrho_{\raisebox{-1pt}{\tiny{$\Mcal$}}}^\prime$; moreover, by Pinsker's inequality one has $\tfrac{1}{2}\| \varrho_{\raisebox{-1pt}{\tiny{$\Mcal$}}} - \varrho_{\raisebox{-1pt}{\tiny{$\Mcal$}}}^\prime\|_1^2 \leq D(\varrho_{\raisebox{-1pt}{\tiny{$\Mcal$}}}\| \varrho_{\raisebox{-1pt}{\tiny{$\Mcal$}}}^\prime)$, so one can bound the trace distance between the initial and final state of the machine for any desired value of the relative entropy. Although $\varrho_{\raisebox{-1pt}{\tiny{$\Mcal$}}} = \varrho_{\raisebox{-1pt}{\tiny{$\Mcal$}}}^\prime$ implies a trivial process that cannot cool the (initially thermal) target system, as our protocols that saturate the Landauer limit demonstrate, there are processes that asymptotically display the behaviour $ \varrho_{\raisebox{-1pt}{\tiny{$\Mcal$}}}^\prime \to \varrho_{\raisebox{-1pt}{\tiny{$\Mcal$}}}$ \emph{and} cool the target system. For the asymptotic machine states to converge, in particular, their eigenvalues must become approximately equal asymptotically. Demanding this on the spectrum in Eq.~\eqref{eq:finalmachinestateoptimalenergy} leads to a generic term that must be asymptotically satisfied of the form:
\begin{align}\label{eq:lemmagenerictermeigenvalue} 
    \frac{\frac{1}{2} \left(\lambda_{\raisebox{-1pt}{\tiny{$\Mcal$}}}^{\downarrow \lfloor \frac{i}{2} \rfloor} + \lambda_{\raisebox{-1pt}{\tiny{$\Mcal$}}}^{\downarrow  \frac{d_{\raisebox{-1pt}{\tiny{$\Mcal$}}}}{2}  + \lfloor \frac{i}{2} \rfloor}\right)}{\lambda_{\raisebox{-1pt}{\tiny{$\Mcal$}}}^{\downarrow i}} \to 1 \qquad \forall \; i \in \{ 0, \hdots, d_{\raisebox{-1pt}{\tiny{$\Mcal$}}}-1\}.
\end{align}
\end{proof}
\noindent In order to achieve perfect cooling at the Landauer limit, one thus must simultaneously satisfy the conditions outlined in Theorem~\ref{thm:maineigenvaluecondition}. In other words, to minimise the relative-entropy term with the additional constraints $\sum_{i=0}^{\frac{d_{\raisebox{-1pt}{\tiny{$\Mcal$}}}}{2}-1} \lambda_{\raisebox{-1pt}{\tiny{$\Mcal$}}}^{\downarrow i} \to 1$ and $\sum_{i=\frac{d_{\raisebox{-1pt}{\tiny{$\Mcal$}}}}{2}}^{d_{\raisebox{-1pt}{\tiny{$\Mcal$}}}-1} \lambda_{\raisebox{-1pt}{\tiny{$\Mcal$}}}^{\downarrow i} \to 0$. The first thing to note is that since the eigenvalues $\lambda_{\raisebox{-1pt}{\tiny{$\Mcal$}}}^{\downarrow i}$ contribute to different sums depending on whether $i$ is in the larger half $\{ 0, \hdots, \tfrac{d_{\raisebox{-1pt}{\tiny{$\Mcal$}}}}{2}-1\}$ or smaller half $\{ \tfrac{d_{\raisebox{-1pt}{\tiny{$\Mcal$}}}}{2}, \hdots, d_{\raisebox{-1pt}{\tiny{$\Mcal$}}}\}$, one cannot have $\lambda_{\raisebox{-1pt}{\tiny{$\Mcal$}}}^{\downarrow \frac{d_{\raisebox{-1pt}{\tiny{$\Mcal$}}}}{2}  + \lfloor \frac{i}{2} \rfloor } = \lambda_{\raisebox{-1pt}{\tiny{$\Mcal$}}}^{\downarrow \lfloor \frac{i}{2} \rfloor} \; \forall \;i$ (i.e., a completely degenerate machine), since then both summations would be over identical values and there is no way for them to converge to distinct values. This precludes the trivial solution that satisfies the constraints of Eq.~\eqref{eq:maingenerictermeigenvalue} alone, namely the maximally mixed machine state, which cannot be used to perform any cooling [as, in particular, it does not satisfy the constraints of Eq.~\eqref{eq:maineigenvaluesum}]. For the conditions to be simultaneously satisfied, we intuitively require that, although they must be distinct, for each $i$ both $\lambda_{\raisebox{-1pt}{\tiny{$\Mcal$}}}^{\downarrow \lfloor \frac{i}{2} \rfloor}$ and $\lambda_{\raisebox{-1pt}{\tiny{$\Mcal$}}}^{\downarrow \frac{d_{\raisebox{-1pt}{\tiny{$\Mcal$}}}}{2} + \lfloor \frac{i}{2} \rfloor }$ become ``close'' to each other, but with a difference that decays rapidly as $d_{\raisebox{-1pt}{\tiny{$\Mcal$}}} \to \infty$, such that in the infinite-dimensional limit the larger ``half'' of the eigenvalues sum to one and the smaller ``half'' sum to zero. A subtle point to note is that because the relative entropy involves the ratio of final to original eigenvalues it is not enough that the absolute difference $|\lambda^{\prime \downarrow i}_{\raisebox{-1pt}{\tiny{$\Mcal$}}} - \lambda^{\downarrow i}_{\raisebox{-1pt}{\tiny{$\Mcal$}}}|$ goes to zero, as in the infinite $d_{\raisebox{-1pt}{\tiny{$\Mcal$}}}$ limit, it is possible for this to happen for all of the eigenvalues approaching zero without the ratios of final to initial eigenvalues approaching unity (and hence the relative entropy not vanishing). One manner of satisfying such a constraint, as evidenced by the construction we proceed with next, is for the ratios of final to initial eigenvalues go to unity for all but a small number energy levels, with the population in this exceptional subspace going to zero in the infinite $d_{\raisebox{-1pt}{\tiny{$\Mcal$}}}$ limit (along with the ratios not diverging within said subspace).

The natural question that arises here is whether or not it is possible to satisfy these constraints concurrently. (Note that none of the cooling protocols provided throughout this paper use the max-cooling operation, so do not necessarily serve as examples.) To this end, we now construct a family of machine Hamiltonians $H_{\raisebox{-1pt}{\tiny{$\Mcal$}}}$ of increasing dimension that in the limit $d_{\raisebox{-1pt}{\tiny{$\Mcal$}}} \rightarrow \infty$ manages to attain both perfect cooling of a maximally mixed qubit and the Landauer limit for the energy cost using the maximal cooling operation discussed above. The form of the Hamiltonian is instructive regarding the complexity requirements for perfect cooling at the Landauer limit. The construction is inspired by the infinite-dimensional Hamiltonian found in Ref.~\cite{Reeb_2014} (Appendix D), therein used to perfectly cool a qubit with energy cost arbitrarily close to the Landauer limit. Their construction already begins with infinitely many machine eigenvalues, as well as infinitely many of them corresponding to diverging energy levels. In the following, we demonstrate that one can arbitrarily closely attain perfect cooling and the Landauer limit with finite-dimensional Hamiltonians, and by taking the limit $d_{\raisebox{-1pt}{\tiny{$\Mcal$}}} \rightarrow \infty$, recover the result of Ref.~\cite{Reeb_2014}.

The Hamiltonian of the machine is $d_{\raisebox{-1pt}{\tiny{$\Mcal$}}} := 2^{N+1}$ dimensional,
\begin{align}
    H_{\raisebox{-1pt}{\tiny{$\Mcal$}}} &= \sum_{n=0}^N \sum_{j=1}^{2^n} \bigg( n \Delta \ket{n;j}\!\bra{n;j}_{\raisebox{-1pt}{\tiny{$\Mcal$}}} \bigg) + N\Delta \ket{N;2^N\!\!+\!\!1}\!\bra{N;2^N\!\!+\!\!1}_{\raisebox{-1pt}{\tiny{$\Mcal$}}}
\end{align}
Here, each energy eigenvalue labelled by $n$ is  $2^n$-fold degenerate. Thus the ground state is unique, the first excited state is twofold degenerate, the second excited state fourfold degenerate, and so on, with the degeneracy doubling every energy level. In order to make the Hamiltonian of even dimensionality for convenience, we add an extra degenerate state to the final level [which makes this level $(2^N+1)$-fold degenerate]. Also note that the Hamiltonian is equally spaced with energy gap $\Delta$. In the following, we use the index $n$ to denote any one of the degenerate states in the $n^{\textup{th}}$ energy level from $n=0$ to $n=N$, and the index $i$ to denote individual energy eigenstates from $i=1$ to $i=2^{N+1}$ (note that in contrast to the previous section, we are here beginning with $i=1$ in order to simplify some future notation). With these indices, the eigenvalues are related by
\begin{alignat}{2}
    \lambda_{\raisebox{-1pt}{\tiny{$\Mcal$}}}^{\downarrow i} &=  e^{-\beta \Delta} \lambda_{\raisebox{-1pt}{\tiny{$\Mcal$}}}^{\downarrow \lfloor \frac{i}{2} \rfloor} \qquad &&\forall \; i \in \{ 2, \hdots, d_{\raisebox{-1pt}{\tiny{$\Mcal$}}} - 1\}, \label{eq:rwcooling1}\\
    \lambda_{\raisebox{-1pt}{\tiny{$\Mcal$}}}^{\downarrow n} &=  e^{-\beta \Delta} \lambda_{\raisebox{-1pt}{\tiny{$\Mcal$}}}^{\downarrow n-1} \qquad &&\forall \; n \in \{ 1, \hdots, N\}. \label{eq:rwcooling2}
\end{alignat}
\noindent We introduce a parameter $\epsilon$ to express the Gibbs ratio as
\begin{align}
    e^{-\beta \Delta} &= \frac{1 - \epsilon}{2},
\end{align}
where $0<\epsilon<1$, and we eventually take the limit $\epsilon \rightarrow 0$ appropriately as the dimension diverges. Note that this constrains the Gibbs ratio to be smaller than $\tfrac{1}{2}$, which in turn ensures that the total population over all of the degenerate eigenstates in the $n^{\textup{th}}$ level is smaller than that in the $(n-1)^{\textup{th}}$ level (as it has twice the number of eigenstates, but less than half the population in each). If this constraint failed to hold, then in the asymptotic limit, all of the population would lie in energy levels that diverge.

We now consider using this machine to cool a maximally mixed qubit target. The final ground-state population of the qubit under the maximal cooling operation is the sum over the larger half of the eigenvalues of the machine, corresponding to the eigenvalues from $i=1$ to $i=2^N$ (equivalently, from $n=0$ to $n=N-1$ plus a single eigenvalue from the $n=N$ energy level), and is thus given by
\begin{align}
    p_0^\prime &= \frac{1}{\mathcal{Z}_{\raisebox{-1pt}{\tiny{$\Mcal$}}}} \left( \sum_{n=0}^{N-1} 2^n \left( \frac{1-\epsilon}{2} \right)^n + \left( \frac{1-\epsilon}{2} \right)^N \right), \\
    \text{where} \quad \mathcal{Z}_{\raisebox{-1pt}{\tiny{$\Mcal$}}} &= \sum_{n=0}^{N} 2^n \left( \frac{1-\epsilon}{2} \right)^n + \left( \frac{1-\epsilon}{2} \right)^N \noindent
\end{align}
is the partition function of the machine. The geometric series above evaluates to
\begin{align}
    p_0^\prime &= \left( 1 + \frac{\epsilon (1 - \epsilon)^N}{1 - (1-\epsilon)^N + \epsilon (1-\epsilon)^N 2^{-N}} \right)^{-1}.
\end{align}
As an ansatz, supposing that $\epsilon$ scales inversely with $N$ as $\epsilon := \tfrac{\theta}{N}$ leads to the simplification $(1-\epsilon)^N \rightarrow e^{-\theta}$ as $d_{\raisebox{-1pt}{\tiny{$\Mcal$}}}$ (and hence $N$) diverges. The asymptotic behaviour of the ground-state population is thus
\begin{align}
    p_0^\prime &= 1 - \frac{1}{N} \left( \frac{\theta}{e^{\theta} - 1} \right) + O \left( \frac{1}{N^2} \right),\label{eq:asymptoticgroundpop}
\end{align}
and so $p_0^\prime \rightarrow 1$ in the $N \rightarrow \infty$ limit.

We now move to calculate the energy cost. Rather than considering the optimal max-cooling operation described above, we consider a slight modification in order to make the connection to the construction in Ref.~\cite{Reeb_2014} clear as well as to simplify notation. Nonetheless, the energy cost of this modified protocol upper bounds that of the max-cooling operation (for the same achieved ground-state population), and so showing that the Landauer limit is attained for the modified protocol implies that it would be too for the max-cooling protocol. The modification is simply to relabel the smallest eigenvalue of the machine $\lambda^{ 2^{N+1}}_{\raisebox{-1pt}{\tiny{$\Mcal$}}}$ as $\lambda^{0}_{\raisebox{-1pt}{\tiny{$\Mcal$}}}$, and treat it as the ground-state eigenvalue in the max-cooling operation. For general machine states, such a switch would lead to less cooling (if the same unitary were applied), but in this case it does not because the sum of the first half of the machine eigenvalues, from $i=0$ to $i=2^N-1$, is the same as the original sum from $i=1$ to $i=2^N$, due to the relabelling $\lambda_0 = \lambda_{2^N}$, since they are both eigenvalues of states corresponding the maximum excited energy level of the machine spectrum. The spectrum of the final state of the machine is then given by
\begin{align}
    \lambda_{\raisebox{-1pt}{\tiny{$\Mcal$}}}^{\prime\downarrow i} &=  \frac{1}{2} \left(  \lambda_{\raisebox{-1pt}{\tiny{$\Mcal$}}}^{\downarrow \lfloor \frac{i}{2} \rfloor} + \lambda_{\raisebox{-1pt}{\tiny{$\Mcal$}}}^{\downarrow \lfloor \frac{i}{2} \rfloor + \frac{d_{\raisebox{-1pt}{\tiny{$\Mcal$}}}}{2}} \right) \qquad \forall \; i \in \{ 0, \hdots, d_{\raisebox{-1pt}{\tiny{$\Mcal$}}} - 1\},
\end{align}
which leads to
\begin{align}
    \lambda^{\prime \downarrow 0}_{\raisebox{-1pt}{\tiny{$\Mcal$}}} &= \frac{1}{2} \left( \lambda^{\downarrow 0}_{\raisebox{-1pt}{\tiny{$\Mcal$}}} + \lambda^{\downarrow 2^N}_{\raisebox{-1pt}{\tiny{$\Mcal$}}} \right) = \lambda^{\downarrow 0}_{\raisebox{-1pt}{\tiny{$\Mcal$}}}, \qquad 
    \lambda^{\prime \downarrow 1}_{\raisebox{-1pt}{\tiny{$\Mcal$}}} = \frac{1}{2} \left( \lambda^{\downarrow 0}_{\raisebox{-1pt}{\tiny{$\Mcal$}}} + \lambda^{\downarrow 2^N}_{\raisebox{-1pt}{\tiny{$\Mcal$}}} \right) = \lambda^{\downarrow 0}_{\raisebox{-1pt}{\tiny{$\Mcal$}}}, \notag \\
    \lambda^{\prime \downarrow i}_{\raisebox{-1pt}{\tiny{$\Mcal$}}} &= \frac{1}{2} \left(  \lambda_{\raisebox{-1pt}{\tiny{$\Mcal$}}}^{\downarrow \lfloor \frac{i}{2} \rfloor} + \lambda_{\raisebox{-1pt}{\tiny{$\Mcal$}}}^{\downarrow \lfloor \frac{i}{2} \rfloor + \frac{d_{\raisebox{-1pt}{\tiny{$\Mcal$}}}}{2}} \right) \qquad \forall \; i \in \{ 2, \hdots, d_{\raisebox{-1pt}{\tiny{$\Mcal$}}} - 1\} \notag \\
    &= \frac{1}{2} \left( \frac{2}{1-\epsilon} \lambda^{\downarrow i}_{\raisebox{-1pt}{\tiny{$\Mcal$}}} + \lambda^{n=N}_{\raisebox{-1pt}{\tiny{$\Mcal$}}} \right) = \frac{1}{2} \frac{1}{\mathcal{Z}_{\raisebox{-1pt}{\tiny{$\Mcal$}}}} \left[ \left( \frac{1-\epsilon}{2} \right)^{n-1} + \left( \frac{1-\epsilon}{2} \right)^N \right],
\end{align}
where we observe that the index $\lfloor \tfrac{i}{2} \rfloor + \tfrac{d_{\raisebox{-1pt}{\tiny{$\Mcal$}}}}{2}$ corresponds to the largest energy level of the machine for all $i$, and we use Eq.~\eqref{eq:rwcooling1} for the spectrum of initial eigenvalues. Using the index $n$ instead to denote a generic eigenvalue of the $n^{\textup{th}}$ energy level, we have the simpler expression
\begin{align}
    \lambda^{\prime \downarrow n}_{\raisebox{-1pt}{\tiny{$\Mcal$}}} &= \frac{1}{2} \left( \lambda^{\downarrow n-1}_{\raisebox{-1pt}{\tiny{$\Mcal$}}} + \lambda^{\downarrow N}_{\raisebox{-1pt}{\tiny{$\Mcal$}}} \right), \qquad \forall \; n \in \{1,2, \hdots, N\}.
\end{align}
The energy cost can now be simply calculated from the difference in the average energy of the machine state,
\begin{align}
     \Delta E_{\raisebox{-1pt}{\tiny{$\Mcal$}}} &= \sum_{i=0}^{d_{\raisebox{-1pt}{\tiny{$\Mcal$}}} - 1} \left( \lambda^{\prime \downarrow i}_{\raisebox{-1pt}{\tiny{$\Mcal$}}} - \lambda^{\downarrow i}_{\raisebox{-1pt}{\tiny{$\Mcal$}}} \right) \omega_i,
\end{align}
where we denote the $i^\textup{th}$ energy eigenvalue by $\omega_i$. $\lambda^{\downarrow 0}_{\raisebox{-1pt}{\tiny{$\Mcal$}}}$ is unchanged, and although $\lambda^{\downarrow 1}_{\raisebox{-1pt}{\tiny{$\Mcal$}}}$ does change, $\omega_1=0$ corresponds to the ground state and thus this eigenvalue change does not affect the energy cost. We can thus express the energy cost in terms of the index $n$ instead, starting from $n=1$ (corresponding to $i=2$ onward), as
\begin{align}
    \Delta E_{\raisebox{-1pt}{\tiny{$\Mcal$}}} &= \sum_{n=1}^{N} \left( \lambda^{\prime \downarrow n}_{\raisebox{-1pt}{\tiny{$\Mcal$}}} - \lambda^{ \downarrow n}_{\raisebox{-1pt}{\tiny{$\Mcal$}}} \right) \omega_n = \frac{1}{\beta} \left[ 1 - \frac{2 (1-\epsilon)^N}{1 - (1-\epsilon)^N + (1 - 2^{-N}) (1-\epsilon)^N \epsilon} \right] \log \left( \frac{2}{1-\epsilon} \right).
\end{align}
As we did above, we parameterise $\epsilon = \tfrac{\theta}{N}$. The asymptotic behaviour of the energy cost is then
\begin{align}
    \beta \Delta E_{\raisebox{-1pt}{\tiny{$\Mcal$}}} &= \log (2) + \frac{1}{N} \left( 1 - \frac{2 \log (2)}{e^{\theta} - 1} \right) \theta + O \left( \frac{1}{N^2} \right),\label{eq:asymptoticworkcost}
\end{align}
or in terms of the decrease in entropy of the system,
\begin{align}
    \beta \Delta E_{\raisebox{-1pt}{\tiny{$\Mcal$}}} &= \tilde{\Delta} S_{\raisebox{-1pt}{\tiny{$\Mcal$}}} + \frac{\log N}{N} \left( \frac{\theta}{e^{\theta} - 1} \right) + O \left( \frac{1}{N} \right). 
\end{align}
Combining \eqref{eq:asymptoticgroundpop} and \eqref{eq:asymptoticworkcost}, we have that in the limit $N \rightarrow \infty$, which is also $d_{\raisebox{-1pt}{\tiny{$\Mcal$}}} \rightarrow \infty$, the ground-state population approaches $1$---corresponding to perfect cooling---and the energy cost approaches $\beta^{-1} \log (2)$, which is the Landauer limit for the perfect erasure of a maximally mixed qubit.

To connect this construction to the constraints of Eq.~\eqref{eq:maingenerictermeigenvalue}, note that in the limit $N\rightarrow \infty$ (recalling that $\epsilon = \frac{\theta}{N}$),
\begin{align}
    \frac{\lambda^{\prime \downarrow n}}{\lambda^{\downarrow n}} &= \lim_{N\rightarrow \infty} \frac{ \frac{1}{2} \left[ \left( \frac{1-\epsilon}{2} \right)^{n-1} + \left( \frac{1-\epsilon}{2} \right)^N \right]}{\left( \frac{1-\epsilon}{2} \right)^n} = \lim_{N \rightarrow \infty} \left[ \frac{1}{1-\epsilon} + \frac{1}{2^{N-n+1}} \left( \frac{e^{-\theta}}{(1-\epsilon)^n} \right) \right] = 1,
\end{align}
for all $n \geq 1$, leaving only the ground-state eigenvalue (corresponding to $n=0$ and $i=1$) not satisfying the condition. However, this term is actually a negative contribution to the relative entropy as this eigenvalue decreases, and in any case can be verified independently to approach zero.

To see this, note that a necessary condition that ensures the contribution of any set of eigenvalues that do not satisfy Eq.~\eqref{eq:maingenerictermeigenvalue} to the relative entropy to be negligible is that the total population of the relevant subspace is vanishingly small. Writing the relative entropy between two states in terms of their eigenvalues, we have $D(\varrho^\prime\|\varrho) = \sum_{n} \lambda^\prime_n \log\left( \tfrac{\lambda^\prime_n}{\lambda_n}\right)$, which we split up into two sets: $S_{0}$ containing all $n$ for which Eq.~\eqref{eq:maingenerictermeigenvalue} is satisfied and $S_{\pm}$ containing the all $n$ for which Eq.~\eqref{eq:maingenerictermeigenvalue} is not satisfied. The contribution of the first term to the relative entropy is asymptotically zero, so we are left with $D(\varrho^\prime\|\varrho) = \sum_{n\in S_{\pm}} \lambda^\prime_n \log\left( \tfrac{\lambda^\prime_n}{\lambda_n}\right)$. For each term in the sum here, one can write $\lambda^\prime_n = \lambda(1+\Delta_n)$ with the condition $|\Delta_n| \geq \theta > 0$ for some $\theta$, i.e., the ratio of eigenvalues is bounded away from unity (on either side) by at least $\theta$. This leads to the expression
\begin{align}\label{eq:appdrelativeentropydeviation-start}
    D(\varrho^\prime\|\varrho) = - \sum_{n\in S_{\pm}} \lambda^\prime_n \log(1+\Delta_n) = - N_{\pm} \sum_{n\in S_{\pm}} p_n \log(1+\Delta_n),
\end{align}
where we renormalise the eigenvalues (which here correspond to a subnormalised probability distribution) by writing $\lambda_n^\prime = N_{\pm} p_n$, with $N_{\pm} := \sum_{n\in S_{\pm}} \lambda^\prime_n$ being the total population of the subspace $S_{\pm}$ and $\{p_n\}$ here forming a probability distribution. Note that the ratio of eigenvalues going to unity in the $S_0$ subspace implies that the total populations of initial and final eigenvalues in this subspace are equal, i.e., $\sum_{n \in S_0} \lambda_n = \sum_{n \in S_0} \lambda_n^\prime$, which in turn implies that the same is true for the $S_{\pm}$ subspace, leading to $\sum_{n \in S_{\pm}} p_n \Delta_n = 0$.

We argue from the concavity of the logarithm function that
\begin{align}
    \tfrac{1}{2} \log(1+\theta) + \tfrac{1}{2} \log(1-\theta) \geq \sum_{n\in S_{\pm}} p_n \log(1+\Delta_n).
\end{align}
Visualising the graph of the function $y=\log(1+x)$, the latter expression above must evaluate to a point that lies within the intersection of the convex hull of $(\Delta_n, \log(1+\Delta_n) )$ and the linear equality $\sum_{n \in S_{\pm}} p_n \Delta_n = 0$, the latter of which is the line $x=0$. By the concavity of the logarithm, the aforementioned convex hull lies entirely below the line segment connecting $(1-\theta, \log(1-\theta))$ to $(1+\theta, \log(1+\theta))$, and thus the expression is upper bounded by the intersection of this line segment with $x=0$, which is precisely the l.h.s. of the inequality above. Thus we have the inequality
\begin{align}\label{eq:appdrelativeentropydeviation-end}
    D(\varrho^\prime\|\varrho) \geq -N_{\pm} \left[ \tfrac{1}{2} \log(1+\theta) + \tfrac{1}{2} \log(1-\theta) \right] = -\frac{N_{\pm}}{2}\log(1-\theta^2) \geq \frac{N_{\pm}}{2} \theta^2,
\end{align}
where we use $\log(1-\theta^2) \leq - \theta^2$ for all $\theta \in [-1,1]$. As $\theta > 0$, the only way that this contribution to the relative entropy by the eigenvalues that do not satisfy Eq.~\eqref{eq:maingenerictermeigenvalue} can be asymptotically negligible is if the total population of their associated subspace $N_{\pm}$ goes to zero.

Finally note that, as mentioned in the main text, the above result pertains to the restricted setting where the target system is cooled as much as possible. However, this is not the only way to approach perfect cooling at the Landauer cost: instead of the largest half of global eigenvalues being placed into the ground-state subspace of the target system, any amount of them such that their sum is sufficiently close to one would suffice. Although it is complicated to derive an exact set of conditions that would need to be satisfied in such cases (since it depends upon exactly which eigenvalues are permuted to which subspaces), the fact that fine-tuned control over particular degrees of freedom is required remains. Lastly, note that even in the restricted setting of cooling the target as much as possible, the situation becomes even more complicated when considering target systems that begin at a finite temperature. Here, the choice of which global eigenvalues should be permuted to which subspaces to cool the system as much as possible at minimal energy cost depends on the microscopic structure of both the system and machine. This means that one can no longer determine the final eigenvalue distributions of the reduced states in terms of the initial machine eigenvalues alone, as we were able to do for the maximally mixed state. In turn, one can no longer derive a condition on properties of the machine itself, independently of the target system. Nonetheless, again, the key message that cooling at minimal energy cost requires fine-tuned control to access precisely distributed populations still holds true. We leave the further exploration of such scenarios, for instance constructing optimal machines for particular initial target systems, to future work. 

\subsection{Energy-Gap Variety as a Notion of Control Complexity}
\label{app:energygapvariety}

The insights drawn above regarding sufficient conditions for cooling a system at the Landauer limit lead us to propose a more nuanced notion of control complexity than the preliminary effective dimension that satisfies the natural desiderata outlined in the main text. In particular, here we demonstrate that the \emph{energy-gap variety} (see Definition~\ref{def:energygapvariety}) provides a good measure of control complexity, both from a theoretical, thermodynamic standpoint as well as a practical one. 

Firstly, it is quite clear that coupling a system to a diverging number of distinct machine energy gaps is a difficult task to achieve in almost any conceivable physical platform, especially when the energy gaps are closely spaced; thus, this definition indeed corresponds to our intuitive understanding of ``complex'' as an operation that is inherently difficult to perform in practice. Secondly, from all of the optimal cooling protocols that we outline in this paper, we see that, in contrast to the effective dimension, having a diverging energy-gap variety that densely covers an appropriate interval is sufficient for saturating the Landauer limit, thereby making it a better quantifier of control complexity. The remaining point is to show that its divergence is necessary to cool a system to the ground state using a single control operation with energy cost saturating the Landauer bound, so that it is fully consistent also with Nernst's unattainability principle. We argue that this is indeed the case below by proving Theorem~\ref{thm:energygapvariety-main}.

\textbf{Proof:} First of all, note that how cold the final system state can be made is bounded by the inequality:
\begin{align}\label{eq:appgeneralpuritybound-2}
    \lambda_{\textup{min}}(\varrho_{\raisebox{-1pt}{\tiny{$\Scal$}}}^\prime) \geq e^{-\beta\, \omega_{\raisebox{0pt}{\tiny{$\Mcal$}}}^{\textup{max}}} \lambda_{\textup{min}}(\varrho_{\raisebox{-1pt}{\tiny{$\Scal$}}}),
\end{align}
where $\lambda_{\textup{min}}$ denotes the minimal eigenvalue. For a pure final system state, the l.h.s. of the above equation goes to 0; thus, for any nontrivial initial system state [i.e., such that $\lambda_{\textup{min}}(\varrho_{\raisebox{-1pt}{\tiny{$\Scal$}}})>0$] and finite temperature $\beta < \infty$, we must have $\omega_{\raisebox{0pt}{\tiny{$\Mcal$}}}^{\textup{max}} \to \infty$. This determines the upper limit of the required interval of energy gaps. The lower limit of the required interval comes from the fact that the only subspaces of the machine that are relevant for cooling the target system are those associated to energy gaps that are at least as large as the smallest energy gap of the target, $\omega_0$~\cite{Clivaz_2019E}.

Next, recall the equality form of the Landauer limit, which holds true for \emph{any} global unitary transformation with a thermal machine:
\begin{align}\label{eq:landauerequality-appd}
    \beta \Delta E_{\raisebox{-1pt}{\tiny{$\Mcal$}}} = \widetilde{\Delta} S_{\raisebox{-1pt}{\tiny{$\Scal$}}} + I(\Scal: \Mcal)_{\varrho_{\raisebox{-1pt}{\tiny{$\Scal \Mcal$}}}^\prime} + D(\varrho_{\raisebox{-1pt}{\tiny{$\Mcal$}}}^\prime \| \varrho_{\raisebox{-1pt}{\tiny{$\Mcal$}}}),
\end{align}
Cooling the target system to a pure state necessitates that the final system and machine are uncorrelated and we therefore have $I(\Scal: \Mcal)_{\varrho_{\raisebox{-1pt}{\tiny{$\Scal \Mcal$}}}^\prime} = 0$ for the optimal process. We thus need to focus on minimising the relative-entropy term, which we do in the following steps. 

Consider for simplicity the target system to be a qubit initially in the maximally mixed state. A generic cooling machine should be able to cool \emph{any} system state, include the maximally mixed one; therefore the following insights pertaining to this special case apply generically. In this case, the initial joint spectrum of the system and machine is given by Eq.~\eqref{eq:maximallymixedglobalspectrum}. Cooling the target system to the ground state necessitates taking a subset $\mathcal{A}$ of these global eigenvalues such that $\sum_{i \in \mathcal{A}} \lambda_{\raisebox{-1pt}{\tiny{$\Mcal$}}}^{\downarrow i} = 1 - \epsilon$ for arbitrarily small $\epsilon$ and placing these into the ground-state subspace of the target, with the remaining (small) $\epsilon$ amount of population contributing only to any higher-energy eigenstates [this is essentially a generalisation of the conditions put forth in Eqs. \eqref{eq:maineigenvaluesum}, accounting for an arbitrarily small cooling error]. 

As discussed previously, there are many possible ways to achieve such a configuration, but there is a \emph{unique} one that minimises the total energy cost of doing so: namely, that in which the reduced final state of the machine is rendered passive. This is because if one compares two protocols achieving the same cooling for the target system, one in which the final machine is passive and any other in which it is not, then the former protocol has the smaller energy cost since a positive amount of energy can be (unitarily) extracted from the latter machine in order to render it passive. 

Thus, for any protocol saturating the Landauer limit, the final machine state must be arbitrarily close to a passive state, which implies that it must be approximately diagonal in the local machine energy eigenbasis with the globally ordered spectrum as per Eq.~\eqref{eq:appdgloballyorderedfinalspectrum}. Moreover, in order to minimise the relative-entropy term and therefore saturate the Landauer limit, the final machine state must be arbitrarily close to the initial (i.e., thermal) machine state; following the argumentation put forth in the previous Appendix, this leads to the set of conditions outlined in Eq.~\eqref{eq:lemmagenerictermeigenvalue}, which must be satisfied up to arbitrary precision.

Since we have the exact relationship between the initial and final machine eigenvalues, the contribution to the energy cost from the relative-entropy term can be calculated explicitly, i.e., $D(\varrho^\prime\|\varrho) = \sum_{n} \lambda^\prime_n \log\left( \tfrac{\lambda^\prime_n}{\lambda_n}\right)$. Following the argumentation from Eq.~\eqref{eq:appdrelativeentropydeviation-start} until Eq.~\eqref{eq:appdrelativeentropydeviation-end} in the previous appendix, we see that by assuming a finite deviation from \emph{any} of the conditions of Eq.~\eqref{eq:lemmagenerictermeigenvalue}, i.e., writing $\lambda_n^\prime = \lambda_n(1+\Delta_n)$ with $|\Delta_n| \geq \theta > 0$ for some $\theta$, one can derive a lower bound on the relative entropy:
\begin{align}\label{eq:appdlowerboundrelativeentropysimple}
    D(\varrho^\prime\|\varrho) \geq  \frac{N_{\pm}}{2} \theta^2,
\end{align}
where $N_{\pm}$ is the total population of the subspaces corresponding to the terms in the sum such that $\tfrac{\lambda^\prime_n}{\lambda_n}$ differs from unity by at least $\theta$. In other words, these are the relevant additional contributions to the energy cost; whenever $N_{\pm}$ is nonzero, the Landauer bound \emph{cannot} be approached arbitrarily closely.

The final piece is to relate the machine eigenvalues to its energy-gap spectrum, which can be done straightforwardly due to the initial thermality of the machine, i.e., $\lambda_{\raisebox{-1pt}{\tiny{$\Mcal$}}}^{\downarrow i} = e^{-\beta \omega_i}/\mathcal{Z}_{\raisebox{-1pt}{\tiny{$\Mcal$}}}(\beta, H_{\raisebox{-1pt}{\tiny{$\Mcal$}}})$. We now argue that if there is ever a finite ``jump'' in the energy-gap structure of the machine, then one cannot achieve both a ground-state population of the target that is arbitrarily close to unity \emph{and} have $N_{\pm}$ be arbitrarily close to zero concurrently. Suppose now that one has a machine with a dense energy-gap structure from $\omega_0$ up until some (finite) $\omega_a$, followed by a finite jump until the energy level $\omega_a + \Omega$ (for some strictly finite $\Omega>0$), and then again a dense set of energy gaps throughout the interval $[\omega_a + \Omega, \infty)$. Then, one can utilise the energy-gap structure in the ``lower band'' $[\omega_0,\omega_a)$ in an optimal fashion in order to cool the target system to a minimum temperature (set by $\omega_a$) at Landauer energy cost~\cite{Clivaz_2019L,Clivaz_2019E}. However, assuming that the jump in the energy-gap structure begins at some finite $\omega_a$, then there is always a finite amount of population in the machine that is supported on the energy levels corresponding to the ``upper band'' $[\omega_a + \Omega, \infty)$. To cool the target to arbitrarily close to the ground state, one must therefore access this population and transfer it to the ground-state subspace of the target; this precisely corresponds to the $N_{\pm}$ that contributes to the excess energy cost in a non-negligible manner for finite population exchanges. In particular, we have the bound $N_\pm \geq \textup{min}({\frac{e^{-\beta \omega_a}}{1+e^{-\beta \omega_a}},\frac{1}{1+e^{-\beta \omega_a}}})$. Thus, whenever $\omega_a$ takes a finite value, $N_\pm$ is a strictly positive number. The only way that the relative-entropy term can vanish then is if $\theta$ vanishes; however, this can occur only if $\Omega \to 0$, because for any finite $\Omega$, the ratio $\frac{\lambda^\prime_n}{\lambda_n}$ for at least one value of $n$ differs from 1 by a finite amount as argued above, which finally leads to a nonzero lower bound in Eq.~\eqref{eq:appdlowerboundrelativeentropysimple} and implies that the Landauer limit cannot be saturated. In other words, the endpoints of the lower and upper energy gap intervals considered above must coincide (up to arbitrary precision) in order to saturate the Landauer bound. This implies that the energy-gap variety must diverge and moreover, since the above logic holds for arbitrary $\omega_\alpha$, which can be smoothly varied as a parameter, it follows that the diverging number of energy gaps must additionally approximately densely cover the interval in question. 


\section{Diverging Time and Diverging Control Complexity Cooling Protocols for Harmonic Oscillators}
\label{app:harmonicoscillators}

We now analyse the case of cooling infinite-dimensional quantum systems in detail. More specifically, we consider ensembles of harmonic oscillators. For the sake of completeness, we first briefly present some key concepts that will become relevant throughout this analysis. Following this, in Appendix~\ref{app:hodivergingtimegaussian}, we construct a protocol that achieves perfect cooling at the Landauer limit using a diverging number of Gaussian operations. Although such operations are typically considered to be relatively ``simple'' both when it comes to experimental implementation and theoretical description, according to the effective dimension notion of control complexity that we have shown must necessarily diverge to cool at the Landauer limit [see Eq.~\eqref{eq:maineffectivedimension}], such Gaussian operations have infinite control complexity. Subsequently, in Appendix~\ref{app:hodivergingtimenongaussian}, we consider the task of perfect cooling with diverging time but restricting the individual operations to be of finite control complexity. In particular, note that such operations are non-Gaussian in general. Here, we present a protocol that approaches perfect cooling of the target system as the number of operations diverges, with finite energy cost---albeit not at the Landauer limit. Whether or not a similar protocol exists that also saturates the Landauer bound remains an open question. Finally, in Appendix~\ref{app:hodivergingcontrolcomplexitygaussian}, we reconsider the protocol from Appendix~\ref{app:hodivergingtimenongaussian} in terms of a single transformation, i.e., unit time. By explicitly constructing the joint unitary transformation that is applied throughout the entire protocol, we show this to be a multimode Gaussian operation acting on a diverging number of harmonic oscillators. The key message to be taken away from these protocols is that, while the distinction between Gaussian and non-Gaussian operations is a significant one in terms of experimental feasibility, and it certainly plays a role regarding the task of cooling---in particular, the energy cost incurred---these concepts alone cannot be used to characterise a notion of control complexity that must diverge to approach perfect cooling at the Landauer limit. On the other hand, the effective dimension of the machine used does precisely that; however, in a manner that is far from sufficient (for the case of harmonic oscillators), as even a single two-mode swap, which cannot cool perfectly at Landauer cost, would have infinite control complexity. Indeed, a more nuanced characterisation of control complexity in the infinite-dimensional setting, which takes more structure regarding the operations and energy levels into account, remains an open problem to be addressed. 

\subsection{Preliminaries}
\label{app:hopreliminaries}

We consider ensembles of $N$ harmonic oscillators (i.e., infinite-dimensional systems consisting of $N$ bosonic modes), which are associated to a tensor product Hilbert space $\mathcal{H}_{\textup{tot}}=\bigotimes_{j=1}^N \mathcal{H}_{j}$ and (respectively: lowering, raising) mode operators $\{a_k$ , \,$a_k^{\dagger}\}$ satisfying the bosonic commutation relations:
\begin{equation}
    [a_k,a_{k'}^{\dagger}]=\delta_{kk'} , \quad \quad  [a_k,a_{k'}]=0, \quad \quad \quad \forall \; k,k'=1,2,\hdots,N. \label{eq:bosonicommrelation}
\end{equation}
The free Hamiltonian of any such system can be written as $H_{\textup{tot}}=\sum_{k=1}^N \omega_k a_{k}^{\dagger}a_{k}$, where $\omega_{k}$ represents the energy gap of the $k$-th mode (in units where $\hbar=1$). Position- and momentum-like operators for each mode can be defined as follows (for simplicity, we use the rescaled version below where the $\omega_k$ are omitted from the prefactors)
\begin{equation}
    q_k:=\frac{1}{\sqrt{2}}(a_{k}+a_{k}^{\dagger}), \quad \quad p_{k}:=\frac{1}{i \sqrt{2}}(a_{k}-a_{k}^{\dagger}).
\end{equation}
As a consequence of the commutation relations in Eq.~\eqref{eq:bosonicommrelation}, the generalised position and momentum operators satisfy the canonical commutation relations
\begin{equation}
    \left[ q_{k},\,p_l\right]=\, i\delta_{kl}. 
\end{equation}
To simplify notation, one may further introduce the vector of quadrature operators $\mathds{X}:=( q_{1},\, p_{1},\hdots,\,  q_{N},\, p_{N})$; then, the commutation relations can be expressed succinctly as
\begin{equation}
     \left[ \mathds{X}_{k},\,\mathds{X}_{l}\right]=\, i \Omega_{kl}, 
\end{equation}
where the $\Omega_{kl} $ are the components of the symplectic form
\begin{equation}
    \Omega =\bigoplus_{j=1}^N \Omega_{j}, \quad \quad \Omega_{j}=\begin{bmatrix}
    0 & 1\\
    -1 & 0
    \end{bmatrix}.
\end{equation}
The density operator associated to $N$ harmonic oscillators can be written in the so-called \emph{phase-space representation} as
\begin{align}
    \varrho = \frac{1}{(2 \pi)^N} \int \chi(\Omega \xi) \mathcal{W}(- \Omega \xi) \, d^{2N} \xi,
\end{align}
where $\mathcal{W}(\xi) := e^{i \xi^{T} \mathds{X}}$ is the Weyl operator and $\chi(\xi) := \tr{\varrho \mathcal{W}(\xi)}$ is called the characteristic function.

Throughout our analysis, we see that a particular class of states and operations, namely those that are known as \emph{Gaussian}, are of particular importance. A Gaussian state is one for which the characteristic function is Gaussian
\begin{align}
    \chi(\xi) = e^{-\frac{1}{4} \xi^T \Gamma \xi + i \overline{\mathds{X}}^{T} \xi}.
\end{align}
Here, $\overline{\mathds{X}} := \left<\mathds{X}\right>_\varrho$ is the \emph{displacement vector} or \emph{vector of first moments}, and $\Gamma$ is a real symmetric matrix that collects the \emph{second statistical moments} of the quadratures, which is known as the \emph{covariance matrix}. Its entries are given by
\begin{equation}
     \Gamma_{mn}:=\left<\mathds{X}_{m}\mathds{X}_{n}+\mathds{X}_{n}\mathds{X}_{m} \right>_\varrho-2\left<\mathds{X}_{n}\right>_\varrho\left<\mathds{X}_{m} \right>_\varrho.
\end{equation}
We see that any Gaussian state is thus uniquely determined by its first and second moments. As an example of specific interest here, we recall that any thermal state $\tau$ of a harmonic oscillator with frequency $\omega$ is a Gaussian state and has vanishing first moments, $\overline{\mathds{X}}=0$. Here and throughout this article, we are assuming that the infinite-dimensional thermal state is well defined (see, e.g., Ref.~\cite{Thirring_2002} for discussion). The covariance matrix of a thermal state is proportional to the $2\times2$ identity, and given by $ \Gamma[\tau(\beta, H)]= \coth{\left(\frac{\beta \omega}{2}\right)}\,\openone_{2}$. 

Gaussian operations are transformations that map the set of Gaussian states onto itself. Such operations, which include, e.g., beam-splitting and phase-space displacement, are generally considered to be relatively easily implementable in the laboratory. Although nonunitary Gaussian operations exist as well, all of the examples mentioned above are Gaussian unitaries. Such Gaussian unitaries are generated by Hamiltonians that are at most quadratic in the raising and lowering operators. Conversely, any Hamiltonian that can be expressed as a polynomial of at most second order in the mode operators generates a Gaussian unitary. Any unitary Gaussian transformation can be represented by an affine map $(M, \kappa)$,
 \begin{equation}
     \mathds{X} \mapsto M \mathds{X}+\kappa ,
 \end{equation}
where $\kappa \in \mathds{R}^{2N}$ is a displacement vector in the phase-space representation and $M$ is a symplectic $2N \times 2N$ matrix that leaves the symplectic form $\Omega$ invariant, i.e., 
\begin{equation}
    M \, \Omega \, M^T=  \Omega.
\end{equation}
Under such a mapping, the first and second moments transform according to
\begin{align}\label{eq:gaussianaffinemoments}
    \overline{\mathds{X}} \mapsto M \overline{\mathds{X}} + \kappa, \quad \quad
    \Gamma \mapsto M \Gamma M^T.
\end{align}
Lastly, note that the energy of a Gaussian state $\varrho_G$ with respect to its free Hamiltonian $H = \sum_k \omega_k a_k^\dagger a_k$ can be calculated in terms of the first and second moments as follows~\cite{Friis2018}
\begin{equation}
    E(\varrho_{G})=\sum_{k} \omega_{k}\left(\frac{1}{4}\tr{\Gamma^{(k)}-\,2}\, +\frac{1}{2}\,||\overline{\mathds{X}}^{(k)}||^2\right), \label{eq:gaussianenergy}
\end{equation}
where $\| \cdot \|$ denotes the Euclidean norm. Here, $\Gamma^{(k)}$ indicates the  $(2 \times 2)$ submatrix of the full covariance matrix $\Gamma$ corresponding to the reduced state of the $k^{\textup{th}}$ mode. Similarly $\overline{\mathds{X}}^{(k)}$ denotes the two-component subvector of first moments for the $k^{\textup{th}}$ mode of the displacement vector $\overline{\mathds{X}}$.\\

\subsection{Diverging-Time Cooling Protocol for Harmonic Oscillators}
\label{app:hodivergingtime}

\subsubsection{Diverging-Time Protocol using Gaussian Operations (with Diverging Control Complexity)}
\label{app:hodivergingtimegaussian}

We now consider a simple protocol for lowering the temperature of a single-mode system within the coherent-control paradigm using a single harmonic oscillator machine. This protocol will form the basic step of a protocol for achieving perfect cooling at the Landauer limit using diverging time, which we subsequently present.

In the situation we consider here, the target system $\Scal$ to be cooled is a harmonic oscillator with frequency $\omega_{\raisebox{0pt}{\tiny{$\Scal$}}}$ interacting with a harmonic oscillator machine $\Mcal$ at frequency $\omega_{\raisebox{0pt}{\tiny{$\Mcal$}}} \geq \omega_{\raisebox{0pt}{\tiny{$\Scal$}}}$ via a (non-energy-conserving) unitary acting on the joint system $\Scal\Mcal$ initialised as a tensor product of thermal states $\tau_{\raisebox{-1pt}{\tiny{$\Scal$}}}(\beta, H_{\raisebox{-1pt}{\tiny{$\Scal$}}}) \otimes \tau_{\raisebox{-1pt}{\tiny{$\Mcal$}}}(\beta,H_{\raisebox{-1pt}{\tiny{$\Mcal$}}}) $ at inverse temperature $\beta$ with respect to their local Hamiltonians $H_{\raisebox{-1pt}{\tiny{$\Scal$}}}$ and $H_{\raisebox{-1pt}{\tiny{$\Mcal$}}}$, respectively. 
The joint covariance matrix of the system and machine modes is block diagonal since the initial state is of product form, i.e.,
\begin{equation}
    \Gamma[\tau_{\raisebox{-1pt}{\tiny{$\Scal$}}}(\beta, H_{\raisebox{-1pt}{\tiny{$\Scal$}}}) \otimes \tau_{\raisebox{-1pt}{\tiny{$\Mcal$}}}(\beta, H_{\raisebox{-1pt}{\tiny{$\Mcal$}}})]=  \Gamma[\tau_{\raisebox{-1pt}{\tiny{$\Scal$}}}(\beta, H_{\raisebox{-1pt}{\tiny{$\Scal$}}}) ]\oplus  \Gamma[ \tau_{\raisebox{-1pt}{\tiny{$\Mcal$}}}(\beta, H_{\raisebox{-1pt}{\tiny{$\Mcal$}}})],
\end{equation}
and the $2\times2$ blocks of the individual modes are also diagonal, 
with the explicit expression $ \Gamma[\tau_{X}(\beta, H_{\raisebox{-1pt}{\tiny{$\Xcal$}}})]= \coth{\left(\frac{\beta \omega_{\raisebox{0pt}{\tiny{$\Xcal$}}}}{2}\right)}\,\openone_{2}$. 

In this setting, it has been shown that the minimum reachable temperature of the target system is given by $T_{\textup{min}}=\frac{\omega_{\raisebox{0pt}{\tiny{$\Scal$}}}}{\omega_{\raisebox{0pt}{\tiny{$\Mcal$}}}}\,T$ (for the case $ \omega_{\raisebox{0pt}{\tiny{$\Mcal$}}} \geq \omega_{\raisebox{0pt}{\tiny{$\Scal$}}}$)~\cite{Clivaz_2019E}. The non-energy-conserving unitary transformation that achieves this is of the form
\begin{equation} \label{eq:swapwithi}
    U=e^{-i\frac{\pi}{2}( a^{\dagger} b+\, a b^{\dagger})},
\end{equation}
where the operators $a~(a^{\dagger})$ and $b~(b^{\dagger})$ denote the annihilation (creation) operators of the target system and machine, respectively. This beam-splitter-like unitary acts as a \texttt{SWAP} with a relative phase factor imparted on the resultant state; nonetheless, this phase is irrelevant at the level of the covariance matrix, which fully characterises the (Gaussian) thermal states considered, and transforms it according to a standard swapping of the systems. After acting with such a \texttt{SWAP} operator, which is a Gaussian operation, the first moment remains vanishing and the covariance matrix transforms as [see Eq.~\eqref{eq:gaussianaffinemoments}]
\begin{align}
    \begin{bmatrix}
    \coth{\left( \frac{\beta \omega_{\raisebox{0pt}{\tiny{$\Scal$}}}}{2} \right)\,\openone_2}& 0\\
    0& \coth{\left( \frac{\beta \omega_{\raisebox{0pt}{\tiny{$\Mcal$}}}}{2} \right)} \,\openone_2
    \end{bmatrix}\, \overset{\textup{SWAP}}{\longmapsto }\,  \begin{bmatrix}
    \coth{\left( \frac{\beta \omega_{\raisebox{0pt}{\tiny{$\Mcal$}}}}{2} \right)} \,\openone_2& 0\\
    0& \coth{\left( \frac{\beta \omega_{\raisebox{0pt}{\tiny{$\Scal$}}}}{2} \right)\,\openone_2}
\label{cov. transformation}    \end{bmatrix}.
\end{align}
This means that both the output target system and machine are thermal states at different temperatures $T'_{\raisebox{-1pt}{\tiny{$\Scal$}}}=\frac{\omega_{\raisebox{0pt}{\tiny{$\Scal$}}}}{\omega_{\raisebox{0pt}{\tiny{$\Mcal$}}}}\,T $ and $T'_{\raisebox{-1pt}{\tiny{$\Mcal$}}}=\frac{\omega_{\raisebox{0pt}{\tiny{$\Mcal$}}}}{\omega_{\raisebox{0pt}{\tiny{$\Scal$}}}}\,T$. Making use of Eq.~\eqref{eq:gaussianenergy}, we can calculate the energy change for the system and machine as
\begin{align}
    \Delta E_{\raisebox{-1pt}{\tiny{$\Scal$}}} &= E\left[\tau_{\raisebox{-1pt}{\tiny{$\Scal$}}}\left(\frac{\omega_{\raisebox{0pt}{\tiny{$\Mcal$}}}}{\omega_{\raisebox{0pt}{\tiny{$\Scal$}}}}\beta, H_{\raisebox{-1pt}{\tiny{$\Scal$}}}\right) \right]-\, E\left[\tau_{\raisebox{-1pt}{\tiny{$\Scal$}}}(\beta, H_{\raisebox{-1pt}{\tiny{$\Scal$}}})\right]= \frac{\omega_{\raisebox{0pt}{\tiny{$\Scal$}}}}{2}\,\left[\coth{\left(\frac{\beta \omega_{\raisebox{0pt}{\tiny{$\Mcal$}}} }{2}\right)}-\coth{\left(\frac{\beta \omega_{\raisebox{0pt}{\tiny{$\Scal$}}} }{2}\right)}\right],\nonumber\\
    \Delta E_{\raisebox{-1pt}{\tiny{$\Mcal$}}} &= E\left[\tau_{\raisebox{-1pt}{\tiny{$\Mcal$}}}\left(\frac{\omega_{\raisebox{0pt}{\tiny{$\Scal$}}}}{\omega_{\raisebox{0pt}{\tiny{$\Mcal$}}}}\beta, H_{\raisebox{-1pt}{\tiny{$\Mcal$}}} \right)\right]-\, E\left[\tau_{\raisebox{-1pt}{\tiny{$\Mcal$}}}(\beta, H_{\raisebox{-1pt}{\tiny{$\Mcal$}}})\right]= \frac{\omega_{\raisebox{0pt}{\tiny{$\Mcal$}}}}{2}\,\left[\coth{\left(\frac{\beta \omega_{\raisebox{0pt}{\tiny{$\Scal$}}} }{2}\right)}-\coth{\left(\frac{\beta \omega_{\raisebox{0pt}{\tiny{$\Mcal$}}} }{2}\right)}\right].
    \label{EnergyExch.2moeds}
\end{align}
The total energy cost associated to this \texttt{SWAP} operation is thus
\begin{align}
     \Delta E_{\raisebox{-1pt}{\tiny{$\Scal \Mcal$}}}
     \,=\,
     \Delta E_{\raisebox{-1pt}{\tiny{$\Scal$}}} + \Delta E_{\raisebox{-1pt}{\tiny{$\Mcal$}}}
    \,=\,
    \frac{(\omega_{\raisebox{0pt}{\tiny{$\Mcal$}}}-\omega_{\raisebox{0pt}{\tiny{$\Scal$}}})}{2}\,\left[\coth{\left(\frac{\beta \omega_{\raisebox{0pt}{\tiny{$\Scal$}}}}{2}\right)}-\coth{\left(\frac{\beta \omega_{\raisebox{0pt}{\tiny{$\Mcal$}}} }{2}\right)}\right]
     \,=\,
     (\omega_{\raisebox{0pt}{\tiny{$\Mcal$}}}-\omega_{\raisebox{0pt}{\tiny{$\Scal$}}})\,\frac{e^{-\beta \omega_{\raisebox{0pt}{\tiny{$\Scal$}}}} \,(1-e^{-\beta( \omega_{\raisebox{0pt}{\tiny{$\Mcal$}}}-\omega_{\raisebox{0pt}{\tiny{$\Scal$}}})})}{(1-e^{-\beta\omega_{\raisebox{0pt}{\tiny{$\Scal$}}}})(1-e^{-\beta \omega_{\raisebox{0pt}{\tiny{$\Mcal$}}}})}.
     \label{freeEn bosonic modes}
\end{align}
Note that this form is similar to that for finite-dimensional systems with equally spaced Hamiltonian [cf., Eq.~\eqref{eq:tot energy d}]; the dimension-dependent term vanishes as $d \to \infty$, simplifying the expression in the infinite-dimensional case.

With this simple protocol for lowering the temperature of a harmonic oscillator target using a single harmonic oscillator machine at hand, we are now in a position to describe an energy-optimal (in the sense of saturating the Landauer bound) cooling protocol when a diverging number of operations, i.e., diverging time, is permitted. In other words, we now show how to achieve perfect cooling with minimal energy at the expense of requiring diverging time, i.e., infinitely many steps of finite duration. As mentioned above, in this specific protocol, the control complexity as per Eq.~\eqref{eq:maineffectivedimension} is infinite in each of these infinitely many steps. As we argue after having presented the protocol, this is an artefact of the simple structure of the Gaussian operations used. Indeed, we later present a non-Gaussian diverging-time protocol for cooling a single harmonic oscillator to the ground state using finite control complexity in each of the infinitely many steps, and at an overall finite (albeit not minimal, i.e., not at the Landauer limit) energy cost. Before presenting this non-Gaussian protocol, let us now discuss the details of the Gaussian diverging-time protocol for cooling at the Landauer limit.

We consider a harmonic oscillator with the frequency $\omega_{\raisebox{0pt}{\tiny{$\Scal$}}}$ as the target system and the machine to comprise $N$ harmonic oscillators, where the $n^{\textup{th}}$ oscillator has frequency $\omega_{M_n}= \omega_{\raisebox{0pt}{\tiny{$\Scal$}}} +\, n\,\epsilon$. In addition, we assume that all modes are initially uncorrelated and in thermal states at the same inverse temperature $\beta$ with respect to their free Hamiltonians, i.e., the target system is $\tau_{\raisebox{-1pt}{\tiny{$\Scal$}}}(\beta,H_{\raisebox{-1pt}{\tiny{$\Scal$}}})$ and the multimode thermal machine is $\tau_{\raisebox{-1pt}{\tiny{$\Mcal$}}}(\beta, H_{\raisebox{-1pt}{\tiny{$\Mcal$}}}) = \bigotimes_{n=1}^N\,\tau_{\raisebox{-1pt}{\tiny{$\Mcal$}$_n$}}(\beta, H_{\raisebox{-1pt}{\tiny{$\Mcal$}$_n$}})$. 

In this case, the cooling process is divided into $N$ time steps. During each step, there is an interaction between the target system and one of the harmonic oscillators in the machine. Here, we assume that at the $n^{\textup{th}}$ time step, the target system interacts only with the $n^{\textup{th}}$ harmonic oscillator, which has frequency $\omega_{\raisebox{0pt}{\tiny{$\Scal$}}} +\, n\,\epsilon$. To obtain the minimum temperature for the target system, we perform the previously outlined cooling process at each step, which is given by swapping the corresponding two modes. Using Eq.~\eqref{cov. transformation}, the covariance matrix transformation of the two-mode process at the first time step takes the form
\begin{align}
 \Gamma^{(1)}(\tau_{\raisebox{-1pt}{\tiny{$\Scal$}}}(\beta) \otimes \tau_{\raisebox{-1pt}{\tiny{$\Mcal$}$_1$}}(\beta))&= \begin{bmatrix}
    \coth{\left(\frac{\beta \omega_{\raisebox{0pt}{\tiny{$\Scal$}}}}{2}\right)}\openone_2& 0\\
    0& \coth{\left(\frac{\beta (\omega_{\raisebox{0pt}{\tiny{$\Scal$}}}+\epsilon)}{2}\right)}\,\openone_2
    \end{bmatrix}\, 
    \overset{\textup{SWAP}}{\longmapsto }\, \Gamma^{(1)}_{\textup{opt}}&=\, \begin{bmatrix}
    \coth{\left(\frac{\beta (\omega_{\raisebox{0pt}{\tiny{$\Scal$}}}+\epsilon)}{2}\right)}\,\openone_2& 0\\
    0& \coth{\left(\frac{\beta \omega_{\raisebox{0pt}{\tiny{$\Scal$}}}}{2}\right)}\,\openone_2
\label{cov. transformation sm}    \end{bmatrix}.
\end{align}
By repeating this process on each of the harmonic oscillators in the machine, after the $(n-1)^{\textup{th}}$ step, the $2\times 2$ block corresponding to the target system $\Scal$ in the covariance matrix is given by $\coth{\left(\frac{\beta (\omega_{\raisebox{0pt}{\tiny{$\Scal$}}}+(n-1)\epsilon)}{2}\right)}\,\openone_2$. Therefore, one can show inductively that the covariance matrix transformation associated to the $n^{\textup{th}}$ interaction is given by
\begin{align}
 \Gamma^{(n)}(\tau_{\raisebox{-1pt}{\tiny{$\Scal$}}}(\beta) \otimes \tau_{\raisebox{-1pt}{\tiny{$\Mcal$}$_n$}}(\beta))=& \begin{bmatrix}
    \coth{\left(\frac{\beta (\omega_{\raisebox{0pt}{\tiny{$\Scal$}}}+(n-1)\epsilon)}{2}\right)}\openone_2& 0\\
    0& \coth{\left(\frac{\beta (\omega_{\raisebox{0pt}{\tiny{$\Scal$}}}+n\epsilon)}{2}\right)}\,\openone_2
    \end{bmatrix}\, \notag \\ 
    \overset{\textup{SWAP}}{\longmapsto }\, \Gamma^{(n)}_{\textup{opt}}=&\, \begin{bmatrix}
    \coth{\left(\frac{\beta (\omega_{\raisebox{0pt}{\tiny{$\Scal$}}}+n\epsilon)}{2}\right)}\,\openone_2& 0\\
    0& \coth{\left(\frac{\beta (\omega_{\raisebox{0pt}{\tiny{$\Scal$}}}+(n-1)\epsilon)}{2}\right)}\,\openone_2
\label{cov. transformation sm 2}    \end{bmatrix}.
\end{align} 
Based on this process, after $N$ steps (i.e., after the system has interacted with all $N$ harmonic oscillators), the minimal achievable temperature of the target system is $T^{(N)}_{\textup{min}} = \frac{\omega_{\raisebox{0pt}{\tiny{$\Scal$}}}}{\omega_{\raisebox{0pt}{\tiny{$\Scal$}}} + N \epsilon} T$. Moreover, by using Eq.~(\ref{EnergyExch.2moeds}), one can calculate the energy changes of the target system and the machine at each time step as
\begin{align}
    \Delta E^{(n)}_{\raisebox{-1pt}{\tiny{$\Scal$}}}&=  \frac{\omega_{\raisebox{0pt}{\tiny{$\Scal$}}}}{2}\,\left[\coth{\left(\frac{\beta (\omega_{\raisebox{0pt}{\tiny{$\Scal$}}}+n\epsilon)}{2}\right)}-\coth{\left(\frac{\beta (\omega_{\raisebox{0pt}{\tiny{$\Scal$}}}+(n-1)\epsilon)}{2}\right)}\right],\notag\\
    \Delta E^{(n)}_{\raisebox{-1pt}{\tiny{$\Mcal$}$_n$}} &=  \frac{(\omega_{\raisebox{0pt}{\tiny{$\Scal$}}}+n\epsilon)}{2}\,\left[\coth{\left(\frac{\beta (\omega_{\raisebox{0pt}{\tiny{$\Scal$}}}+(n-1)\epsilon)}{2}\right)}-\coth{\left(\frac{\beta (\omega_{\raisebox{0pt}{\tiny{$\Scal$}}}+n\epsilon)}{2}\right)}\right].
    \label{EnergyExch.2moeds step m}  
\end{align}
The total energy change for the target system during the overall process (i.e., throughout the $N$ steps) is thus given by 
\begin{align}
    \Delta E_{\raisebox{-1pt}{\tiny{$\Scal$}}} &=\sum_{n=1}^{N}\Delta E^{(n)}_{\raisebox{-1pt}{\tiny{$\Scal$}}}
    =  \sum_{n=1}^{N}\frac{\omega_{\raisebox{0pt}{\tiny{$\Scal$}}}}{2}\,\left[\coth{\left(\frac{\beta (\omega_{\raisebox{0pt}{\tiny{$\Scal$}}}+n\epsilon)}{2}\right)}-\coth{\left(\frac{\beta (\omega_{\raisebox{0pt}{\tiny{$\Scal$}}}+(n-1)\epsilon)}{2}\right)}\right]\nonumber\\
    &=\frac{\omega_{\raisebox{0pt}{\tiny{$\Scal$}}}}{2}\,\left[\coth{\left(\frac{\beta (\omega_{\raisebox{0pt}{\tiny{$\Scal$}}}+N\epsilon)}{2}\right)}-\coth{\left(\frac{\beta \omega_{\raisebox{0pt}{\tiny{$\Scal$}}}}{2}\right)}\right]= \omega_{\raisebox{0pt}{\tiny{$\Scal$}}}\,\left[\frac{e^{-\beta (\omega_{\raisebox{0pt}{\tiny{$\Scal$}}}+N\epsilon)}}{1-e^{-\beta (\omega_{\raisebox{0pt}{\tiny{$\Scal$}}}+N\epsilon)}}-\frac{e^{-\beta \omega_{\raisebox{0pt}{\tiny{$\Scal$}}}}}{1-e^{-\beta \omega_{\raisebox{0pt}{\tiny{$\Scal$}}}}}\right].
    \label{TotEnergyExch S}  
\end{align}
Here, we write $\coth{(x)} = 1+(2e^{-2x})/(1-e^{-2x})$. Similarly, one can obtain the total energy change of the overall machine
\begin{align}
   \Delta E_{\raisebox{-1pt}{\tiny{$\Mcal$}}} &=\sum_{n=1}^{N}\Delta E^{(n)}_{\raisebox{-1pt}{\tiny{$\Mcal$}$_n$}}=  \sum_{n=1}^{N}\frac{\omega_{\raisebox{0pt}{\tiny{$\Scal$}}}+ n\epsilon}{2}\,\left[\coth{\left(\frac{\beta (\omega_{\raisebox{0pt}{\tiny{$\Scal$}}}+(n-1)\epsilon)}{2}\right)}-\coth{\left(\frac{\beta (\omega_{\raisebox{0pt}{\tiny{$\Scal$}}}+n\epsilon)}{2}\right)}\right]\nonumber\\
     &=  \sum_{n=1}^{N}(\omega_{\raisebox{0pt}{\tiny{$\Scal$}}}+ n\epsilon)\,\left[\frac{e^{-\beta (\omega_{\raisebox{0pt}{\tiny{$\Scal$}}}+(n-1)\epsilon)}}{1-e^{-\beta (\omega_{\raisebox{0pt}{\tiny{$\Scal$}}}+(n-1)\epsilon)}}-\frac{e^{-\beta (\omega_{\raisebox{0pt}{\tiny{$\Scal$}}}+n\epsilon)}}{1-e^{-\beta (\omega_{\raisebox{0pt}{\tiny{$\Scal$}}}+n\epsilon)}}\right].
    \label{eqTotEnergyExchMach.}  
\end{align}
It is straightforward to check that the total energy change, i.e., the sum of Eqs.~\eqref{TotEnergyExch S} and \eqref{eqTotEnergyExchMach.}, is equal to the energy cost obtained in Eq.~\eqref{eq:tot energy d} with $d \to \infty$. In particular, this can be seen by considering the second line of Eq.~\eqref{eq:tot energy d}, where the second term in round parenthesis vanishes as $d \to \infty$ for any value of $N$. Thus, when the number of operations diverges $N\to \infty$ and $\epsilon =(\omega_{\textup{max}}-\omega_{\raisebox{0pt}{\tiny{$\Scal$}}})/N \to 0$, where $\omega_{\textup{max}} := \frac{\beta_{\textup{max}}}{\beta}\omega_{\raisebox{0pt}{\tiny{$\Scal$}}}$ is the maximum frequency of the machines, the heat dissipated by the machines throughout the process saturates the Landauer bound and is therefore energetically optimal. Moreover, by taking $\omega_{\textup{max}} \to \infty$ one approaches perfect cooling. 

At this point, a comment on the notion of control complexity is in order. According to Eq.~\eqref{eq:maineffectivedimension}, the effective dimension of the machine in the protocol we consider here diverges in addition to time. Indeed, the notion of control complexity thusly defined diverges for \emph{any} Gaussian operation acting on the machine, in particular, it diverges for any single one of the infinitely many steps of the protocol, as each operation is a two-mode Gaussian operation. At first glance, this appears to be in contrast to the common conception that Gaussian operations are typically easily implementable (cf. Refs.~\cite{BrownFriisHuber2016,Friis2018}). However, an alternative way of interpreting this protocol is that, exactly because of the simple structure of Gaussian operations, reaching the ground state at finite energy cost requires a diverging number of two-mode Gaussian unitaries, and thus divergingly many modes on which to act (see also Appendix~\ref{app:hodivergingcontrolcomplexitygaussian}). In fact, if non-Gaussian unitaries are employed, then the ground state can be approached at finite energy cost using just a single harmonic oscillator machine, as we now show. 

\subsubsection{Diverging-Time Protocol using Non-Gaussian Operations (with Finite Control Complexity)}
\label{app:hodivergingtimenongaussian}

We now consider a protocol for cooling a single harmonic oscillator at frequency $\omega_{\raisebox{0pt}{\tiny{$\Scal$}}}$ to the ground state using a diverging amount of time, but requiring only a finite overall energy input as well as finite control complexity in each of the diverging number of steps of the protocol. In this protocol, the machine $\Mcal$ is also represented by a single harmonic oscillator whose frequency matches that of the target oscillator that is to be cooled, $\omega_{\raisebox{0pt}{\tiny{$\Mcal$}}}=\omega_{\raisebox{0pt}{\tiny{$\Scal$}}}=:\omega$. The initial states of both the target system $\Scal$ and machine $\Mcal$ are assumed to be thermal at the same inverse temperature $\beta$, and are hence both described by thermal states of the form
\begin{align}
    \tau(\beta) &=\,\frac{
    e^{-\beta H}}{\tr{e^{-\beta H}}}
        \,=\,
    \sum\limits_{n=0}^{\infty}
    e^{-\beta\omega\nr n}
    (1-e^{-\beta\omega})
    \,\ket{n}\!\bra{n}
    \,=\,
    \sum\limits_{n=0}^{\infty}
    p_{n}
    \,\ket{n}\!\bra{n}{\raisebox{-1pt}{\tiny{$\Scal \Mcal$}}},
\end{align}
where the Hamiltonian $H$ is given by $H=\sum_{n=0}^{\infty}n\omega\,\ket{n}\!\bra{n}$ and the $p_{n}=e^{-\beta\omega\nr n}(1-e^{-\beta\omega})$ are the eigenvalues of $\tau$. The joint initial state is a product state that we can then write as
\begin{align}
\tau_{\raisebox{-1pt}{\tiny{$\Scal$}}}(\beta)\otimes\tau_{\raisebox{-1pt}{\tiny{$\Mcal$}}}(\beta)  &=\,\sum\limits_{m,n=0}^{\infty}
    p_{m}p_{n}
    \,\ket{m}\!\bra{m}_{\raisebox{-1pt}{\tiny{$\Scal$}}}\otimes\ket{n}\!\bra{n}_{\raisebox{-1pt}{\tiny{$\Mcal$}}}
    \,=\,\sum\limits_{m,n=0}^{\infty}
    \tilde{p}_{m+n}
    \,\ket{m,n}\!\bra{m,n},
\end{align}
where we define $\tilde{p}_{k}:=e^{-\beta\omega\nr k}(1-e^{-\beta\omega})^{2}$. We then note that the eigenvalues $\tilde{p}_{k}$ of the joint initial state have degeneracy $k+1$. For instance, the largest value $\tilde{p}_{0}=p_{0}p_{0}$, corresponding to both the system and machine being in the ground state, is the single largest eigenvalue, but there are two eigenstates, $\ket{0,1}$ and $\ket{1,0}$, corresponding to the second largest eigenvalue $\tilde{p}_{1}$, three states, $\ket{0,2}$, $\ket{1,1}$, and $\ket{2,0}$ for the third largest eigenvalue $\tilde{p}_{2}$, and so forth. Obviously, not all of these eigenvalues correspond to eigenstates for which the target system is in the ground state. 

In order to increase the ground-state population of the target system oscillator, we can now apply a sequence of `two-level' unitaries, i.e., unitaries that act only on a subspace spanned by two particular eigenstates and exchange their respective populations. The two-dimensional subspaces are chosen such that one of the two eigenstates corresponds to the system $\Scal$ being in the ground state, $\ket{0,k}$, while the other eigenstate corresponds to $\Scal$ being in an excited state, $\ket{i\neq0,j}$. In addition, these pairs of levels are selected such that, at the time the unitary operation is to be performed, the population of $\ket{0,k}$ is smaller than that of $\ket{i\neq0,j}$, such that the two-level exchange increases the ground-state population of $\Scal$ at each step. 

More specifically, at the $k^\text{th}$ step of this sequence, the joint system $\Scal\Mcal$ is in the state $\varrho\suptiny{0}{0}{(k)}_{\raisebox{-1pt}{\tiny{$\Scal\Mcal$}}}$ and one determines the set $\Omega_{k}$ of index pairs $(i\neq 0,j)$ such that $\tilde{p}_{k}<\bra{i,j}\varrho\suptiny{0}{0}{(k)}_{\raisebox{-1pt}{\tiny{$\Scal\Mcal$}}}\ket{i,j}$, i.e., the set of eigenstates for which $\Scal$ is not in the ground state and which have a larger associated population (at the beginning of the $k^\text{th}$ step) than $\ket{0,k}$. One then determines an index pair $(m_{k},n_{k})$ for which this population is maximal, i.e., $\bra{m_{k},n_{k}}\varrho_{\raisebox{-1pt}{\tiny{$\Scal\Mcal$}}}\suptiny{0}{0}{(k)}\ket{m_{k},n_{k}}=\max\{\bra{i,j}\varrho_{\raisebox{-1pt}{\tiny{$\Scal\Mcal$}}}\suptiny{0}{0}{(k)}\ket{i,j}|(i,j)\in \Omega_{k}\}$, and performs the unitary
\begin{align}
    U^{(k)}_{\raisebox{-1pt}{\tiny{$\Scal\Mcal$}}}   &=\,\openone_{\raisebox{-1pt}{\tiny{$\Scal\Mcal$}}}-\ket{0,k}\!\bra{0,k}-\ket{m_{k},n_{k}}\!\bra{m_{k},n_{k}}
    +\Bigl(
    \ket{0,k}\!\bra{m_{k},n_{k}}+
    \ket{m_{k},n_{k}}\!\bra{0,k}
    \Bigr).
\end{align}
If there is no larger population that is not already in the subspace of the ground state of the target system, i.e., when $\Omega_{k}=\emptyset$, which is only the case for the first step ($k=1$), then no unitary is performed. After the $k^\text{th}$ step, the joint state $\varrho\suptiny{0}{0}{(k+1)}_{\raisebox{-1pt}{\tiny{$\Scal\Mcal$}}}$ is still diagonal in the energy eigenbasis, and the subspace of the joint Hilbert spaces for which $\Scal$ is in the ground state is populated with the $k+1$ largest eigenvalues $\tilde{p}_{i}$ in nonincreasing order with respect to nondecreasing energy eigenvalues of the subspace's basis vectors $\ket{0,i}$. That is, for all $i\in\{0,1,2,\ldots,k\}$ and for all $j\in\mathbbm{N}$ with $j>i$, we have $\bra{0,i}\varrho_{\raisebox{-1pt}{\tiny{$\Scal\Mcal$}}}\suptiny{0}{0}{(k+1)}\ket{0,i}\geq \bra{0,j}\varrho_{\raisebox{-1pt}{\tiny{$\Scal\Mcal$}}}\suptiny{0}{0}{(k+1)}\ket{0,j}$.

Since the Hilbert spaces of both $\Scal$ and $\Mcal$ are infinite dimensional, we can continue with such a sequence of two-level exchanges indefinitely, starting with $k=1$ and continuing step by step as $k\rightarrow\infty$. Here we note that the choice of $(m_{k},n_{k})$ is generally not unique at the $k$-th step, but as $k\rightarrow\infty$, the resulting final state is independent of the particular choices of $(m_{k},n_{k})$ made along the way. In particular, in a fashion that is reminiscent of the famed Hilbert hotel paradox (see, e.g., Ref.~\cite[p.~17]{Gamow1947}), this sequence places \emph{all} of the infinitely many eigenvalues $\tilde{p}_{k}$ of the joint state of $\Scal\Mcal$ (which must hence sum to one) into the subspace where $\Scal$ is in the ground state. In other words, in the limit of infinitely many steps, the population of the ground-state subspace evaluates to
\begin{align}
    \sum\limits_{k=0}^{\infty}(k+1)\tilde{p}_{k}\,=\,
    \sum\limits_{k=0}^{\infty}(k+1)\,e^{-\beta\omega\nr k}\,(1-e^{-\beta\omega})^{2}\,=\,1,
\end{align}
where we take into account the $(k+1)$-fold degeneracy of the $k^\text{th}$ eigenvalue $\tilde{p}_{k}$. We thus have $\lim_{k\rightarrow\infty}\ptr{\Mcal}{\varrho_{\raisebox{-1pt}{\tiny{$\Scal\Mcal$}}}\suptiny{0}{0}{(k)}}=\ket{0}\!\bra{0}_{\raisebox{-1pt}{\tiny{$\Scal$}}}$, the reduced state of the system is asymptotically the pure state $\ket{0}_{\raisebox{-1pt}{\tiny{$\Scal$}}}$. 

As per our requirement on the structural complexity (see Appendix~\ref{app:conditionsstructuralcontrolcomplexity}), the Hilbert space of the machine required to achieve this is infinite-dimensional, and since each step of the protocol is assumed to take a finite amount of time, the overall time for reaching the ground state diverges. At the same time, the control complexity for each individual step is finite, since each $U_{k}$ acts nontrivially only on a two-dimensional subspace. To see that also the energy cost for this protocol is finite, we first note that the protocol results in a final state of the machine that is diagonal in the energy eigenbasis $\ket{n}_{\raisebox{-1pt}{\tiny{$\Mcal$}}}$, with probability weights $\tilde{p}_{k}$ decreasing (but not strictly) with increasing energy. Due to the degeneracy of the eigenvalues $\tilde{p}_{k}$, each one appears $(k+1)$ times on the diagonal (w.r.t. the energy eigenbasis) of the resulting machine state, populating adjacent energy levels. The label $n(k)$ of the lowest energy level that is populated by a particular value $\tilde{p}_{k}$ can be calculated as 
\begin{align}
    \tilde{n}(k):=\sum\limits_{n=0}^{k-1}(n+1)\,=\,\tfrac{1}{2}k(k+1),
\end{align}
while the largest energy populated by $\tilde{p}_{k}$ is given by $\tilde{n}(k+1)-1$. With this, we calculate the energy of the machine after the protocol, which evaluates to
\begin{align}
    \frac{E_{\raisebox{-1pt}{\tiny{$\Mcal$}}}^{\raisebox{0pt}{\tiny{final}}}}{\omega} &=\,\sum\limits_{k=1}^{\infty}e^{-\beta\omega\nr k}\,(1-e^{-\beta\omega})^{2}\,\sum\limits_{n=\tilde{n}(k)}^{\tilde{n}(k+1)-1}n
    =\,\sum\limits_{k=1}^{\infty}e^{-\beta\omega\nr k}\,(1-e^{-\beta\omega})^{2}\,\tfrac{1}{2}k(k+1)(k+2)\,=\,
    \tfrac{3}{4}\operatorname{cosech}^{2}\bigl(\tfrac{\beta\omega}{2}\bigr).
\end{align}
Since the energy of the initial thermal state is given by
\begin{align}
    \frac{E\left[\tau(\beta)\right]}{\omega}   &=\,\sum\limits_{n=0}^{\infty}
    n\,e^{-\beta\omega\nr n}\,(1-e^{-\beta\omega})\,=\,\frac{e^{-\beta\omega}}{1-e^{-\beta\omega}},
\end{align}
we thus arrive at the energy cost
\begin{align}
    \frac{\Delta E_{\raisebox{-1pt}{\tiny{$\Mcal$}}}}{\omega}  &=\,\frac{E_{\raisebox{-1pt}{\tiny{$\Mcal$}}}^{\raisebox{0pt}{\tiny{final}}}-E\left[\tau(\beta)\right]}{\omega} \,=\,
    \frac{e^{-\beta\omega}(2+e^{-\beta\omega})}{(1-e^{-\beta\omega})^{2}}.
    \label{eq:DeltaE harmonic osci nonGaussian infinite control protocol}
\end{align}
We thus see that this energy cost is finite for all finite initial temperatures (although note that the energy cost diverges when $\beta\rightarrow0$). 

However, as we show next, the energy cost for attaining the ground state is not minimal, i.e., the protocol achieves perfect cooling (with finite energy and control complexity, but infinite time) but not at the Landauer limit. To see this, we first observe that the entropy of the final pure state of the system $\Scal$ vanishes, such that $\widetilde{\Delta} S_{\raisebox{-1pt}{\tiny{$\Scal$}}}=S\left[\tau(\beta)\right]$. Evaluating this entropy, one obtains
\begin{align}
    S\left[\tau(\beta)\right]   &=-\tr{\tau\log(\tau)}
    =
    -\sum\limits_{n=0}^{\infty}e^{-\beta\omega\nr n}
    (1\nl-\nl e^{-\beta\omega})\log\left[e^{-\beta\omega\nr n}(1\nl -\nl e^{-\beta\omega})\right]
    =
    -\sum\limits_{n=0}^{\infty}e^{-\beta\omega\nr n}(1\nl-\nl e^{-\beta\omega})\left[-\beta\omega\nr n+\log(1\nl-\nl e^{-\beta\omega})\right]\nonumber\\
    &=\,\frac{\beta\omega e^{-\beta\omega}}{1-e^{-\beta\omega}}\,+\,
    \beta\omega\,+\,\log\Bigl(\frac{ e^{-\beta\omega}}{1-e^{-\beta\omega}}\Bigr)
    \,=\,
    \frac{\beta\omega}{1-e^{-\beta\omega}}\,+\,\log\Bigl(\frac{ e^{-\beta\omega}}{1-e^{-\beta\omega}}\Bigr).
    \label{eq:thermal state entropy harmonic osci}
\end{align}
Using the results from Eqs.~\eqref{eq:DeltaE harmonic osci nonGaussian infinite control protocol} and~\eqref{eq:thermal state entropy harmonic osci}, we can thus compare the expressions for $\beta\Delta E_{\raisebox{-1pt}{\tiny{$\Mcal$}}}$ and  $\widetilde{\Delta} S_{\raisebox{-1pt}{\tiny{$\Scal$}}}$, and we find that $\beta\Delta E_{\raisebox{-1pt}{\tiny{$\Mcal$}}}-\widetilde{\Delta} S_{\raisebox{-1pt}{\tiny{$\Scal$}}}>0$ for all nonzero initial temperatures. The origin of this difference is easily identified: although the protocol results in an uncorrelated final state because the system is left in a pure state, that is, $I(\Scal: \Mcal)_{\varrho_{\raisebox{-1pt}{\tiny{$\Scal \Mcal$}}}^\prime}=0$, the last term $D(\varrho_{\raisebox{-1pt}{\tiny{$\Mcal$}}}^\prime \| \tau_{\raisebox{-1pt}{\tiny{$\Mcal$}}})$ in Eq.~\eqref{eq:landauerequality} is nonvanishing for nonzero temperatures because the protocol does not result in a thermal state of the machine. 

With this, we thus show that perfect cooling is indeed possible using a finite energy cost and a finite control complexity in every one of infinitely many steps (thus using diverging time). As we have seen, the structural requirement of an infinite-dimensional effective machine Hilbert space can be met by realising $\Mcal$ as a single harmonic oscillator. Although the presented protocol does not minimise the energy cost to saturate the Landauer bound, we cannot at this point conclusively say that it is not possible to do so in this setting. However, we suspect that a more complicated energy-level structure of the machine is necessary.

Finally, let us comment again on the notion of control complexity in terms of effective machine dimension as opposed to the notion of complexity that is often (loosely) associated with the distinction between Gaussian and non-Gaussian operations. As we see from the protocols presented here, the concept of control complexity based on the nontrivially accessed Hilbert-space dimension of the machine indeed captures the resource that must diverge in order to reach the ground state, while the intuition of complexity associated with (non)-Gaussian operations, albeit valid as a characterisation of a certain practical difficulty in realising such operations, seems to be irrelevant for determining if the ground state can be reached. In the protocol presented in this section, non-Gaussian operations with finite control complexity are used in each step to reach the ground state. Infinitely many steps (i.e., diverging time) could then be traded for a single (also non-Gaussian) operation of infinite control complexity, performed in unit time. In the previous protocol based on Gaussian operations (Appendix~\ref{app:hodivergingtimegaussian}), the control complexity diverges in every single step of the cooling protocol, but only when there are infinitely many such steps (diverging time) or one operation in unit time on infinitely many modes (see below), can we reach the ground state. However, in the latter case, the operation, although acting on a diverging number of harmonic oscillators, remains Gaussian, as we now show explicitly. 

\subsection{Diverging Control Complexity Cooling Protocol for Harmonic Oscillators}
\label{app:hodivergingcontrolcomplexitygaussian}

Here we give a protocol for perfectly cooling a harmonic oscillator in unit time and with the minimum energy cost, but with diverging control complexity. In accordance with Theorem~\ref{thm:variety}, the machines used to cool the target system will likewise be harmonic oscillators. Let the operators $a~(a^{\dagger})$ and $b_{k}~(b_{k}^{\dagger})$, respectively, denote the annihilation (creation) operators of the target system and a machine subsystem labelled $k$. We then consider the the unitary transformation in Eq.~\eqref{eq:swapwithi}, namely
\begin{equation}
    U_{k}:= e^{i\frac{\pi}{2}(a^{\dag}b_k +ab_{k}^\dag)} .
\end{equation}
One can then apply the diverging-time cooling protocol from Appendix~\ref{app:hodivergingtimegaussian} to cool the system to the ground state at the Landauer limit via the total unitary transformation
\begin{align}
    U_{\textup{tot}} := \lim_{N \to \infty} U_{(N)}, \qquad \text{ with } \qquad U_{(N)}:= \prod_{k=1}^N U_k . \label{eq:total unitary} 
\end{align}
We now seek the Hamiltonian that generates $U_{\textup{tot}}$. First note that $U_{(N)}aU_{(N)}^\dag =i b_{1}$ and 
\begin{equation}
    U_{(N)}b_{k}U_{(N)}^\dag = \begin{cases}
    -b_{k+1}, & \text{for } k < N \\
    i a, & \text{for } k = N \\
    b_{k}, & \text{for } k > N 
    \end{cases},
\end{equation}
which can be proven by induction. In contrast with Appendix~\ref{app:hodivergingtimegaussian}, here we use the complex representation of the symplectic group to describe the transformation, i.e., the set of matrices $S$ satisfying $SKS^\dag =K$, where $K:=\openone_{N} \oplus (-\openone_{N})$. Gathering the raising and lowering operators of the target system and the first $N$ machines into the vector $\vec{\xi}:=\bigl(\begin{matrix} a & b_{1} & b_{2} & \ldots & b_{N} & a^{\dag} & b_{1}^{\dag} & b_{2}^{\dag} & \ldots &  b_{N}^{\dag}\end{matrix}\bigr)^\mathrm{T}$, we can write the transformation above as $U_{(N)}\vec{\xi} \, U_{(N)}^\dag = S^\mathrm{T} \vec{\xi}$~\cite{AdessoRagyLee2014}, where
\begin{equation}
    S=\begin{pmatrix}
    \alpha_{(N)} & 0 \\
    0 & \alpha_{(N)} 
    \end{pmatrix}, \qquad \text{ with } \qquad \alpha_{(N)} := 
    \begin{pmatrix}
    0 & 0 & 0 & \ldots & 0 & i \\
    i & 0 & 0 & \ldots & 0 & 0 \\
    0 & -1 & 0 & \ldots & 0 & 0 \\
    0 & 0 & -1 & \ldots & 0 & 0 \\
    \vdots & \vdots & \vdots & \ddots & \vdots & \vdots \\
    0 & 0 & 0 & \ldots & -1 & 0
    \end{pmatrix} . \label{eq:alphamatrix}
\end{equation}
Now, defining the matrix of Hamiltonian coefficients $h_{(N)}$ implicitly by $U_{(N)} =: \exp(-i \vec{\xi}^{\,\dagger} \cdot h_{(N)} \cdot \vec{\xi})$, we have that $S=\exp(-i K h_{(N)})$~\cite{AdessoRagyLee2014}, i.e., $h_{(N)}=i K \log(S^\mathrm{T})=i K \log(S)^\mathrm{T}$, where we take the principal logarithm. To calculate this, we must diagonalise the matrix $\alpha_{(N)}$ in Eq.~\eqref{eq:alphamatrix}. The eigenvalues of $\alpha_{(N)}$ are
\begin{equation}
    \lambda_{k}:=-e^{-i\pi \frac{2k-1}{N+1}}, \qquad \text{ with } \qquad k\in\lbrace 1,2,\ldots,N+1\rbrace ,
\end{equation}
i.e., the negative of the $(N+1)^\mathrm{th}$ roots of $-1$, and it is diagonalised by the unitary matrix $V$ constructed from the eigenvectors $\vec{v}_{k}$:
\begin{equation}
    V:=\begin{pmatrix}
    \vec{v}_{1} & \vec{v}_{2} & \vec{v}_{3} & \ldots & \vec{v}_{N+1} 
    \end{pmatrix}
\qquad \text{ with } \qquad
    \vec{v}_{k}:= \frac{-1}{\sqrt{N+1}} \begin{pmatrix}
    i (-\lambda_{k})^{-1} \\ (-\lambda_{k})^{-2} \\ (-\lambda_{k})^{-3} \\ \vdots \\ (-\lambda_{k})^{-(N+1)} 
    \end{pmatrix} .
\end{equation}
Specifically, $\alpha_{(N)}=VDV^\dag$, where $D:=\diag ( \lambda_{1},\lambda_{2},\ldots,\lambda_{N+1} )$, and thus
\begin{equation}
    h_{(N)}^\mathrm{T} =i K \log \begin{pmatrix}
     V D V^\dag & 0 \\ 0 & V D V^\dag
    \end{pmatrix} = i K \begin{pmatrix}
    V & 0 \\ 0 & W
    \end{pmatrix} \begin{pmatrix}
    \log (D)  & 0 \\ 0 & \log (D)
    \end{pmatrix} \begin{pmatrix}
    V^\dag & 0 \\ 0 & V^\dag
    \end{pmatrix} =: \begin{pmatrix}
    A & 0 \\ 0 & -A
    \end{pmatrix} 
\end{equation}
for some matrix $A$. By direct calculation, one finds that
\begin{equation} \label{eq:AmatrixInfComplex}
    A_{jk} = i^{\delta_{j1}}i^{\delta_{k1}}\frac{\pi}{(N+1)^2} \sum_{p=1}^{N+1} \, (2p-2-N)e^{-i\pi\frac{2p-1}{N+1}(j-k)}.
\end{equation}
Now, considering the identity
\begin{equation}
    \sum_{p=1}^{N+1} \, e^{i \theta p} =  \frac{e^{i \theta (N+1)}-1}{1-e^{i \theta}}
\end{equation}
for $\theta \in \mathbb{R}$, as well as its derivative with respect to $\theta$, one can calculate the sum in Eq.~\eqref{eq:AmatrixInfComplex}. We then have
\begin{equation}
    \lim_{N\to\infty} A_{jk} = \begin{cases}
    0, & \text{for } j = k \\
   i i^{\delta_{j1}}i^{\delta_{k1}} \frac{1}{j-k}, & \text{for } j \neq k 
    \end{cases}.
\end{equation}
Then, finally, we have that $U_{\textup{tot}} = e^{-i H_{\textup{tot}}}$, where $H_{\textup{tot}}= \lim_{N\to\infty} \left( \vec{v}^{\,\dagger} \cdot h_{(N)} \cdot \vec{v} \right)$, i.e.,
\begin{equation}
   H_{\textup{tot}} = - \sum_{j=2}^{\infty} \left( \frac{1}{j-1} b_{j}^\dag a + \mathrm{H.c.} \right) + \sum_{j,k=1; \, j\neq k}^{\infty} \frac{i}{j-k} b_{j}^\dag b_{k}   . \label{eq:complexHamiltonian}
\end{equation}
Thus, the system is cooled to the ground state at an energy cost saturating the Landauer bound, and in unit time, but via a procedure that implements a multimode Gaussian unitary on a diverging number of modes.


\section{Cooling Protocols in the Incoherent-Control Paradigm}
\label{app:coolingprotocolsincoherentcontrolparadigm}

In this section, we investigate the required resources to cool the target system within the incoherent-control paradigm. For simplicity, we consider only the finite-dimensional setting. Here, we have a qudit target system $\Scal$ interacting resonantly (i.e., in an energy-conserving manner) with a qudit machine $\Mcal$, which is partitioned into one part, $\Ccal$, in thermal contact with the ambient environment at inverse temperature $\beta$ and another part, $\Hcal$, in contact with a hot bath at inverse temperature $\beta_{\raisebox{-1pt}{\tiny{$H$}}} < \beta$. The Hamiltonians for each subsystem are $H_{\raisebox{-1pt}{\tiny{$\Xcal$}}}=\sum_{n=0}^{d_X-1} n\, \omega_{\raisebox{0pt}{\tiny{$\Xcal$}}} \ket{n}\!\bra{n}_{\raisebox{-1pt}{\tiny{$\Xcal$}}}$; the energy resonance condition enforces that $\omega_{\raisebox{0pt}{\tiny{$\Hcal$}}}=\omega_{\raisebox{0pt}{\tiny{$\Ccal$}}}-\omega_{\raisebox{0pt}{\tiny{$\Scal$}}}$. For the most part in this section, we focus on equally spaced Hamiltonians for simplicity; we comment specifically whenever we consider otherwise.

In order to cool the target system, we aim to compress as much population as possible into the its lowest energy eigenstates via interactions that are restricted to the energy-degenerate subspaces of the joint $\Scal\Ccal\Hcal$ system. Thus we are restricted to global energy-conserving unitaries $U_{\raisebox{-1pt}{\tiny{EC}}}$ that satisfy
\begin{equation}
    [{H}_{\raisebox{-1pt}{\tiny{$\Scal$}}}+H_{\raisebox{-1pt}{\tiny{$\Ccal$}}}+H_{\raisebox{-1pt}{\tiny{$\Hcal$}}}, U_{\raisebox{-1pt}{\tiny{EC}}}]=0.
    \label{eq:tripartite commutator UH0}
\end{equation}
In Ref.~\cite{Clivaz_2019E}, it was shown that for the case where all three subsystems are qubits, the optimal global unitary in this setting (inasmuch as they render the target system in the coldest state possible given the restrictions) is
\begin{align}
    U_{\raisebox{-1pt}{\tiny{EC}}}=\ket{0,1,0}\! \bra{1,0,1}_{\raisebox{-1pt}{\tiny{$\Scal \Ccal \Hcal$}}}+\ket{1,0,1}\! \bra{0,1,0}_{\raisebox{-1pt}{\tiny{$\Scal \Ccal \Hcal$}}}+\bar{\openone},
    \label{eq:tripartite unitary op qubit}
\end{align}
where $\bar{\openone}$ denotes the identity matrix on all subspaces that are not energy degenerate. Considering the generalisation to qudit subsystems, it is straightforward to see that, for equally spaced Hamiltonians, the optimal global unitaries must be of the form 
\begin{equation}
    U_{\raisebox{-1pt}{\tiny{EC}}}=\Bigg[\sum_{m,n,l=0}^{d-2}\ket{m,n+1,l}\! \bra{m+1,n,l+1}_{\raisebox{-1pt}{\tiny{$\Scal \Ccal \Hcal$}}}+\ket{m+1,n,l+1}\! \bra{m,n+1,l}_{\raisebox{-1pt}{\tiny{$\Scal \Ccal \Hcal$}}}\Bigg]+\bar{\openone}.
    \label{eq:tripartite unitary op}
\end{equation}
For the most general case where the Hamiltonians of each subsystem are arbitrary, it is not possible to write down a generic form of the optimal unitary, since the energy-resonant transitions that lead to cooling the target now depend on the microscopic details of the energetic structure. Nonetheless, in Appendix~\ref{app:incoherentcoolingfinitetemperature}, we provide a protocol (i.e., not the unitary \emph{per se}, but a sequence of steps) in this setting that attains perfect cooling and saturates the Carnot-Landauer limit.

Intuitively, the above types of unitaries simply reshuffle populations that are accessible through resonant transitions. For the purpose of cooling, one wishes to do this in such a way that the largest population is placed in the lowest energy eigenstate of the target system, the second largest in the second lowest energy eigenstate, and so on (in line with the optimal unitaries in the coherent-control setting); indeed, on the energy-degenerate subspaces accessible, such unitaries act precisely in this way. It is straightforward to show that interactions of this form satisfy Eq.~\eqref{eq:tripartite commutator UH0}.

For the sake of simplicity, we now focus on the case where all systems are qubits, although the results generalise to the qudit setting. Consider the initial joint state $\varrho_{\raisebox{-1pt}{\tiny{$\Scal \Ccal \Hcal$}}}=\sum_{m,n,l=0}^1\, p_{mnl}\ket{m,n,l}\!\bra{m,n,l}_{\raisebox{-1pt}{\tiny{$\Scal \Ccal \Hcal$}}}$. By applying a unitary $U_{\raisebox{-1pt}{\tiny{EC}}}$ of the form given in Eq.~\eqref{eq:tripartite unitary op}, the post-transformation joint state is
\begin{equation}
    \varrho_{\raisebox{-1pt}{\tiny{$\Scal \Ccal \Hcal$}}}^\prime=U_{\raisebox{-1pt}{\tiny{EC}}}\varrho_{\raisebox{-1pt}{\tiny{$\Scal \Ccal \Hcal$}}} U_{\raisebox{-1pt}{\tiny{EC}}}^{\dagger}=\varrho_{\raisebox{-1pt}{\tiny{$\Scal \Ccal \Hcal$}}}+ \Delta p \, \ket{0,1,0}\!\bra{0,1,0}_{\raisebox{-1pt}{\tiny{$\Scal \Ccal \Hcal$}}} -\Delta p\, \ket{1,0,1}\!\bra{1,0,1}_{\raisebox{-1pt}{\tiny{$\Scal \Ccal \Hcal$}}},
\end{equation}
where $\Delta p := p_{101}-p_{010}$ indicates the amount of population that has been transferred from the excited state of the target system to the ground state throughout the interaction. Naturally, in order to cool the target system, $\Delta p \geq 0$, i.e., the initial population $p_{101}$ must be at least as large as $p_{010}$. 

Due to the energy-conserving nature of the global interaction, the energy exchanged between the subsystems throughout a single such interaction, $\Delta E_{\raisebox{-1pt}{\tiny{$\Xcal$}}}= \tr{H_{\raisebox{-1pt}{\tiny{$\Xcal$}}} (\varrho_{\raisebox{-1pt}{\tiny{$\Xcal$}}}^{\prime}-\varrho_{\raisebox{-1pt}{\tiny{$\Xcal$}}})}$, can be calculated via
\begin{align}
    \Delta E_{\raisebox{-1pt}{\tiny{$\Scal$}}}= -\omega_{\raisebox{0pt}{\tiny{$\Scal$}}} \Delta p, ~~~~~~ \Delta E_{\raisebox{-1pt}{\tiny{$\Ccal$}}}= \omega_{\raisebox{0pt}{\tiny{$\Ccal$}}} \Delta p, ~~~~~\Delta E_{\raisebox{-1pt}{\tiny{$\Hcal$}}}= -\omega_{\raisebox{0pt}{\tiny{$\Hcal$}}} \Delta p.
    \label{eq:energy exchange INc}
\end{align}
Thus, for a fixed energy-level structure of all subsystems (i.e., given the local Hamiltonians), one requires only knowledge of the pre- and post-transformation state of any one of the subsystems to calculate the energy change for all of them.

\subsection{Diverging Energy: Proof of Theorem~\ref{thm:infenergyauto}}
\label{app:divergingenergyincoherentnogotheorem}

The first thing to note is that in the incoherent-control paradigm, even when one allows for the energy cost, i.e., the heat drawn from the hot bath, to be diverging, it is not possible to perfectly cool the target system, as presented in Theorem~\ref{thm:infenergyauto}. The intuition behind this result is that the target system can interact only with energy-degenerate \emph{subspaces} of the hot and cold machine subsystems. The optimal transformation that one can do here to achieve cooling is to transfer the highest populations of any such subspace to the lowest energy eigenstate of the target system; however, any such subspace has population strictly less than one for any $0 \leq \beta_{\raisebox{-1pt}{\tiny{$H$}}} \leq \beta < \infty$ independently of the energy structure. Moreover, the difference from one can be bounded by a finite amount that does not vanish independent of the energy-level structure of any machine of finite dimension. This makes it impossible to attain a subspace population of one even as the energy cost diverges for any fixed and finite control complexity. It follows that the ground-state population of the target system can never reach unity in a single operation of finite control complexity and hence perfect cooling cannot be achieved. 

Precisely, we show the following. Let $\Scal$ be a finite-dimensional system of dimension $d_{\raisebox{-1pt}{\tiny{$\Scal$}}}$ with associated Hamiltonian with finite but otherwise arbitrary energy gaps $H_{\raisebox{-1pt}{\tiny{$\Scal$}}} = \sum_{i=0}^{d_{\raisebox{-1pt}{\tiny{$\Scal$}}} - 1} \omega_{\raisebox{0pt}{\tiny{$\Scal$}}}^i \ket{i}\!\bra{i}_{\raisebox{-1pt}{\tiny{$\Scal$}}}$, and let $d_{\raisebox{-1pt}{\tiny{$\Ccal$}}}$ and $d_{\raisebox{-1pt}{\tiny{$\Hcal$}}}$ be integers denoting the dimensions of the cold and hot parts of the machine respectively. Then it is impossible to cool the system $\Scal$ in the incoherent-control paradigm, i.e., using energy-conserving unitaries involving $\Ccal$ and $\Hcal$ at some initial inverse temperatures $\beta, \beta_{\raisebox{-1pt}{\tiny{$H$}}}$ respectively, arbitrarily close to the ground state. Note that, in particular, this result holds irrespective of the energy-level structure of $\Ccal$ and $\Hcal$ and no matter how much energy is drawn from the hot bath as a resource. 

In order to set notation for the following, we assume $\omega_{\raisebox{0pt}{\tiny{$\Xcal$}}}^i \geq \omega_{\raisebox{0pt}{\tiny{$\Xcal$}}}^j$ for $i \geq j$ and $\omega_{\raisebox{0pt}{\tiny{$\Xcal$}}}^0=0$, where $\omega_{\raisebox{0pt}{\tiny{$\Xcal$}}}^i$ denotes the $i^{\text{th}}$ energy eigenvalue of system $\Xcal$ with $\Xcal \in \{ \Scal, \Ccal, \Hcal\}$. We also assume the initial states on $\Scal$ and $\Ccal$ to be thermal at inverse temperature $\beta$, and $\Hcal$ is assumed to be initially in a thermal state at inverse temperature $\beta_{\raisebox{-1pt}{\tiny{$H$}}} \leq \beta$. We denote by $p_{\raisebox{-1pt}{\tiny{$\Xcal$}}}^i$ the $i^{\text{th}}$ population of system $\Xcal$ in a given state, i.e., $p_{\raisebox{-1pt}{\tiny{$\Xcal$}}}^i= \bra{i} \varrho_{\raisebox{-1pt}{\tiny{$\Xcal$}}} \ket{i}$, where $\ket{i}$ denotes the $i^{\text{th}}$ energy eigenstate of $\varrho_{\raisebox{-1pt}{\tiny{$\Xcal$}}}$. We also write $p_{ijk} := p_{\raisebox{-1pt}{\tiny{$\Scal$}}}^i p_{\raisebox{-1pt}{\tiny{$\Ccal$}}}^j p_{\raisebox{-1pt}{\tiny{$\Hcal$}}}^k$.

The intuition behind the proof is as follows. The global ground-state level of the joint hot-and-cold machine has some nonzero initial population for any finite-dimensional machine; in particular it can always be lower bounded by $\tfrac{1}{d_\Ccal d_\Hcal}$ for any Hamiltonians and initial temperatures, which is strictly greater than zero as long as the dimensions remain finite. Fixing the control complexity of any protocol considered here to be finite in value thus implies a lower bound on the initial ground-state population of the total machine that is larger than zero by a finite amount. Depending on the energy-level structure of the hot and cold parts of the machine, there may be other nonzero initial populations, but in order to cool the target system $\Scal$ perfectly, at least all of the previously mentioned populations must be transferred into spaces spanned by energy eigenstates of the form $\ket{0 j k}_{{\raisebox{-1pt}{\tiny{$\Scal \Ccal \Hcal$}}}}$. This intuition is formalised via Lemma~\ref{lemma:necdeg}, where we show that independent of the energy structure of $\Ccal$ and $\Hcal$, one must be able to make such transfers of population in order to perfectly cool $\Scal$. However, in order to make such transfers in an energy-conserving manner, all energy eigenstates of the form $\ket{i00}_{\raisebox{-1pt}{\tiny{$\Scal \Ccal \Hcal$}}}$ must be degenerate with some of the form $\ket{0jk}_{\raisebox{-1pt}{\tiny{$\Scal \Ccal \Hcal$}}}$. This degeneracy condition, in turn, also implies that every energy eigenstate of the form $\ket{0 j k}_{{\raisebox{-1pt}{\tiny{$\Scal \Ccal \Hcal$}}}}$ has an associated initial population $p_{0jk}$ that is nonvanishing for all machines of finite dimension (i.e., for all protocols with finite control complexity). Thus, upon transferring some population $p_{i00}$ \emph{into} the subspace spanned by $\ket{0jk}_{\raisebox{-1pt}{\tiny{$\Scal \Ccal \Hcal$}}}$, i.e., one of a relevant form for the population to contribute to the final ground-state population of the target, one inevitably transfers some finite amount of population \emph{away} from the relevant space and into $\ket{i00}_{\raisebox{-1pt}{\tiny{$\Scal \Ccal \Hcal$}}}$, which does not contribute to the final ground-state population of the target. We formalise this intuition in the discussion following Lemma~\ref{lemma:necdeg}. In this way, no matter what one does, there is always a finite amount of population, which is lower bounded by some strictly positive number due to the constraint on control complexity, that does not contribute to the final ground-state population of the target, implying that perfect cooling is not possible.

The formal proof occurs in two steps. We first show that some specific degeneracies in the joint $\Scal \Ccal \Hcal$ system must be present in order to be able to even potentially cool $\Scal$ arbitrarily close to the ground state. We then prove that, given the above degeneracies, one cannot cool the system $\Scal$ beyond a fixed ground-state population that is independent of the energy structure of $\Ccal$ and $\Hcal$; in particular, one can draw as much energy from the hot bath as they like and still do no better. We begin with the following lemma.

\begin{lem}\label{lemma:necdeg}
Given $\Scal$, $d_\Ccal$, and $d_\Hcal$ as above, one can reach a final ground-state population of the system $\Scal$ arbitrarily close to one in the incoherent-control setting only if each $\ket{i00}{\raisebox{-1pt}{\tiny{$\Scal \Ccal \Hcal$}}}$, where $i \in \{1, \dots, d_{\raisebox{-1pt}{\tiny{$\Scal$}}}-1\}$, energy eigenstate is degenerate with at least one $\ket{0jk}{\raisebox{-1pt}{\tiny{$\Scal \Ccal \Hcal$}}}$ energy eigenstate, where $j \in \{0, \dots d_{\raisebox{-1pt}{\tiny{$\Ccal$}}}-1\}, k \in \{0, \dots d_{\raisebox{-1pt}{\tiny{$\Hcal$}}}-1\}$.
\end{lem}

\begin{proof}
Suppose that there exists an $i^* \in \{1, \dots, d_{\raisebox{-1pt}{\tiny{$\Scal$}}}-1\}$  such that $\ket{i^*00}_{\raisebox{-1pt}{\tiny{$\Scal \Ccal \Hcal$}}}$ is not degenerate with any $\ket{0jk}_{\raisebox{-1pt}{\tiny{$\Scal \Ccal \Hcal$}}}$, where $j \in \{0, \dots d_{\raisebox{-1pt}{\tiny{$\Ccal$}}}-1\}, k \in \{0, \dots d_{\raisebox{-1pt}{\tiny{$\Hcal$}}}-1\}$. We show that, then, one cannot cool $\Scal$ arbitrarily close to zero. 
	
Let $B_{i}$ denote the degenerate subspace of the total Hamiltonian $H_{\raisebox{-1pt}{\tiny{$\Scal$}}}+H_{\raisebox{-1pt}{\tiny{$\Ccal$}}}+H_{\raisebox{-1pt}{\tiny{$\Hcal$}}}$, where $H_{\raisebox{-1pt}{\tiny{$\Xcal$}}}$ denotes the Hamiltonian of system $\Xcal \in \{ \Scal, \Ccal, \Hcal\}$, that contains the eigenvector $\ket{i00}_{\raisebox{-1pt}{\tiny{$\Scal \Ccal \Hcal$}}}$. Then, any energy-conserving unitary $U_{\raisebox{-1pt}{\tiny{EC}}}$ used to cool the system in the incoherent-control paradigm must act within such $B_i$ subspaces, i.e., $U_{\raisebox{-1pt}{\tiny{EC}}} = \bigoplus_{i} U_{B_i}$ (this is a direct consequence of $[U_{\raisebox{-1pt}{\tiny{EC}}},H_{\raisebox{-1pt}{\tiny{$\Scal$}}}+H_{\raisebox{-1pt}{\tiny{$\Ccal$}}}+H_{\raisebox{-1pt}{\tiny{$\Hcal$}}}]=0$, see, e.g., Lemma 5 of Ref.~\cite{Clivaz2020Thesis}). This means, in particular, that the initial population of $\ket{i^* 00}_{\raisebox{-1pt}{\tiny{$\Scal \Ccal \Hcal$}}}$ can only be distributed within $B_{i^*}$, and as no eigenvector of the form $\ket{0jk}{\raisebox{-1pt}{\tiny{$\Scal \Ccal \Hcal$}}}$ is contained in $B_{i^*}$ by assumption, it can never contribute to the final ground-state population of $\Scal$, which we denote $\widetilde{p}_{\raisebox{-1pt}{\tiny{$\Scal$}}}^{0}$. So we have
\begin{equation}
	\widetilde{p}_{\raisebox{-1pt}{\tiny{$\Scal$}}}^0 \leq 1- p_{i^*00}.
\end{equation}
Now, as for $\Xcal \in \{ \Ccal, \Hcal \}$, with any $\{\omega_{\raisebox{0pt}{\tiny{$\Xcal$}}}^i\}$ such that each $\omega_{\raisebox{0pt}{\tiny{$\Xcal$}}}^i \geq 0$ with $\omega_{\raisebox{0pt}{\tiny{$\Xcal$}}}^0 = 0$ and any inverse temperature $\beta \geq 0$, we have for the partition function $\mathcal{Z}_{\raisebox{-1pt}{\tiny{$\Scal$}}}$ that
\begin{equation}
\mathcal{Z}_{\raisebox{-1pt}{\tiny{$\Xcal$}}} = 1 + e^{-\beta \omega_{\raisebox{0pt}{\tiny{$\Xcal$}}}^1} + \dots + e^{-\beta \omega_{\raisebox{0pt}{\tiny{$\Xcal$}}}^{d_X-1}} \leq d_\Xcal,
\end{equation}
and so we have the following bound on the initial populations associated to each eigenvector $\ket{i00}_{\raisebox{-1pt}{\tiny{$\Scal \Ccal \Hcal$}}}$
\begin{equation}
p_{i00} = \frac{e^{-\beta \omega_{\raisebox{0pt}{\tiny{$\Scal$}}}^i}}{\mathcal{Z}_{\raisebox{-1pt}{\tiny{$\Scal$}}} \mathcal{Z}_{\raisebox{-1pt}{\tiny{$\Ccal$}}} \mathcal{Z}_{\raisebox{-1pt}{\tiny{$\Hcal$}}}} \geq \frac{e^{-\beta \omega_{\raisebox{0pt}{\tiny{$\Scal$}}}^i}}{\mathcal{Z}_{\raisebox{-1pt}{\tiny{$\Scal$}}} d_{\raisebox{-1pt}{\tiny{$\Ccal$}}} d_{\raisebox{-1pt}{\tiny{$\Hcal$}}}} > 0 \quad \quad \forall \, i \in \{ 1, \hdots , d_{\raisebox{-1pt}{\tiny{$\Scal$}}} - 1 \}.
\end{equation}
Combining the above, we have that
\begin{equation}
\widetilde{p}_{\raisebox{-1pt}{\tiny{$\Scal$}}}^0 \leq 1-\frac{e^{-\beta \omega_{\raisebox{0pt}{\tiny{$\Scal$}}}^{i^*}}}{\mathcal{Z}_{\raisebox{-1pt}{\tiny{$\Scal$}}} d_{\raisebox{-1pt}{\tiny{$\Ccal$}}} d_{\raisebox{-1pt}{\tiny{$\Hcal$}}}}<1.
\end{equation}
So as desired, we show that one cannot cool beyond $1-\frac{e^{-\beta \omega_{\raisebox{0pt}{\tiny{$\Scal$}}}^{i^*}}}{\mathcal{Z}_{\raisebox{-1pt}{\tiny{$\Scal$}}} d_{\raisebox{-1pt}{\tiny{$\Ccal$}}} d_{\raisebox{-1pt}{\tiny{$\Hcal$}}}}$, a bound strictly smaller than 1 for any finite-dimensional machine (i.e., for any protocol using only finite control complexity) and independent of the energies of $\Ccal$ and $\Hcal$.
\end{proof}

\noindent We can now proceed to the second step of the proof of Theorem~\ref{thm:infenergyauto}. 
\begin{proof}
To this end, consider any $i^* \in \{ 1, \dots, d_{\raisebox{-1pt}{\tiny{$\Scal$}}}-1\}$. If $\ket{i^*00}_{\raisebox{-1pt}{\tiny{$\Scal \Ccal \Hcal$}}}$ is not degenerate with any $\ket{0jk}{\raisebox{-1pt}{\tiny{$\Scal \Ccal \Hcal$}}}$, our assertion is proven by Lemma~\ref{lemma:necdeg}. On the other hand, if there is a $j^* \in \{0,\dots, d_{\raisebox{-1pt}{\tiny{$\Ccal$}}}-1\}$ and a $k^* \in \{0, \dots, d_{\raisebox{-1pt}{\tiny{$\Hcal$}}}-1\}$ for which $\ket{i^*00}_{\raisebox{-1pt}{\tiny{$\Scal \Ccal \Hcal$}}}$ and $\ket{0j^*k^*}_{\raisebox{-1pt}{\tiny{$\Scal \Ccal \Hcal$}}}$ are degenerate, then $B_{i^*}$, the degenerate subspace containing $\ket{i^*00}_{\raisebox{-1pt}{\tiny{$\Scal \Ccal \Hcal$}}}$, also contains $\ket{0j^* k^*}$. Now $B_{i^*}$ may also contain other eigenvectors of the form $\ket{0jk}_{\raisebox{-1pt}{\tiny{$\Scal \Ccal \Hcal$}}}$, i.e., some other $\ket{0j'k'}_{\raisebox{-1pt}{\tiny{$\Scal \Ccal \Hcal$}}}$ with $j' \in \{0,\dots,d_{\raisebox{-1pt}{\tiny{$\Ccal$}}}-1\}, k' \in \{0, \dots, d_{\raisebox{-1pt}{\tiny{$\Hcal$}}}-1\}$. Crucially, each such eigenvector in $B_{i^*}$ must have an associated minimal amount of initial population as long as the machine is finite dimensional. Indeed, for any such $\ket{0 j^* k^*}{\raisebox{-1pt}{\tiny{$\Scal \Ccal \Hcal$}}}$ in $B_{i^*}$, we have the condition $\omega_{\raisebox{0pt}{\tiny{$\Ccal$}}}^{j^*} + \omega_{\raisebox{0pt}{\tiny{$\Hcal$}}}^{k^*}=\omega_{\raisebox{0pt}{\tiny{$\Scal$}}}^{i^*}$ and so $\omega_{\raisebox{0pt}{\tiny{$\Ccal$}}}^{j^*} \leq \omega_{\raisebox{0pt}{\tiny{$\Scal$}}}^{i^*}$, $\omega_{\raisebox{0pt}{\tiny{$\Hcal$}}}^{k^*} \leq \omega_{\raisebox{0pt}{\tiny{$\Scal$}}}^{i^*}$, implying that $ \beta \omega_{\raisebox{0pt}{\tiny{$\Ccal$}}}^{j^*} \leq \beta \omega_{\raisebox{0pt}{\tiny{$\Scal$}}}^{i^*}$ and $\beta_H \omega_{\raisebox{0pt}{\tiny{$\Hcal$}}}^{k^*} \leq \beta \omega_{\raisebox{0pt}{\tiny{$\Scal$}}}^{i^*}$. Thus we have the bound 
\begin{equation}\label{eq:minimalpopbound}
p_{0j^*k^*} = \frac{e^{-\beta \omega_{\raisebox{0pt}{\tiny{$\Ccal$}}}^{j^*}} e^{-\beta_H \omega_{\raisebox{0pt}{\tiny{$\Hcal$}}}^{k^*}}}{\Zcal_{\raisebox{-1pt}{\tiny{$\Scal$}}} \Zcal_{\raisebox{-1pt}{\tiny{$\Ccal$}}} \Zcal_{\raisebox{-1pt}{\tiny{$\Hcal$}}}} \geq \frac{e^{-2 \beta \omega_{\raisebox{0pt}{\tiny{$\Scal$}}}^{i^*}}}{\Zcal_{\raisebox{-1pt}{\tiny{$\Scal$}}} \Zcal_{\raisebox{-1pt}{\tiny{$\Ccal$}}} \Zcal_{\raisebox{-1pt}{\tiny{$\Hcal$}}}} \geq \frac{e^{-2 \beta \omega_{\raisebox{0pt}{\tiny{$\Scal$}}}^{i^*}}}{\Zcal_{\raisebox{-1pt}{\tiny{$\Scal$}}} d_{\raisebox{-1pt}{\tiny{$\Ccal$}}} d_{\raisebox{-1pt}{\tiny{$\Hcal$}}}}=: q_{i^*}.
\end{equation}
Now, take any particular $i^* \in \{ 1, \hdots, d_{\raisebox{-1pt}{\tiny{$\Scal$}}} -1\}$ and let $\pi_{i^*}$ be the dimension of $B_{i^*}$, $\mu$ the number of energy eigenstates of the form $\ket{0jk}_{\raisebox{-1pt}{\tiny{$\Scal \Ccal \Hcal$}}}$ that $B_{i^*}$ contains and $\nu=\pi-\mu$ the number of energy eigenstates of the form $\ket{ijk}{\raisebox{-1pt}{\tiny{$\Scal \Ccal \Hcal$}}}$, where $i \neq 0$, that $B_{i^*}$ contains. So
\begin{equation}
B_{i^*} = \text{span} \{ \ket{0jk}, \ket{0j_2k_2}, \dots, \ket{0j_\mu k_\mu}, \ket{i^*00}, \ket{i_2 \ell_2 m_2}, \dots, \ket{i_\nu \ell_\nu m_\nu}\}.
\end{equation}
Let $\boldsymbol{v}=\{p_{0jk},p_{0j_2l_2}, \dots, p_{0j_\mu k_\mu},p_{i^*00},p_{i_2 \ell_2 m_2}, \dots, p_{i_\nu \ell_\nu m_\nu}\}$ be the vector of initial populations associated to the eigenvectors of $B_{i^*}$, and $\boldsymbol{v}^{\uparrow}$ be the vector whose components are those of $\boldsymbol{v}$ arranged in nondecreasing order. Using Schur's theorem \cite{2011Marshall}, we know that after applying any unitary transformation $U_{B_{i^*}}$ on the relevant energy-degenerate subspace, then the vector of transformed populations, $\boldsymbol{\widetilde{v}}$, is majorised by $\boldsymbol{v}$. In particular, labelling the vector elements by $\boldsymbol{v}_\alpha$, we have 
\begin{equation}\label{eqs:nocontributing}
\widetilde{p}_{i^* 00} + \sum_{\alpha=2}^\nu \widetilde{p}_{i_\alpha \ell_\alpha m_\alpha} \geq \sum_{\alpha=1}^\nu \boldsymbol{v}_{\alpha}^{\uparrow}.
\end{equation}
We now claim that $\sum_{\alpha=1}^\nu \boldsymbol{v}_{\alpha}^{\uparrow} \geq q_{i^*}$ from Eq.~\eqref{eq:minimalpopbound}. Indeed, as $\boldsymbol{v}$ has at most $\nu-1$ elements that do not belong to the set $A:=\{p_{0jk},p_{0j_2k_2}, \dots, p_{0j_\mu k_\mu},p_{i^*00}\}$, at least one element of $A$ must contribute to the sum $\sum_{\alpha=1}^\nu\boldsymbol{v}^{\uparrow}_\alpha$. Let $x$ be that element. As $\boldsymbol{v}^{\uparrow}_\alpha \geq 0$ for all $\alpha=1,\dots, \pi = \mu + \nu$, we have 
\begin{equation}
\sum_{\alpha=1}^{\nu}\boldsymbol{v}^{\uparrow}_\alpha \geq x.
\end{equation}
Now as $p_{0j_\gamma k_\gamma} \geq q_{i^*}$ for all $\gamma=2, \dots, \mu$, we have
\begin{equation}
x \geq \min (q_{i^*}, p_{i^*00})= q_{i^*},
\end{equation} 
where $p_{i^*00} \geq q_{i^*}$ can be seen from Eq.~\eqref{eq:minimalpopbound}, as claimed. 

As the l.h.s. of Eq.~\eqref{eqs:nocontributing} represents the amount of population in the subspace $B_{i^*}$ that does \emph{not} contribute to the final ground-state population of the target system, we have 
\begin{align}
\widetilde{p}_{\raisebox{-1pt}{\tiny{$\Scal$}}}^0 &\leq 1- \left( \widetilde{p}_{i^* 00} + \sum_{\alpha=2}^\nu \widetilde{p}_{i_\alpha \ell_\alpha m_\alpha} \right)
 \leq 1- q_{i^*} = 1- \frac{e^{-2 \beta \omega_{\raisebox{0pt}{\tiny{$\Scal$}}}^{i^*}}}{\Zcal_{\raisebox{-1pt}{\tiny{$\Scal$}}} d_{\raisebox{-1pt}{\tiny{$\Ccal$}}} d_{\raisebox{-1pt}{\tiny{$\Hcal$}}}}.
\end{align} 
So, for any finite-dimensional machine, one cannot cool the system $\Scal$ beyond $1-\frac{e^{-\beta \omega_{\raisebox{0pt}{\tiny{$\Scal$}}}^{i^*}}}{\Zcal_{\raisebox{-1pt}{\tiny{$\Scal$}}} d_{\raisebox{-1pt}{\tiny{$\Ccal$}}} d_{\raisebox{-1pt}{\tiny{$\Hcal$}}}}$, a bound strictly smaller than 1 and independent of the energy structure of $\Ccal$ and $\Hcal$, as desired.
\end{proof}

As a concrete example, consider the case where all systems are qubits. The initial joint state is
\begin{align}
    \varrho_{\raisebox{-1pt}{\tiny{$\Scal \Ccal \Hcal$}}}^{(0)} = \frac{(\ketbra{0}{0} + e^{-\beta \omega_{\raisebox{0pt}{\tiny{$\Scal$}}}} \ketbra{1}{1})_{\raisebox{-1pt}{\tiny{$\Scal$}}} \otimes (\ketbra{0}{0} + e^{-\beta \omega_{\raisebox{0pt}{\tiny{$\Ccal$}}}} \ketbra{1}{1})_{\raisebox{-1pt}{\tiny{$\Ccal$}}} \otimes (\ketbra{0}{0} + e^{-\beta_{\raisebox{-1pt}{\tiny{$H$}}} \omega_{\raisebox{0pt}{\tiny{$\Hcal$}}}} \ketbra{1}{1})_{\raisebox{-1pt}{\tiny{$\Hcal$}}}}{\mathcal{Z}_{\raisebox{-1pt}{\tiny{$\Scal$}}}(\beta, \omega_{\raisebox{0pt}{\tiny{$\Scal$}}}) \mathcal{Z}_{\raisebox{-1pt}{\tiny{$\Ccal$}}}(\beta, \omega_{\raisebox{0pt}{\tiny{$\Ccal$}}}) \mathcal{Z}_{\raisebox{-1pt}{\tiny{$\Hcal$}}}(\beta_{\raisebox{-1pt}{\tiny{$H$}}}, \omega_{\raisebox{0pt}{\tiny{$\Hcal$}}})}.
\end{align}
The only energy-conserving unitary interaction that is relevant for cooling is the one that exchanges the populations in the levels spanned by $\ket{010}$ and $\ket{101}$, which have initial populations $\tfrac{e^{-\beta \omega_{\raisebox{0pt}{\tiny{$\Ccal$}}}}}{\mathcal{Z}_{\raisebox{-1pt}{\tiny{$\Scal$}}}(\beta, \omega_{\raisebox{0pt}{\tiny{$\Scal$}}}) \mathcal{Z}_{\raisebox{-1pt}{\tiny{$\Ccal$}}}(\beta, \omega_{\raisebox{0pt}{\tiny{$\Ccal$}}}) \mathcal{Z}_{\raisebox{-1pt}{\tiny{$\Hcal$}}}(\beta_{\raisebox{-1pt}{\tiny{$H$}}}, \omega_{\raisebox{0pt}{\tiny{$\Hcal$}}})}$ and $\tfrac{e^{-\beta \omega_{\raisebox{0pt}{\tiny{$\Scal$}}}} e^{-\beta_{\raisebox{-1pt}{\tiny{$H$}}} \omega_{\raisebox{0pt}{\tiny{$\Hcal$}}}}}{\mathcal{Z}_{\raisebox{-1pt}{\tiny{$\Scal$}}}(\beta, \omega_{\raisebox{0pt}{\tiny{$\Scal$}}}) \mathcal{Z}_{\raisebox{-1pt}{\tiny{$\Ccal$}}}(\beta, \omega_{\raisebox{0pt}{\tiny{$\Ccal$}}}) \mathcal{Z}_{\raisebox{-1pt}{\tiny{$\Hcal$}}}(\beta_{\raisebox{-1pt}{\tiny{$H$}}}, \omega_{\raisebox{0pt}{\tiny{$\Hcal$}}})}$ respectively, which are both strictly less than one. The necessary condition for any cooling to be possible implies that $e^{-\beta \omega_{\raisebox{0pt}{\tiny{$\Scal$}}}} e^{-\beta_{\raisebox{-1pt}{\tiny{$H$}}} \omega_{\raisebox{0pt}{\tiny{$\Hcal$}}}} \geq e^{-\beta \omega_{\raisebox{0pt}{\tiny{$\Ccal$}}}}$; now, performing the optimal cooling unitary leads to the final ground-state population of the target system
\begin{align}
    p_{\raisebox{-1pt}{\tiny{$\Scal$}}}^{\prime}(0) = \bra{0} \,\ptr{\Ccal \Hcal}{U \varrho_{\raisebox{-1pt}{\tiny{$\Scal \Ccal \Hcal$}}}^{(0)} U^\dagger}\!\ket{0}_{\raisebox{-1pt}{\tiny{$\Scal$}}} = \frac{1 + e^{-\beta_{\raisebox{-1pt}{\tiny{$H$}}} \omega_{\raisebox{0pt}{\tiny{$\Hcal$}}}}(1 + e^{-\beta \omega_{\raisebox{0pt}{\tiny{$\Scal$}}}} + e^{-\beta \omega_{\raisebox{0pt}{\tiny{$\Ccal$}}}})}{\mathcal{Z}_{\raisebox{-1pt}{\tiny{$\Scal$}}}(\beta, \omega_{\raisebox{0pt}{\tiny{$\Scal$}}}) \mathcal{Z}_{\raisebox{-1pt}{\tiny{$\Ccal$}}}(\beta, \omega_{\raisebox{0pt}{\tiny{$\Ccal$}}}) \mathcal{Z}_{\raisebox{-1pt}{\tiny{$\Hcal$}}}(\beta_{\raisebox{-1pt}{\tiny{$H$}}}, \omega_{\raisebox{0pt}{\tiny{$\Hcal$}}})} < 1.
\end{align}
Indeed, using $e^{-\beta \omega_{\raisebox{0pt}{\tiny{$\Scal$}}}} e^{-\beta_{\raisebox{-1pt}{\tiny{$H$}}} \omega_{\raisebox{0pt}{\tiny{$\Hcal$}}}} \geq e^{-\beta \omega_{\raisebox{0pt}{\tiny{$\Ccal$}}}}$,
\begin{equation}
    p_{\raisebox{-1pt}{\tiny{$\Scal$}}}^{\prime}(0) \leq \frac{1 + e^{-\beta_{\raisebox{-1pt}{\tiny{$H$}}} \omega_{\raisebox{0pt}{\tiny{$\Hcal$}}}}e^{-\beta \omega_{\raisebox{0pt}{\tiny{$\Scal$}}}} }{\mathcal{Z}_{\raisebox{-1pt}{\tiny{$\Scal$}}}(\beta, \omega_{\raisebox{0pt}{\tiny{$\Scal$}}}) \mathcal{Z}_{\raisebox{-1pt}{\tiny{$\Ccal$}}}(\beta, \omega_{\raisebox{0pt}{\tiny{$\Ccal$}}}) } \leq \frac{1}{\mathcal{Z}_{\raisebox{-1pt}{\tiny{$\Ccal$}}}(\beta, \omega_{\raisebox{0pt}{\tiny{$\Ccal$}}})} \leq 1.
\end{equation}
The second inequality is strict unless $\beta_{\mathcal H} = 0$ or $\omega_{\raisebox{0pt}{\tiny{$\Hcal$}}} = 0$. In the both cases, for equality in the first inequality, we need $\beta \omega_{\raisebox{0pt}{\tiny{$\Scal$}}} = \beta \omega_{\raisebox{0pt}{\tiny{$\Ccal$}}}$. If $\beta = 0$, then $\mathcal{Z}_{\raisebox{-1pt}{\tiny{$\Ccal$}}}(\beta, \omega_{\raisebox{0pt}{\tiny{$\Ccal$}}}) = 2$ and the last inequality is strict. If $\omega_{\raisebox{0pt}{\tiny{$\Scal$}}} = \omega_{\raisebox{0pt}{\tiny{$\Ccal$}}}$, no cooling is possible; hence $p_{\raisebox{-1pt}{\tiny{$\Scal$}}}^{\prime}(0) = p_{\raisebox{-1pt}{\tiny{$\Scal$}}}(0) < 1$. 

\subsection{Diverging Time and Diverging Control Complexity}
\label{app:incoherentdivergingtimecontrolcomplexity}
    
We now move to analyse the case where diverging time is allowed, where we wish to minimise the energy cost and control complexity throughout the protocol over a diverging number of energy-conserving interactions between the target system and the hot and cold subsystems of the machine. We again consider all three systems to be qubits, but the results generalise to arbitrary (finite) dimensions. Here, the machines and ancillas begin as thermal states with initial inverse temperatures $\beta$ and $\beta_{\raisebox{-1pt}{\tiny{$H$}}} \leq \beta$ respectively. Just as in the diverging time cooling protocol in the coherent-control setting presented in Appendix~\ref{app:divergingtimecoolingprotocolfinitedimensionalsystems}, we consider a diverging number of machines, with slightly increasing energy gaps, in a configuration such that the target system interacts with the $n^{\textup{th}}$ machine at time step $n$. Suppose that after $n$ steps of the protocol, the target qubit has been cooled to some inverse temperature $\beta_n > \beta$; equivalently, this can be expressed as a thermal state with corresponding energy gap $\omega_n=\frac{\beta_n}{\beta}\omega_{\raisebox{0pt}{\tiny{$\Scal$}}}$. We now wish to interact the target system $\tau_{\raisebox{-1pt}{\tiny{$\Scal$}}}(\beta_n, \omega_{\raisebox{0pt}{\tiny{$\Scal$}}})$ with a machine $\Mcal_{n+1}$ with slightly increased energy gaps with respect to the most recent one $\Mcal_{n}$, i.e., we increase the energy gaps of the cold subsystem $\Ccal$ from $\omega_n$ to $\omega_{n+1} = \omega_n + \epsilon_n$; the resonance condition enforces the energy gap of the hot subsystem $\Hcal$ to be similarly increased to $\omega_n + \epsilon_n - \omega_{\raisebox{0pt}{\tiny{$\Scal$}}}$. Thus, the next step of the protocol is a unitary acting on the global state
\begin{align}
    \varrho_{\raisebox{-1pt}{\tiny{$\Scal \Ccal \Hcal$}}}^{(n)}=\tau_{\raisebox{-1pt}{\tiny{$\Scal$}}}(\beta_n, \omega_{\raisebox{0pt}{\tiny{$\Scal$}}})\otimes \tau_{\raisebox{-1pt}{\tiny{$\Ccal$}}}(\beta, \omega_n+\epsilon _n) \otimes \tau_{\raisebox{-1pt}{\tiny{$\Hcal$}}}(\beta_{\raisebox{-1pt}{\tiny{$H$}}},\omega_n+\epsilon _n-\omega_{\raisebox{0pt}{\tiny{$\Scal$}}} ).
    \label{eq: initial 3qubits}
\end{align}
In order to cool the target system via said unitary, we must have that $p_{101} \geq p_{010}$ for the state in Eq.~\eqref{eq: initial 3qubits}, which implies that $\epsilon_n$ must satisfy the following condition:
\begin{align}
   e^{-\beta \omega_n-\beta_{\raisebox{-1pt}{\tiny{$H$}}}( \omega_n+\epsilon _n- \omega_{\raisebox{0pt}{\tiny{$\Scal$}}})}\,\geq\,e^{-\beta(\omega_n+ \epsilon_n)} \Rightarrow \epsilon_n\geq  \gamma (\omega_n-\omega_{\raisebox{0pt}{\tiny{$\Scal$}}}) \quad \quad \mathrm{where} \quad \gamma:=\frac{\beta_{\raisebox{-1pt}{\tiny{$H$}}}}{\beta-\beta_{\raisebox{-1pt}{\tiny{$H$}}}}.\label{eq: epsilon bound}
\end{align}
This condition is crucial. It means that if the hot subsystem $\Hcal$ is coupled to a heat bath at any finite temperature, i.e., $\beta_{\raisebox{-1pt}{\tiny{$H$}}} > 0$, $\epsilon_n$ depends linearly on the inverse temperature of the target system at the previous step $\beta_n$, and can thus not be taken to be arbitrarily small. As we now show, this condition prohibits the ability to perfectly cool the target system at the Landauer limit for the energy cost whenever the heat bath is at finite temperature.

On the other hand, for infinite-temperature heat baths, perfect cooling at the Landauer limit is seemingly achievable; here, $\beta_{\raisebox{-1pt}{\tiny{$H$}}} \to 0$ and so $\gamma \to 0$, leading to the trivial constraint $\epsilon_n \geq 0$ which allows it to be arbitrarily small, as is required. Nonetheless, the explicit construction of any protocol doing so in the incoherent-control setting is \emph{a priori} unclear, as the restriction of energy conservation makes for a fundamentally different setting from the coherent-control paradigm. We now explicitly derive the optimal diverging-time protocol to perfectly cool at the Landauer limit for an infinite-temperature heat bath, thereby proving Theorem~\ref{thm:autoinftempinftime}.

\subsection{Saturating the Landauer Limit with an Infinite-Temperature Heat Bath}
\label{app:incoherentinfiniteheatbath}

Before calculating the energy cost, we briefly discuss the attainability of the optimally cool target state. We begin with all subsystems as qubits, for the sake of simplicity, but the logic generalises to higher dimensions. In the incoherent paradigm, the target system $\Scal$ interacts with a virtual qubit of the total machine $\Mcal = \Ccal\Hcal$ that consists of the energy eigenstates $\ket{0,1}_{\Ccal\Hcal}$ and $ \ket{1,0}_{\Ccal\Hcal}$, with populations $p_{0_{\raisebox{-1pt}{\tiny{$\Ccal$}}} 1_{\raisebox{-1pt}{\tiny{$\Hcal$}}}}$ and $p_{1_{\raisebox{-1pt}{\tiny{$\Ccal$}}} 0_{\raisebox{-1pt}{\tiny{$\Hcal$}}}}$ respectively. Suppose that at step $n+1$ the cold subsystem involved in the interaction has energy gap $\omega_n + \epsilon_n$. In Ref.~\cite{Clivaz_2019E}, it is shown that by repeating the incoherent cooling process (i.e., implementing the unitary in Eq.~\eqref{eq:tripartite unitary op}) and taking the limit of infinite cycles, this scenario equivalently corresponds to the general (coherent) setting where arbitrary unitaries are permitted and the target system interacts with a virtual qubit machine with effective energy gap $\omega^{\textup{eff}}_n$ given by
\begin{align}
    e^{-\beta \omega^{\textup{eff}}_n }:=\frac{p_{1_{\raisebox{-1pt}{\tiny{$\Ccal$}}} 0_{\raisebox{-1pt}{\tiny{$\Hcal$}}}}}{p_{0_{\raisebox{-1pt}{\tiny{$\Ccal$}}} 1_{\raisebox{-1pt}{\tiny{$\Hcal$}}}}}= e^{-\beta(\omega_n+ \epsilon_n)}\, e^{\beta_{\raisebox{-1pt}{\tiny{$H$}}}( \omega_n+\epsilon_n- \omega_{\raisebox{0pt}{\tiny{$\Scal$}}})} \quad \quad \Rightarrow \omega^{\textup{eff}}_n=\omega_n +\epsilon_n-\frac{\beta_{\raisebox{-1pt}{\tiny{$H$}}}}{\beta}(\omega_n+\epsilon _n- \omega_{\raisebox{0pt}{\tiny{$\Scal$}}}).
    \label{eq: eff gap}
\end{align}
It is clear that for finite-temperature heat baths, i.e., $\beta_{\raisebox{-1pt}{\tiny{$H$}}}> 0$, the effective energy gap $\omega^{\textup{eff}}_n$ is always smaller than the energy gap of the machine at any given step, i.e., $\omega^{\textup{eff}}_n\leq  \omega_n+\epsilon _n$; on the other hand, equality holds iff the heat bath is at infinite temperature, i.e., $\beta_{\raisebox{-1pt}{\tiny{$H$}}} \to 0$. Thus, in the infinite-temperature case, given a target system beginning at some step of the protocol in the state $\varrho_{\raisebox{-1pt}{\tiny{$\Scal$}}}^*(\beta, \omega_n)$, it is possible to get close to the asymptotic state $\varrho_{\raisebox{-1pt}{\tiny{$\Scal$}}}^*(\beta,  \omega_n+\epsilon_n)$; if the temperature is finite, however, this state is not attainable (even asymptotically). Following the arguments in Appendix~\ref{app:divergingtimecoolingprotocolfinitedimensionalsystems}, i.e., considering a diverging number of machines, each of which having the part connected to the cold bath with energy gap $\omega_{\Ccal_n} = \omega_n + \epsilon_n$ and taking the limit of $\epsilon_n \to 0$, which one can \emph{only} do if the hot-bath temperature is infinite, allows one to cool perfectly in diverging time in the incoherent paradigm at the Landauer limit. 

We now calculate the energy cost explicitly for the infinite-temperature heat bath case, precisely demonstrating attainability of the Landauer limit. We use a similar approach to that described in Appendix~\ref{app:divergingtimecoolingprotocolfinitedimensionalsystems}: we have a diverging number of cold machines for each energy gap $\omega_n$, with which the target system at the $n-1^{\textup{th}}$ time step interacts; for an infinite-temperature heat bath, i.e., $\Hcal$ is in the maximally mixed state independent of its energy structure, the state of the target system at each step $\varrho_{\raisebox{-1pt}{\tiny{$\Scal$}}}^*(\beta, \omega_{n-1})$ is achievable. From Eq.~\eqref{eq:energy exchange INc}, the energy change between all subsystems for a given step of the protocol, i.e., taking $\varrho_{\raisebox{-1pt}{\tiny{$\Scal$}}}^*(\beta, \omega_{n-1}) \to \varrho_{\raisebox{-1pt}{\tiny{$\Scal$}}}^*(\beta, \omega_{n})$, can be calculated as
\begin{align}
 \Delta E_{\raisebox{-1pt}{\tiny{$\Scal$}}}^{(n)}&=\tr{{H}_{\raisebox{-1pt}{\tiny{$\Scal$}}}(\omega_{\raisebox{0pt}{\tiny{$\Scal$}}})(\varrho_{\raisebox{-1pt}{\tiny{$\Scal$}}}^* (\beta, \omega_{n})- \varrho_{\raisebox{-1pt}{\tiny{$\Scal$}}}^* (\beta, \omega_{n-1})}\nonumber\\
 \Delta E_{\raisebox{-1pt}{\tiny{$\Ccal$}}}^{(n)}&=-\tr{{H}_{\raisebox{-1pt}{\tiny{$\Ccal$}}}(\omega_{n})(\varrho_{\raisebox{-1pt}{\tiny{$\Scal$}}}^* (\beta, \omega_{n})- \varrho_{\raisebox{-1pt}{\tiny{$\Scal$}}}^* (\beta, \omega_{n-1})}\nonumber\\
 \Delta E_{\raisebox{-1pt}{\tiny{$\Hcal$}}}^{(n)}&=\tr{{H}_{\raisebox{-1pt}{\tiny{$\Hcal$}}}(\omega_{n}-\omega_{\raisebox{0pt}{\tiny{$\Scal$}}})(\varrho_{\raisebox{-1pt}{\tiny{$\Scal$}}}^* (\beta, \omega_{n})- \varrho_{\raisebox{-1pt}{\tiny{$\Scal$}}}^* (\beta, \omega_{n-1})}
  \label{eq:energy exch incoherent ind}
\end{align}
In general, i.e., for finite-temperature heat baths, we would have $\omega_n = \omega_{n-1} + \epsilon_{n-1}$, with a lower bound on $\epsilon_{n-1}$ for cooling to be possible [in accordance with Eq.~\eqref{eq: epsilon bound}]. However, for infinite-temperature heat baths, this lower bound trivialises since the energy structure of the hot-machine subsystem plays no role in its state; thus we can choose the energy gap structure for the machines as $\{\omega_n=\omega_{\raisebox{0pt}{\tiny{$\Scal$}}}+n \epsilon\}_{n=1}^N$ with $\epsilon$ arbitrarily small. Taking the limit $\epsilon \to 0$, the diverging time limit $N \to \infty$, and writing $\omega_{\raisebox{-1pt}{\tiny{$N$}}} = \omega_{\textup{max}}$ for the maximum energy gap of the cold-machine subsystems, the energy exchanged throughout the entire cooling protocol here is given by
\begin{align}
  \Delta E_{\raisebox{-1pt}{\tiny{$\Scal$}}}&=\lim_{N\to \infty}\sum_{n=1}^N\Delta E_{\raisebox{-1pt}{\tiny{$\Scal$}}}^{(n)}=\tr{{H}_{\raisebox{-1pt}{\tiny{$\Scal$}}}(\omega_{\raisebox{0pt}{\tiny{$\Scal$}}})(\varrho_{\raisebox{-1pt}{\tiny{$\Scal$}}}^* (\beta, \omega_{\textup{max}})- \varrho_{\raisebox{-1pt}{\tiny{$\Scal$}}}^* (\beta, \omega_{\raisebox{0pt}{\tiny{$\Scal$}}})}\nonumber\\
 \Delta E_{\raisebox{-1pt}{\tiny{$\Ccal$}}}&=\lim_{N\to \infty}\sum_{n=1}^N\Delta E_{\raisebox{-1pt}{\tiny{$\Ccal$}}}^{(n)}=\frac{1}{\beta} \left\{ S[\varrho_{\raisebox{-1pt}{\tiny{$\Scal$}}}^*(\beta,\omega_{\raisebox{0pt}{\tiny{$\Scal$}}})] - S[\varrho_{\raisebox{-1pt}{\tiny{$\Scal$}}}^*(\beta,\omega_{\textup{max}})] \right\} = \frac{1}{\beta}\widetilde{\Delta} S_{\raisebox{-1pt}{\tiny{$\Scal$}}} \nonumber\\
 \Delta E_{\raisebox{-1pt}{\tiny{$\Hcal$}}}&=\lim_{N\to \infty}\sum_{n=1}^N\Delta E_{\raisebox{-1pt}{\tiny{$\Hcal$}}}^{(n)}=-\Delta E_{\raisebox{-1pt}{\tiny{$\Scal$}}}-\Delta E_{\raisebox{-1pt}{\tiny{$\Ccal$}}}.
  \label{eq:energy exch incoherent}
\end{align}
Here, the expression for $\Delta E_{\raisebox{-1pt}{\tiny{$\Ccal$}}}$ can be derived using the same arguments as presented in Appendix~\ref{app:prooftheoreminftime}. In particular, the heat dissipated by the cold part of the machine, which is naturally connected to the heat sink in the incoherent setting as an infinite-temperature heat bath can be considered a work source since any energy drawn comes with no entropy change, is in accordance with the Landauer limit. It is straightforward to obtain the same result for qudit systems. Lastly, in a similar way to the other protocols we have presented, one could compress all of the diverging number of operations into a single one whose control complexity diverges, thereby trading off between time and control complexity.

\subsection{Analysis of Finite-Temperature Heat Baths}
\label{app:analysisfinitetempheatbaths}

We now return to the more general consideration of finite-temperature heat baths, i.e., $0 < \beta_{\raisebox{-1pt}{\tiny{$H$}}}\leq\beta$. In the case where $\beta_{\raisebox{-1pt}{\tiny{$H$}}}=\beta$, from Eq.~\eqref{eq: eff gap}, it is straightforward to see that for any machine energy gap $\omega_n$, the effective gap $\omega^{\textup{eff}}_n$ is equal to the gap of the target system, which means that no cooling can be achieved in the incoherent paradigm. Nonetheless, for any $\Hcal$ subsystem coupled to a heat bath of inverse temperature $\beta_{\raisebox{-1pt}{\tiny{$H$}}} < \beta$, cooling is possible. We first provide more detail regarding why cooling at the Landauer limit is not possible in this setting, before deriving the minimal energy cost in accordance with the Carnot-Landauer limit presented in Theorem~\ref{thm:main-landauer-incoherent}; in Appendix~\ref{app:incoherentcoolingfinitetemperature}, we provide explicit protocols that saturate this bound for any finite-temperature heat bath and arbitrary finite-dimensional systems and machines. 

Suppose that at some step $n$ one has the initial joint state of Eq.~\eqref{eq: initial 3qubits}, where $\epsilon_n=  \gamma (\omega_n-\omega_{\raisebox{0pt}{\tiny{$\Scal$}}})+\epsilon$ and $\omega_n=\omega_{\raisebox{0pt}{\tiny{$\Scal$}}}+n\epsilon$. Here, $\gamma$ is as in Eq.~\eqref{eq: epsilon bound}. We now wish to cool the target system to $\varrho_{\raisebox{-1pt}{\tiny{$\Scal$}}}^* (\beta, \omega_{n}+\epsilon)$. For cooling to be possible in the incoherent setting here, we need the cold-machine subsystem to have an energy gap of at least $\omega_n + \epsilon_n$; moreover, with a finite-temperature heat bath, this energy gap is insufficient to achieve the desired transformation [see Eq.~\eqref{eq: epsilon bound}]. Based on Eq.~\eqref{eq:energy exchange INc}, we can see that nonetheless, if we calculate the \emph{hypothetical} energy change in this scenario if it were possible, we can derive a lower bound for the actual energy cost incurred. Employing Eq.~\eqref{eq:energy exch incoherent ind}, we have
\begin{align}
    \Delta E_{\raisebox{-1pt}{\tiny{$\Ccal$}}}^{(n+1)}&\geq -\mathrm{tr}\{{H}_{\raisebox{-1pt}{\tiny{$\Ccal$}}}(\omega_{n}+\epsilon_n)[\varrho_{\raisebox{-1pt}{\tiny{$\Scal$}}}^* (\beta, \omega_{n}+\epsilon)- \varrho_{\raisebox{-1pt}{\tiny{$\Scal$}}}^* (\beta, \omega_n)]\}\nonumber\\
    &=-\mathrm{tr}\{{H}_{\raisebox{-1pt}{\tiny{$\Ccal$}}}[(\gamma+1)\omega_{n}-\gamma \omega_{\raisebox{0pt}{\tiny{$\Scal$}}} +\epsilon][\varrho_{\raisebox{-1pt}{\tiny{$\Scal$}}}^* (\beta, \omega_{n}+\epsilon)- \varrho_{\raisebox{-1pt}{\tiny{$\Scal$}}}^* (\beta, \omega_n)]\}\nonumber\\
     &=-\mathrm{tr}\{{H}_{\raisebox{-1pt}{\tiny{$\Ccal$}}}[(\gamma+1)\omega_{n}-\gamma \omega_{\raisebox{0pt}{\tiny{$\Scal$}}} +\epsilon+\gamma \epsilon-\gamma \epsilon][\varrho_{\raisebox{-1pt}{\tiny{$\Scal$}}}^* (\beta, \omega_{n}+\epsilon)- \varrho_{\raisebox{-1pt}{\tiny{$\Scal$}}}^* (\beta, \omega_n)]\}\nonumber\\
    &=-(\gamma+1)\mathrm{tr}\{{H}_{\raisebox{-1pt}{\tiny{$\Ccal$}}}(\omega_{n} +\epsilon)[\varrho_{\raisebox{-1pt}{\tiny{$\Scal$}}}^* (\beta, \omega_{n}+\epsilon)- \varrho_{\raisebox{-1pt}{\tiny{$\Scal$}}}^* (\beta, \omega_n)]\}+\gamma\mathrm{tr}\{{H}_{\raisebox{-1pt}{\tiny{$\Ccal$}}}(\omega_{\raisebox{0pt}{\tiny{$\Scal$}}}+\epsilon)[\varrho_{\raisebox{-1pt}{\tiny{$\Scal$}}}^* (\beta, \omega_{n}+\epsilon)- \varrho_{\raisebox{-1pt}{\tiny{$\Scal$}}}^* (\beta, \omega_n)]\}\nonumber\\
    &= (\gamma+1)\Delta E_{\raisebox{-1pt}{\tiny{$\Ccal$}}}^{* (n+1)}+\gamma \Delta E_{\raisebox{-1pt}{\tiny{$\Scal$}}}^{* (n+1)}+\gamma\mathrm{tr}\{{H}_{\raisebox{-1pt}{\tiny{$\Ccal$}}}(\epsilon)[\varrho_{\raisebox{-1pt}{\tiny{$\Scal$}}}^* (\beta, \omega_{n}+\epsilon)- \varrho_{\raisebox{-1pt}{\tiny{$\Scal$}}}^* (\beta, \omega_n)]\},
\end{align}
where we make use of the fact that for equally spaced Hamiltonians, the structure of the Hamiltonians on each subsystem take the same form [i.e., we can write, with slight abuse of notation, $H_{\raisebox{-1pt}{\tiny{$\Ccal$}}}(\omega+\omega_{\raisebox{0pt}{\tiny{$\Scal$}}})=H_{\raisebox{-1pt}{\tiny{$\Ccal$}}}(\omega)+H_{\raisebox{-1pt}{\tiny{$\Scal$}}}(\omega_{\raisebox{0pt}{\tiny{$\Scal$}}})$]. We use the star in $ \Delta E^*_{\raisebox{-1pt}{\tiny{$\Acal$}}}$ to denote the idealised energy cost [i.e., that corresponding to what would be achievable in the infinite-temperature setting; see Eq.~\eqref{eq:energy exch incoherent ind}] and the energy costs without the star to represent those for when the temperature of the heat bath is finite. The additional term $\mathrm{tr}\{{H}(\gamma\epsilon)[\varrho_{\raisebox{-1pt}{\tiny{$\Scal$}}}^* (\beta, \omega_{n}+\epsilon)- \varrho_{\raisebox{-1pt}{\tiny{$\Scal$}}}^* (\beta, \omega_n)]\}$ vanishes for $\epsilon \to 0$.

Summing up these contributions for a diverging number of steps gives the lower bound for the heat dissipated throughout the entire protocol for cooling an initial state $\tau_{\raisebox{-1pt}{\tiny{$\Scal$}}}(\beta, \omega_{\raisebox{0pt}{\tiny{$\Scal$}}})$ to some final $\tau_{\raisebox{-1pt}{\tiny{$\Scal$}}}(\beta_{\textup{max}}, \omega_{\raisebox{0pt}{\tiny{$\Scal$}}})$ is given by 
\begin{align}
    \Delta E_{\raisebox{-1pt}{\tiny{$\Ccal$}}} &= \lim_{N\to\infty} \sum_{n=1}^{N} \Delta E_{\raisebox{-1pt}{\tiny{$\Ccal$}}}^{(n+1)} \nonumber\\
    &\geq (\gamma+1) \frac{1}{\beta}\,\widetilde{\Delta} S_{\raisebox{-1pt}{\tiny{$\Scal$}}} +\,\gamma\,\Delta E_{\raisebox{-1pt}{\tiny{$\Scal$}}} \notag \\
    &= \frac{1}{\beta}\,\widetilde{\Delta} S_{\raisebox{-1pt}{\tiny{$\Scal$}}} + \gamma \left( \Delta E_{\raisebox{-1pt}{\tiny{$\Scal$}}} + \frac{1}{\beta}\,\widetilde{\Delta} S_{\raisebox{-1pt}{\tiny{$\Scal$}}}\right).
\end{align}
Note that for infinite-temperature heat baths, $\gamma \to 0$ and the usual Landauer limit is recovered; nonetheless, for finite-temperature heat baths, $\gamma > 0$ and there is an additional energy contribution, implying that the Landauer limit cannot be achieved. Moreover, note that the expression inside the parenthesis in the second term above is always non-negative, as it is the free energy difference of the system during the cooling process. Lastly, it is straightforward to show that this lower bound is equivalent to the Carnot-Landauer limit in Eq.~\eqref{eq:carnotlandauercold}, which was derived in a protocol-independent manner as the ultimate limit in the incoherent-control setting. We now present explicit protocols that saturate this bound. 


\section{Perfect Cooling at the Carnot-Landauer Limit in the Incoherent-Control Paradigm}
\label{app:incoherentcoolingfinitetemperature}

The precise statement that we wish to prove regarding saturation of the Carnot-Landauer limit is the following:
\begin{lem}\label{lem:qubitincoherentfinitetemp}
For any $\beta^*\geq\beta>\beta_{\raisebox{-1pt}{\tiny{$H$}}}$ and $\epsilon_{1,2} > 0$, there exists a cooling protocol in the incoherent-control setting comprising a number of unitaries of finite control complexity, which, when the number of operations diverges, cools to some final temperature $\beta^\prime$ that is arbitrarily close to the ideal temperature value $\beta^*$, i.e.,
\begin{align}
    \left| \beta^{\prime} - \beta^* \right| < \epsilon_1,
\end{align}
with an energy cost, measured as heat drawn from the hot bath, that is arbitrarily close to the ideal Carnot-Landauer limit, i.e.,
\begin{align}
    \left| \Delta E_{\raisebox{-1pt}{\tiny{$\Hcal$}}} - \eta^{-1} \widetilde{\Delta} F_{\raisebox{-1pt}{\tiny{$\Scal$}}}^{(\beta)} \right| < \epsilon_2,
\end{align}
where $\eta = 1 - \beta_{\raisebox{-1pt}{\tiny{$H$}}} / \beta$ and $\Delta F_{\raisebox{-1pt}{\tiny{$\Scal$}}}^{(\beta)} = F_{\raisebox{-1pt}{\tiny{$\beta$}}}(\varrho_{\raisebox{-1pt}{\tiny{$\Scal$}}}^\prime) - F_{\raisebox{-1pt}{\tiny{$\beta$}}}(\varrho_{\raisebox{-1pt}{\tiny{$\Scal$}}})$ is the free energy difference between the initial $\varrho_{\raisebox{-1pt}{\tiny{$\Scal$}}} = \tau_{\raisebox{-1pt}{\tiny{$\Scal$}}}(\beta, H_{\raisebox{-1pt}{\tiny{$\Scal$}}})$ and final $\varrho_{\raisebox{-1pt}{\tiny{$\Scal$}}}^\prime = \tau_{\raisebox{-1pt}{\tiny{$\Scal$}}}(\beta^*, H_{\raisebox{-1pt}{\tiny{$\Scal$}}})$ system states (w.r.t. inverse temperature $\beta$).
\end{lem}

We begin by presenting the diverging-time protocol that saturates the Carnot-Landauer limit when all three subsystems $\Scal, \Ccal, \Hcal$ are qubits. The simplicity of this special case allows us to calculate precisely bounds on the number of operations required to reach any chosen error threshold. Building on this intuition, we then present the generalisation to the case where all systems are qudits. The protocols with diverging control complexity follow directly via the same line of reasoning presented in the main text.

\subsection{Qubit Case}
\label{app:incoherentqubit}

We begin with setting some notation and intuition for the proof, before expanding on mathematical details. 

\textbf{Sketch of Protocol.---}The protocol consists of the following. There are $N$ stages, each labelled by $n \in \{1,2, ... , N\}$. Each stage proceeds as follows:
\begin{itemize}
    \item A qubit with energy gap $\omega_{\raisebox{0pt}{\tiny{$\Scal$}}} + n \theta$ is taken from the cold part of the machine, and a qubit with energy gap $n \theta$ is taken from the hot part (see below). The initial state of the machine at the beginning of the $n^\textup{th}$ stage is thus $\tau_{\raisebox{-1pt}{\tiny{$\Ccal$}}}(\beta, \omega_{\raisebox{0pt}{\tiny{$\Scal$}}} + n \theta) \otimes \tau_{\raisebox{-1pt}{\tiny{$\Hcal$}}} (\beta_{\raisebox{-1pt}{\tiny{$H$}}}, n\theta)$.
    \item The energy-preserving three qubit unitary cycle in the $\{010,101\}_{\raisebox{-1pt}{\tiny{$\Scal \Ccal \Hcal$}}}$ subspace is performed [see Eq.~\eqref{eq:tripartite unitary op}], after which the cold and hot qubits are rethermalised to their respective initial temperatures.
    \item The above steps are repeated $m_n$ times.
\end{itemize}

\noindent The energy increment $\theta$ is defined as
\begin{align}\label{eq:incrementsize}
    \theta &:= \frac{\omega_{\raisebox{0pt}{\tiny{$\Scal$}}}}{N} \left( \frac{\beta^* - \beta}{\beta - \beta_{\raisebox{-1pt}{\tiny{$H$}}}} \right),
\end{align}
while the number of repetitions within each stage is given by
\begin{align}\label{eq:noofreps}
    m_n = \Bigg\lceil \frac{ \log (\delta) }{ \log (1 - N_{\raisebox{-1pt}{\tiny{$V$}}}^{(n)} ) } \Bigg\rceil.
\end{align}
$\lceil \cdot \rceil$ is the ceiling function, and $N_{\raisebox{-1pt}{\tiny{$V$}}}^{(n)}$ is the sum of the initial thermal populations in the $\{01,10\}_{\Ccal\Hcal}$ subspace of the machine, i.e.,
\begin{align}
    N_{\raisebox{-1pt}{\tiny{$V$}}}^{(n)} &:= \braket{01| \tau_{\raisebox{-1pt}{\tiny{$\Ccal$}}}(\beta, \omega_{\raisebox{0pt}{\tiny{$\Scal$}}} + n \theta) \otimes \tau_{\raisebox{-1pt}{\tiny{$\Hcal$}}} (\beta_{\raisebox{-1pt}{\tiny{$H$}}}, n\theta) | 01} + \braket{10| \tau_{\raisebox{-1pt}{\tiny{$\Ccal$}}}(\beta, \omega_{\raisebox{0pt}{\tiny{$\Scal$}}} + n \theta) \otimes \tau_{\raisebox{-1pt}{\tiny{$\Hcal$}}} (\beta_{\raisebox{-1pt}{\tiny{$H$}}}, n\theta) | 10}.
\end{align}
The parameter $\delta$ is chosen appropriately to complete the proof ($\delta = 1/N^2$ works).

The intuition for the proof is as follows. We first consider how the populations of the target system changes in the \emph{idealised} protocol where $m_n \rightarrow \infty$, so that in each stage, the system reaches the virtual temperature determined by the $\Ccal\Hcal$ qubits. We can use this ideal setting to find expressions for the final temperature and energy cost, which serves as a baseline that we wish to attain to within arbitrary precision. We then consider the protocol as constructed above with a finite number of repetitions $m_n$ in each stage, and show that its expressions for temperature and work cost are close (w.r.t. $1/N$) to the original expressions, and by taking $N$ to be sufficiently large but still finite (i.e., in the diverging time limit), we prove that the protocol can be arbitrarily close in temperature and energy cost to the ideal values.

\begin{proof}

We label the population in the \emph{excited} state of the target system at the end of stage $n$ as $p_n$. Thus $p_0$ is the initial population and $p_{\raisebox{-1pt}{\tiny{$N$}}}$ is the final population in the excited level of the target system qubit, i.e., that spanned by $\ket{1}\!\bra{1}_{\raisebox{-1pt}{\tiny{$\Scal$}}}$. We also label by $q_n$ what the corresponding population $p_n$ would hypothetically be in the limit $m_n \rightarrow \infty$. This value can be calculated by matching the temperature of the target system qubit to the temperature of the $\{01,10\}_{\Ccal \Hcal}$ virtual qubit within the machine (see Appendix G in Ref.~\cite{Clivaz_2019E}). Thus $q_n$ is defined via the Gibbs ratio
\begin{align}
    \frac{q_n}{1-q_n} &= e^{-\beta(\omega_{\raisebox{0pt}{\tiny{$\Scal$}}} + n\theta)} e^{+\beta_{\raisebox{-1pt}{\tiny{$H$}}} n\theta} = e^{-\beta \omega_{\raisebox{0pt}{\tiny{$\Scal$}}}} e^{-(\beta - \beta_{\raisebox{-1pt}{\tiny{$H$}}}) n \theta}. \label{eq:abstractpop}
\end{align}
Note that
\begin{enumerate}
    \item $\{ p_n\},\{q_n \}$ are both monotonically decreasing sequences, as each stage cools the target qubit further.
    \item $p_n > q_n$ for all $n$, as more repetitions within each stage keep cooling the target qubit further.
\end{enumerate}

To keep track of the energetic resource cost, which we take here to be the total heat drawn from the hot bath, we must sum the energetic contribution from each time the hot qubit is rethermalised to $\beta_{\raisebox{-1pt}{\tiny{$H$}}}$ after the application of the three-party cycle unitary. Due to the fact that the only manner in which the population of the hot qubit changes is due to the $\{010,101\}_{\raisebox{-1pt}{\tiny{$\Scal \Ccal \Hcal$}}}$ exchange, it follows that any population change in the hot qubit is identical to the population change in the target system qubit.

Focusing on a single stage, where the machine qubits are fixed in energy gap, the total population change in the hot qubit that must be restored by the hot bath is therefore equal to the population change in the target system throughout that stage. The heat drawn from the hot bath throughout the entire stage is therefore
\begin{align}
    \widetilde{\Delta} E_{\raisebox{-1pt}{\tiny{$\Hcal$}}}^{(n)} &= \omega_{\raisebox{0pt}{\tiny{$\Hcal$}}}^{(n)} (p_{n-1} - p_n) = n \theta (p_{n-1} - p_n).
\end{align}
With these expressions derived, we can study the properties of the abstract protocol where the number of repetitions within each stage goes to infinity: $m_n\to\infty$. First, the final temperature asymptotically achieved here is given by finding the temperature $\widetilde{\beta}$ associated with the qubit with excited-state population $q_{\raisebox{-1pt}{\tiny{$N$}}}$
\begin{align}
    \frac{q_{\raisebox{-1pt}{\tiny{$N$}}}}{1-q_{\raisebox{-1pt}{\tiny{$N$}}}} = e^{-\widetilde{\beta} \omega_{\raisebox{0pt}{\tiny{$\Scal$}}}} 
    \Rightarrow \quad e^{-\beta \omega_{\raisebox{0pt}{\tiny{$\Scal$}}}} e^{-(\beta - \beta_{\raisebox{-1pt}{\tiny{$H$}}}) N \theta} = e^{-\widetilde{\beta} \omega_{\raisebox{0pt}{\tiny{$\Scal$}}}} 
    \Rightarrow \quad \widetilde{\beta} = \beta^*,
\end{align}
where we make use of the definition of $\theta$ in Eq.~\eqref{eq:incrementsize}. We can thus identify $q_{\raisebox{-1pt}{\tiny{$N$}}} = q^*$, since it is the population associated with the ideal final temperature $\beta^*$.

We also have the following expression for the total energetic cost of the ideal protocol after $N$ stages
\begin{align}\label{eq:abstractcost}
    \widetilde{\Delta} E_{\raisebox{-1pt}{\tiny{$\Hcal$}}}^* &= \sum_{n=1}^N n \theta (q_{n-1} - q_n),
\end{align}
which can alternatively be expressed as
\begin{align}\label{eq:leftsum}
    \widetilde{\Delta} E^*_{\raisebox{-1pt}{\tiny{$\Hcal$}}}  &= \sum_{n=1}^N \left[(n-1) \theta (q_{n-1} - q_n)\right] + \theta (q_0 - q_{\raisebox{-1pt}{\tiny{$N$}}})
\end{align}
The sums appearing in the two alternative expressions are the left and right Riemann sums of the integral of the variable $y = n\theta$ integrated with respect to the variable $q$, i.e.,
\begin{align}\label{eq:freeintegral}
    I :=& - \int_{q_0}^{q^*} y \; \textup{d}q,  \notag\\
    \text{where} \quad \frac{q(y)}{1-q(y)} =& e^{-\beta \omega_{\raisebox{0pt}{\tiny{$\Scal$}}}} e^{-(\beta - \beta_{\raisebox{-1pt}{\tiny{$H$}}}) y},
\end{align}
from Eq.~\eqref{eq:abstractpop}. For $y>0$, $q(y)$ is monotonically decreasing and so the converse is also true, i.e., $y$ is monotonically decreasing w.r.t. $q(y)$. This implies that the integral is bounded by the left and right Riemann sums, so we have
\begin{align}
    \sum_{n=1}^N (n-1) \theta (q_{n-1} - q_n) \leq \int_{q_0}^{q^*} y \; \textup{d}q \leq \sum_{n=1}^N n \theta (q_{n-1} - q_n),
\end{align}
from which we can deduce that the value of $\Delta E^*_{\raisebox{-1pt}{\tiny{$\Hcal$}}}$ is itself is bounded both ways from Eqs.~\eqref{eq:abstractcost} and~\eqref{eq:leftsum}:
\begin{align}
    \int_{q_0}^{q^*} y \; \textup{d}q \leq \widetilde{\Delta} E^*_{\raisebox{-1pt}{\tiny{$\Hcal$}}} \leq \int_{q_0}^{q^*} y \; \textup{d}q + \theta (q_0 - q^*).
\end{align}
The integral itself can by expressed in terms of the free energy of the qubit target system with respect to the environment inverse temperature $\beta$. Expressing the free energy as a function of the excited-state population $q$ and differentiating w.r.t. $q$ gives
\begin{align}
    F(q) &= \braket{E}(q) - \frac{S(q)}{\beta}= q \; \omega_{\raisebox{0pt}{\tiny{$\Scal$}}} + \frac{1}{\beta} \left[ q \log (q) + (1-q) \log (1-q) \right]. \\
    \frac{\partial F}{\partial q} &= \omega_{\raisebox{0pt}{\tiny{$\Scal$}}} + \frac{1}{\beta} \log \left( \frac{q}{1-q} \right) = \left( \omega_{\raisebox{0pt}{\tiny{$\Scal$}}} + \frac{1}{\beta} \left( -\beta \omega_{\raisebox{0pt}{\tiny{$\Scal$}}} - (\beta - \beta_{\raisebox{-1pt}{\tiny{$H$}}}) y \right) \right) = -\frac{\beta - \beta_{\raisebox{-1pt}{\tiny{$H$}}}}{\beta} y.
\end{align}
Using the above expression, the definite integral in Eq.~\eqref{eq:freeintegral} amounts to
\begin{align}
    I &= \frac{1}{\eta} \left[ F(q^*) - F(q_0) \right] =: \frac{1}{\eta} \left( F^* - F_0 \right),
\end{align}
where we identify the Carnot efficiency $\eta =1 - \beta_{\raisebox{-1pt}{\tiny{$H$}}}/\beta$ and for ease of notation written $F^* := F(q^*)$ and $F_0 := F(q_0)$. Thus we can bound $\widetilde{\Delta} E^*_{\raisebox{-1pt}{\tiny{$\Hcal$}}}$ on both sides
\begin{align}
    \frac{1}{\eta} \left( F^* - F_0 \right) \leq \widetilde{\Delta} E^*_{\raisebox{-1pt}{\tiny{$\Hcal$}}} &\leq \frac{1}{\eta} \left( F^* - F_0 \right) + \theta (q_0 - q^*) \leq \frac{1}{\eta} \left( F^* - F_0 \right) + \frac{\omega_{\raisebox{0pt}{\tiny{$\Scal$}}}}{N} \left( \frac{\beta^* - \beta}{\beta - \beta_{\raisebox{-1pt}{\tiny{$H$}}}} \right), \label{eq:abstractcostdiff}
\end{align}
where the inequality in the second line follows from the fact that $\{q_n\}$ forms a decreasing sequence.

We now proceed to consider the cooling protocol with a finite number of repetitions $m_n$ within each stage. We first bound the difference between $p_n$ and $q_n$. Using the properties of the exchange unitary under repetitions~\cite{Silva_2016,Clivaz_2019E} (in particular, see Appendix G in Ref.~\cite{Clivaz_2019E}), we have that in each stage
\begin{align}\label{eq: swapreps}
    \frac{p_n - q_n}{p_{n-1}-q_n} &= \left(1 - N_{\raisebox{-1pt}{\tiny{$V$}}}^{(n)} \right)^{m_n}.
\end{align}
Thus, the population difference to the asymptotically achievable population given by the virtual temperature shrinks as a power law w.r.t. the number of repetitions. Since $0 < N_{\raisebox{-1pt}{\tiny{$V$}}}^{(n)} < 1$ (all strict inequalities), three points follow: first, the population $q_n$ can never be attained with a finite number of steps within the stage $n$; second, that every repetition cools the system further by some finite amount; third, that one can get arbitrarily close to $q_n$ by taking $m_n$ sufficiently large. In fact, by our definition of $m_n$, we have that
\begin{align} \label{eq: swapreps-delta}
    \frac{p_n - q_n}{p_{n-1}-q_n} \leq \delta.
\end{align}
From this, we can prove that
\begin{align}\label{eq:popdiff}
    p_n - q_n \leq \delta^n q_0 - \delta q_n + (1-\delta)\delta \sum_{j=1}^{n-1} \delta^{n-j-1} q_j.
\end{align}
The proof is by induction. For $n=0$, $p_0 = q_0$ (initial state), and for $n=1$, using Eq.~\eqref{eq: swapreps-delta}
\begin{align}
    p_1 - q_1 &\leq \delta ( p_0 - q_1 )\notag \\
    &= \delta (q_0 - q_1).
\end{align}
Suppose that the above statement holds true for $p_k$. Then from Eq.~\eqref{eq: swapreps-delta}
\begin{align}
    p_{k+1} - q_{k+1} &\leq \delta (p_k - q_{k+1}) \notag \\ &= \delta (p_k - q_k + q_k - q_{k+1}) \notag\\
        &\;\,\vdots \notag\\
        &\leq \delta^{k+1} q_0  - \delta q_{k+1} (1-\delta)\delta +\sum_{j=1}^{(k+1)-1} \delta^{(k+1)-j-1} q_j. 
\end{align}
With this result, we can now bound the difference between the energy cost of this finite-repetition protocol and that of the idealised one. We now proceed to prove that
\begin{align}\label{eq:generaldiffhypothesis}
    \widetilde{\Delta} E_{\raisebox{-1pt}{\tiny{$\Hcal$}}} - \widetilde{\Delta} E^*_{\raisebox{-1pt}{\tiny{$\Hcal$}}} &= \sum_{n=1}^N n\theta (p_{n-1} - p_n) - \sum_{n=1}^N n\theta (q_{n-1} - q_n) \leq \theta \left( q_0 \sum_{j=1}^{N-1} \delta^{N-j} - \sum_{j=1}^{N-1} \delta^{N-j} q_j \right).
\end{align}
We again use proof by induction. First note that we can rewrite
\begin{align}
    \sum_{n=1}^N n\theta (f_{n-1} - f_n) &= \theta \left( \sum_{n=1}^N f_{n-1} \right) - N \theta f_{\raisebox{-1pt}{\tiny{$N$}}},
\end{align}
for $f_n \in \{ p_n, q_n\}$. Therefore, we can rewrite the difference
\begin{align}
    \widetilde{\Delta} E_{\raisebox{-1pt}{\tiny{$\Hcal$}}} - \widetilde{\Delta} E^*_{\raisebox{-1pt}{\tiny{$\Hcal$}}} &= \theta  \sum_{n=1}^N \left(p_{n-1} - q_{n-1} \right) - N \theta (p_{\raisebox{-1pt}{\tiny{$N$}}} - q_{\raisebox{-1pt}{\tiny{$N$}}}) \leq \theta \left( \sum_{n=1}^N (p_{n-1} - q_{n-1}) \right),
\end{align}
since the last subtracted term is always strictly positive. Consider now the partial sum
\begin{align}
    \mathcal{E}_k &= \sum_{n=1}^k \left(p_{n-1} - q_{n-1} \right).
\end{align}
For $k=1$, $\mathcal{E}_1 = 0$, since $p_0 = q_0$. For $k=2$, we have
\begin{align}
    \mathcal{E}_1 &= (p_1 - q_1) \leq \delta (q_0 - q_1) = \left( q_0 \sum_{j=1}^{1} \delta^{2-j} - \sum_{j=1}^{1} \delta^{2-j} q_j \right),
\end{align}
which matches the hypothesis of Eq.~\eqref{eq:generaldiffhypothesis}. Assuming that the same holds true for $\mathcal{E}_k$, then for $\mathcal{E}_{k+1}$, we have
\begin{align}
    \mathcal{E}_{k+1} &= \mathcal{E}_k + (p_k - q_k)\notag  \\
        &\leq \left( q_0 \sum_{j=1}^{k-1} \delta^{k-j} - \sum_{j=1}^{k-1} \delta^{k-j} q_j \right) + \left( \delta^k q_0 + (1-\delta)\delta \sum_{j=1}^{k-1} \delta^{k-j-1} q_j - \delta q_k \right) \notag \\
        &\;\vdots \notag \\
        &= q_0\sum_{j=1}^k \delta^{k+1-j} - \sum_{j=1}^{k} \delta^{k+1-j} q_j. 
\end{align}
Then, by dropping the second sum, which is a strictly positive quantity, the difference in Eq.~\eqref{eq:generaldiffhypothesis} can be further simplified to
\begin{align}
    \widetilde{\Delta} E_{\raisebox{-1pt}{\tiny{$\Hcal$}}} - \widetilde{\Delta} E^*_{\raisebox{-1pt}{\tiny{$\Hcal$}}}  &\leq \theta q_0 \sum_{j=1}^{N-1} \delta^{N-j} = \theta q_0 \, \delta \sum_{k=0}^{N-2} \delta^k < \theta q_0 \, \delta (N-1) < \theta q_0 \, \delta  N < \omega_{\raisebox{0pt}{\tiny{$\Scal$}}} \left( \frac{\beta^* - \beta}{\beta - \beta_{\raisebox{-1pt}{\tiny{$H$}}}} \right) \delta, \label{eq:costdiff}
\end{align}
where we use that $\delta < 1$. Finally, to upper bound the number of operations required in the protocol, we bound the number of repetitions within each stage by bounding the total population of the virtual qubit spanned by the levels $\{01,10\}_{\Ccal\Hcal}$ as follows:
\begin{align}
    N_{\raisebox{-1pt}{\tiny{$V$}}}^{(n)} &= \braket{01| \tau_{\raisebox{-1pt}{\tiny{$\Ccal$}}}(\beta, \omega_{\raisebox{0pt}{\tiny{$\Scal$}}} + n \theta) \otimes \tau_{\raisebox{-1pt}{\tiny{$\Hcal$}}} (\beta_{\raisebox{-1pt}{\tiny{$H$}}}, n\theta)  | 01} + \braket{10| \tau_{\raisebox{-1pt}{\tiny{$\Ccal$}}}(\beta, \omega_{\raisebox{0pt}{\tiny{$\Scal$}}} + n \theta) \otimes \tau_{\raisebox{-1pt}{\tiny{$\Hcal$}}} (\beta_{\raisebox{-1pt}{\tiny{$H$}}}, n\theta) | 10} \notag \\
        &= \frac{e^{-\beta_{\raisebox{-1pt}{\tiny{$H$}}} n \theta} + e^{-\beta (\omega_{\raisebox{0pt}{\tiny{$\Scal$}}} + n \theta)}}{(1+e^{-\beta_{\raisebox{-1pt}{\tiny{$H$}}} n \theta})(1+e^{-\beta (\omega_{\raisebox{0pt}{\tiny{$\Scal$}}} + n \theta)})} \notag \\
        &> \frac{e^{-\beta (\omega_{\raisebox{0pt}{\tiny{$\Scal$}}} + n \theta)}}{4}. \notag \\
    \Rightarrow \quad \log \left[ 1 - N_{\raisebox{-1pt}{\tiny{$V$}}}^{(n)} \right] &< \log \left[ 1 - \frac{e^{-\beta (\omega_{\raisebox{0pt}{\tiny{$\Scal$}}} + n \theta)}}{4} \right] \notag \\
        &< - \frac{e^{-\beta (\omega_{\raisebox{0pt}{\tiny{$\Scal$}}} + n \theta)}}{4} \quad\quad\quad \text{if $x \in (0,1) \;\Rightarrow\; \log(1-x) < -x$.} \notag \\
    \Rightarrow \quad - \frac{1}{\log \left[ 1 - N_{\raisebox{-1pt}{\tiny{$V$}}}^{(n)} \right]} &< 4 e^{+\beta (\omega_{\raisebox{0pt}{\tiny{$\Scal$}}} + n \theta)}
\end{align}
Thus we can bound the number of repetitions in each stage from Eq.~\eqref{eq:noofreps}. Noting that $\log (\delta) < 0$, we have
\begin{align}
    m_n < 4\log \left( 1/\delta \right) e^{+\beta (\omega_{\raisebox{0pt}{\tiny{$\Scal$}}} + n \theta)} + 1.
\end{align}
For a crude bound, we can replace $n$ by its maximum value $N$, and sum over all the stages to find an upper bound on the total number of three-qubit exchange unitaries implemented throughout the entire protocol, which gives
\begin{align}
    M = \sum_{n=1}^N m_n &< N \left[ 4\log \left( 1/\delta \right) e^{+\beta (\omega_{\raisebox{0pt}{\tiny{$\Scal$}}} + N \theta)} + 1 \right] = N \left[ 4\log \left( 1/\delta \right) e^{ \omega_{\raisebox{0pt}{\tiny{$\Scal$}}} (\beta^* - \beta_{\raisebox{-1pt}{\tiny{$H$}}})/\eta} + 1 \right].
    \label{eq:totalreps}
\end{align}
Also, note that $\lim_{\delta \rightarrow 0} p_{\raisebox{-1pt}{\tiny{$N$}}} = q_{\raisebox{-1pt}{\tiny{$N$}}} = q^*$. More precisely, using Eq.~\eqref{eq:popdiff}, we have
\begin{align}
    p_{\raisebox{-1pt}{\tiny{$N$}}} - q^* &< \delta \left( \delta^{N-1} q_0 + (1-\delta) \sum_{j=1}^{N-1} \delta^{n-j-1} q_j - q_{\raisebox{-1pt}{\tiny{$N$}}} \right)\notag  \\
        &< \delta \left( 1 + (1-\delta)(N-1) \right) < \delta N.
\end{align}
In summary, we have the following bounds on the protocol in which each stage consists of a finite number of steps
\begin{align}
    p_{\raisebox{-1pt}{\tiny{$N$}}} - q^* &< \delta N \notag\\
    \widetilde{\Delta} E_{\raisebox{-1pt}{\tiny{$\Hcal$}}} &< \frac{1}{\eta} \left( F^* - F_0 \right) + \omega_{\raisebox{0pt}{\tiny{$\Scal$}}} \left( \frac{\beta^* - \beta}{\beta - \beta_{\raisebox{-1pt}{\tiny{$H$}}}} \right) \left( \frac{1}{N} + \delta \right),
\end{align}
where we combine Eqs.~\eqref{eq:abstractcostdiff} and \eqref{eq:costdiff} for the second expression. For simplicity, we choose $\delta = 1/N^2$, so that
\begin{align}
    p_{\raisebox{-1pt}{\tiny{$N$}}} - q^* &< \frac{1}{N} \notag\\
    \widetilde{\Delta} E_{\raisebox{-1pt}{\tiny{$\Hcal$}}} &< \frac{1}{\eta} \left( F^* - F_0 \right) + \omega_{\raisebox{0pt}{\tiny{$\Scal$}}} \left( \frac{\beta^* - \beta}{\beta - \beta_{\raisebox{-1pt}{\tiny{$H$}}}} \right) \left( \frac{2}{N} \right).
\end{align}
Thus, given any final temperature (encoded by the population $q^*$), and allowed errors $\epsilon_1$ and $\epsilon_2$ for the final population and energy cost respectively, one can always choose $N$ large enough so that both quantities are within the error threshold. Specifically, choosing $N$ as
\begin{align}
    N &= \Bigg\lceil \textup{max} \left\{ \epsilon_1^{-1}, 2 \omega_{\raisebox{0pt}{\tiny{$\Scal$}}} \left( \frac{\beta^* - \beta}{\beta - \beta_{\raisebox{-1pt}{\tiny{$H$}}}} \right) \epsilon_2^{-1} \right\} \Bigg\rceil,
\end{align}
we automatically have that $p_{\raisebox{-1pt}{\tiny{$N$}}} - q^* <\epsilon_1$ and $\Delta E_{\raisebox{-1pt}{\tiny{$\Hcal$}}} < (F^* - F_0)/\eta + \epsilon_2$. The total number of unitary operations (each of which is followed by rethermalisation of the machine) is then bounded by Eq.~\eqref{eq:totalreps}
\begin{align}
    M < N \left(8 \log [N] e^{ \omega_{\raisebox{0pt}{\tiny{$\Scal$}}} (\beta^* - \beta_{\raisebox{-1pt}{\tiny{$H$}}})/\eta} + 1 \right).
\end{align}
\noindent We can see from Theorem \ref{thm:landauer-incoherent} that the protocol is asymptotically optimal with respect to the energy extracted from the hot bath. \end{proof}

\subsection{Qudit Case}
\label{app:incoherentqudit}

The extension of the proof above to the case of qudits is nontrivial. This is because, while for qubits there is only one energy-resonant subspace that leads to cooling and hence a unique protocol [see Eq.~\eqref{eq:tripartite unitary op}] that asymptotically attains perfect cooling at the Carnot-Landauer bound, this is no longer the case for higher-dimensional systems; here, there can be a number of energy-resonant subspaces that cool the target and the question of optimality hinges crucially on the complex energy-level structure of all systems involved. Hence, it is not possible to provide a unique unitary that generates the optimal protocol independently of the subsystem Hamiltonians. Nonetheless, we slightly modify the protocol for the qubit case above to be implemented on a number of particular three-qubit subspaces of the three-qudit global state such that, at the end of each stage, the state of the target system is arbitrarily close to the (known) state, which would be achieved in an abstract protocol in the diverging-time limit. This asymptotically attainable state is precisely that which would be achieved in the coherent-control paradigm with a machine the same dimension as the joint hot-cold qudits. Thus, we first begin by presenting the necessary steps for the proof in the coherent-control setting, which we then adapt as appropriate for the incoherent setting control. Finally, summing the overall energy cost of said protocol over all stages saturates the Carnot-Landauer bound, as required. \\

\begin{proof}
\textbf{An idealised sequence of temperatures and system states.} We construct the incoherent protocol in the following manner. We seek to take the system through a sequence of thermal states starting at inverse temperature $\beta$ and ending at inverse temperature $\beta^*$ with $N$ equally spaced intermediary steps, i.e.,
\begin{align}
    \beta_n &= \beta + n \theta \left( \beta - \beta_{\raisebox{-1pt}{\tiny{$\Hcal$}}} \right), \label{eq:idealtemperatureincoherentstagen} \\
        \theta &= \frac{1}{N} \left( \frac{\beta^* - \beta}{\beta - \beta_{\raisebox{-1pt}{\tiny{$\Hcal$}}}} \right),
\end{align}
so that $\beta_{\raisebox{-1pt}{\tiny{$N$}}} = \beta^*$ by construction. This corresponds to taking the system through the following sequence of thermal states
\begin{align}
    \varrho^{(n)}_{\raisebox{-1pt}{\tiny{$\Scal$}}} &= \frac{e^{-\beta_n H_{\raisebox{-1pt}{\tiny{$\Scal$}}}}}{\mathcal{Z}_{\raisebox{-1pt}{\tiny{$\Scal$}}}(H_{\raisebox{-1pt}{\tiny{$\Scal$}}}, \beta_n)}.
\end{align}
Note that, in contrast to the coherent protocol where such a sequence can be traversed by simply swapping the target system with a sequence of appropriate machines, in the incoherent setting such a protocol is generally not possible as such swaps are not energy conserving. Nonetheless, we develop a modified protocol that is energy conserving and mimics this idealised one.

Corresponding to each step in the sequence, we define the following quantity, which we eventually show to be related to the heat drawn from the hot bath:
\begin{align}
    G^{(n)} &= - n \theta \Delta E_{\raisebox{-1pt}{\tiny{$\Scal$}}}^{(n)} = - n \theta \, \mathrm{tr} \left[ H_{\raisebox{-1pt}{\tiny{$\Scal$}}} \left( \varrho^{(n)}_{\raisebox{-1pt}{\tiny{$\Scal$}}} - \varrho^{(n-1)}_{\raisebox{-1pt}{\tiny{$\Scal$}}} \right) \right].
\end{align}
We proceed to show that the total $\sum_n G^{(n)}$ that we label \textit{the idealised heat cost $\widetilde{\Delta} E^*_{\raisebox{-1pt}{\tiny{$\Hcal$}}}$} is close to the free energy difference over the entire sequence. We have
\begin{align}
    \widetilde{\Delta} E^*_{\raisebox{-1pt}{\tiny{$\Hcal$}}} &= \sum_{n=1}^N G^{(n)} \notag \\
        &= \sum_{n=1}^N n \theta \, \mathrm{tr} \left[ H_{\raisebox{-1pt}{\tiny{$\Scal$}}} \left( \varrho_{\raisebox{-1pt}{\tiny{$\Scal$}}}^{(n-1)} - \varrho_{\raisebox{-1pt}{\tiny{$\Scal$}}}^{(n)} \right) \right] \label{eq:rightRiemannsum} \\
        &= \left\{ \sum_{n=1}^N (n-1) \theta \, \mathrm{tr} \left[ H_{\raisebox{-1pt}{\tiny{$\Scal$}}} \left( \varrho_{\raisebox{-1pt}{\tiny{$\Scal$}}}^{(n-1)} - \varrho_{\raisebox{-1pt}{\tiny{$\Scal$}}}^{(n)} \right) \right] \right\} + \theta \, \mathrm{tr} \left[ H_{\raisebox{-1pt}{\tiny{$\Scal$}}} \left( \varrho_{\raisebox{-1pt}{\tiny{$\Scal$}}}^{(0)} - \varrho_{\raisebox{-1pt}{\tiny{$\Scal$}}}^{(N)} \right) \right].\label{eq:leftRiemannsum}
\end{align}
The sums on the second and third lines above, Eqs.~\eqref{eq:rightRiemannsum} and \eqref{eq:leftRiemannsum} respectively, are the \textit{right} and \textit{left} Riemann sums corresponding to the following integral:
\begin{align}
    I &= \int_{q_i}^{q_f} q \left( - \textup{d}x \right) = \int_{q_f}^{q_i} q \, \textup{d}x, \notag\\
    \text{where} \quad n\theta &\rightarrow q, \notag\\
    \quad x &= \mathrm{tr} \left[ H_{\raisebox{-1pt}{\tiny{$\Scal$}}} \varrho_{\raisebox{-1pt}{\tiny{$\Scal$}}}(q) \right], \notag \\
    \varrho_{\raisebox{-1pt}{\tiny{$\Scal$}}}(q) &= \frac{ e^{-[\beta + q(\beta - \beta_{\raisebox{-1pt}{\tiny{$\Hcal$}}})] H_{\raisebox{-1pt}{\tiny{$\Scal$}}}} }{ \mathrm{tr}\left[ e^{-[\beta + q(\beta - \beta_{\raisebox{-1pt}{\tiny{$\Hcal$}}})] H_{\raisebox{-1pt}{\tiny{$\Scal$}}}} \right]}.\label{eq:incoherentworkint}
\end{align}
We observe that $x$ is the average energy of the thermal state of temperature $\beta + q(\beta - \beta_{\raisebox{-1pt}{\tiny{$\Hcal$}}})$, and thus $x$ and $q$ are strictly monotonically decreasing w.r.t. each other (which explains why the left and right sums are switched). It follows that the Riemann sums bound the integral
\begin{align}
    \sum_{n=1}^N (n-1) \theta \, \mathrm{tr} \left[ H_{\raisebox{-1pt}{\tiny{$\Scal$}}} \left( \varrho_{\raisebox{-1pt}{\tiny{$\Scal$}}}^{(n-1)} - \varrho_{\raisebox{-1pt}{\tiny{$\Scal$}}}^{(n)} \right) \right] \leq \int_{q_f}^{q_i} q \, \textup{d}x \leq        \sum_{n=1}^N n \theta \, \mathrm{tr} \left[ H_{\raisebox{-1pt}{\tiny{$\Scal$}}} \left( \varrho_{\raisebox{-1pt}{\tiny{$\Scal$}}}^{(n-1)} - \varrho_{\raisebox{-1pt}{\tiny{$\Scal$}}}^{(n)} \right) \right].
\end{align}
We can thus bound the idealised heat cost in both directions via
\begin{align}
    I \leq \widetilde{\Delta} E^*_{\raisebox{-1pt}{\tiny{$\Hcal$}}} \leq I + \theta \, \mathrm{tr} \left[ H_{\raisebox{-1pt}{\tiny{$\Scal$}}} \left( \varrho_{\raisebox{-1pt}{\tiny{$\Scal$}}}^{(0)} - \varrho_{\raisebox{-1pt}{\tiny{$\Scal$}}}^{(N)} \right) \right].
\end{align}
The integral in Eq.~\eqref{eq:incoherentworkint} can be shown to be equal to the change in free energy of the target system (w.r.t. inverse temperature $\beta$)
\begin{align}
    F_{\raisebox{-1pt}{\tiny{$\beta$}}}[\varrho_{\raisebox{-1pt}{\tiny{$\Scal$}}}(q)] &= \mathrm{tr} \left[ H_{\raisebox{-1pt}{\tiny{$\Scal$}}} \varrho_{\raisebox{-1pt}{\tiny{$\Scal$}}}(q) \right] + \frac{1}{\beta} \mathrm{tr} \left[ \varrho_{\raisebox{-1pt}{\tiny{$\Scal$}}}(q) \log \varrho_{\raisebox{-1pt}{\tiny{$\Scal$}}}(q) \right], \notag \\
    \frac{\textup{d}}{\textup{d}q} F_{\raisebox{-1pt}{\tiny{$\beta$}}}[\varrho_{\raisebox{-1pt}{\tiny{$\Scal$}}}(q)] &= \mathrm{tr} \left[ \left( H_{\raisebox{-1pt}{\tiny{$\Scal$}}} + \frac{\mathbbm{1}_{\raisebox{-1pt}{\tiny{$\Scal$}}} + \log \varrho_{\raisebox{-1pt}{\tiny{$\Scal$}}}(q)}{\beta} \right) \frac{\textup{d} \varrho_{\raisebox{-1pt}{\tiny{$\Scal$}}}(q)}{\textup{d}q} \right].
\end{align}
Note that $\varrho_{\raisebox{-1pt}{\tiny{$\Scal$}}}(q)$ and $\textup{d} \varrho_{\raisebox{-1pt}{\tiny{$\Scal$}}}(q)$ are both always diagonal in $H_{\raisebox{-1pt}{\tiny{$\Scal$}}}$ and full rank for all $q \in \mathbb{R}$, so we have no problems with $\log \varrho_{\raisebox{-1pt}{\tiny{$\Scal$}}}(q)$, and all of the operators in the expression are well defined and commute. Proceeding, we repeatedly use $\mathrm{tr} \left[ \textup{d} \varrho_{\raisebox{-1pt}{\tiny{$\Scal$}}}(q) \right] = \textup{d} \, \mathrm{tr} \left[ \varrho_{\raisebox{-1pt}{\tiny{$\Scal$}}}(q) \right] = 0$ and label the partition function $\mathcal{Z}(q) := \mathrm{tr} \left[ e^{-[\beta + q(\beta - \beta_{\raisebox{-1pt}{\tiny{$\Hcal$}}})] H_{\raisebox{-1pt}{\tiny{$\Scal$}}}} \right]$ to obtain
\begin{align}
    \frac{\textup{d}}{\textup{d}q} F_{\raisebox{-1pt}{\tiny{$\beta$}}}[\varrho_{\raisebox{-1pt}{\tiny{$\Scal$}}}(q)] &= \mathrm{tr} \left[ \left( H_{\raisebox{-1pt}{\tiny{$\Scal$}}} + \frac{\log \varrho_{\raisebox{-1pt}{\tiny{$\Scal$}}}(q)}{\beta} \right) \frac{\textup{d} \varrho_{\raisebox{-1pt}{\tiny{$\Scal$}}}(q)}{\textup{d}q} \right] \notag \\
        &= \mathrm{tr} \left[ \left( H_{\raisebox{-1pt}{\tiny{$\Scal$}}} - \frac{ \beta + q (\beta - \beta_{\raisebox{-1pt}{\tiny{$\Hcal$}}})}{\beta} H_{\raisebox{-1pt}{\tiny{$\Scal$}}} - \mathbbm{1}_{\raisebox{-1pt}{\tiny{$\Scal$}}} \frac{\log \mathcal{Z}(q)}{\beta} \right) \frac{\textup{d} \varrho_{\raisebox{-1pt}{\tiny{$\Scal$}}}(q)}{\textup{d}q} \right] \notag \\
        &= - q \left( 1 - \frac{\beta_{\raisebox{-1pt}{\tiny{$\Hcal$}}}}{\beta} \right) \frac{\textup{d}}{\textup{d}q} \mathrm{tr} \left[ H_{\raisebox{-1pt}{\tiny{$\Scal$}}} \varrho_{\raisebox{-1pt}{\tiny{$\Scal$}}}(q) \right] = - q \eta \frac{\textup{d}x}{\textup{d}q},
\end{align}
where we identify the Carnot efficiency $\eta$ for an engine operating between $\beta$ and $\beta_{\raisebox{-1pt}{\tiny{$\Hcal$}}}$. The integral thus simplifies to
\begin{align}
    I &= \eta^{-1} \left( F_{\raisebox{-1pt}{\tiny{$\beta$}}}[\varrho_{\raisebox{-1pt}{\tiny{$\Scal$}}}(q_f)] - F_{\raisebox{-1pt}{\tiny{$\beta$}}}[\varrho_{\raisebox{-1pt}{\tiny{$\Scal$}}}(q_i)] \right) =: \eta^{-1} \Delta F_{\raisebox{-1pt}{\tiny{$\Scal$}}}^{(\beta)}.
\end{align}
The idealised heat cost is thus bounded by
\begin{align}
    \eta^{-1} \Delta F_{\raisebox{-1pt}{\tiny{$\Scal$}}}^{(\beta)} \leq \widetilde{\Delta} E^*_{\raisebox{-1pt}{\tiny{$\Hcal$}}}  \leq \eta^{-1} \Delta F_{\raisebox{-1pt}{\tiny{$\Scal$}}}^{(\beta)} + \theta \, \mathrm{tr} \left[\Hcal_{\raisebox{-1pt}{\tiny{$\Scal$}}} \left( \varrho_{\raisebox{-1pt}{\tiny{$\Scal$}}}^{(0)} - \varrho_{\raisebox{-1pt}{\tiny{$\Scal$}}}^{(N)}\right) \right].
\end{align}
The left inequality is Landauer's bound applied to cooling a target system with Hamiltonian $H_{\raisebox{-1pt}{\tiny{$\Scal$}}}$ (see Theorem~\ref{thm:main-landauer-incoherent}), and the error term on the right can be bounded quite easily; for instance, for $\beta>0$, we have
\begin{align}
    \mathrm{tr} \left[ \Hcal_{\raisebox{-1pt}{\tiny{$\Scal$}}} \left(\varrho_{\raisebox{-1pt}{\tiny{$\Scal$}}}^{(0)} - \varrho_{\raisebox{-1pt}{\tiny{$\Scal$}}}^{(N)}\right) \right] &= \mathrm{tr} \left[ \left( \Hcal_{\raisebox{-1pt}{\tiny{$\Scal$}}} - E_{\raisebox{-1pt}{\tiny{$\Scal$}}}^{\textup{min}} \mathbbm{1}_{\raisebox{-1pt}{\tiny{$\Scal$}}} \right) \left(\varrho_{\raisebox{-1pt}{\tiny{$\Scal$}}}^{(0)} - \varrho_{\raisebox{-1pt}{\tiny{$\Scal$}}}^{(N)}\right) \right] \notag \\
    &\leq \mathrm{tr} \left[ \left( \Hcal_{\raisebox{-1pt}{\tiny{$\Scal$}}} - E_{\raisebox{-1pt}{\tiny{$\Scal$}}}^{\textup{min}} \mathbbm{1}_{\raisebox{-1pt}{\tiny{$\Scal$}}} \right) \varrho_{\raisebox{-1pt}{\tiny{$\Scal$}}}^{(0)} \right] & &\text{since $\Hcal_{\raisebox{-1pt}{\tiny{$\Scal$}}}  - E_{\raisebox{-1pt}{\tiny{$\Scal$}}}^{\textup{min}} \mathbbm{1}_{\raisebox{-1pt}{\tiny{$\Scal$}}}$ is a positive operator,} \notag \\
    &\leq \mathrm{tr} \left[ \left( \Hcal_{\raisebox{-1pt}{\tiny{$\Scal$}}} - E_{\raisebox{-1pt}{\tiny{$\Scal$}}}^{\textup{min}} \mathbbm{1}_{\raisebox{-1pt}{\tiny{$\Scal$}}} \right) \frac{\mathbbm{1}_{\raisebox{-1pt}{\tiny{$\Scal$}}}}{d_{\raisebox{-1pt}{\tiny{$\Scal$}}}} \right] \leq \frac{ \omega_{\raisebox{0pt}{\tiny{$\Scal$}}}^{\textup{max}}}{d_{\raisebox{-1pt}{\tiny{$\Scal$}}}},
\end{align}
where $\omega^{\textup{max}}_{\raisebox{-1pt}{\tiny{$\Scal$}}} := E_{\raisebox{-1pt}{\tiny{$\Scal$}}}^{\textup{max}} -  E_{\raisebox{-1pt}{\tiny{$\Scal$}}}^{\textup{min}}$ is the largest energy gap in the target system Hamiltonian and $d_{\raisebox{-1pt}{\tiny{$\Scal$}}}$ is the system dimension. We use the fact that since $\rho_{\raisebox{-1pt}{\tiny{$\Scal$}}}^{(0)}$ is a thermal state of positive temperature, its average energy is less than that of the infinite temperature thermal state, $\mathbbm{1}_{\raisebox{-1pt}{\tiny{$\Scal$}}}/d_{\raisebox{-1pt}{\tiny{$\Scal$}}}$. Since $\theta \propto 1/N$, it follows that one can always find an $N$ large enough such that the error is smaller than a given value, thereby saturating the Landauer bound.

\hrulefill

\textbf{A sequence of machine Hamiltonians to mimic the idealised sequence.} Next we construct a protocol that mimics the above sequence and obeys the global energy conservation condition imposed in the incoherent-control setting. The protocol is split into $N$ stages (like above). In each stage, the Hamiltonian of the machine is fixed. The machine here comprises to two parts: the ``cold" part and the ``hot" part. The cold part is chosen to begin in a thermal state at temperature $\beta$ of the Hamiltonian
\begin{align}
    H_{\raisebox{-1pt}{\tiny{$\Ccal$}}} &= \left( 1 + n \theta \right) H_{\raisebox{-1pt}{\tiny{$\Scal$}}}
\end{align}
At this point we note that this sequence of cold-machine states is exactly the same as in the coherent protocol, which would proceed by simply swapping the full state of target system and machine in each stage. However, that is not possible here since this is not an energy-preserving operation. To allow for energy-preserving operations, the hot part of the machine consists of $d_{\raisebox{-1pt}{\tiny{$\Scal$}}}(d_{\raisebox{-1pt}{\tiny{$\Scal$}}}-1)/2$ qubits, each corresponding to a pair of levels $(i,j)$ of the target system (henceforth we take $i < j$ to avoid double counting), whose energy gap is equal to the difference in energies of the target and cold qubit subspaces (hence rendering the desired exchange energy resonant)
\begin{align}
    H^{(ij)}_{\raisebox{-1pt}{\tiny{$\Hcal$}}} &= \left[\omega_i + (1+n\theta) \omega_j - \left( \omega_j + (1+n\theta) \omega_i \right) \right] \ket{1}\!\bra{1}^{(ij)}_{\raisebox{-1pt}{\tiny{$\Hcal$}}} = n\theta \left( \omega_j - \omega_i \right) \ket{1}\!\bra{1}^{(ij)}_{\raisebox{-1pt}{\tiny{$\Hcal$}}},
\end{align}
where we label the energy eigenvalues of $H_{\raisebox{-1pt}{\tiny{$\Scal$}}}$ by $\{\omega_i\}$. Each of these hot qubits begins at inverse temperature $\beta_{\raisebox{-1pt}{\tiny{$H$}}}$. After every unitary operation, the cold and hot parts of the machine are rethermalised to their respective initial temperatures. 

To understand the choice of machine Hamiltonians, consider the following two energy eigenstates of the machine: $\ket{i}_{\raisebox{-1pt}{\tiny{$\Ccal$}}} \otimes \ket{1}_{\raisebox{-1pt}{\tiny{$\Hcal$}}}^{(ij)}$ and $\ket{j}_{\raisebox{-1pt}{\tiny{$\Ccal$}}} \otimes \ket{0}_{\raisebox{-1pt}{\tiny{$\Hcal$}}}^{(ij)}$. The energy difference is
\begin{align}
    \Delta^{(ij)} &= \omega_j (1 + n\theta) - \omega_i (1+n\theta) - n\theta (\omega_j - \omega_i) = \omega_j - \omega_i,
\end{align}
matching the energy difference between the corresponding pair of energy eigenstates of the target system. Furthermore, calculating the ratio of populations of the two levels we find
\begin{align}
    g^{(ij)} &= \frac{ e^{-\beta \omega_j (1+n\theta)} }{ e^{-\beta \omega_i (1+n\theta)} e^{-\beta_{\raisebox{-1pt}{\tiny{$\Hcal$}}} n \theta (\omega_j - \omega_i)} } = e^{-(\omega_j - \omega_i) (\beta + n\theta (\beta - \beta_{\raisebox{-1pt}{\tiny{$\Hcal$}}}))}.
\end{align}
This corresponds to the Gibbs ratio of a qubit at the temperature $\beta + n\theta (\beta - \beta_{\raisebox{-1pt}{\tiny{$\Hcal$}}})$, which is the temperature that defines stage $n$ [see Eq.~\eqref{eq:idealtemperatureincoherentstagen}]. In summary, we construct a machine featuring $d_{\raisebox{-1pt}{\tiny{$\Scal$}}}(d_{\raisebox{-1pt}{\tiny{$\Scal$}}}-1)/2$ qubit subspaces (or virtual qubits), each of the same energy gap as one pair of energy eigenstates of the system, and all of which have a Gibbs ratio (or virtual temperature) corresponding the $n^{\text{th}}$ temperature of our desired sequence.

\hrulefill

\textbf{A single step of the protocol: The max exchange.}  Within each stage of the protocol, a single step consists of a unitary operation on $\Scal\Ccal\Hcal$, followed by the rethermalisation of the machine parts to their respective initial temperatures. We construct the unitary operation as follows: for every pair $(i,j)$ of system energy levels, one can calculate the absolute value of the difference in populations of the following two degenerate eigenstates $\ket{i}_{\raisebox{-1pt}{\tiny{$\Scal$}}}\ket{j}_{\raisebox{-1pt}{\tiny{$\Ccal$}}}\ket{0}_{\raisebox{-1pt}{\tiny{$\Hcal$}}}^{(ij)}$ and $\ket{j}_{\raisebox{-1pt}{\tiny{$\Scal$}}}\ket{i}_{\raisebox{-1pt}{\tiny{$\Ccal$}}}\ket{1}_{\raisebox{-1pt}{\tiny{$\Hcal$}}}^{(ij)}$. This value corresponds to the amount of population that would move under an exchange $\ket{i}_{\raisebox{-1pt}{\tiny{$\Scal$}}} \ket{j}_{\raisebox{-1pt}{\tiny{$\Ccal$}}} \ket{0}_{\raisebox{-1pt}{\tiny{$\Hcal$}}}^{(ij)} \leftrightarrow \ket{j}_{\raisebox{-1pt}{\tiny{$\Scal$}}} \ket{i}_{\raisebox{-1pt}{\tiny{$\Ccal$}}} \ket{1}_{\raisebox{-1pt}{\tiny{$\Hcal$}}}^{(ij)}$. We then choose the pair with the largest absolute value of this difference and perform that exchange, with an identity operation applied to all other subspaces. We call this unitary operation the \emph{max exchange}. We proceed to prove two statements about the max-exchange operation. First, that the heat extracted from the hot bath is proportional to the change in average energy of the system; and second, that system state under repetition of said operation converges to the thermal state of the temperature that defines the stage $n$.

Consider the change in average energy of the target system under the exchange unitary. The only two populations that change are those of the $\ket{i}_{\raisebox{-1pt}{\tiny{$\Scal$}}}$ and $\ket{j}_{\raisebox{-1pt}{\tiny{$\Scal$}}}$. We label the increase in the population of $\ket{i}_{\raisebox{-1pt}{\tiny{$\Scal$}}}$ as $\delta p$. Then, we have
\begin{align}
    \Delta E_{\raisebox{-1pt}{\tiny{$\Scal$}}} &= \mathrm{tr} \left[ H_{\raisebox{-1pt}{\tiny{$\Scal$}}} \left( \varrho_{\raisebox{-1pt}{\tiny{$\Scal$}}}^\prime - \varrho_{\raisebox{-1pt}{\tiny{$\Scal$}}} \right) \right] = -\delta p \left( \omega_j - \omega_i \right).
\end{align}
On the other hand, the populations of the corresponding hot qubit (i.e., tracing out the target system and cold machine) change by the same amount, i.e., there is a move of $\delta p$ from $\ket{1}_{\raisebox{-1pt}{\tiny{$\Hcal$}}}^{(ij)}$ to $\ket{0}_{\raisebox{-1pt}{\tiny{$\Hcal$}}}^{(ij)}$. In order to rethermalise the hot qubit, the heat drawn from the hot bath is thus
\begin{align}
    \widetilde{\Delta} E_{\raisebox{-1pt}{\tiny{$\Hcal$}}} &= \delta p \; n \theta (\omega_j - \omega_i) = - n \theta \Delta E_{\raisebox{-1pt}{\tiny{$\Scal$}}}.
\end{align}
This is an expression conveniently independent of the pair $(i,j)$ that applies after an arbitrary number of repetitions of the max-exchange operation (which will use different pairs in general). 

\hrulefill

\textbf{Convergence of the max-exchange protocol to the virtual temperature.} To show that the max-exchange protocol indeed converges to the desired system state in each stage of the protocol, we first prove a rather general statement: given a state $\varrho$ diagonal in the energy eigenbasis, if we exchange any qubit subspace within this system with a virtual qubit of a particular virtual temperature, then the relative entropy of the target system w.r.t. the thermal state of that (virtual) temperature decreases.

To this end, consider the relative entropy of a state $\varrho$ that is diagonal in the energy eigenbasis to a thermal state $\tau$. Labelling the populations of $\varrho$ as $p_i$ and those of $\tau$ as $q_i$, this can be expressed as
\begin{align}
    D(\varrho||\tau) &= \sum_k p_k \log \left(\frac{p_k}{q_k}\right).
\end{align}
We now focus on a single-qubit subspace labelled by $\{i,j\}$, which leads to
\begin{align}
    D(\varrho||\tau) &= p_i \log \left(\frac{p_i}{q_i}\right) + p_j \log \left(\frac{p_j}{q_j}\right) + \sum_{k \notin \{i,j\}} p_k \log \left(\frac{p_k}{q_k}\right) \notag \\
        &= (p_i + p_j) \left[ \frac{p_i}{p_i + p_j} \log \left( \frac{\frac{p_i}{p_i+p_j}}{\frac{q_i}{q_i+q_j}} \frac{p_i+p_j}{q_i+q_j} \right) + \frac{p_j}{p_i + p_j} \log \left( \frac{\frac{p_j}{p_i+p_j}}{\frac{q_j}{q_i+q_j}} \frac{p_i+p_j}{q_i+q_j} \right) \right] + \sum_{k \notin \{i,j\}} p_k \log \left(\frac{p_k}{q_k}\right) \notag \\
        &= N \left( \bar{p}_i \log \frac{\bar{p}_i}{\bar{q}_i} + \bar{p}_j \log \frac{\bar{p}_j}{\bar{q}_j} + \log \frac{N}{N_{\raisebox{-1pt}{\tiny{$V$}}}} \right) + \sum_{k \notin \{i,j\}} p_k \log \frac{p_k}{q_k}.
\end{align}
In the last line we renormalise the populations within the qubit subspace and labelled the total populations of the system and thermal state qubit subspaces of interest by $N$ and $N_{\raisebox{-1pt}{\tiny{$V$}}}$, respectively. Labelling the normalised states within these subspaces as $\varrho_{\raisebox{-1pt}{\tiny{$V$}}}$ and $\tau_{\raisebox{-1pt}{\tiny{$V$}}}$ respectively, we have
\begin{align}
    D(\varrho||\tau) &= N \left[ D(\varrho_{\raisebox{-1pt}{\tiny{$V$}}}||\tau_{\raisebox{-1pt}{\tiny{$V$}}}) + \log \left(\frac{N}{N_{\raisebox{-1pt}{\tiny{$V$}}}} \right)\right] + \sum_{k \notin \{i,j\}} p_k \log \left(\frac{p_k}{q_k}\right).
\end{align}
Suppose now that this qubit subspace of the target system is exchanged with a qubit subspace of any machine that has the same temperature as the thermal state above. The only object that changes in the the above expression is $\varrho_{\raisebox{-1pt}{\tiny{$V$}}}$, since the norm $N$ remains the same. In addition, $\varrho_{\raisebox{-1pt}{\tiny{$V$}}}$ always gets closer to $\tau_{\raisebox{-1pt}{\tiny{$V$}}}$ under such an exchange~\cite{Silva_2016,Clivaz_2019E}, implying that the relative entropy always strictly decreases under such an operation. 

Returning to the max-exchange protocol, note that by construction, every virtual qubit in the machine that is exchanged with the qubit subspace $\{i,j\}$ of the target system in a given stage $n$ has the same virtual temperature, $\beta_n = \beta + n\theta (\beta - \beta_{\raisebox{-1pt}{\tiny{$\Hcal$}}})$. Thus the relative entropy of the system to the thermal state at this temperature always decreases under this operation, unless the operation does not shift any population, which happens only at the unique fixed point where every qubit subspace of the system is already at the virtual temperature $\beta_n$. By monotone convergence, the relative entropy must converge, and moreover converge to the value that it has at the fixed point of the operation, which is the thermal state at inverse temperature $\beta_n$. Note that rather than choosing the qubit subspace with maximum population difference to exchange we could also have picked at random from among the pairs $\{i,j\}$ and convergence would still hold; the max-exchange protocol simply ensures the fastest rate of convergence among these choices.

\hrulefill

\textbf{Choosing a large enough number of repetitions in each stage so that the overall heat cost is close to the idealised heat cost.} Given that the max-exchange protocol in stage $n$ converges to the thermal state that we label $\varrho_{\raisebox{-1pt}{\tiny{$\Scal$}}}^{(n)}$, given any error $\delta_{\raisebox{-1pt}{\tiny{$E$}}}$, we choose a number of repetitions $m_n$ that is large enough so that the difference between the average energy of the actual final state of this stage, which we label $\widetilde{\varrho}_{\raisebox{-1pt}{\tiny{$\Scal$}}}^{(n)}$, and that of the ideal state $\varrho_{\raisebox{-1pt}{\tiny{$\Scal$}}}^{(n)}$ is less than $\delta_{\raisebox{-1pt}{\tiny{$E$}}}$. In this case, the total heat cost over all stages is close to the idealised heat cost
\begin{align}
    \left| \widetilde{\Delta} E_{\raisebox{-1pt}{\tiny{$\Hcal$}}} - \widetilde{\Delta} E_{\raisebox{-1pt}{\tiny{$\Hcal$}}}^* \right| &=  \left| \sum_{n=1}^N \left\{- n \theta \, \mathrm{tr} \left[ H_{\raisebox{-1pt}{\tiny{$\Scal$}}} \left( \widetilde{\varrho}_{\raisebox{-1pt}{\tiny{$\Scal$}}}^{(n)} - \widetilde{\varrho}_{\raisebox{-1pt}{\tiny{$\Scal$}}}^{(n-1)} \right) \right]\right\} - \sum_{n=1}^N \left\{- n \theta \, \mathrm{tr} \left[ H_{\raisebox{-1pt}{\tiny{$\Scal$}}} \left( \varrho_{\raisebox{-1pt}{\tiny{$\Scal$}}}^{(n)} - \varrho_{\raisebox{-1pt}{\tiny{$\Scal$}}}^{(n-1)} \right) \right] \right\} \right| \notag \\
        &= \left| \sum_{n=0}^{N-1} \theta \, \mathrm{tr} \left[ H_{\raisebox{-1pt}{\tiny{$\Scal$}}} \left( \widetilde{\varrho}_{\raisebox{-1pt}{\tiny{$\Scal$}}}^{(n)} - \varrho_{\raisebox{-1pt}{\tiny{$\Scal$}}}^{(n)} \right) \right] - N \theta \left( \widetilde{\varrho}_{\raisebox{-1pt}{\tiny{$\Scal$}}}^{(N)} - \varrho_{\raisebox{-1pt}{\tiny{$\Scal$}}}^{(N)} \right) \right| \notag  \\
        &\leq 2 N \theta \delta_E = 2 \left(\frac{\beta^* - \beta}{\beta - \beta_{\raisebox{-1pt}{\tiny{$\Hcal$}}}} \right) \delta_E.
\end{align}
The number of repetitions in each stage $m_n$ required depends only upon the initial choice of $\beta^*$ and $N$.

\hrulefill

\textbf{Completing the proof.} Finally, suppose that one is given any target temperature $\beta^*$ and two arbitrarily small errors, $\epsilon_{\raisebox{-1pt}{\tiny{$\beta$}}}$ for the cooling and $\epsilon_{\raisebox{-1pt}{\tiny{$E$}}}$ for the heat cost, and asked to cool incoherently in such a way that achieves
\begin{align}
    \left| \beta^\prime - \beta^* \right| \leq \epsilon_{\raisebox{-1pt}{\tiny{$\beta$}}}, \\
    \left| \widetilde{\Delta} E_{\raisebox{-1pt}{\tiny{$\Hcal$}}} - \eta^{-1} \Delta F_{\raisebox{-1pt}{\tiny{$\Scal$}}}^{(\beta)} \right| \leq \epsilon_{\raisebox{-1pt}{\tiny{$E$}}}.
\end{align}
We proceed by first choosing a number of stages $N$ so that the idealised heat cost $\widetilde{\Delta} E_{\raisebox{-1pt}{\tiny{$\Hcal$}}}^*$ is within $\tfrac{\epsilon_{\raisebox{-1pt}{\tiny{$E$}}}}{2}$ to the Carnot-Landauer bound above. The idealised sequence of temperatures satisfies $\beta_{\raisebox{-1pt}{\tiny{$N$}}} = \beta^*$ by construction. Once $N$ is fixed, for each stage from $n=1$ to $N-1$ we choose a number of repetitions for each stage $m_n$ such that the actual heat cost is within $\tfrac{\epsilon_E}{2}$ of the idealised heat cost, as discussed above. This ensures that the total heat cost is within $\epsilon_{\raisebox{-1pt}{\tiny{$E$}}}$ of the bound. Finally, we check that the number of repetitions of the last stage $m_{\raisebox{-1pt}{\tiny{$N$}}}$ is large enough for us to be within $\epsilon_{\raisebox{-1pt}{\tiny{$\beta$}}}$ of $\beta^*$. If not, we increase the number of repetitions (this can only decrease the error in the heat cost anyway) until we are close enough, as required.

\end{proof}

\section{Comparison of Cooling Paradigms and Resources for Imperfect Cooling}
\label{app:imperfectcooling}

Although we have looked at a number of cooling protocols throughout to demonstrate the ability for perfect cooling in the asymptotic limit, here we focus on imperfect cooling behaviour, i.e., when all resources are restricted to be finite and thus a perfectly pure state cannot be attained. We have three main goals in doing so.
\begin{enumerate}
    \item To illustrate the finite trade-offs between the trinity of resources (energy, time, control complexity).
    \item To compare the behaviour of different constructions of the cooling unitary for machines of the same size (i.e., analysing the energy-time trade-off for for fixed control complexity).
    \item To demonstrate the increase in resources required for cooling in the thermodynamically self-contained paradigm of energy-preserving unitaries (i.e., incoherent control), as compared to coherently driven unitaries.
\end{enumerate}

\subsection{Rates of Resource Divergence for Linear Qubit Machine Sequence}
\label{app:linearqubitscaling}

Consider cooling a qubit target system with energy gap $\omega_{\raisebox{0pt}{\tiny{$\Scal$}}}$ by swapping it sequentially with a sequence of $N$ machine qubits of linearly increasing energy gaps. In Appendix~\ref{app:incoherentqubit}, we derived the deviation from the idealised heat dissipation in the incoherent control setting for a sequence of $N$ machines [see Eq.~\eqref{eq:abstractcostdiff}], which we repeat below:
\begin{align}
    \frac{1}{\eta} \left( F^* - F_0 \right) \leq \widetilde{\Delta} E^*_{\raisebox{-1pt}{\tiny{$\Hcal$}}} \leq \frac{1}{\eta} \left( F^* - F_0 \right) + \frac{\omega_{\raisebox{0pt}{\tiny{$\Scal$}}}}{N} \left( \frac{\beta^* - \beta}{\beta - \beta_{\raisebox{-1pt}{\tiny{$H$}}}} \right).
\end{align}
We can immediately adapt this result to the paradigm of coherent control by taking $\beta_{\raisebox{-1pt}{\tiny{$H$}}} = 0$ and replacing the heat by work, which yields
\begin{align}
    \Delta F_{\raisebox{0pt}{\tiny{$\Scal$}}} \leq W \leq \Delta F_{\raisebox{0pt}{\tiny{$\Scal$}}} + \frac{\omega_{\raisebox{0pt}{\tiny{$\Scal$}}}}{N} \left( \frac{\beta^*}{\beta} - 1 \right).
\end{align}
Since the above inequalities are derived from the left and right Riemann sums of an integral, as $N$ becomes large, one can expect that $W$ lies roughly halfway between both extremes; we can thus cast the scaling in the approximate form
\begin{align}\label{eq:finitescaling}
    \left[ \frac{W - \Delta F}{\omega_{\raisebox{0pt}{\tiny{$\Scal$}}}} \right] N \sim \frac{1}{2} \left( \frac{\beta^*}{\beta} - 1 \right).
\end{align}
Thus, we see that the relevant quantifier of the energy resource here is the extra work cost above the Landauer limit relative to the system energy. Additionally, the quantifier of how much said resource is required (per machine qubit) is $\beta^*/\beta - 1$, which, for cold enough final temperatures, is approximately the ratio $\beta^*/\beta$.

Returning to the incoherent control paradigm, analysing the scaling behaviour between energy and time is more complicated. On the one hand, the expression above is only slightly modified, with the work being replaced by the heat dissipated multiplied by the Carnot factor:
\begin{align}\label{eq:finitescalingincoherent}
    \left[ \frac{\eta \; \Delta E_{\raisebox{0pt}{\tiny{$\Hcal$}}} - \Delta F}{\omega_{\raisebox{0pt}{\tiny{$\Scal$}}}} \right] N \sim \frac{1}{2} \left( \frac{\beta^*}{\beta} - 1 \right),
\end{align}
which is consistent with the work-to-heat efficiency of a Carnot engine. However, in the case of incoherent control, since the population swap only takes place within a subspace of the two-qubit machine, the total population is not completely exchanged in a single operation (in contrast to that in the coherent control setting). Thus the number of operations here required to transfer a desired amount of population to the ground state of the target is greater than the number of machine qubits $N$. To make a fair comparison, one could either compare the same number of machine qubits but swap repeatedly (with rethermalisation of the machine in between operations)---thereby fixing the control complexity at the expense of longer time---or one could increase the number of machine qubits and count time by the number of two-level swaps---thereby fixing time to be equal at the expense of increased control complexity overall. We investigate both methods in the coming section.

\subsection{Comparison of Coherent and Incoherent Control}
\label{app:comparisoncoherentincoherent}

Intuitively, the incoherent control paradigm requires the utilisation of a greater amount of resources (albeit less overall control in general) than the coherent control counterpart because of two distinct disadvantages. First, the temperature of the baths plays a substantial role in cooling performance. Consider the example of a \texttt{SWAP} gate applied between a system and machine qubit: in the coherent control case, this operation transforms the target system to the state of the thermal machine qubit, characterised by the Gibbs ratio of ground-state to excited-state population. In the incoherent control case, one requires the addition of a thermal qubit from the hot bath to render said operation energy preserving; as a result, the Gibbs ratio of the virtual qubit that the target system swaps with is, in general, worse than that of the coherent control setting, and only becomes equal in the limit of an infinite temperature hot bath. This is the first disadvantage. The second disadvantage is that in the incoherent control setting, the target system swaps with only a subspace of the machine rather than the entire one, i.e., it is swapped with a virtual qubit. Thus, the exchange of population is only partial as compared to the coherent control case: in the limiting case of an infinite temperature hot bath, said factor goes to $\tfrac{1}{2}$ for all relevant two-level subspaces. This implies that a greater number of operations, and thus time, is required in the incoherent control paradigm in order to achieve a similar result as its coherent control counterpart.

We illustrate this behaviour via the following example. The system is a degenerate qubit (beginning in the maximally mixed state), and we fix the final target ground-state population ($p=0.99$, corresponding to $\epsilon = 1-p=0.01$). Even in this simple case, the optimal finite-resource protocols with coherent and incoherent control are not known; we therefore compare protocols from each setting that make use of machines of a similar structure, namely swapping with machine qubits (virtual ones, in the incoherent control setting) of linearly increasing energy gaps.

More specifically, the coherent control cooling protocol employed is that of a sequence of swaps with machine qubits of linearly increasing energy gaps, and for the fixed target population, we can calculate the surplus work cost over the Landauer limit as a function of the number $N$ of operations (which corresponds in this case to the number of machine qubits). In the incoherent control case, we take the hot bath to be at infinite temperature, allowing for the potential saturation of the Landauer limit as in the coherent case. In this way we isolate the disadvantage that arises due to working in degenerate subspaces in our analysis. Here too we take a linear sequence of energy gaps for the cold (and hot) baths, with a single operation step corresponding to a three-level energy-conserving exchange involving the qubit taken from each of the hot and cold parts of the machine, i.e., $\ket{1}_{\raisebox{-1pt}{\tiny{$\Scal$}}} \ket{0}_{\raisebox{-1pt}{\tiny{$\Ccal$}}} \ket{0}_{\raisebox{-1pt}{\tiny{$\Hcal$}}} \leftrightarrow \ket{0}_{\raisebox{-1pt}{\tiny{$\Scal$}}} \ket{1}_{\raisebox{-1pt}{\tiny{$\Ccal$}}} \ket{1}_{\raisebox{-1pt}{\tiny{$\Hcal$}}}$. As mentioned previously, for an incoherent control protocol of fixed overall machine size, there are essentially two extremal methods of implementation. The first is to identify $N$ two-level subspaces of the total machine with distinct energy gaps and perform the sequence of virtual swaps between them and the target; in the language of Appendix~\ref{app:incoherentcoolingfinitetemperature}, we therefore have $N$ different stages with a single step within each stage (no repetitions) before moving on to the next stage. The second is to take $N/m$ two-level subspaces and swap the target with each virtual qubit $m$ times before moving on to the next; in other words, we here have $N/m$ different stages with $m$ steps (repetitions) within each stage. For the same fixed ground-state population, we plot the surplus work cost (energy drawn from the hot bath in the case of incoherent control) against the total machine size and number of two-level unitary swaps, as characterised by $N$, for both of these incoherent control adaptations, comparing them to the coherent control paradigm in Fig.~\ref{fig:cohvincoh}. 

In both control paradigms, we see that the deviation of the energy cost above the Landauer limit scales inversely with the number of operations [as expected from Eqs.~\eqref{eq:finitescaling} and~\eqref{eq:finitescalingincoherent}], but the proportionality constant is worse in the case of incoherent control. Moreover, the incoherent control paradigm with no repetitions within stages outperforms that with multiple repetitions, as intuitively expected since the former protocol corresponds to one for which the spacing between distinct energy gaps that are utilised is smaller, allowing us to stay closer to the reversible limit in each step. In our example, the no repetition incoherent control protocol is around 3 times worse than the coherent control protocol and the incoherent control protocol with $m=5$ repetitions is around 5.3 times worse, implying that one would require that many times the number of operations (i.e., that much more time) to achieve the same performance with incoherent control paradigm as with coherent control.

\begin{figure}[h]
    \centering
    \includegraphics[width=0.6\linewidth]{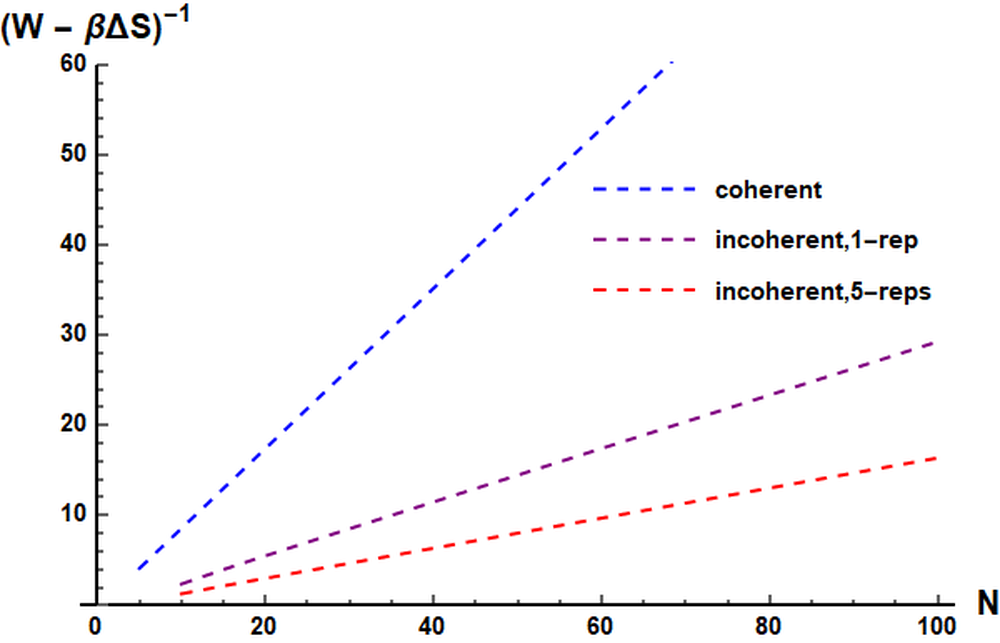}
    \caption{\emph{Imperfect Cooling with Coherent and Incoherent Control.} We compare the performance of coherent and incoherent control protocols for cooling a degenerate qubit target by swapping it with machine qubits with linearly increasing energy. The final ground-state population is fixed to be $0.99$. The inverse of the surplus work cost $W - \beta \widetilde{\Delta} S_{\raisebox{-1pt}{\tiny{$\Scal$}}}$ (with $\beta = 1$) is plotted (in units of the smallest machine energy gap, $\omega_{\raisebox{-1pt}{\tiny{$\Mcal$}}}^{\textup{min}}$) against the total number of unitary operations, with the temperature of the hot bath in the incoherent control protocols set to $\beta_H = \infty$ in order to make meaningful comparison to the coherent control case. We see that the coherent control protocol (blue) outperforms the two incoherent ones (purple, red) at any given time. As discussed in the text, there are two choices for how to implement an incoherent control protocol of this type with fixed control complexity: The red line corresponds to a protocol in which a machine (subspace) with the same energy gap is reused 5 times before moving on to the next; on the other hand, the purple line depicts the case where there are no repetitions within each stage defined by a distinct energy gap in the machine. By inspection, the single-use incoherent protocol (purple) requires approximately 3 times more unitaries to achieve the same efficiency as the coherent one (blue), whereas the five-repetition incoherent protocol (red) requires approximately 5.3 times as many unitaries as the coherent one.}
    \label{fig:cohvincoh}
\end{figure}

\end{document}